\documentclass[11pt]{article}%
\usepackage{amsfonts}
\usepackage{amsthm}
\usepackage{enumerate}
\usepackage{natbib}
\usepackage{amsmath}
\usepackage{amssymb}
\usepackage{mathtools}
\usepackage{graphicx}
\usepackage{subfig}
\usepackage{color}
\usepackage{multirow}
\usepackage{booktabs}
\usepackage{enumitem}
\usepackage[bottom]{footmisc}
\usepackage{algorithm}
\usepackage{algorithmic}
\usepackage[bookmarks=false,colorlinks,linkcolor=black,anchorcolor=black,citecolor=black]%
{hyperref}%
\providecommand{\U}[1]{\protect\rule{.1in}{.1in}}
\newtheorem{theorem}{Theorem}
\newtheorem{corollary}{Corollary}

\newtheorem{lemma}{Lemma}
\newtheorem{proposition}{Proposition}
\newtheorem{assumption}{Assumption}
\theoremstyle{definition}
\newtheorem{definition}{Definition}

\allowdisplaybreaks
\providecommand{\keywords}[1]
{
 \small	
 \textbf{\textit{Keywords---}} #1
}
\begin{document}

\title{Quantifying Distributional Input Uncertainty via Inflated Kolmogorov-Smirnov Confidence Band}
\date{}
\author{Motong Chen\thanks{Department of Industrial
Engineering and Operations Research, Columbia University.}
\and Henry Lam\footnotemark[1]
\and Zhenyuan Liu\footnotemark[1]}
\maketitle

\begin{abstract}
In stochastic simulation, input uncertainty refers to the propagation of the statistical noise in calibrating input models to impact output accuracy, in addition to the Monte Carlo simulation noise. The vast majority of the input uncertainty literature focuses on estimating target output quantities that are real-valued. However, outputs of simulation models are random and real-valued targets essentially serve only as summary statistics. To provide a more holistic assessment, we study the input uncertainty problem from a distributional view, namely we construct confidence bands for the entire output distribution function. Our approach utilizes a novel test statistic whose asymptotic consists of the supremum of the sum of a Brownian bridge and a suitable mean-zero Gaussian process, which generalizes the Kolmogorov-Smirnov statistic to account for input uncertainty. Regarding implementation, we also demonstrate how to use subsampling to efficiently estimate the covariance function of the Gaussian process, thereby leading to an implementable estimation of the quantile of the test statistic and a statistically valid confidence band. Numerical results demonstrate how our new confidence bands provide valid coverage for output distributions under input uncertainty that is not achievable by conventional approaches.

\end{abstract}


\section{Introduction}\label{sec:intro}

\noindent Stochastic simulation is widely used for operational decision-making and
scientific modeling (\cite{law2007simulation,banks2014discrete}). It consists of
the generation of random input variates, fed into the system logic, to produce
outputs that are subsequently analyzed for downstream tasks. In many
real-world settings, the distributions that govern the input variates are
unknown and need to be calibrated from external data. In this case, proper
output analysis needs to account for both the Monte Carlo noises in running
the simulation and also the input data noises that propagate to the outputs.
This issue, which is beyond the handling of only the Monte Carlo noises in
traditional output analysis, is known commonly as the input uncertainty
problem in the simulation literature. For an overview on this topic, see,
e.g.,
\cite{henderson2003input,barton2012tutorial,song2014advanced,lam2016advanced,corlu2020stochastic,bls22,lam_tutorial2023}%
, and \cite{nelson2013foundations} Chapter 7.

More concretely, the goal of handling input uncertainty can often be cast as
the construction of confidence intervals that are statistically valid under
both sources of noises from Monte Carlo and input data. A closely related task
is variance estimation, or the estimation of the different components of the
variance that attribute to each source of noise. There are several major
methods in the literature. The first is bootstrap resampling, which includes
the variance bootstrap (\cite{cheng1997sensitivity,song2015quickly}),
quantile-based bootstrap (\cite{barton1993uniform,barton2001resampling}), and
enhancements such as the metamodel-assisted bootstrap
(\cite{barton2014quantifying,xie2016multivariate}), subsampling (\cite{lam2022subsampling}), adaptive budget allocation (\cite{yi2017efficient}) and shrinkage (\cite{song2024shrinkage}). The second line of methods uses the delta method directly which involves
estimating gradient information of the model with respect to the inputs. This
includes the two-point method (\cite{cheng1998two,cheng2004calculation}),
regression approach (\cite{lin2015single,song2019input}), and infinitesimal jackknife (\cite{lam2019random}). Different from the above approaches that are based on consistent variance estimation, \cite{glynn2018constructing} proposes $t$-based interval construction by sectioning input data, and \cite{lam2022cheap,lam2023bootstrap} further propose to run a so-called cheap bootstrap on the input data. Other recent approaches
include the use of robust optimization (\cite{ghosh2019robust,lam2016robust,lam2018sensitivity}).
There is also a stream of work on Bayesian approaches that rely on
posterior sampling
(\cite{chick2001input,zouaoui2003accounting,zouaoui2004accounting,xie2014bayesian,zhu2020risk,xie2021nonparametric}). Finally, several works have studied input uncertainty in simulation optimization, including ranking and selection problems (\cite{fan2020distributionally,song2019input,wu2022data}) and Bayesian optimization (\cite{pearce2017bayesian,ungredda2022bayesian}), as well as the tradeoff between input data collection and simulation effort (\cite{xu2020joint}).

Despite the growing literature in input uncertainty, to our best knowledge all the previous works have been focusing on the estimation of target quantities that are
real-valued. However, in stochastic simulation the outputs are random (hence the Monte Carlo noise) and the real-valued target quantities often only serve as a summary statistic. For instance, suppose the simulation model represents a queueing system and we are interested in the waiting time or queue length at a certain time horizon. Then, in the input uncertainty literature, existing approaches would consider a summary target such as the expected waiting time. Such type of summary statistic is arguably convenient but, depending on our downstream decision-making tasks, it could hide valuable information, which in this example could be the collection of quantiles or tail distribution function values of the waiting time. Moreover, although some existing work quantifies uncertainty for performance measures that are possibly nonlinear in the input distributions (i.e., beyond expected values; e.g., \cite{lam2022subsampling}), their theoretical guarantees rely on model assumptions that can only be verified when the performance measure is linear. Also, even if verification beyond linear is doable, this would need to be conducted case by case. In contrast, having an approach to understand the interplay of Monte Carlo and input uncertainty at the \emph{distributional} level would help us bypass these issues. That is, doing so would not only capture the \emph{epistemic} uncertainty, i.e., the uncertainty coming from the data noise, but also more explicitly shape the \emph{aleatory} uncertainty, i.e., the intrinsic uncertainty from the randomness in the stochastic model. In particular, when the performance measures are monotone (e.g., quantiles) or more generally piece-wise monotone with respect to the distribution function, quantifying this distributional uncertainty will directly imply the uncertainty quantification of these performance measures, even simultaneously.  

Motivated by the above, our main contribution of this paper is to propose an approach to construct a confidence band for the entire output distribution that accounts for both Monte Carlo and input data noises. Our approach relies on a novel statistic whose asymptotic limit comprises the supremum of the sum of a Brownian bridge and a mean-zero Gaussian process with a covariance structure governed by the so-called influence function of the target output distribution. In this representation, the Brownian bridge signifies the Monte Carlo noises and resembles the limit of the celebrated Kolmogorov-Smirnov (KS) statistic. The mean-zero Gaussian process, on the other hand, signifies the input data noises, where the influence function that controls the covariance structure can be regarded as the derivative of the target output distribution with respect to the input parameters and carry sensitivity information like in the standard delta method. Our approach thus can be viewed as a generalization of the KS statistic, from the traditional use to construct confidence bands under a single source of randomness, to the simulation input uncertainty setting with both epistemic and aleatory uncertainties. In fact, we will illustrate that when there is input uncertainty, using the standard KS machinery to construct confidence bands for the output distribution would lead to a systematic under-coverage. The additional Gaussian process serves to inflate the variability to lift up the coverage by the right amount. 

We call our above constructed confidence band the \emph{input-uncertainty-inflated KS confidence band}. Technically, analyzing the corresponding generalized KS statistic to derive this confidence band relies on two significant novelties compared to the development of the classic KS statistic. First, the input data noise and its entanglement with the Monte Carlo noise necessitate new statistical tools to establish the weak convergence of our test statistic. In particular, the input data noise possesses the form of the so-called V-process that describes the structure of common discrete-event simulation models (see, e.g.,  \cite{silverman1976limit,silverman1983convergence,akritas1986empirical} for details on V-processes and the closely related U-processes). We derive the weak convergence of the input data noise by bounding its modulus of continuity via a generalization of the techniques in \cite{silverman1976limit} from single-sample U-processes to multi-sample V-processes. On the other hand, the Monte Carlo noise is now also dependent on the input data noise, which precludes a direct argument on its convergence to the classic KS asymptotic limit statistic. To this end, we utilize the Koml\'{o}s-Major-Tusn\'{a}dy approximation (\cite{komlos1975approximation}) to strengthen the classic KS convergence to have a uniform error control over all possible underlying distribution, which in turn implies the joint weak convergence of both Monte Carlo and input noises in the so-called Skorohod product space. Our second novelty is to circumvent the lack of continuity for the supremum operator in the latter space, which is otherwise used via the continuous mapping theorem to derive the classic KS statistic. In particular, we create and show that a coupling process, which composites the original output with its true distribution function, retains approximately the same supremum. The new process now lies in a Skorohod space guaranteed to have a compact domain where the supremum operator is indeed continuous, which in turn induces our desired test statistic.

With our framework and statistic, to enhance implementability, we further propose a subsampling approach to efficiently
estimate the covariance function of the Gaussian process that signifies the input uncertainty. Note that the
influence function described above, while conceptually understandable, is
difficult to use in practice because it requires computing essentially
infinite-dimensional derivatives. Thus we will use the bootstrap in lieu of direct estimation of the influence function. However, the standard variance
bootstrap is inefficient because of the interaction of the Monte Carlo and
input data noises, which necessitates washing away the Monte Carlo noises via
a large amount of simulation runs. This phenomenon has been pointed out in
\cite{lam2022subsampling} even when only real-valued target quantities are
considered. To address this computational challenge, we generalize the idea in
\cite{lam2022subsampling} to devise a subsampling method where we draw a
resample of smaller size than the original input data size, at the outer layer
of a nested simulation scheme used to bootstrap estimate the covariance function. However, while \cite{lam2022subsampling} uses a debiased ANOVA-based estimator, we propose an estimator for covariance estimation that does not include the debiasing term. This is because we need to retain the positive semi-definite requirement of the covariance matrix with a finite simulation budget, which is crucial to generate suitable multivariate Gaussian vectors to implement the confidence band. This latter advantage outweighs the relatively larger estimation error in our biased estimator. We provide full implementation details and both theoretical and empirical justifications of our resulting confidence bands.



The rest of this paper is as follows. Section \ref{sec:formulation} first introduces the conventional and the new distributional input uncertainty problems. Section \ref{sec:methodology} describes our methodology for quantifying distributional input uncertainty and its main statistical guarantees. Section \ref{sec:guarantee} discusses implementation, including proposing a subsampling approach to estimate the covariance function required in constructing our proposed interval. Section \ref{sec:theory} overviews the main technical developments to obtain our guarantees. Section \ref{sec:num} presents numerical results. Additional theoretical results and all proofs are left to Appendix.
	
\mbox{}\\
\noindent {\bf Notations. }$\lfloor x\rfloor$ denotes the largest integer that is less than or equal to $x$. $\mathrm{frac}(x)=x-\lfloor x\rfloor$ denotes the fractional part of $x$. $\overset{d}{\approx}$ means ``having approximately the same distribution as'' (not used in mathematically rigorous contexts). $\delta_{x}$ denotes the Dirac delta measure at $x$. Given two non-negative sequences $\{a_n\}$ and $\{b_n\}$, we write $a_n=o(b_n)$ if $\lim_{n\rightarrow\infty}a_n/b_n=0$, write $a_n=O(b_n)$ if $\limsup_{n\rightarrow\infty}a_n/b_n<\infty$, write $a_n=\omega(b_n)$ if $\lim_{n\rightarrow\infty}a_n/b_n=\infty$ and write $a_n=\Theta(b_n)$ if $0<\liminf_{n\rightarrow\infty}a_n/b_n\le\limsup_{n\rightarrow\infty}a_n/b_n<\infty$. The same notations apply to non-negative functions. Given a sequence of random variables $\{X_n\}$ and a non-negative sequence $\{a_n\}$, we write $X_n=o_p(a_n)$ if $X_n/a_n\overset{p}{\rightarrow}0$, write $X_n=O_p(a_n)$ if $X_n/a_n$ is stochastically bounded, i.e., for any $\varepsilon>0$, there exists a finite $M > 0$ and a finite $N > 0$ such that $\mathbb{P}(|X_n/a_n|>M)<\varepsilon,\forall n>N$, and write $X_n=\Theta_p(a_n)$ if $X_n=O_p(a_n)$ and there exists a finite $M>0$, $\varepsilon\in (0,1)$ and a finite $N>0$ such that $\mathbb{P}(|X_n/a_n|>M)>\varepsilon,\forall n>N$. In particular, if $X_n/a_n\overset{p}{\rightarrow}X$ for some random variable $X$, we have $X_n=O_p(a_n)$ and further $X_n=\Theta_p(a_n)$ if $X$ is not almost surely zero. For any event $A$ (in an underlying $\sigma$-field $\mathcal{F}$), the indicator function $I(A)$ is 1 if $A$ occurs and 0 otherwise. Given a function $f$, $f(t-)$ denotes the left limit of $f$ at $t$ (if it exists, e.g., $f$ is a cumulative distribution function). 
	
\section{From Conventional to Distributional Input Uncertainty}\label{sec:formulation}
We review the traditional input uncertainty problem in stochastic simulation (Section \ref{sec:conventional}), and from there we motivate and describe our target distributional input uncertainty problem (Section \ref{sec:distributional}).

\subsection{The Conventional Problem}\label{sec:conventional}
Consider a stochastic simulation model with a target performance measure to evaluate $\psi\in\mathbb R$. Suppose we can generate independent unbiased Monte Carlo samples, say $\hat\psi_r,r=1,\ldots,R$. By using these Monte Carlo samples, a standard output analysis would aim to construct a confidence interval (CI) for $\psi$ at some prescribed level, say $1-\alpha$, i.e., an interval that covers the ground-truth $\psi$ with probability close to $1-\alpha$. 


Typically, a simulation model consists of input distributions that govern the randomness of the underlying random variates. Consider $m$ input
distributions $\underline{P}:=(P_{1},\ldots,P_{m})$, each governing a different source of randomness. Correspondingly, we write $\psi=\psi(\underline{P})$ to highlight the dependence of the performance measure on the input distributions, and we write $\hat\psi_r(\underline P),r=1,\ldots,R$ as the Monte Carlo samples. For example,
we might be interested in some performance measure of a simulatable queueing system. In this case, there could be two input distributions $P_{1}$ and $P_{2}$ that represent the inter-arrival and service time
distributions respectively. The performance measure $\psi$ can be the expected average waiting time
of the first 10 customers in the queue, and we can generate $R$ unbiased trajectories to obtain $\hat\psi_r(P_1,P_2)$ given the input distributions $P_1$ and $P_2$.


Suppose we face input uncertainty for the stochastic simulation model. That is, the true input distributions are unknown and, instead, we have finite data for each distribution, say $\{X_{i,j} : i = 1, \ldots, m, j = 1,
\ldots, n_{i}\}$ where $X_{i,j},j = 1, \ldots, n_{i}$ are i.i.d. random
variables drawn from $P_{i}$ and they are independent for different $i$'s. To obtain an approximation of $\psi(\underline{P})$, we need to drive the
simulation model using an estimate of $\underline{P}$. A natural choice is the
empirical distribution $\underline{\hat P}:=(\hat P_{1},\ldots,\hat P_{m})$
constructed from the data
\[
\hat P_{i}(\cdot)=\frac{1}{n_{i}}\sum_{j=1}^{n_{i}} \delta_{X_{i,j}}(\cdot).
\]
With these, we drive $R$ independent simulation runs to obtain $\hat{\psi}_{r}(\underline{\hat P}),r=1,\ldots,R$. Here, $\hat{\psi}_{r}(\underline{\hat P})$ is unbiased for the
performance measure $\psi(\underline{\hat P})$ evaluated under the empirical input distributions. A natural point estimate for $\psi=\psi(\underline P)$ is the sample mean $\bar{\psi}(\underline{\hat P}):=\sum_{r=1}^{R}
\hat{\psi}_{r}(\underline{\hat P})/R$.


The output estimate $\bar{\psi}(\underline{\hat P})$ contains statistical
noises coming from both the external data used to calibrate the input
distributions, and the Monte Carlo noise in running the simulation model. To
conduct statistically valid inference for $\psi$, we need to quantify both
sources of noises. A way to quantify them is via the following
variance decomposition:
\begin{equation}
\mathrm{Var}(\bar{\psi}(\underline{\hat P})) = \sigma^{2}_{S} + \sigma^{2}%
_{I}, \label{conventional}%
\end{equation}
where
\[
\sigma^{2}_{S} = \frac{\mathbb{E}[\mathrm{Var}(\hat{\psi}_{r}(\underline{\hat
P})|\underline{\hat P})]}{R}%
\]
is the variance component from the Monte Carlo error in simulating the model and
\[
\sigma^{2}_{I} = \mathrm{Var}(\psi(\underline{\hat P}))
\]
is the component from the input data noise, which we also call the input variance for convenience. Furthermore, under proper conditions, we have a central limit theorem (CLT) for $\bar{\psi}(\underline{\hat P})$ (\cite{cheng1997sensitivity,glynn2018constructing,lam_tutorial2023}) so that a $(1-\alpha)$-level CI for $\psi$ can be
constructed as $\bar{\psi}(\underline{\hat P})\pm z_{1-\alpha/2}\sqrt
{\sigma^{2}_{S} + \sigma^{2}_{I}}$ where $z_{1-\alpha/2}$ is the
$(1-\alpha/2)$-th quantile of the standard normal. In this approach, the Monte Carlo variance component
$\sigma^{2}_{S}$ can be readily estimated by the sample variance of $\hat{\psi}%
_{r}(\underline{\hat P})$ (divided by $R$). The input variance $\sigma^{2}%
_{I}$ is far more difficult to estimate and thus is the focus of the input uncertainty literature, which entails the several major approaches described in the introduction (see, e.g., \cite{cheng1997sensitivity,cheng1998two,song2015quickly,barton2014quantifying,lam2022subsampling}).


In this paper, we consider a much more general type of target measure.
Note that in the above discussion, the random sample $\hat\psi_r$ corresponds to a realized output trajectory of the system of interest, which is a more fundamental quantity than $\psi$ formed as a summary statistic of the output. From this view, let us denote random variable
$Y$ as the random output of the model. Then a real-valued performance measure $\psi$ discussed above would be a summary statistic of $Y$ such as $\mathbb{E}[Y]$. While convenient, depending on the downstream tasks, such a summary statistic
may hide valuable information of the output distribution, for instance the quantile or tail information other than the expectation of the average waiting time in our described example. If we are interested in this other distributional information, then we need to revise and re-run our inference procedure -- if at all feasible, both in terms of computational expense and methodological readiness (as mentioned in the introduction, existing literature that quantifies input uncertainty for performance measures that are nonlinear in the input distributions is very limited and, if any, would require assumptions that can only be verified for very specific cases). These motivate our investigation to quantify input uncertainty on the \emph{entire output distribution} instead of only a summary statistic. This more holistic measurement of input uncertainty in turn allows us to, for instance, generate simultaneous confidence regions for any collections of performance measures that can depend even nonlinearly on the input distributions.

\subsection{Distributional Input Uncertainty}\label{sec:distributional}
Our goal is to construct a confidence band for the
output distribution function. For convenience, let $Q(t,\underline{P}%
)=\mathbb{P}_{\underline{P}}(Y\leq t)$ denote the true output distribution
function where the subscript $\underline{P}$ makes clear that $Y$ is the
output generated using the true unknown input distributions $\underline{P}$.
We would like to find two functions $L_{IU}(t)$ and $U_{IU}(t)$ such that
\begin{equation}
\mathbb{P}(L_{IU}(t)\leq Q(t,\underline{P})\leq U_{IU}(t)\ \forall t\in\mathbb{R}%
)\to1-\alpha, \label{target}%
\end{equation}
where $1-\alpha$ is a pre-specified confidence level, e.g., $95\%$, and the
limit is taken as $n_{i}$ and $R$ get large. We call a confidence band $\{[L_{IU}(t),U_{IU}(t)]:t\in\mathbb R\}$ \emph{asymptotically exact} if it satisfies \eqref{target}.

To get a sense of how one can construct these $L_{IU}(t)$ and $U_{IU}(t)$, note that
when we do not have noises in calibrating the input distributions (i.e.,
suppose we know the input distributions completely), the only uncertainty is
from the Monte Carlo runs, and in this case we can use, for instance, the KS
statistic. Suppose the true distribution function $Q(t,\underline{P})$ is
continuous in $t$. Let $\hat Q(t,\underline{P})$ be the output empirical
distribution constructed from $R$ simulation runs driven by the input
distributions $\underline{P}$, i.e.,
\[
\hat Q(t,\underline{P})=\frac{1}{R}\sum_{r=1}^{R} I(Y_{r}\leq t),
\]
where $Y_{r},r=1,\ldots,R$ denote i.i.d. Monte Carlo outputs from the
simulation model. Then the error between $Q(t,\underline{P})$ and $\hat
Q(t,\underline{P})$ can be quantified by the following weak convergence
result
\begin{equation}
\sup_{t\in\mathbb{R}}|\sqrt{R}(Q(t,\underline{P})-\hat{Q}(t,\underline{P}%
))|\overset{d}{\rightarrow}\sup_{t\in\lbrack0,1]}|BB(t)|, \label{KS-statistic}%
\end{equation}
where $BB(\cdot)$ is a standard Brownian bridge. The convergence \eqref{KS-statistic} implies that the $(1-\alpha)$-level confidence band \eqref{target} can be accomplished asymptotically by adding/subtracting the empirical output distribution by a standard error that captures its variability, which gives us
\begin{equation}
L_{KS}(t)=\hat Q(t,\underline{P})-\frac{q_{1-\alpha}}{\sqrt R},\quad U_{KS}(t)=\hat
Q(t,\underline{P})+\frac{q_{1-\alpha}}{\sqrt R},\label{basic LU}
\end{equation}
where $q_{1-\alpha}$ is the $(1-\alpha)$-th quantile of $\sup_{t\in[0,1]}|BB(t)|$. More precisely, with the band \eqref{basic LU}, we have
\begin{align*}
\mathbb{P}(L_{KS}(t)\leq Q(t,\underline{P})\leq U_{KS}(t)\ \forall t\in\mathbb{R}
)&=\mathbb P(\sqrt R|\hat Q(t,\underline{P})-Q(t,\underline{P})|\leq q_{1-\alpha}\ \forall t\in\mathbb{R})\\
&=\mathbb P\left(\sup_{t\in\mathbb R}|\sqrt R(\hat Q(t,\underline{P})-Q(t,\underline{P}))|\leq q_{1-\alpha}\right)\\
&\to\mathbb P\left(\sup_{t\in[0,1]}|BB(t)|\leq q_{1-\alpha}\right)=1-\alpha,
\end{align*}
where the convergence follows from \eqref{KS-statistic} and the fact that $\sup_{t\in[0,1]}|BB(t)|$ has a continuous distribution function, thus achieving asymptotic exactness in covering the entire distribution $Q(t,\underline P)$. 
Moreover, the above $L_{KS}$ and $U_{KS}$ are piece-wise constant and can be computed
readily. Figure \ref{ks_graph}(a) illustrates how a typical instance of this confidence
band covers $Q(t,\underline P)$. That is, when there is no input uncertainty, the problem straightforwardly
reduces to the classical goodness-of-fit problem.

\begin{figure}[pth]
\centering
\par
\subfloat[CB without input uncertainty]{\includegraphics[width=0.45\textwidth]{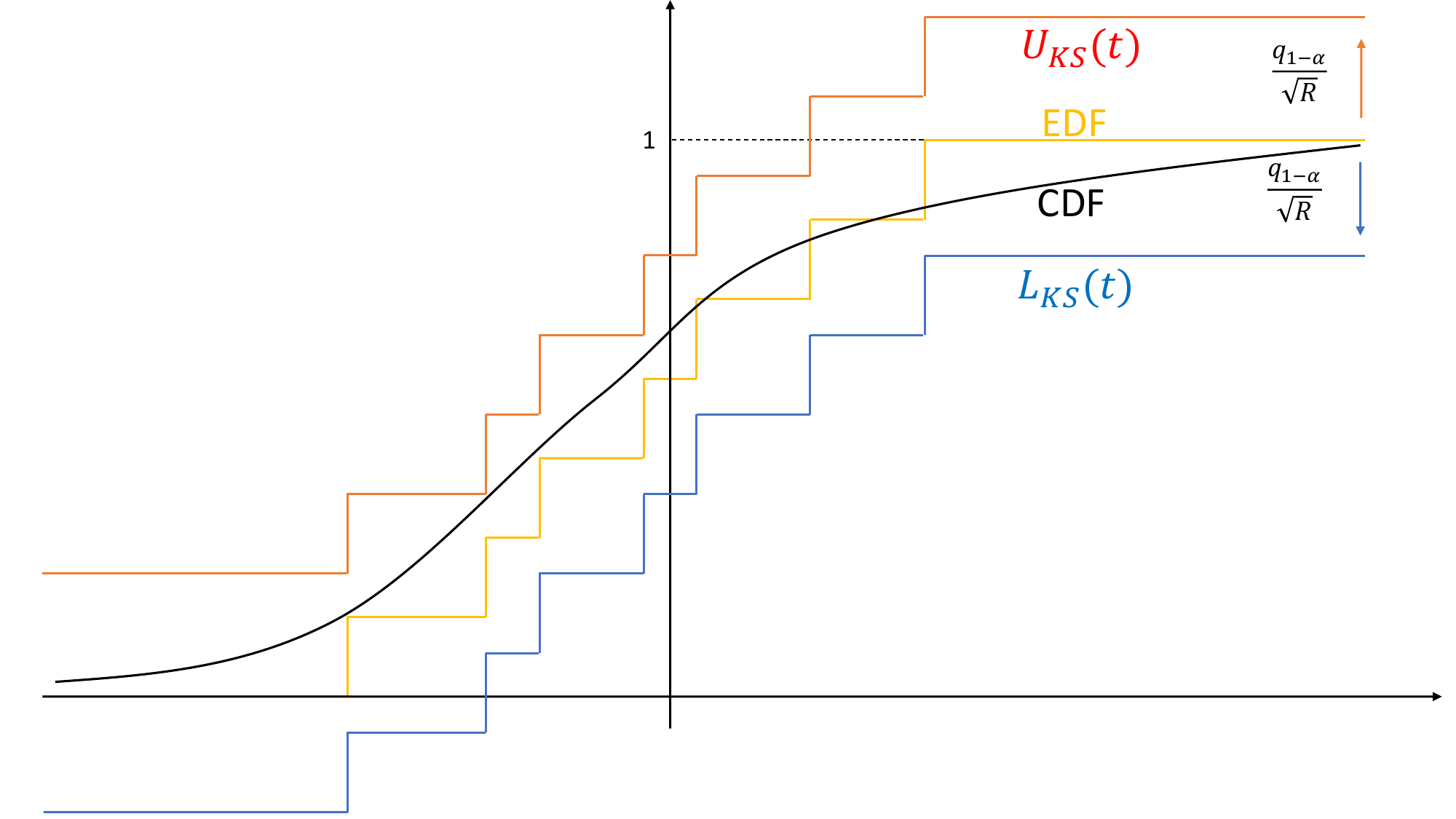}}
\subfloat[CB with input uncertainty]{\includegraphics[width=0.45\textwidth]{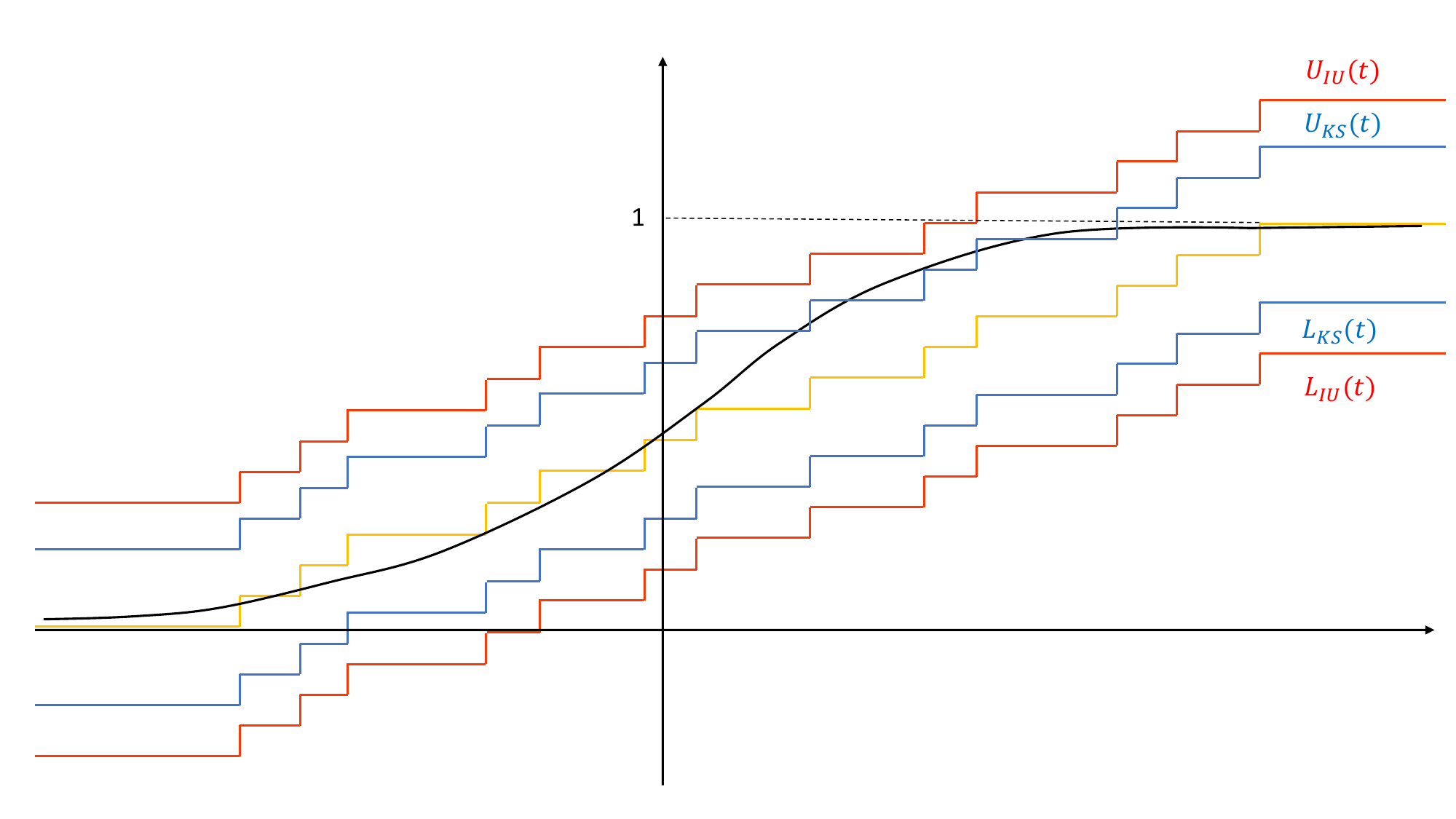}}\caption{A typical instance of various confidence
bands (CBs) with and without input uncertainty. The curve in both figures denotes the
true output distribution function. The step functions in (a) are $U_{KS}(t)$,
the empirical output distribution function, and $L_{KS}(t)$ respectively reading from top to bottom. The step functions in (b) are $U_{IU}(t)$, $U_{KS}(t)$, the empirical output distribution function, $L_{KS}(t)$ and $L_{IU}(t)$ respectively reading from top to bottom.}%
\label{ks_graph}%
\end{figure}

However, when the input distributions need to be calibrated from external
data, we can only get access to the empirical distribution $\hat
{Q}(t,\underline{\hat{P}})$ in the form of
\[
\hat{Q}(t,\underline{\hat{P}})=\frac{1}{R}\sum_{r=1}^{R} I(Y_{r}\leq t),
\]
where $Y_{r},r=1,\ldots,R$ denote i.i.d. Monte Carlo outputs (conditional on
the input data) from the simulation model driven by the empirical input
distributions $\underline{\hat P}$. If we still construct the confidence band $[L_{KS}(t),U_{KS}(t)]$ by the quantile of the standard KS statistic, we would
\emph{under-cover} the true output distribution $Q(t,\underline{P})$, because
it has ignored the additional source of input noise and under-estimate the variability of the empirical output distribution. In Figure
\ref{ks_graph}(b), we see that a typical instance of the standard KS confidence band is now too narrow to cover the entire true output
distribution function. To remedy this, we need to properly inflate the
variability of the limiting statistic to incorporate the input uncertainty and obtain a correspondingly wider confidence band $L_{IU}$ and $U_{IU}$. We call this band the \emph{input-uncertainty-inflated KS confidence band}. Figure
\ref{ks_graph}(b) illustrates such an inflated interval, where we can see that it now covers the true output distribution. The construction of this new band is the main focus of this paper and what we will discuss next.


\section{Methodology to Quantify Distributional Input Uncertainty}

\label{sec:methodology} This section presents our methodology to produce the input-uncertainty-inflated KS confidence band for the entire output distribution. We first present our main idea in Section \ref{sec:main idea}, followed by more detailed assumptions and theoretical results in Section \ref{sec:main theory}.


\subsection{Main Idea on Interval Construction}\label{sec:main idea}
As discussed in Section \ref{sec:formulation}, our key idea is to properly inflate the variability of the limiting
distribution in (\ref{KS-statistic}) to incorporate the
input uncertainty. More concretely, let $Q(t,\underline{\hat{P}})$ be the true
output distribution of the simulation model driven by the empirical input
distribution $\underline{\hat{P}}$. Motivated from the conventional two-term
variance decomposition in (\ref{conventional}), when the input data sizes $n_i$ and simulation run size $R$ are large, the error
between the empirical output distribution and true output distribution can be
decomposed as
\begin{equation}
\hat{Q}(t,\underline{\hat{P}})-Q(t,\underline{P})=[Q(t,\underline{\hat{P}%
})-Q(t,\underline{P})]+[\hat{Q}(t,\underline{\hat{P}})-Q(t,\underline{\hat{P}%
})] \label{process_decomp}%
\end{equation}
which separates the input data noise $Q(t,\underline{\hat{P}}%
)-Q(t,\underline{P})$ and Monte Carlo noise $\hat{Q}(t,\underline{\hat{P}%
})-Q(t,\underline{\hat{P}})$. Next, each noise can be approximated by a
stochastic process
\begin{equation}
Q(t,\underline{\hat{P}})-Q(t,\underline{P}) \overset{d}{\approx}%
\frac{\mathbb{G}(t)}{\sqrt{n}}, \quad\hat{Q}(t,\underline{\hat{P}%
})-Q(t,\underline{\hat{P}})\overset{d}{\approx}\frac{BB({Q}(t,\underline{P}%
))}{\sqrt{R}} \label{process_appro}%
\end{equation}
where $\mathbb{G}(t)$ is a mean-zero Gaussian process, $BB({Q}%
(t,\underline{P}))$ is a Brownian bridge with a changed time scale
${Q}(t,\underline{P})$, and $n$ is a scaling parameter on the data sizes, which we set as the average data size
$n:=\sum_{i=1}^{m} n_{i}/m$. The second process is the same as in the KS statistic and is independent of
$\mathbb{G}(t)$. The first process, on the other hand, captures the variability from input uncertainty by bearing a covariance structure that relates to the so-called influence function of $Q(t,\underline P)$, which can be viewed as a notion of derivative with respect to $\underline P$ and which we will detail further in the sequel. Consequently, (\ref{process_decomp}) and (\ref{process_appro}%
) yield
\begin{equation}
\sup_{t\in\mathbb{R}}|\hat{Q}(t,\underline{\hat{P}})-Q(t,\underline{P}%
)|\overset{d}{\approx}\sup_{t\in\mathbb{R}}\left\vert \frac{\mathbb{G}%
(t)}{\sqrt{n}}+\frac{BB({Q}(t,\underline{P}))}{\sqrt{R}}\right\vert ,
\label{approximation}%
\end{equation}

In view of the approximation \eqref{approximation}, in order to construct a
($1-\alpha$)-level confidence band for $Q(t,\underline{P})$, the key is to
estimate the ($1-\alpha$)-th quantile of the random variable
\begin{equation}
\sup_{t\in\mathbb{R}}\left\vert \frac{\mathbb{G}(t)}{\sqrt{n}}+\frac
{BB({Q}(t,\underline{P}))}{\sqrt{R}}\right\vert\label{key stat}
\end{equation}
say $\tilde q_{1-\alpha}$. Then a statistically valid pair of confidence bands
is given by
\begin{equation}
L_{IU}(t)=\hat Q(t,\underline{\hat P})-\tilde q_{1-\alpha}, \quad U_{IU}(t)=\hat
Q(t,\underline{\hat P})+\tilde q_{1-\alpha}, \label{confidence_band}%
\end{equation}
like in \eqref{basic LU}, which are readily computable since they are piece-wise constant. 

\subsection{Main Theoretical Guarantees}\label{sec:main theory}
We justify our construction in Section \ref{sec:main idea} more concretely. We first state our main assumptions. Recall that
$n=\sum_{i=1}^{m}n_{i}/m$ is the average data size and $R$ is the number of
simulation runs to construct the empirical output distribution function.

\begin{assumption}[Balanced data sizes]
\label{balanced_data} $\lim_{n\rightarrow\infty}n_{i}/n=\beta_{i}>0,\forall
i=1,\ldots,m.$
\end{assumption}

\begin{assumption}[Balanced data and Monte Carlo noises]
\label{balanced_randomness}$\lim_{n,R\rightarrow\infty}R/n=\gamma^{2}%
\in(0,\infty)$.
\end{assumption}

Assumptions \ref{balanced_data} and \ref{balanced_randomness} ensure that all the noises from input data and simulation are significant. More precisely, Assumption \ref{balanced_data} assumes the
data sizes from all input distributions follow the same order, and Assumption
\ref{balanced_randomness} assumes the same between the Monte Carlo replication size and the input data sizes. If any of these input data or simulation runs are much more than others, or in other words the input noises or the Monte Carlo noise become relatively negligible, then the problem reduces to either the classical goodness-of-fit problem where standard KS statistic suffices ($\gamma=0$, where Monte Carlo noises dominate) or the functional delta method ($\gamma=\infty$, where input noises dominate). 

Next, we present a generic structural form of the simulation model.

\begin{assumption}[Finite-horizon model]
\label{finite_horizon_model}Our stochastic simulation model is a
finite-horizon model, i.e., given $m$ arbitrary input distributions, the
output quantity $Y$ is a function of i.i.d. random variables drawn from the input
distributions, in the form of
\[
Y=h(\mathbf{X}_{1},\ldots,\mathbf{X}_{m}),
\]
where $\mathbf{X}_{1},\ldots,\mathbf{X}_{m}$ are independent, $\mathbf{X}%
_{i}=(X_{i}(1),\ldots,X_{i}(T_{i}))$ is composed of $T_{i}$ i.i.d. random
variables from the $i$-th input distribution for a constant $T_{i}%
\in\mathbb{N}$, and $h$ is a measurable function. Furthermore, the output
distribution function under the true input distributions $\underline{P}$,
i.e.,
\[
Q(t,\underline{P})=\mathbb{P}(h(\mathbf{X}_{1},\ldots,\mathbf{X}_{m})\leq t)
\]
is a continuous function in $t$.
\end{assumption}


The finite-horizon model described in Assumption \ref{finite_horizon_model} occurs commonly in discrete-event simulation. For instance, in the queueing context, $\mathbf X_1$ and $\mathbf X_2$ can refer to sequences of inter-arrival and service times of customers, and $h$ denotes layers of max-plus operators. With minimal restriction on the form of $h$, Assumption \ref{finite_horizon_model} is general enough to cover many interesting discrete-event examples, while allows us to establish weak convergence results for KS-like statistics. In particular, under the finite-horizon model, the exact output distribution function under $\underline{\hat{P}}$ is given by the so-called V-statistic, which can be viewed as a generalization of the sample mean from one observation in each summand to several observations in each summand:
\begin{equation}
Q(t,\underline{\hat{P}})=\frac{1}{\prod_{i=1}^{m}n_{i}^{T_{i}}}\sum_{1\leq
j_{i1},\ldots,j_{iT_{i}}\leq n_{i},1\leq i\leq m}I(h(X_{1,j_{11}}%
,\ldots,X_{1,j_{1T_{1}}},\ldots,X_{m,j_{m1}},\ldots,X_{m,j_{mT_{m}}})\leq t).\label{V_statistic}%
\end{equation}

In order to argue the statistical validity of our proposed band $\{[L_{IU}(t),U_{IU}(t)]:t\in\mathbb R\}$ in \eqref{confidence_band}, we need a weak convergence result similar to \eqref{KS-statistic}, but under both input and Monte Carlo uncertainties. In order to attain this result, we first introduce the so-called influence function. This object can be viewed as the functional gradient or the Gateaux derivative of $Q(t,\underline{P})$ with respect to the $i$-th input distribution $P_{i}$ (see \cite{wasserman2006all} Section 2.3 and \cite{serfling2009approximation} Chapter 6). For arbitrary input distributions $\underline{P}=(P_{1},\ldots,P_{m})$ and any $x\in\mathbb{R}$, the
influence function $IF_{i}(t,x;\underline{P})$ of $Q(t,\underline{P})$ for the
$i$-th input distribution is defined as
\[
IF_{i}(t,x;\underline{P})=\lim_{\varepsilon\downarrow0}\frac{Q(t,P_{1},\ldots,P_{i-1},(1-\varepsilon)P_{i}+\varepsilon\delta_{x},P_{i+1},\ldots,P_{m})-Q(t,\underline{P})}{\varepsilon}.
\]
Influence function has zero mean under $P_{i}$, i.e., $\mathbb{E}_{P_{i}}[IF_{i}(t,X_{i};\underline{P})]=0$. Besides, with its gradient interpretation, it gives rise to the first order Taylor expansion of the functional $\underline{P} \mapsto Q(t,\underline{P})$ (e.g., see Assumptions \ref{1st_expansion_truth} and \ref{1st_expansion_empirical} in Appendix \ref{sec:general_assumptions}). For the finite-horizon model we presented in Assumption \ref{finite_horizon_model}, its influence function has the following representation
\begin{equation}
IF_{i}(t,x;\underline{P})=\sum_{j=1}^{T_{i}}\mathbb{P}(h(\mathbf{X}_{1},\ldots,\mathbf{X}_{m})\leq t|X_{i}(j)=x)-T_i\mathbb{P}(h(\mathbf{X}_{1},\ldots,\mathbf{X}_{m})\leq t),\label{IF finite horizon}
\end{equation}
where the probability $\mathbb{P}$ is a product measure of the $\underline{P}$ (i.e., $\mathbf X_1,\ldots,\mathbf X_m$ are all independent).

With all the above, we have the following weak convergence result that generalizes the
KS statistic to incorporate input uncertainty:

\begin{theorem}
\label{weak_convergence_sup}Suppose Assumptions \ref{balanced_data},
\ref{balanced_randomness} and \ref{finite_horizon_model} hold. We have%
\[
\sup_{t\in\mathbb{R}}|\sqrt{R}(\hat{Q}(t,\underline{\hat{P}})-Q(t,\underline{P}))|\overset{d}{\rightarrow}\sup_{t\in\mathbb{R}}|\gamma\mathbb{G}(t)+BB({Q}(t,\underline{P}))|
\]
as $n,R\rightarrow\infty$, where $BB(\cdot)$ and $\mathbb{G}(\cdot)$ are two
independent stochastic processes, $BB(\cdot)$ is the standard Brownian bridge
on $[0,1]$, $\mathbb{G}(\cdot)$ is a mean-zero Gaussian process defined on
$\mathbb{R}$ having bounded continuous paths almost surely and the following
covariance function
\begin{equation}
\mathrm{Cov}(\mathbb{G}(t),\mathbb{G}(s))=\sum_{i=1}^{m}\frac{1}{\beta_{i}}\mathrm{Cov}_{P_{i}}(IF_{i}(t,X_{i};\underline{P}),IF_{i}(s,X_{i};\underline{P})),\label{covariance_function}
\end{equation}
where $P_i$ is the $i$-th true input distribution. Moreover, the limiting distribution is absolutely continuous with respect to Lebesgue measure on $[0,\infty)$.
\end{theorem}

Finally, based on Theorem \ref{weak_convergence_sup}, we show that the confidence band $\{[L_{IU}(t),U_{IU}(t)]:t\in\mathbb R\}$ in \eqref{confidence_band} is an asymptotically exact $(1-\alpha)$-level confidence band. 

\begin{corollary}\label{CB_direct}
Suppose Assumptions \ref{balanced_data}, \ref{balanced_randomness} and \ref{finite_horizon_model} hold. Let $q^{IS}_{1-\alpha}$ be the $(1-\alpha)$-th quantile of $\sup_{t\in\mathbb{R}}|\gamma\mathbb{G}(t)+BB({Q}(t,\underline{P}))|$. We have the following asymptotic exact $(1-\alpha)$-level confidence band
\[
\mathbb{P}\left(  \hat{Q}(t,\underline{\hat{P}})-\frac{q^{IS}_{1-\alpha}%
}{\sqrt{R}}\leq Q(t,\underline{P})\leq\hat{Q}(t,\underline{\hat{P}}%
)+\frac{q^{IS}_{1-\alpha}}{\sqrt{R}}\ \forall t\in\mathbb{R}\right)
\rightarrow1-\alpha
\]
as $n,R\rightarrow\infty$.
\end{corollary}



The proof of Theorem \ref{weak_convergence_sup} and hence Corollary \ref{CB_direct} follows the intuitive argument discussed in Section \ref{sec:main idea}. Nonetheless, the weak convergence established in Theorem \ref{weak_convergence_sup} requires two substantial technical novelties beyond merely applying the classical theory for the KS statistic. First of all, our proof approach derives a joint weak convergence in the Skorohod product space of the two sources of randomness, namely the input data noise and the Monte Carlo noise in the decomposition (\ref{process_decomp}), followed by applying a supremum operator to obtain the asymptotic of our ultimate statistic. Our first novelty is to develop new statistical tools to handle the convergence of the input data noise and its entanglement with the Monte Carlo noise. In particular, we utilize the V-process structure in common discrete-event simulation models described in Assumption \ref{finite_horizon_model} to bound the modulus of continuity of the input data noise. We also leverage the Koml\'{o}s-Major-Tusn\'{a}dy approximation (\cite{komlos1975approximation}) to derive the KS-type convergence of the Monte Carlo noise, by providing a uniform error bound that allows for perturbation due to the entangled input noise. After adding both noises to obtain the weak convergence of the total noise $\sqrt{R}(\hat{Q}(\cdot,\underline{\hat{P}})-Q(\cdot,\underline{P}))$, the classical route to obtain a KS-type statistic and confidence band is to apply the continuous mapping theorem for the supremum operator $f(\cdot)\mapsto \sup_{t\in\mathbb{R}} |f(t)|$ to get the desired weak convergence in Theorem \ref{weak_convergence_sup}. However, the supremum operator is actually not continuous in our Skorohod space due to its unbounded domain inherited from the possibly unbounded support of $Q(\cdot,\underline{P})$. To this end, our second novelty is to bypass this issue by creating an artificial coupling output $Q(Y,\underline{P})$ from the original output $Y$ whose distribution becomes $U[0,1]$ and thus has bounded support. We show that the total noise processes $\sqrt{R}(\hat{Q}(\cdot,\underline{\hat{P}})-Q(\cdot,\underline{P}))$ associated with the original output $Y$ and the new output $Q(Y,\underline{P})$ have approximately the same supremum, while the new process lies in a Skorohod space with compact domain where the supremum operator is indeed continuous, which in turn induces the convergence in Theorem \ref{weak_convergence_sup}. We provide further details on the above discussion in Section \ref{sec:theory_convergence}.

\section{Efficient Implementation}\label{sec:guarantee}
This section is devoted to the implementation guidelines of the input-uncertainty-inflated KS confidence band provided in Corollary \ref{CB_direct}. Section \ref{sec:proc} provides procedural details and algorithmic specifications for estimating the quantile $q^{IS}_{1-\alpha}$ (or equivalently $\tilde q_{1-\alpha}$ in (\ref{confidence_band})), including in particular the use of subsampling to estimate the covariance function of $\mathbb{G}(\cdot)$. Section \ref{sec:cov subsampling} presents the statistical guarantees and optimal algorithmic configuration of our subsampling scheme.


\subsection{Procedural Specifications and Covariance Function Subsampling}\label{sec:proc}

Computing the confidence band \eqref{confidence_band} requires several approximation procedures to estimate the involved quantile $\tilde q_{1-\alpha}$. First, note that exact simulation of the limiting distribution is hard
since it contains the supremum over the whole real line. To this end, we first discretize $t$ on a grid over the real line, say $t_{1},\ldots,t_{k}$ and
approximate the supremum over $t\in\mathbb{R}$ by the maximum over $t\in
\{t_{1},\ldots,t_{k}\}$, i.e.,
\begin{equation}
\max_{i=1,\ldots,k}\left\vert \frac{\mathbb{G}(t_{i})}{\sqrt{n}}+\frac
{BB({Q}(t_{i},\underline{P}))}{\sqrt{R}}\right\vert \overset{d}{\approx}%
\sup_{t\in\mathbb{R}}\left\vert \frac{\mathbb{G}(t)}{\sqrt{n}}+\frac
{BB({Q}(t,\underline{P}))}{\sqrt{R}}\right\vert . \label{discrete_time}%
\end{equation}
Since the limiting process $\gamma\mathbb{G}(t)+BB({Q}(t,\underline{P}))$ has continuous paths, when the discrete grid gets finer, the finite maximum is guaranteed to converge to the true supremum, i.e., the finite maximum with a fine enough grid is a good approximation of the supremum. Therefore, we next aim at (approximately) simulating the finite maximum in \eqref{discrete_time} to obtain its quantile. Here, by plugging in $\hat Q$ and $\hat{\underline P}$ for $Q$ and $\underline P$, we can approximate $BB({Q}(t_{i},\underline{P}))$ by
$BB(\hat{Q}(t_{i},\underline{\hat{P}}))$, which can be readily simulated. On the other hand, $\mathbb{G}(\cdot)$ is a mean-zero Gaussian process, which can be simulated as long as we can compute its covariance function. Note that directly computing the influence function \eqref{IF finite horizon} is expensive due to the many involved conditional expectation estimations. In the following, we explain how to use the bootstrap as a viable alternative, more specifically subsampling that allows us to substantially reduce the computation burden.


Based on the approximation
\eqref{process_appro}, we have
\begin{equation}
n\mathrm{Cov}(Q(t,\underline{\hat{P}}),Q(t^{\prime},\underline{\hat{P}%
}))\approx\mathrm{Cov}(\mathbb{G}(t),\mathbb{G}(t^{\prime})),
\label{cov_appro1}%
\end{equation}
where the covariance in the LHS is taken with respect to the true input distributions. As we do not know the latter, we consider the classical variance bootstrap that utilizes the following plug-in approximation
\begin{equation}
n\mathrm{Cov}_{*}(Q(t,\underline{\hat{P}}^{\ast}),Q(t^{\prime},\underline{\hat{P}}^{\ast}))\approx n\mathrm{Cov}(Q(t,\underline{\hat{P}}),Q(t^{\prime},\underline{\hat{P}})),\label{naive_cov_appro}
\end{equation}
where $\mathrm{Cov}_{*}$ denotes the covariance conditional on the empirical input distributions $\underline{\hat{P}}$, and $\underline{\hat{P}}^{*}=(\hat{P}^{*}_1,\ldots,\hat{P}^{*}_m)$ with $\hat{P}^{*}_i$ denoting the empirical distribution of $n_i$ i.i.d. observations drawn from $\hat{P}_i$. Since the exact output distribution function $Q(\cdot,\underline{\hat{P}}^{\ast})$ is also unknown, we can approximate it by the empirical output distribution. The implementation of the above classical variance bootstrap is as follows: We first draw i.i.d. data $X_{i,1}^{*},\ldots,X_{i,n_{i}}^{*}$ from the empirical distribution $\hat{P}_i$ (i.e., resample uniformly with replacement). The resample empirical distributions $\underline{\hat{P}}^{*}=(\hat{P}^{*}_1,\ldots,\hat{P}^{*}_m)$ with $\hat{P}^{*}_i=(1/n_{i})\sum_{j=1}^{n_{i}} \delta_{X_{i,j}^{*}}$ are used to drive many simulation runs (in order to wash away the Monte Carlo error), which then form the empirical output distribution function $\hat{Q}(\cdot,\underline{\hat{P}}^{*})$ as an approximation of $Q(\cdot,\underline{\hat{P}}^{*})$. We repeat the resampling and the above derived procedure many times (say $B$ times) and get $B$ empirical output distribution functions $\hat{Q}(\cdot,\underline{\hat{P}}^{*1}),\ldots,\hat{Q}(\cdot,\underline{\hat{P}}^{*B})$. Then the LHS of (\ref{naive_cov_appro}) can be approximated via the empirical covariance formula
\[
\frac{n}{B-1}\sum_{b=1}^{B}(\hat{Q}(t,\underline{\hat{P}}^{\ast b})-\bar{Q}(t))(\hat{Q}(t^{\prime},\underline{\hat{P}}^{\ast b})-\bar{Q}(t^{\prime})),
\]
where $\bar{Q}(\cdot)=\sum^B_{b=1} \hat Q(\cdot,\underline{\hat{P}}^{*b})/B$. 

Note that the variance bootstrap described above requires nested
simulation where, in the outer layer, we resample the input distributions and,
in the inner layer, we use the resample distributions to drive simulation runs
to obtain the output distributions. The issue with this naive use of the
bootstrap is that the total required simulation effort can be
huge (see \cite{lam2022subsampling} Section 2.3), essentially an amount of larger order than the data size
itself so that the Monte Carlo noises can be washed away. To remedy this computational demand, we utilize a generalization of
\cite{lam2022subsampling} to devise a subsampling procedure. This procedure
aims to exploit the form of the covariance in \eqref{cov_appro1} that scales
with the data size $n$ and turns out to reduce the computational demand from a
larger order than $n$ to an order that is independent of $n$.


To explain further, our subsampling procedure is a modification of the above bootstrap scheme where we draw a smaller-size instead of a full-data-size resample in the outer layer. More precisely, we first specify a subsample ratio $\theta\in(0,1]$, and
compute $s_{i} = \lfloor\theta n_{i}\rfloor, \forall i = 1,\ldots, m$ as the
subsample sizes. Let $X_{i,1}^{*},\ldots,X_{i,s_{i}}^{*}$ be resamples
uniformly drawn with replacement from $X_{i,1}, \ldots, X_{i,n_{i}}$ and let
$\hat P_{i,s_{i}}^{*}=(1/s_{i})\sum_{j=1}^{s_{i}} \delta_{X_{i,j}^{*}}$ be the
corresponding subsample empirical input distribution. For convenience, denote $\underline{\hat{P}}^{*}_{\theta} = (\hat P_{1,s_{1}}^{*},\ldots, \hat P_{m,s_{m}}^{*})$. Using
the same rationale as the full-data-size bootstrap, we can approximate the LHS of
\eqref{cov_appro1} by
\begin{equation}
\theta n\mathrm{Cov}_{*}(Q(t,\underline{\hat{P}}_{\theta}^{\ast}),Q(t^{\prime
},\underline{\hat{P}}_{\theta}^{\ast}))\approx n\mathrm{Cov}%
(Q(t,\underline{\hat{P}}),Q(t^{\prime},\underline{\hat{P}}))\approx
\mathrm{Cov}(\mathbb{G}(t),\mathbb{G}(t^{\prime})). \label{cov_appro2}%
\end{equation}
Compared to the classical variance bootstrap (\ref{naive_cov_appro}), in the subsample variance bootstrap (\ref{cov_appro2}) we now rescale the bootstrap covariance in \eqref{cov_appro2} by $\theta
n$ instead of $n$. This is because we now use the smaller average data size
$s=\sum_{i=1}^{m}s_{i}/m$ in our resampling, and this rescaling factor
$\theta$ retains the estimation correctness for the original quantity of interest. Importantly, as we will state more clearly in the next subsection, this modification will allow us to use less simulation effort than the
classical bootstrap. The overall implementation of \eqref{cov_appro2} is quite similar to that of (\ref{naive_cov_appro}), the main difference being the replacement of the full size resample empirical distributions $\underline{\hat{P}}^{*}$ by the subsample empirical distributions $\underline{\hat{P}}^{*}_{\theta}$. Algorithm \ref{alg:subsample} summarizes our subsampling procedure to
estimate the covariance matrix of $\mathbb{G}(t)$ for $t=t_{1},\ldots,t_{k}$.

\begin{algorithm}[hbt!]
\caption{Subsampling to Estimate Covariance Matrix}
\label{alg:subsample}
\textbf{Inputs:} number of bootstrap resample $B$, number of inner simulation runs $R_s$, input data $\{X_{i, j}: i=1,\dots,m, j=1,\dots,n_i\}$, subsample ratio $\theta$, a grid $\{t_l: l=1,\dots,k\}$
\begin{algorithmic}
\STATE Compute $s_i =\lfloor\theta n_i\rfloor, i =1,\dots, m$.
\STATE Compute $n=\sum_{i=1}^m n_i/m$.
\FOR{$b= 1, \dots, B$ }
\STATE For each $i=1,\dots, m$, draw a subsample $X_i^{*b} = \{X_{i,1}^{*b},\dots,X_{i,s_i}^{*b}\}$ uniformly with replacement from the input data $\{X_{i,1},\ldots,X_{i,n_i}\}$. Form resample empirical input distribution $\hat P_{i,s_i}^{*b}$.
\STATE Let $\underline{\hat{P}}^{*b}_{\theta} = (\hat P_{1,s_1}^{*b},\ldots, \hat P_{m,s_m}^{*b})$.
\FOR{$r= 1, \dots, R_s$ }
\STATE Simulate $\hat{Y}_r^{*b}$ from the stochastic simulation model driven by $\underline{\hat{P}}^{*b}_{\theta}$.
\ENDFOR
\STATE Form the empirical output distribution $\hat{Q}(\cdot,\underline{\hat{P}}^{*b}_{\theta})=\sum_{r=1}^{R_s}I(Y_r^{*b}\leq \cdot)/R_s$.
\ENDFOR
\STATE Form the average empirical output distribution $\bar{Q}(\cdot) = \sum^B_{b=1} \hat Q(\cdot,\underline{\hat{P}}^{*b}_{\theta})/B$.
\STATE Compute the estimated covariance matrix $V\in\mathbb{R}^{k\times k}$ at points $\{t_l: l=1,\ldots,k\}$:
\begin{equation*}
V_{i,j}  =  \frac{\theta n}{B-1}\sum_{b=1}^{B}(\hat{Q}(t_{i},\underline{\hat{P}}_{\theta}^{\ast b})-\bar{Q}(t_{i}))(\hat{Q}(t_{j},\underline{\hat{P}}_{\theta}^{\ast b})-\bar{Q}(t_{j})).
\end{equation*}
\RETURN $V$
\end{algorithmic}
\end{algorithm}

Once we obtain the estimated covariance matrix $V$ returned by Algorithm
\ref{alg:subsample}, we can simulate a $k$-dimensional Gaussian vector
$(Z_{1},\ldots,Z_{k})\sim N(0,V)$ as the estimator for $(\mathbb{G}%
(t_{1}),\ldots,\mathbb{G}(t_{k}))$. As explained before, we can also simulate
$(W_{1},\ldots,W_{k})\overset{d}{=}(BB(\hat{Q}(t_{1},\underline{\hat{P}%
})),\ldots,BB(\hat{Q}(t_{k},\underline{\hat{P}})))$ as an approximation for
$(BB({Q}(t_{1},\underline{P})),\ldots,BB({Q}(t_{k},\underline{P})))$. Since
the processes $\mathbb{G}(t)$ and $BB({Q}(t,\underline{P}))$ are independent,
the simulation procedures of $(Z_{1},\ldots,Z_{k})$ and $(W_{1},\ldots,W_{k})$
are also set to be independent, conditional on the input data. With them in hand,
we can obtain one simulation replication of the finite maximum in (\ref{discrete_time}).
Repeating the above procedure many times, we can estimate the quantile of the
finite maximum in (\ref{discrete_time}) and construct the desired confidence
band as in (\ref{confidence_band}). We summarize our confidence band
construction procedure in Algorithm \ref{alg:cdf_cb}.

\begin{algorithm}[hbt!]
\caption{Construction of Input-Uncertainty-Inflated KS Confidence Band for the Output Distribution Function}
\label{alg:cdf_cb}
\textbf{Inputs:} number of simulation runs for empirical output distribution $R$, number of limiting distribution samples $R_q$, confidence level $1-\alpha$, input data $\{X_{i, j}: i=1,\dots,m, j=1,\dots,n_i\}$, estimated covariance matrix $V$, the same grid $\{t_l: l=1,\dots,k\}$ as in Algorithm \ref{alg:subsample}
\begin{algorithmic}
\STATE Compute $n=\sum_{i=1}^m n_i/m$.
\FOR{$r= 1, \dots, R$ }
\STATE Simulate $Y_r$ from the stochastic simulation model driven by $\underline{\hat{P}}$.
\ENDFOR
\STATE Form the empirical output distribution $\hat Q(\cdot,\underline{\hat{P}})=\sum_{r=1}^{R}I(Y_r\leq \cdot)/R$.
\FOR{$i = 1,\dots, R_q$}
\STATE Simulate $Z^i=(Z^i_1,\ldots,Z^i_k)$ from the $k$-dimensional Gaussian random vector $N(0,V)$.
\STATE Simulate $W^i=(W^i_1,\ldots,W^i_k)$ as a time-transformed discretized Brownian bridge $BB(\hat{Q}(t_l,\underline{\hat{P}})), l = 1, \ldots,k$.
\STATE Compute $e_i= \max\limits_{j = 1,\dots,k} |Z^i_j/\sqrt{n}+W^i_j/\sqrt{R}|$.
\ENDFOR
\STATE Find $\hat{\tilde{q}}_{1-\alpha} = (1-\alpha)$-th quantile of $e_1,\dots,e_{R_q}$.
\RETURN $L_{IU}(\cdot)=\hat{Q}(\cdot,\underline{\hat{P}})-\hat{\tilde{q}}_{1-\alpha}$ and $U_{IU}(\cdot)=\hat{Q}(\cdot,\underline{\hat{P}})+\hat{\tilde{q}}_{1-\alpha}$.
\end{algorithmic}
\end{algorithm}

To close this subsection, we make a final note that, to estimate the covariance function, one may also consider a debiased ANOVA-based estimator like the one used to estimate the variance of a conditional expectation (e.g., \cite{sun2011efficient,lam2022subsampling}), which turns out to be conditionally unbiased for $\theta n\mathrm{Cov}_{\ast}(Q(t,\underline{\hat{P}}_{\theta}^{\ast}),Q(t^{\prime},\underline{\hat{P}}_{\theta}^{\ast}))$. However, such a debiased estimator is not always positive semi-definite given a finite simulation budget, in contrary to Algorithm \ref{alg:subsample}. Considering that our confidence band construction requires generating multivariate Gaussian vectors with the estimated covariance matrix (i.e., Algorithm \ref{alg:cdf_cb}), the positive semi-definiteness of the latter estimation appears crucial and hence we propose Algorithm \ref{alg:subsample} instead of a debiased counterpart.

\subsection{Statistical Guarantees on Covariance Function Subsampling}\label{sec:cov subsampling}

This subsection presents the statistical guarantees and optimal algorithmic configuration of our subsampling scheme. Let $\hat{\sigma}^{2}(t,t^{\prime})$ be the output of Algorithm
\ref{alg:subsample} at general positions $t,t^{\prime}\in\mathbb{R}$, i.e.,
\begin{equation*}
\hat{\sigma}^{2}(t,t^{\prime}):=   \frac{\theta n}{B-1}\sum
_{b=1}^{B}(\hat{Q}(t,\underline{\hat{P}}_{\theta}^{\ast b})-\bar{Q}%
(t))(\hat{Q}(t^{\prime},\underline{\hat{P}}_{\theta}^{\ast b})-\bar
{Q}(t^{\prime})).
\end{equation*}
We focus on the estimation error $\hat{\sigma}^{2}(t,t^{\prime})-\mathrm{Cov}%
(\mathbb{G}(t),\mathbb{G}(t^{\prime}))$. First, we show that, with proper algorithmic configurations, this estimation error converges to 0 in probability, i.e., the covariance estimator is \emph{consistent}.

\begin{theorem}
\label{validity_alg_cov} Suppose Assumptions \ref{balanced_data} and \ref{finite_horizon_model} hold. Moreover, suppose the configuration in Algorithm \ref{alg:subsample} satisfies%
\begin{equation}
\theta=\omega(1/n),\quad B=\omega(1),\quad R_{s}=\omega
(s).\label{configuration_B_R_theta}%
\end{equation}
We have%
\begin{equation}
\hat{\sigma}^{2}(t,t^{\prime})-\mathrm{Cov}(\mathbb{G}(t),\mathbb{G}%
(t^{\prime}))=o_{p}(1),\label{decomposition_est_err}%
\end{equation}
and consequently the output $V$ in Algorithm \ref{alg:subsample} is consistent
for estimating the covariance matrix of $(\mathbb{G}(t_{1}),\ldots
,\mathbb{G}(t_{k}))$.
\end{theorem}

In Theorem \ref{validity_alg_cov}, (\ref{configuration_B_R_theta}) requires the total simulation budget $N:=BR_{s}$ for Algorithm \ref{alg:subsample} to be at least
$N=\omega(\theta n)=\omega(1)$. On the other hand, we show that $N=\omega(1)$ is also
sufficient to admit a configuration satisfying (\ref{configuration_B_R_theta}%
). This means that the simulation budget in Algorithm \ref{alg:subsample} can
grow at an arbitrary rate as $n\rightarrow\infty$\ while maintaining estimation consistency.

\begin{corollary}
\label{min_budget} Suppose Assumptions \ref{balanced_data} and \ref{finite_horizon_model} hold. The minimum simulation budget $N$ such that (\ref{configuration_B_R_theta}) holds,
and thus the covariance estimator $\hat{\sigma}^{2}(t,t^{\prime})$ is consistent, is $N=\omega(1)$. 
\end{corollary}

Finally, we characterize the optimal configuration, for all three parameters $\theta$, $B$ and $R_{s}$ in Algorithm \ref{alg:subsample}, to minimize the estimation error $\hat{\sigma}^{2}(t,t^{\prime})-\mathrm{Cov}(\mathbb{G}(t),\mathbb{G}(t^{\prime}))$ given a simulation budget $N=\omega(1)$. Let $\sigma^{2}(t,t^{\prime})$ be the scaled bootstrap covariance in the LHS of (\ref{cov_appro2}), i.e., $\sigma^{2}(t,t^{\prime}):=\theta n\mathrm{Cov}_{\ast}(Q(t,\underline{\hat{P}}_{\theta}^{\ast}),Q(t^{\prime},\underline{\hat{P}}_{\theta}^{\ast}))$, which $\hat{\sigma}^{2}(t,t^{\prime})$ is meant to estimate directly.

\begin{theorem}
\label{overall_opt_config}Suppose Assumptions \ref{balanced_data} and \ref{finite_horizon_model} hold. Assume the simulation budget is $N=\omega(1)$ and the subsample size is chosen as $s=\omega(1)$ such that $s=o(N)$. Suppose
\[
\hat{\sigma}^{2}(t,t^{\prime})-\mathrm{Cov}(\mathbb{G}(t),\mathbb{G}(t^{\prime}))=\Theta_p(\mathcal{O}_1+\mathcal{O}_2)
\]
whenever $\hat{\sigma}^{2}(t,t^{\prime})-\sigma^{2}(t,t^{\prime})$ has a conditional mean squared error (MSE) of order $\mathcal{O}^2_1$, i.e., $\mathbb{E}_{\ast}[(\hat{\sigma}^{2}(t,t^{\prime})-\sigma^{2}(t,t^{\prime}))^{2}]=\Theta_p(\mathcal{O}^2_1)$, and the bootstrap error $\sigma^{2}(t,t^{\prime})-\mathrm{Cov}(\mathbb{G}(t),\mathbb{G}(t^{\prime})$ is of order $\Theta(\mathcal{O}_2)$, where both $\mathcal{O}_1$ and $\mathcal{O}_2$ could depend on the configuration parameters.
If the following non-degeneracy condition holds 
\begin{equation*}
0<Q(t,\underline{P}),Q(t^{\prime},\underline{P})<1,\quad\sum_{i=1}^{m}\mathbb{E}_{P_{i}}[IF_{i}^{2}(t,X_{i};\underline{P})]>0,\quad\sum_{i=1}^{m}\mathbb{E}_{P_{i}}[IF_{i}^{2}(t^{\prime},X_{i};\underline{P})]>0,
\end{equation*}
and 
\[
\mathbb{E}[\mathcal{E}_{1}]=\Theta\left(  \frac{1}{s}+\sum_{i=1}^{m}\left\vert\frac{n}{n_{i}}-\frac{1}{\beta_{i}}\right\vert \right)  ,\quad\mathrm{Var}(\mathcal{E}_{1})=\Theta\left(  \frac{1}{n}\right)  ,
\]
where $\mathcal{E}_1$ is defined in Theorem \ref{error_true_bootstrap}, then the optimal allocation that minimizes the order of $\hat{\sigma}^{2}(t,t^{\prime})-\mathrm{Cov}(\mathbb{G}(t),\mathbb{G}(t^{\prime}))$ is $R_{s}^{\ast}=\Theta(N^{1/3}(s^*)^{2/3})$, $B^{\ast}=N/R_{s}^{\ast}$, $s^{\ast}=\Theta(N^{1/4})$ and consequently $\theta^{\ast}=\Theta(s^{\ast}/n)=\Theta(N^{1/4}/n)$, in which case the minimal order of $\hat{\sigma}^{2}(t,t^{\prime})-\mathrm{Cov}(\mathbb{G}(t),\mathbb{G}(t^{\prime}))$ is%
\begin{equation}
\hat{\sigma}^{2}(t,t^{\prime})-\mathrm{Cov}(\mathbb{G}(t),\mathbb{G}(t^{\prime}))=\Theta_{p}\left(  \frac{1}{N^{1/4}}+\sum_{i=1}^{m}\left\vert\frac{n}{n_{i}}-\frac{1}{\beta_{i}}\right\vert +\frac{1}{\sqrt{n}}\right) .\label{minimal_estimation_error}
\end{equation}
\end{theorem}

The theoretical guarantees above generalize the results in \cite{lam2022subsampling} from variance to covariance estimation. However, since we use a biased estimator that guarantees positive semi-definiteness instead of a debiased estimator in \cite{lam2022subsampling}, even though our required total budget $N$ is still $\omega(1)$, our optimal configuration becomes $R^*_s=B^*=\Theta(N^{1/2})$ and $s^*=\Theta(N^{1/4})$ in Theorem \ref{overall_opt_config}, with the minimal estimation error (\ref{minimal_estimation_error}). This is different from the optimal configuration suggested in Theorem 5 in \cite{lam2022subsampling} given by $R^*_s=\Theta(N^{1/3}),B^*=\Theta(N^{2/3})$ and $s^*=\Theta(N^{1/3})$ (in their case of $1\ll N \le n^{3/2}$). The latter has a slightly smaller minimal estimation error of order $\Theta_p(N^{-1/3}+n^{-1/2})$ (their estimation target is a scaled variance by $n$, so we multiply their error by $n$ to make the comparison fair). By comparing the two orders (ignoring the middle term in (\ref{minimal_estimation_error}) which is the residual error of the limits in Assumption \ref{balanced_data}), we can see that the price of bearing the bias is an increase in one of the components in the error order from $N^{-1/3}$ to $N^{-1/4}$. This error increase, nonetheless, is outweighed by the gain in maintaining positive semi-definiteness of the covariance matrix which is crucial to implement our confidence band. We provide further technical details in Section \ref{sec:theory_cov}.

\section{Main Technical Developments for Theoretical Guarantees}

\label{sec:theory}

This section outlines the main development of the theoretical results presented in Sections \ref{sec:methodology} and \ref{sec:guarantee}. Section \ref{sec:theory_convergence} focuses on the weak convergence of the process $\sqrt{R}(\hat{Q}(\cdot,\underline{\hat{P}})-Q(\cdot,\underline{P}))$ which leads to our main result Theorem \ref{weak_convergence_sup}, while Section \ref{sec:theory_cov} focuses on the statistical guarantees for the covariance function estimation. In addition, we provide an overview of some useful weak convergence results in Appendix \ref{sec:Skorohod}, and generalizations of our statistical guarantees for the covariance function estimation under further relaxed assumptions in Appendix \ref{sec:general_assumptions}. 


\subsection{Weak Convergence of Generalized Kolmogorov-Smirnov Statistic}\label{sec:theory_convergence}


Recall that in (\ref{process_decomp}), we separate
the input data noise $Q(t,\underline{\hat{P}})-Q(t,\underline{P})$ and Monte
Carlo noise $\hat{Q}(t,\underline{\hat{P}})-Q(t,\underline{\hat{P}})$, and
approximate them by (scaled) $\mathbb{G}(t)$ and $BB(Q(t,\underline{P}))$ respectively
in (\ref{process_appro}). We first rigorize (\ref{process_appro}) by showing
that it corresponds to the weak convergence of the two-dimensional stochastic
process
\begin{equation}
(\sqrt{R}(\hat{Q}(\cdot,\underline{\hat{P}})-Q(\cdot,\underline{\hat{P}%
})),\sqrt{n}(Q(\cdot,\underline{\hat{P}})-Q(\cdot,\underline{P})))\Rightarrow
(BB({Q}(\cdot,\underline{P})),\mathbb{G}(\cdot)) \label{weak_conv_target}%
\end{equation}
on a proper metric space where the stochastic processes in
(\ref{weak_conv_target}) lie in. In fact, the metric space is chosen as the Skorohod product
space $D_{\pm\infty}\times D_{\pm\infty}$, where the two-dimensional process naturally lies in. A brief introduction to all the necessary results in
this space can be found in Appendix \ref{sec:Skorohod}. Roughly speaking, to
show the weak convergence, we need to show three results: 1) measurability on
the Skorohod product space; 2) finite dimensional convergence and 3)
tightness. Measurability ensures all the subsequent probabilistic analysis is
well-defined. Finite dimensional convergence means that the weak convergence
(\ref{weak_conv_target}) holds for any finite collection of $t$. It is
necessary for the weak convergence of stochastic processes but not sufficient
unless tightness holds. Tightness means that the LHS of
(\ref{weak_conv_target}) is well-behaved as a stochastic process (or ``random
function'') in $t$, which ensures that we can reasonably deduce
(\ref{weak_conv_target}) from the finite dimensional convergence.

Let $(\Omega,\mathcal{F},\mathbb{P})$ be our underlying probability space. We first prove measurability.

\begin{lemma}
\label{measurability}Suppose for each $t\in\mathbb{R}$, $\hat{Q}%
(t,\underline{\hat{P}})$ and $Q(t,\underline{\hat{P}})$ are random variables,
i.e., they are measurable maps from $(\Omega,\mathcal{F})$ to $(\mathbb{R}%
,\mathcal{B}(\mathbb{R}))$. Then $(\sqrt{R}(\hat{Q}(\cdot,\underline{\hat{P}%
})-Q(\cdot,\underline{\hat{P}})),\sqrt{n}(Q(\cdot,\underline{\hat{P}}%
)-Q(\cdot,\underline{P})))$ is measurable on the Skorohod product space.
\end{lemma}

Throughout this paper, we always assume $\hat{Q}(t,\underline{\hat{P}})$ and
$Q(t,\underline{\hat{P}})$ are random variables for each $t$ $\in\mathbb{R}$
in $(\Omega,\mathcal{F},\mathbb{P})$ so that the measurability holds.

Next we show finite dimensional convergence. For a sequence of stochastic
processes $\{X_{n}(t),t\in\mathbb{R\}}$ and another stochastic process
$\{X(t),t\in\mathbb{R\}}$, we say $X_{n}$ finite-dimensionally converges to
$X$ (written as $X_{n}\overset{f.d.}{\rightarrow}X$) if%
\[
(X_{n}(t_{1}),\ldots,X_{n}(t_{k}))\overset{d}{\rightarrow}(X(t_{1}%
),\ldots,X(t_{k})),\forall t_{1}<t_{2}<\cdots<t_{k},\forall k\in\mathbb{N}.
\]
We would like to show
\[
(\sqrt{R}(\hat{Q}(\cdot,\underline{\hat{P}})-Q(\cdot,\underline{\hat{P}%
})),\sqrt{n}(Q(\cdot,\underline{\hat{P}})-Q(\cdot,\underline{P}%
)))\overset{f.d.}{\rightarrow}(BB({Q}(\cdot,\underline{P})),\mathbb{G}%
(\cdot)),
\]
where the mean-zero Gaussian process $\mathbb{G}$ has the covariance function (\ref{covariance_function}).

\begin{proposition}
\label{fidis_convergence} Suppose Assumptions \ref{balanced_data} and
\ref{finite_horizon_model} hold. We have
\[
(\sqrt{R}(\hat{Q}(\cdot,\underline{\hat{P}})-Q(\cdot,\underline{\hat{P}%
})),\sqrt{n}(Q(\cdot,\underline{\hat{P}})-Q(\cdot,\underline{P}%
)))\overset{f.d.}{\rightarrow}(BB({Q}(\cdot,\underline{P})),\mathbb{G}%
(\cdot))
\]
as $n,R\rightarrow\infty$, where $BB(\cdot)$ is the standard Brownian bridge
on $[0,1]$, and $\mathbb{G}(\cdot)$ is a mean-zero Gaussian process defined on
$\mathbb{R}$ with covariance function
\[
\mathrm{Cov}(\mathbb{G}(t),\mathbb{G}(s))=\sum_{i=1}^{m}\frac{1}{\beta_{i}%
}\mathrm{Cov}_{P_{i}}(IF_{i}(t,X_{i};\underline{P}),IF_{i}(s,X_{i}%
;\underline{P})).
\]
Moreover, $BB(\cdot)$ and $\mathbb{G}(\cdot)$ are independent processes.
\end{proposition}

Now we present the weak convergence result (\ref{weak_conv_target}), where we prove the final condition, tightness, along the way.

\begin{theorem}
\label{weak_convergence}Suppose Assumptions \ref{balanced_data} and
\ref{finite_horizon_model} hold. We have%
\[
(\sqrt{R}(\hat{Q}(\cdot,\underline{\hat{P}})-Q(\cdot,\underline{\hat{P}%
})),\sqrt{n}(Q(\cdot,\underline{\hat{P}})-Q(\cdot,\underline{P})))\Rightarrow
(BB({Q}(\cdot,\underline{P})),\mathbb{G}(\cdot))
\]
as $n,R\rightarrow\infty$, where the distribution of $BB({Q}(\cdot
,\underline{P}))$ and $\mathbb{G}(\cdot)$ are given in Proposition
\ref{fidis_convergence}. Moreover $\mathbb{P}(\mathbb{G}\in C_{b}
(\mathbb{R}))=1$, where $C_{b}(\mathbb{R})\subset D_{\pm\infty}$ is
the class of bounded continuous functions on $\mathbb{R}$.
\end{theorem}

Theorem \ref{weak_convergence} establishes the joint weak convergence for Monte
Carlo noise $\hat{Q}(\cdot,\underline{\hat{P}})-Q(\cdot,\underline{\hat{P}})$
and input data noise $Q(\cdot,\underline{\hat{P}})-Q(\cdot,\underline{P})$, where the tightness of the Monte Carlo noise is proved by means of the Koml\'{o}s-Major-Tusn\'{a}dy approximation and the tightness of the input data noise is handled by the V-process theory. Recall that Assumption \ref{balanced_randomness} ensures the two noises are of the same order. So we can add them together (with a unified
normalizing constant) to obtain the weak convergence for the total noise
$\hat{Q}(\cdot,\underline{\hat{P}})-Q(\cdot,\underline{P})$. However, we caution that the weak convergence of the summation of two weakly
convergent processes should not be taken for granted on the Skorohod space
$D_{\pm\infty}$. This is because the addition operator $+:D_{\pm\infty}\times
D_{\pm\infty}\mapsto D_{\pm\infty}$ is not a continuous mapping in general.
However, we can show that it is indeed continuous on the subset $C(\mathbb{R})$ where both
$BB({Q}(\cdot,\underline{P}))$ and $\mathbb{G}(\cdot)$ lie in, which allows us
to use the continuous mapping theorem.


\begin{corollary}
\label{weak_convergence_sum}Suppose Assumptions \ref{balanced_data},
\ref{balanced_randomness} and \ref{finite_horizon_model} hold. We have%
\[
\sqrt{R}(\hat{Q}(\cdot,\underline{\hat{P}})-Q(\cdot,\underline{P}%
))\Rightarrow\gamma\mathbb{G}(\cdot)+BB({Q}(\cdot,\underline{P}))
\]
as $n,R\rightarrow\infty$, where $BB(\cdot)$ and $\mathbb{G}(\cdot)$ are
defined in Theorem \ref{weak_convergence}.
\end{corollary}

As the final step, we need to take the supremum in Corollary \ref{weak_convergence_sum} to obtain our main Theorem \ref{weak_convergence_sup}, i.e., the weak convergence of the supremum $\sup_{t\in\mathbb{R}}|\sqrt{R}(\hat{Q}(t,\underline{\hat{P}})-Q(t,\underline{P}))|$. As explained at the end of Section \ref{sec:methodology}, we use a coupling technique by creating the artificial output $Q(Y,\underline{P})$ to avoid the issue of unbounded support which damages the continuity of the supremum operator. By showing that the processes in Corollary \ref{weak_convergence_sum} associated with the two outputs have approximately the same supremum, the continuous mapping theorem applied to the new process proves Theorem \ref{weak_convergence_sup}.

\subsection{Covariance Function Estimation}

\label{sec:theory_cov}


Recall that the covariance function estimator $\hat{\sigma}^{2}(t,t^{\prime})$
is the estimator for the subsampled covariance $\sigma^{2}(t,t^{\prime
})=\theta n\mathrm{Cov}_{\ast}(Q(t,\underline{\hat{P}}_{\theta}^{\ast
}),Q(t^{\prime},\underline{\hat{P}}_{\theta}^{\ast}))$ and the latter is the
approximation for the covariance function $\mathrm{Cov}(\mathbb{G}%
(t),\mathbb{G}(t^{\prime}))$ according to the bootstrap rationale in
(\ref{cov_appro2}). Therefore, to analyze the estimation error $\hat{\sigma
}^{2}(t,t^{\prime})-\mathrm{Cov}(\mathbb{G}(t),\mathbb{G}(t^{\prime}))$, we
consider the following decomposition
\begin{equation}
\hat{\sigma}^{2}(t,t^{\prime})-\mathrm{Cov}(\mathbb{G}(t),\mathbb{G}%
(t^{\prime}))=(\hat{\sigma}^{2}(t,t^{\prime})-\sigma^{2}(t,t^{\prime
}))+(\sigma^{2}(t,t^{\prime})-\mathrm{Cov}(\mathbb{G}(t),\mathbb{G}(t^{\prime
}))),\label{cov_est_decomp}%
\end{equation}
where the first term $\hat{\sigma}^{2}(t,t^{\prime})-\sigma^{2}(t,t^{\prime})$
denotes the Monte Carlo error and the second term $\sigma^{2}(t,t^{\prime
})-\mathrm{Cov}(\mathbb{G}(t),\mathbb{G}(t^{\prime}))$ denotes the bootstrap
error. We will analyze each error individually and then add them together to
analyze the total error $\hat{\sigma}^{2}(t,t^{\prime})-\mathrm{Cov}%
(\mathbb{G}(t),\mathbb{G}(t^{\prime}))$.

We first consider the consistency of our estimator $\hat{\sigma}%
^{2}(t,t^{\prime})$. The following theorem establishes the validity of the
bootstrap rationale, i.e., the bootstrap error is small enough:

\begin{theorem}
\label{consistency_sigma} Suppose Assumptions \ref{balanced_data} and \ref{finite_horizon_model} hold. Then we have $\sigma^{2}(t,t^{\prime})\overset{p}{\rightarrow}\mathrm{Cov}(\mathbb{G}(t),\mathbb{G}(t^{\prime}))$ as $\theta n\rightarrow\infty$ for any $t,t^{\prime}\in\mathbb{R}$.
\end{theorem}

To study the Monte Carlo error, we focus on the conditional MSE $\mathbb{E}%
_{\ast}[(\hat{\sigma}^{2}(t,t^{\prime})-\sigma^{2}(t,t^{\prime}))^{2}]$, which
can be decomposed as%
\[
\mathbb{E}_{\ast}[(\hat{\sigma}^{2}(t,t^{\prime})-\sigma^{2}(t,t^{\prime
}))^{2}]=\mathrm{Var}_{\ast}(\hat{\sigma}^{2}(t,t^{\prime}))+(\mathbb{E}%
_{\ast}[\hat{\sigma}^{2}(t,t^{\prime})]-\sigma^{2}(t,t^{\prime}))^{2}.
\]
Algorithm \ref{alg:subsample} that outputs the covariance estimator $\hat{\sigma}^{2}(t,t^{\prime})$ has a generic form presented in Appendix \ref{sec: var_of_cov}. By means of the mean and variance formulas in Lemma \ref{general_covariance_undebias}, we obtain the following theorem.

\begin{theorem}
\label{MSE_sigma_hat}Suppose Assumptions \ref{balanced_data} and \ref{finite_horizon_model} hold and the subsample ratio $\theta$ is chosen such that $\theta n=\omega(1)$.
Then the configuration
\begin{equation}
B=\omega(1),\quad R_{s}=\omega(s),\label{configuration_B_R}%
\end{equation}
is sufficient for $\mathbb{E}_{\ast}[(\hat{\sigma}^{2}(t,t^{\prime}%
)-\sigma^{2}(t,t^{\prime}))^{2}]=o_{p}(1)$. If the following non-degeneracy
condition holds%
\begin{equation}
0<Q(t,\underline{P}),Q(t^{\prime},\underline{P})<1,\quad\sum_{i=1}%
^{m}\mathbb{E}_{P_{i}}[IF_{i}^{2}(t,X_{i};\underline{P})]>0,\quad\sum
_{i=1}^{m}\mathbb{E}_{P_{i}}[IF_{i}^{2}(t^{\prime},X_{i};\underline{P}%
)]>0,\label{non-degeneracy}%
\end{equation}
then the configuration (\ref{configuration_B_R}) is also necessary for
$\mathbb{E}_{\ast}[(\hat{\sigma}^{2}(t,t^{\prime})-\sigma^{2}(t,t^{\prime
}))^{2}]=o_{p}(1)$.
\end{theorem}

The combination of Theorems \ref{consistency_sigma} and \ref{MSE_sigma_hat}
leads to the consistency of $\hat{\sigma}^{2}(t,t^{\prime})$ in Theorem
\ref{validity_alg_cov} and also the minimal simulation budget in Corollary
\ref{min_budget}, which tells us that under a simulation budget $N=\omega(1)$ and a proper algorithm configuration, we have
a negligible estimation error $\hat{\sigma}^{2}(t,t^{\prime})-\mathrm{Cov}%
(\mathbb{G}(t),\mathbb{G}(t^{\prime}))=o_{p}(1)$. As a further step, we would
like to know how to tune the configuration parameters $B,R_{s},\theta$ to
minimize the estimation error given a simulation budget $N=BR_{s}$. Note that
in the decomposition of the estimation error (\ref{cov_est_decomp}), $B$ and
$R_{s}$ only affect the term $\hat{\sigma}^{2}(t,t^{\prime})-\sigma
^{2}(t,t^{\prime})$ while $\theta$ affects both terms $\hat{\sigma}%
^{2}(t,t^{\prime})-\sigma^{2}(t,t^{\prime})$ and $\sigma^{2}(t,t^{\prime
})-\mathrm{Cov}(\mathbb{G}(t),\mathbb{G}(t^{\prime}))$. So we first study the
optimal configuration of $B$ and $R_{s}$ for the term $\hat{\sigma}%
^{2}(t,t^{\prime})-\sigma^{2}(t,t^{\prime})$ by minimizing $\mathbb{E}_{\ast}[(\hat{\sigma}^{2}(t,t^{\prime})-\sigma^{2}(t,t^{\prime}))^{2}]$ given $N$ and $\theta$.

\begin{theorem}
\label{optimal_B_R} Suppose Assumptions \ref{balanced_data}, \ref{finite_horizon_model} and the non-degeneracy condition (\ref{non-degeneracy}) hold. Given a simulation budget $N=\omega(\theta n)$ and a subsample ratio $\theta=\omega(1/n)$, the optimal $B$ and $R_{s}$ that minimize the order of $\mathbb{E}_{\ast}[(\hat{\sigma}^{2}(t,t^{\prime})-\sigma^{2}(t,t^{\prime}))^{2}]$ are%
\begin{equation}
R_{s}^{\ast}=\Theta(N^{1/3}s^{2/3}),B^{\ast}=N/R_{s}^{\ast}%
,\label{optimal_B_Rquation}%
\end{equation}
which leads to the minimal order of the MSE%
\[
\mathbb{E}_{\ast}[(\hat{\sigma}^{2}(t,t^{\prime})-\sigma^{2}(t,t^{\prime
}))^{2}]=\Theta((\theta n)^{2/3}/N^{2/3})(1+o_{p}(1)).
\]
If the non-degeneracy condition (\ref{non-degeneracy}) does not hold, then
under the configuration (\ref{optimal_B_Rquation}), we have%
\[
\mathbb{E}_{\ast}[(\hat{\sigma}^{2}(t,t^{\prime})-\sigma^{2}(t,t^{\prime
}))^{2}]=O((\theta n)^{2/3}/N^{2/3})(1+o_{p}(1)).
\]
\end{theorem}

Next, to obtain the optimal configuration of $\theta$, we need to study the
influence of $\theta$ on the two terms $\hat{\sigma}^{2}(t,t^{\prime}%
)-\sigma^{2}(t,t^{\prime})$ and $\sigma^{2}(t,t^{\prime})-\mathrm{Cov}%
(\mathbb{G}(t),\mathbb{G}(t^{\prime}))$. The MSE in Theorem \ref{optimal_B_R}
already gives the order of $\theta$ in $\hat{\sigma}^{2}(t,t^{\prime}%
)-\sigma^{2}(t,t^{\prime})$ when $B$ and $R_{s}$ are optimally tuned. The
following theorem reveals the order of $\theta$ in $\sigma^{2}(t,t^{\prime
})-\mathrm{Cov}(\mathbb{G}(t),\mathbb{G}(t^{\prime}))$.

\begin{theorem}
\label{error_true_bootstrap}Suppose Assumptions \ref{balanced_data} and \ref{finite_horizon_model} hold. Then we have%
\[
\sigma^{2}(t,t^{\prime})-\mathrm{Cov}(\mathbb{G}(t),\mathbb{G}(t^{\prime}))=\mathcal{E}_{1}+o_{p}\left(  \frac{1}{s}+\frac{1}{\sqrt{n}}\right)  ,
\]
as $\theta n\rightarrow\infty$ for any $t,t^{\prime}\in\mathbb{R}$, where the leading order term $\mathcal{E}_{1}$ satisfies
\begin{align*}
\mathbb{E}[\mathcal{E}_{1}]  &  =\sum_{i=1}^{m}\left(  \frac{n\mathrm{frac}(\theta n_{i})}{s_{i}n_{i}}+\left(  \frac{n}{n_{i}}-\frac{1}{\beta_{i}}\right)  \right)  \mathrm{Cov}(IF_{i}(t,X_{i};\underline{P}),IF_{i}(t^{\prime},X_{i};\underline{P}))\\
&  +\theta n\sum_{i=1}^{m}\frac{1}{2s_{i}^{2}}\mathrm{Cov}(IF_{i}(t,X_{i};\underline{P}),IF_{ii}(t^{\prime},X_{i},X_{i};\underline{P}))\\
&  +\theta n\sum_{i=1}^{m}\sum_{i^{\prime}=1}^{m}\frac{1}{2s_{i}s_{i^{\prime}}}\mathrm{Cov}(IF_{i}(t,X_{i,1};\underline{P}),IF_{ii^{\prime}i^{\prime}}(t^{\prime},X_{i,1},X_{i^{\prime},2},X_{i^{\prime},2};\underline{P}))\\
&  +\theta n\sum_{i=1}^{m}\frac{1}{2s_{i}^{2}}\mathrm{Cov}(IF_{i}(t^{\prime},X_{i};\underline{P}),IF_{ii}(t,X_{i},X_{i};\underline{P}))\\
&  +\theta n\sum_{i=1}^{m}\sum_{i^{\prime}=1}^{m}\frac{1}{2s_{i}s_{i^{\prime}}}\mathrm{Cov}(IF_{i}(t^{\prime},X_{i,1};\underline{P}),IF_{ii^{\prime}i^{\prime}}(t,X_{i,1},X_{i^{\prime},2},X_{i^{\prime},2};\underline{P}))\\
&  +\theta n\sum_{i_{1},i_{2}=1}^{m}\frac{1}{2s_{i_{1}}s_{i_{2}}}\mathrm{Cov}(IF_{i_{1}i_{2}}(t,X_{i_{1},1},X_{i_{2},2};\underline{P}),IF_{i_{1}i_{2}}(t^{\prime},X_{i_{1},1},X_{i_{2},2};\underline{P}))\\
&  =O\left(  \frac{1}{s}+\sum_{i=1}^{m}\left|  \frac{n}{n_{i}}-\frac{1}{\beta_{i}}\right|  \right)  ,
\end{align*}
and%
\begin{align*}
\mathrm{Var}(\mathcal{E}_{1})  &  =\sum_{i=1}^{m}\frac{1}{n_{i}}\mathrm{Var}\left(  \frac{1}{\beta_{i}}IF_{i}(t,X_{i,1};\underline{P})IF_{i}(t^{\prime},X_{i,1};\underline{P})\right. \\
&  +\sum_{i^{\prime}=1}^{m}\frac{1}{\beta_{i^{\prime}}}\mathbb{E}[IF_{i^{\prime}}(t,X_{i^{\prime},2};\underline{P})IF_{i^{\prime}i}(t^{\prime},X_{i^{\prime},2},X_{i,1};\underline{P})|X_{i,1}]\\
&  \left.  +\sum_{i^{\prime}=1}^{m}\frac{1}{\beta_{i^{\prime}}}\mathbb{E}[IF_{i^{\prime}}(t^{\prime},X_{i^{\prime},2};\underline{P})IF_{i^{\prime}i}(t,X_{i^{\prime},2},X_{i,1};\underline{P})|X_{i,1}]\right)  =O\left(\frac{1}{n}\right)  .
\end{align*}
\end{theorem}

Combining the orders in Theorems \ref{optimal_B_R} and
\ref{error_true_bootstrap}, we can specify the optimal configuration of the
subsample size $s$ and the corresponding subsample ratio $\theta$, which is
presented in Theorem \ref{overall_opt_config}.

Finally, we point out that our theorems regarding covariance function estimation hold under more general assumptions than the finite-horizon model. These more general results could be of independent interest in other contexts, and thus we present them in Appendix \ref{sec:general_assumptions}.

\section{Numerical Experiments}

\label{sec:num}

We conduct experiments to validate our theory and illustrate how our methodology can be used to quantify uncertainty for distributional quantities. More specifically, Section \ref{sec:num_verify} compares our input-uncertainty-inflated KS confidence bands with the classic KS bands and the standard bootstrap, and tests various algorithmic configurations in our approach. Section \ref{sec:quantiles} further applies our confidence bands to construct simultaneous confidence regions for output quantiles.

Throughout this section, we use the following two examples:
\begin{enumerate} 
\item An M/M/1 queue. It has 2 input distributions $\underline{P} = (P_1, P_2)$ where $P_1$ denotes the inter-arrival time distribution with rate 0.5 and $P_2$ denotes the service time distribution with rate 1, for which we have $n_1$ and $n_2$ i.i.d. data points available respectively. To satisfy the balanced data Assumption \ref{balanced_data}, we choose $n_2=2n_1$ and only report $\min_i n_i$ for convenience. The output random variable is the average sojourn time (waiting time plus service time) of the first 10 customers entering the system (denoted by $\bar{W}_{10}$).

\item The computer communication network simulation model that has also been used in \citet{cheng1997sensitivity,lin2015single,lam2022subsampling, lam2023bootstrap}. This network contains four nodes and four edges where nodes denote message processing units and edges are transport channels (see Figure \ref{network} in Appendix \ref{sec:details_network}). For every pair of nodes $i,j$ ($i\neq j$), there are external messages which enter into node $i$ from the external and are to be transmitted to node $j$ through a prescribed path. Their arrival time follows a Poisson process with parameter $\lambda_{i,j}$ specified in Table \ref{network_parameter1} in Appendix \ref{sec:details_network}. All the message lengths (unit: bits) are i.i.d. following a common exponential distribution with mean $300$ bits. Suppose each unit spends $0.001$ second to process a message passing it. We assume the node storage is unlimited but the channel storage is restricted to $275000$ bits. Message speed in transport channels is $150000$ miles per second and channel $i$ has length $100i$ miles. Therefore, it takes $l/275000+100i/150000$ seconds for a message with length $l$ bits to pass channel $i$. Suppose the network is empty at the beginning. The output random variable is the average delay for the first 30 messages (denoted by $\bar{W}_{30}$) where delay means the time from the entering node to the destination node. This example has $13$ unknown input distributions, i.e., $12$ inter-arrival time distributions $\mathrm{Exp}(\lambda_{i,j})$ and one message length distribution $\mathrm{Exp}(1/300)$. We keep the data sizes of the 13 input distributions proportional to each other so that the balanced data Assumption \ref{balanced_data} holds and only report $\min_i n_i$.
\end{enumerate}


\subsection{Statistical Validity and Performance Comparisons of Confidence Bands}\label{sec:num_verify}
Our target is to construct confidence bands for the output distribution functions of the above two simulation models. We compare our methodology with both the classic KS confidence band, and an approach that uses the standard bootstrap instead of our proposed subsampling. These comparisons serve to demonstrate our coverage validity compared to the shortfall in using the classic KS and the standard bootstrap under limited simulation budget. In addition, we test various algorithmic configurations of our methodology and investigate whether the optimal order of algorithmic configuration (subsample ratio $\theta$, subsampling split between outer loop $B$ and inner loop $R_s$) suggested in Theorem \ref{overall_opt_config} leads to a better coverage probability in practice.

To calculate the coverage probability of an estimated confidence band, we first run the simulation models driven by the true input distributions $100000$ times and use the empirical output distribution function as an accurate proxy of the true output distribution function. We then perform 500 independent runs of the confidence band construction procedure and report the percentage of runs in which the constructed confidence band covers the true output distribution function entirely (i.e., the empirical coverage probability of the confidence band). For each tested example, we vary the minimum input data size $\min n_i\in\{500,1000,2000\}$, the subsample ratio $\theta\in\{0.01, 0.02, 0.03, 1\}$, $R_s\in\{10,30,100\}$ and $B=\lfloor N/R_s \rfloor$ with simulation budget $N=1000$. Note that in particular $\theta=1$ corresponds to using the standard bootstrap. The discrete grid $\{t_1,\ldots,t_k\}$ in Algorithms \ref{alg:subsample} and \ref{alg:cdf_cb} is chosen as the uniform grid on $[0,10]$ with $k=100$ for the M/M/1 queue and the uniform grid on $[0,0.04]$ with $k=100$ for the computer communication network, where the intervals $[0,10]$ and $[0,0.04]$ are chosen to cover most of the mass of the output random variable. The number of simulation runs $R$ in Algorithm \ref{alg:cdf_cb} is chosen to be equal to $\min n_i$ such that the balanced noise condition in Assumption \ref{balanced_randomness} holds.


Figures \ref{fig:queue_n_500}-\ref{fig:queue_n_2000} report results for the M/M/1 queue and Figures \ref{fig:network_n_500}-\ref{fig:network_n_2000} report the counterpart for the computer communication network, where the dashed lines denote the 95\% nominal confidence level. Besides, in order to compare our input-uncertainty-inflated KS confidence bands with the classic KS confidence bands, we run both methods for each configuration of $\theta,B$ and $R_s$. However, since the classic KS confidence band does not rely on these parameters, it displays the same width in each figure and the variability of its coverage probability is merely due to the Monte Carlo errors of different experimental runs. Looking at the coverage probabilities, we observe that the classic KS confidence bands significantly under-cover the true output distribution functions. On the other hand, our method, but using a standard bootstrap instead of subsampling (i.e., $\theta=1$), significantly over-cover the distribution functions. This observation is also consistent with the large widths when using the standard bootstrap. All these results suggest that both the classic KS and standard bootstrap fail to give statistically valid confidence bands. In contrast, our input-uncertainty-inflated KS confidence bands using suitably configured subsampling have coverage probabilities close to the nominal level. In particular, we observe that for each minimum data size $\min n_i$, the best configuration $(\theta^*,R_s^*)$, i.e., the configuration with a coverage probability that is closest to 95\%, satisfies: 1) $R_s^*$ are all around 30 and 2) $\theta^*$ is roughly decreasing as $\min n_i$ increases. Both observations are consistent with the conclusion in Theorem \ref{overall_opt_config} which suggests that the optimal $R_s$ should be of order $N^{1/3}(s^*)^{2/3}=N^{1/3}(N^{1/4})^{2/3}=N^{1/2}\approx 31$ and the optimal $\theta^*=\Theta(N^{1/4}/n)$ is decreasing as the average data size $n$ increases. With the best configuration $(\theta^*,R_s^*)$, we can see the corresponding coverage probabilities are very close to the 95\% nominal level. 

\begin{figure}[ht]
    \centering
    \includegraphics[scale=0.7]{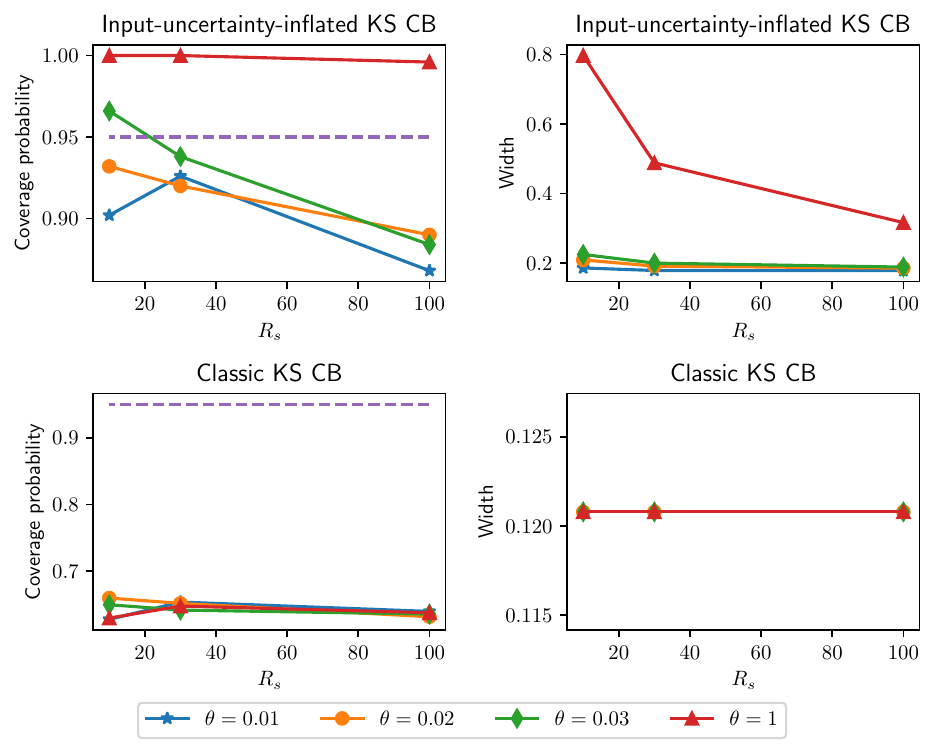}
    \caption{M/M/1 queue with minimum data size $\min n_i=500$. $\theta=1$ corresponds to the standard bootstrap instead of our subsampling bootstrap. 
    %
    }
    \label{fig:queue_n_500}
\end{figure}
\begin{figure}[ht]
    \centering
    \includegraphics[scale=0.7]{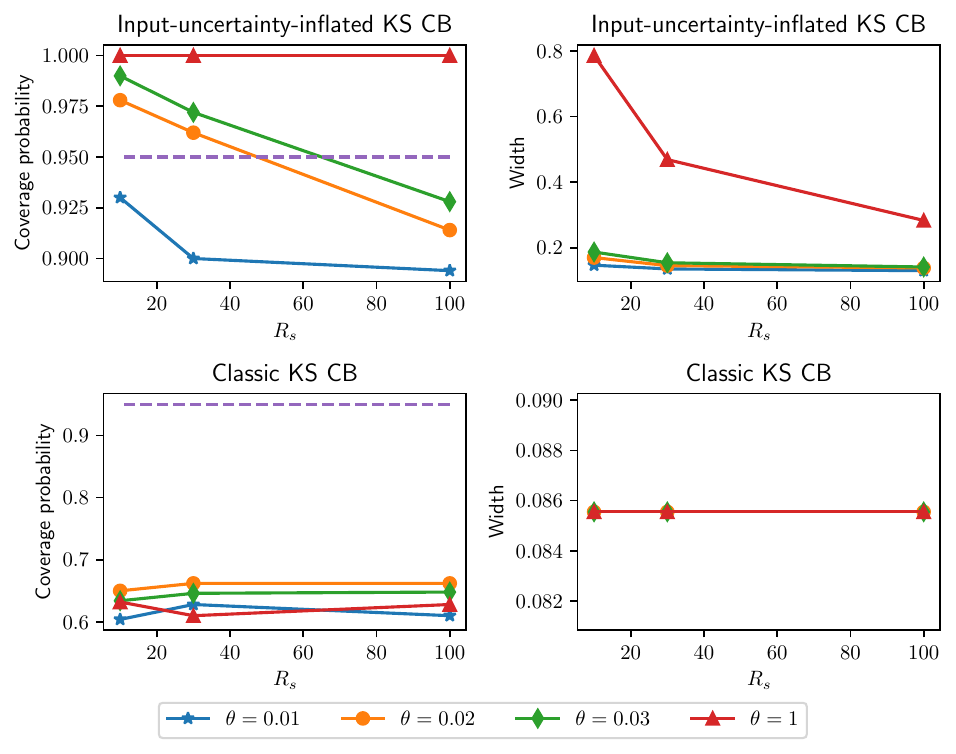}
    \caption{M/M/1 queue with minimum data size $\min n_i=1000$. $\theta=1$ corresponds to the standard bootstrap instead of our subsampling bootstrap. 
    }
    \label{fig:queue_n_1000}
\end{figure}
\begin{figure}[ht]
    \centering
    \includegraphics[scale=0.7]{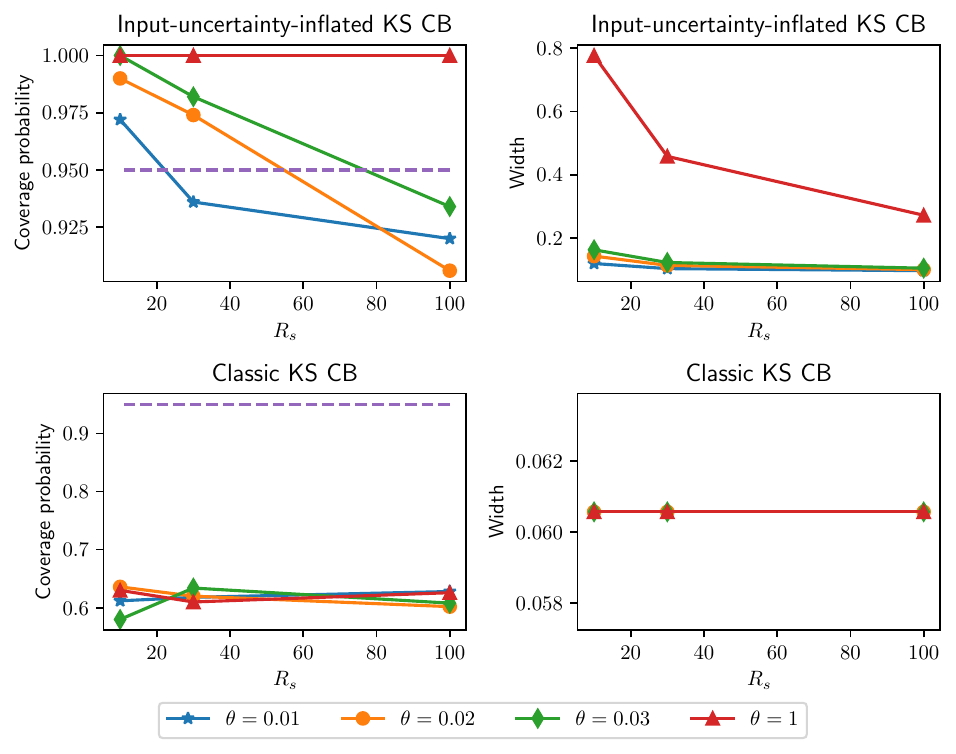}
    \caption{M/M/1 queue with minimum data size $\min n_i=2000$. $\theta=1$ corresponds to the standard bootstrap instead of our subsampling bootstrap. 
    }
    \label{fig:queue_n_2000}
\end{figure}

\begin{figure}
    \centering
    \includegraphics[scale=0.7]{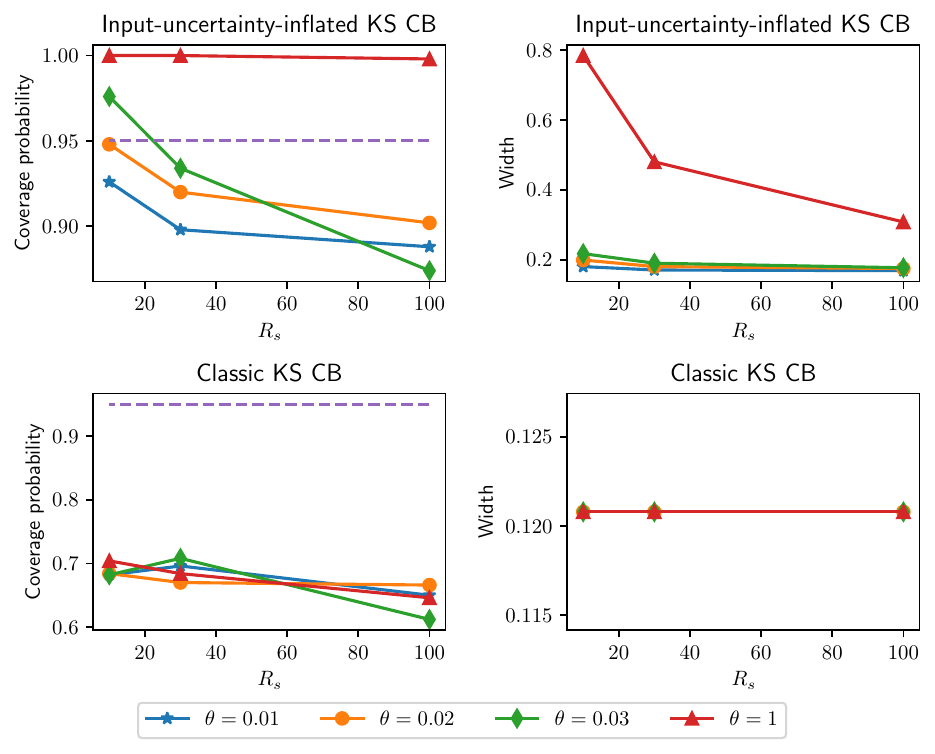}
    \caption{Computer communication network simulation model with minimum data size $\min n_i=500$. $\theta=1$ corresponds to the standard bootstrap instead of our subsampling bootstrap. 
    }
    \label{fig:network_n_500}
\end{figure}
\begin{figure}
    \centering
    \includegraphics[scale=0.7]{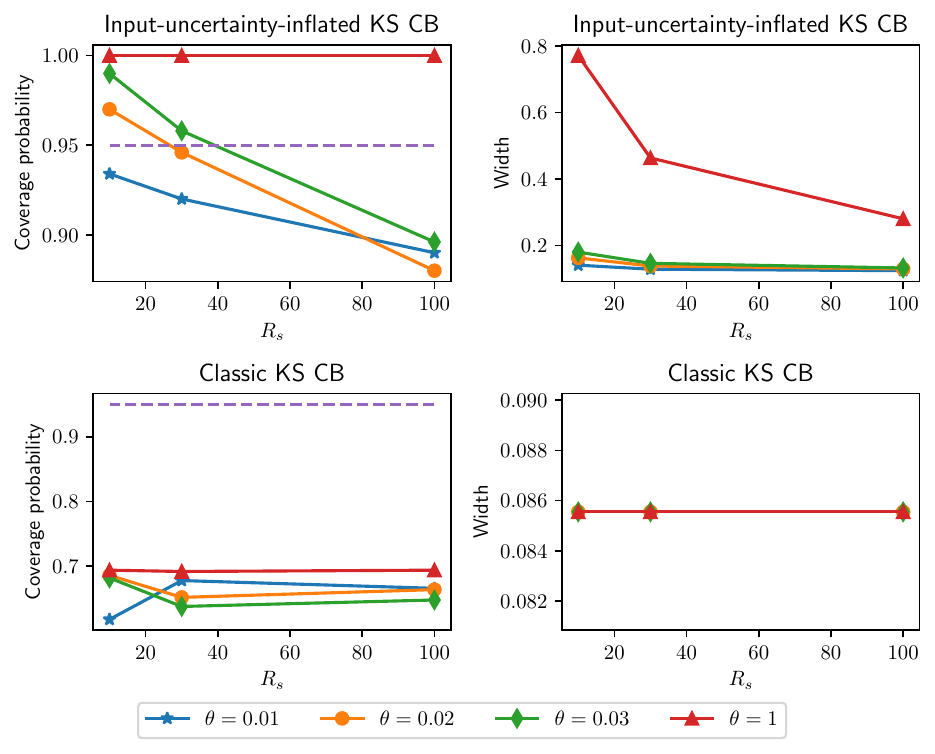}
    \caption{Computer communication network simulation model with minimum data size $\min n_i=1000$. $\theta=1$ corresponds to the standard bootstrap instead of our subsampling bootstrap. 
    }
    \label{fig:network_n_1000}
\end{figure}
\begin{figure}
    \centering
    \includegraphics[scale=0.7]{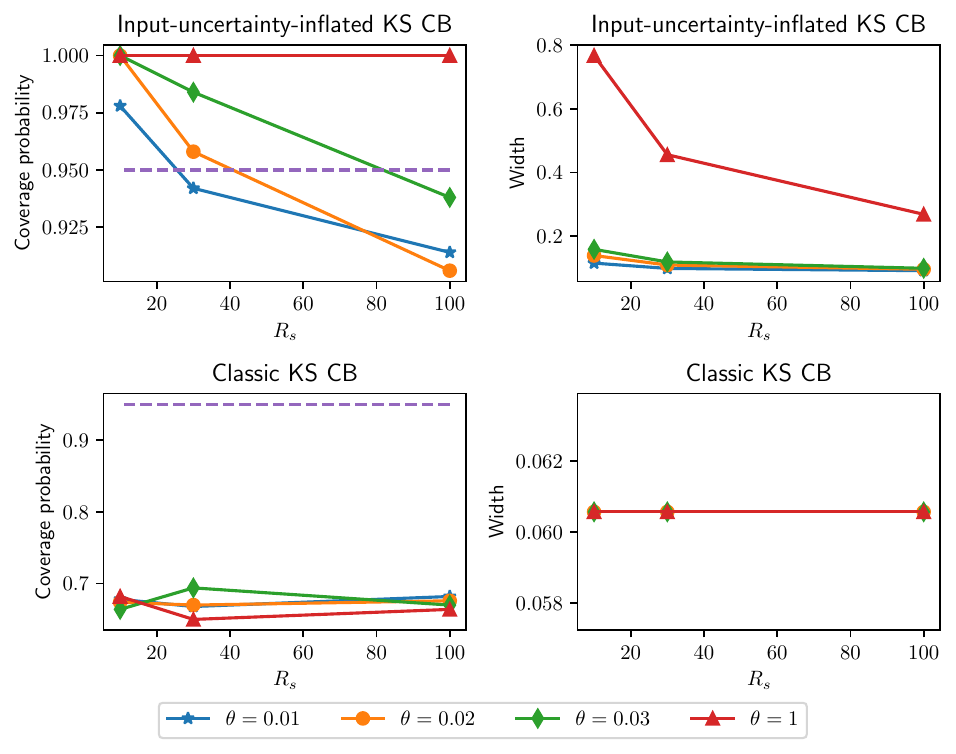}
    \caption{Computer communication network simulation model with minimum data size $\min n_i=2000$. $\theta=1$ corresponds to the standard bootstrap instead of our subsampling bootstrap. 
    }
    \label{fig:network_n_2000}
\end{figure}

\subsection{Extracting Simultaneous Confidence Regions for Quantiles}\label{sec:quantiles}

We illustrate how our input-uncertainty-inflated KS confidence bands capture distributional uncertainty information, specifically by using them to extract simultaneous confidence regions for output quantiles. More concretely, let $q_{0.25},q_{0.5},q_{0.75}$ be the 25\%, 50\%, 75\% quantiles of the output. We aim at constructing a simultaneous confidence region $[\hat{q}_{0.25,l},\hat{q}_{0.25,u}]\times[\hat{q}_{0.5,l},\hat{q}_{0.5,u}]\times[\hat{q}_{0.75,l},\hat{q}_{0.75,u}]$ so that
\[
P(\hat{q}_{0.25,l}\le q_{0.25}\le \hat{q}_{0.25,u}, \hat{q}_{0.5,l}\le q_{0.5}\le \hat{q}_{0.5,u}, \hat{q}_{0.75,l}\le q_{0.75}\le \hat{q}_{0.75,u})\approx 1-\alpha.
\]
To achieve this, we note that, given a confidence band $L_{IU}(t)\le F(t) \le U_{IU}(t) \, \forall t\in\mathbb{R}$ with a nominal level $1-\alpha$, we can invert the involved functions to obtain a simultaneous confidence band for the entire quantile function:
\[
\hat{q}_{s,u}:=\inf\{t:L_{IU}(t)\ge s\} \ge q_{s}\equiv\inf\{t:F(t)\ge s\} \ge \inf\{t:U_{IU}(t)\ge s\}=:\hat{q}_{s,l},\forall s\in(0,1).
\]
Setting $s=0.25, 0.5, 0.75$, we can then obtain the desired simultaneous confidence region for the quantiles $q_{0.25},q_{0.5},q_{0.75}$. We caution that this obtained simultaneous confidence region may over-cover the true quantiles because the following inequality might be strict:
\[
P(\hat{q}_{s,u} \ge q_{s} \ge \hat{q}_{s,l},\forall s\in(0,1))\ge P(L_{IU}(t)\le F(t) \le U_{IU}(t), \forall t\in\mathbb{R}).
\]
However, given that constructing even just one asymptotically valid confidence interval for a single quantile for simulation output under input uncertainty is still open to our best knowledge (the only related work we know is \cite{parmar2022input}, which proposes methods based on the bootstrap or Taylor series but without any theoretical guarantees), our approach provides a reasonable approach towards this direction with an asymptotically lower bound (i.e., $1-\alpha$). Moreover, as the number of target quantiles increases, we expect our confidence regions to get increasingly tighter.

We conduct numerical experiments for the two simulation models to examine our performance. We use the same simulation setups but fix the configuration as the optimal one: $(B,R_s)=(33,30)$, $\theta=0.03$ if $\min n_i=500$, $\theta=0.02$ if $\min n_i=1000$, and $\theta=0.01$ if $\min n_i=2000$. Figure \ref{fig:quantile} displays the results, where the dashed lines denote the 95\% nominal level. We can see that the confidence regions obtained by our method have approximately the correct confidence level while the regions obtained by the classic KS method significantly under-cover the true quantiles. 


\begin{figure}
    \centering
    \subfloat[M/M/1 queue]{\includegraphics[scale=0.5]{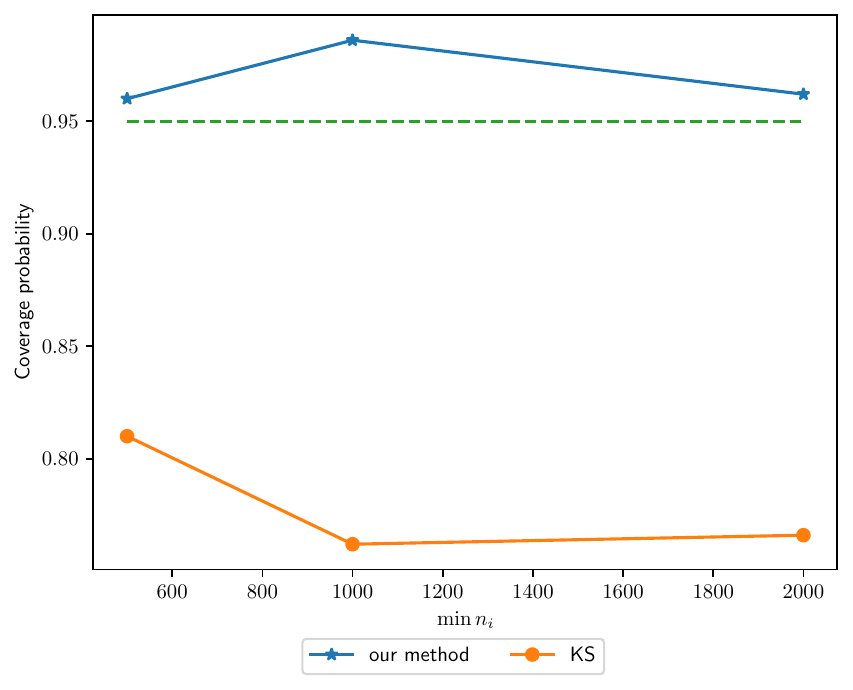}}
    \subfloat[Computer network]{\includegraphics[scale=0.5]{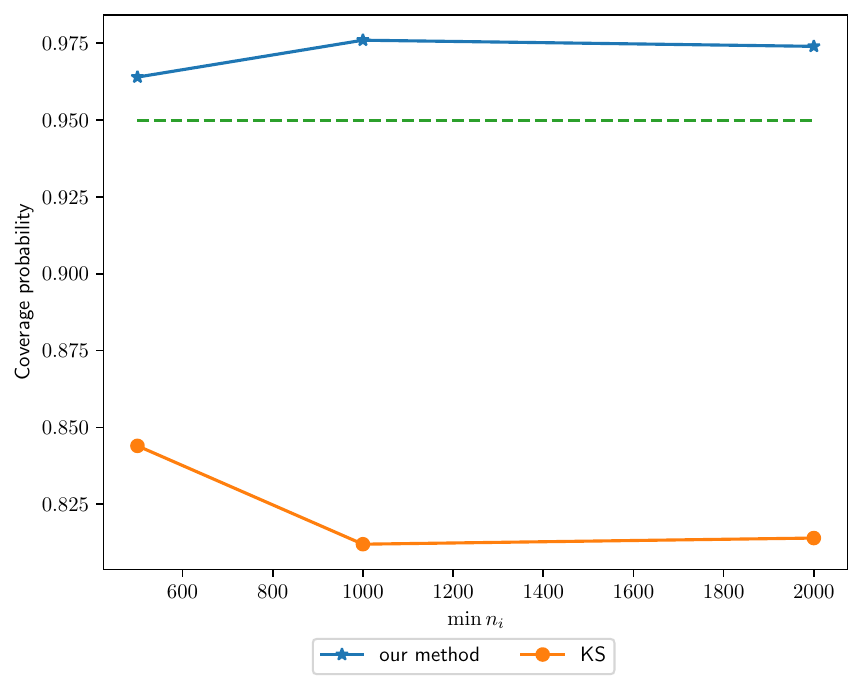}}
    \caption{Empirical coverage probabilities of the simultaneous confidence regions of the quantiles.}
    \label{fig:quantile}
\end{figure}


\section{Conclusion and Discussion}

This paper studies a new simulation input uncertainty problem, viewed at the distributional level where the goal is to construct a confidence band for the entire distribution function of the output random variable. Existing literature in input uncertainty mainly focuses on confidence intervals for real-valued summary statistics of the outputs, more specifically expected values, that account for both input data noise and Monte Carlo noise. While convenient, such an approach could hide valuable output distributional information that is relevant for downstream decision-making. Our methodology provides a more holistic assessment of input uncertainty and, moreover, can be converted into simultaneous confidence regions for many performance measures that are possibly nonlinear in the input distributions. 

Our methodology entails a generalization of the classic KS statistic to properly inflate its uncertainty measurement to include input uncertainty. This new statistic exhibits an asymptotic that comprises the addition of the Brownian bridge inherited from the classic KS statistic with a mean-zero Gaussian process that captures the additional input data noise. In particular, the Gaussian process bears a covariance structure governed by the influence function of the output distribution, which measures the sensitivity with respect to the input distributions. Besides mathematical analyses, we propose an efficient method to approximate the limiting distribution and thereby implement the confidence band, by utilizing a subsampling bootstrap to estimate the covariance function of the Gaussian process. Technically, the  weak convergence analyses of our new statistic leverage the V-process theory and the Koml\'{o}s-Major-Tusn\'{a}dy approximation to handle the entanglements between the input and Monte Carlo noises, as well as a coupling technique to apply the continuous mapping theorem to the supremum operator needed in our construction. Moreover, we derive the statistical guarantees of our subsampling scheme and analyze its optimal algorithmic configuration. 



Our work opens a new line of research on input uncertainty that considers more general target estimation quantities and takes a more holistic, distributional view of simulation outputs. There are many open questions, with some examples below. First, we have so far considered single-dimensional outputs and have demonstrated the challenges and novelties in this setting. A natural direction is to extend to multi-dimensional output distributions and, relatedly, high-dimensional (even just real) outputs where other types of statistics and analysis tools need to be developed. Second, like the classical KS counterpart, our input-uncertainty-inflated KS confidence band has the same width over all points. A direction is to investigate alternative approaches that allow for different widths across points, for example wider for tail regions where less data are available. A third direction is to extend our model settings, for example from finite horizon to steady state, and to the incorporation of dependent input processes.

\section*{Acknowledgments}
We gratefully acknowledge support from the National Science Foundation under grants CAREER CMMI-1834710 and IIS-1849280. A preliminary conference version of this work has appeared in \cite{chen2022distributional}.

\bibliographystyle{plainnat}
\bibliography{distributional.bib}

\clearpage

\newpage

\begin{appendix}
\section{Weak Convergence on the Skorohod Product Space}

\label{sec:Skorohod}

In this section, we give a brief survey of the Skorohod (product) space and a
few probability theory concepts related to it. Most of the materials are taken
from \cite{billingsley2013convergence} and \cite{ferger2010weak}.

\subsection{Skorohod Space}

Let $T$ be a closed (possible infinite) interval on $\mathbb{R}$. Three
representatives are $T=[0,1]$, $T=[0,\infty)$ and $T=(-\infty,\infty)$. A
function $x(t)$ defined on $T$ is called \emph{c\`{a}dl\`{a}g} if it is right
continuous with left limits for any $t\in T$. The Skorohod space
$D(T)$ is defined as the space of all c\`{a}dl\`{a}g functions on
$T$. The spaces $D[0,1]$ and $D[0,\infty)$ are discussed
in details in \cite{billingsley2013convergence} Chapter 3. For our purpose, we
only discuss the space $D_{\pm\infty}:=D(-\infty,\infty)$.
As said in \cite{ferger2010weak}, all the following results about
$D_{\pm\infty}$ are direct generalizations of results for
$D[0,\infty)$ in \cite{billingsley2013convergence} Section 16.

We first define the convergence in $D_{\pm\infty}$. Let $\Lambda$ be
the space of all strictly increasing, continuous and surjective functions from
$\mathbb{R}$ onto $\mathbb{R}$. We say a sequence of functions $\{x_{n}%
\}\subset D_{\pm\infty}$ converges to $x\in D_{\pm\infty}$
if there exists a sequence $\{\lambda_{n}\}\subset\Lambda$ such that
\begin{equation}
\sup_{t\in\mathbb{R}}|\lambda_{n}(t)-t|\rightarrow0, \quad\sup_{t\in
[-m,m]}|x_{n}(\lambda_{n}(t))-x(t)|\rightarrow0, \forall m\in\mathbb{N}.
\label{conv_Skorohod}%
\end{equation}
In fact, this convergence is metrizable by a metric $d_{\pm\infty}^{\circ}$
on $D_{\pm\infty}$, i.e., $x_{n}\rightarrow x$ on $D_{\pm\infty}$ if and only if $d_{\pm\infty}^{\circ}(x_{n},x)\rightarrow0$. The
metric $d_{\pm\infty}^{\circ}$ satisfies $d_{\pm\infty}^{\circ}(x,y)\leq
\sup_{t\in\mathbb{R}}|x(t)-y(t)|$ for $x,y\in D_{\pm\infty}$.
Besides, from the definition (\ref{conv_Skorohod}), we know that topology
induced by $d_{\pm\infty}^{\circ}$ is weaker than the uniform topology induced
by locally uniform convergence. However, the advantage of using this kind of
convergence or metric is that $(D_{\pm\infty},d_{\pm\infty}^{\circ
})$ is a complete separable metric space so that the general weak convergence
theory can be applied.

\subsection{Weak Convergence on the Skorohod Space}

In this section, we present a few results about the weak convergence on the
Skorohod space.

We first define and characterize the measurability on the Skorohod space. Let
$(\Omega,\mathcal{F},\mathbb{P})$ be a probability space on which a collection
of random variables $X(t),t\in\mathbb{R}$ is defined. We endow the Skorohod
space $D_{\pm\infty}$ with its Borel $\sigma$-field $\mathcal{D}_{\pm\infty}$
to obtain the measurable space $(D_{\pm\infty},\mathcal{D}_{\pm\infty})$. The
stochastic process $X(\cdot)$ is measurable (also called a \emph{random element}) on
$(D_{\pm\infty},\mathcal{D}_{\pm\infty})$ if (i) for each $\omega\in\Omega$,
$X(\cdot,\omega)$ is a c\`{a}dl\`{a}g function; (ii) $X$ is a measurable map
from $(\Omega,\mathcal{F})$ to $(D_{\pm\infty},\mathcal{D}_{\pm\infty})$,
i.e., for any $B\in\mathcal{D}_{\pm\infty}$, $\{\omega:X(\cdot,\omega)\in
B\}:=X^{-1}(B)\in\mathcal{F}$. Usually for a given stochastic process $X$,
condition (ii) is hard to directly verify since the pre-image $X^{-1}(B)$ is
not easy to characterize. Luckily, the structure of the c\`{a}dl\`{a}g
function admits a simple verification based on each single random variable
$X(t)$. For $k\in\mathbb{N}$ and $-\infty<t_{1}<\cdots<t_{k}<\infty$, let
$\pi_{t_{1},\ldots,t_{k}}:(D_{\pm\infty},\mathcal{D}_{\pm\infty}%
)\mapsto(\mathbb{R}^{k},\mathcal{B}(\mathbb{R}^{k}))$ be the projection at
$\{t_{1},\ldots,t_{k}\}$, i.e., $\pi_{t_{1},\ldots t_{k}}(x)=(x(t_{1}%
),\ldots,x(t_{k}))$ for $x\in D_{\pm\infty}$. For a subset $T\subset
(-\infty,\infty)$, let $\sigma(\pi_{t},t\in T)$ be the $\sigma$-field
generated by all pre-images $\pi_{t}^{-1}(B),t\in T,B\in
\mathcal{B}(\mathbb{R})$. We have

\begin{theorem}
\label{equivalence_from_projection}If $T$ is dense in $(-\infty,\infty)$, then
$\sigma(\pi_{t},t\in T)=\mathcal{D}_{\pm\infty}$.
\end{theorem}

Theorem \ref{equivalence_from_projection} ensures that condition (ii) is
automatic if $X(t),t\in T$ are random variables for a dense $T\subset
(-\infty,\infty)$.

Now we define the weak convergence on the Skorohod space. Let $X_{n},n\geq1$
and $X$ be random elements on $(D_{\pm\infty},\mathcal{D}_{\pm\infty})$. We
say $X_{n}$ weakly converges to $X$, denoted by $X_{n}\Rightarrow X$ if
$\mathbb{E}[f(X_{n})]\rightarrow\mathbb{E}[f(X)]$ for any bounded continuous
function $f:D_{\pm\infty}\mapsto\mathbb{R}$. Thanks to the completeness and
separability of the Skorohod space, the famous Prokhorov's theorem can be
applied to obtain the following (sufficient) condition of the weak convergence.

\begin{theorem}
Let $X_{n},n\geq1$ and $X$ be random elements on $(D_{\pm\infty}%
,\mathcal{D}_{\pm\infty})$. If the following two conditions hold:
\begin{description}
\item[(1)] $\{X_{n}\}$ is tight, i.e., for any $\varepsilon>0$, there exists a
compact set $K\subset D_{\pm\infty}$ such that $\mathbb{P}(X_{n}\in
K)>1-\varepsilon$ for any $n$;

\item[(2)] the finite dimensional distributions of $X_{n}$ weakly converges to
the finite dimensional distributions of $X$ (denoted by $X_{n}%
\overset{f.d.}{\rightarrow}X$), i.e., $(X_{n}(t_{1}),\ldots,X_{n}%
(t_{k}))\overset{d}{\rightarrow}(X(t_{1}),\ldots,X(t_{k}))$ for any
$k\in\mathbb{N}$ and $-\infty<t_{1}<\cdots<t_{k}<\infty$,
\end{description}
then $X_{n}\Rightarrow X$. Besides, if $X_{n}\Rightarrow X$, then $\{X_{n}\}$
is tight.
\end{theorem}

There are a few tightness characterization based on (modified) modulus of
continuity or moment conditions. Interested readers may refer to
\cite{ferger2010weak} for more details.

\subsection{Weak Convergence on the Skorohod Product Space}

Based on the weak convergence on the Skorohod space, the weak convergence on
the Skorohod product space $D_{\pm\infty}\times D_{\pm\infty}$ is just a
natural generalization.

We endow $D_{\pm\infty}\times D_{\pm\infty}$ with the product topology induced
by the metric
\[
\rho((x_{1},y_{1}),(x_{2},y_{2}))=d_{\pm\infty}^{\circ}(x_{1},x_{2})\vee
d_{\pm\infty}^{\circ}(y_{1},y_{2}),\quad(x_{i},y_{i})\in D_{\pm\infty}\times
D_{\pm\infty}.
\]
It is well-known that $(D_{\pm\infty}\times D_{\pm\infty},\rho)$ is also a
complete separable metric space and its Borel $\sigma$-field is exactly the
product $\sigma$-field $\mathcal{D}_{\pm\infty}\times\mathcal{D}_{\pm\infty}$.
In this case, the two-dimensional stochastic process $(X,Y)$ is measurable on
$(D_{\pm\infty}\times D_{\pm\infty},\mathcal{D}_{\pm\infty}\times
\mathcal{D}_{\pm\infty})$ if and only if $X$ and $Y$ are both measurable on
$(D_{\pm\infty},\mathcal{D}_{\pm\infty})$.

The weak convergence is defined similarly. Let $(X_{n},Y_{n}),n\geq1$ and
$(X,Y)$ be random elements on $(D_{\pm\infty}\times D_{\pm\infty}%
,\mathcal{D}_{\pm\infty}\times\mathcal{D}_{\pm\infty})$. We say $(X_{n}%
,Y_{n})$ weakly converges to $(X,Y)$, denoted by $(X_{n},Y_{n})\Rightarrow
(X,Y)$ if $\mathbb{E}[f(X_{n},Y_{n})]\rightarrow\mathbb{E}[f(X,Y)]$ for any
bounded continuous function $f:D_{\pm\infty}\times D_{\pm\infty}%
\mapsto\mathbb{R}$. Moreover, we have the following sufficient condition:

\begin{theorem}
Let $(X_{n},Y_{n}),n\geq1$ and $(X,Y)$ be random elements on $(D_{\pm\infty
}\times D_{\pm\infty},\mathcal{D}_{\pm\infty}\times\mathcal{D}_{\pm\infty})$.
If the following two conditions hold:
\begin{description}
\item[(1)] $\{X_{n}\}$ and $\{Y_{n}\}$ are both individually tight;

\item[(2)] $(X_{n},Y_{n})\overset{f.d.}{\rightarrow}(X,Y)$, i.e.,
\[
(X_{n}(t_{1}),\ldots,X_{n}(t_{k}),Y_{n}(t_{1}),\ldots,Y_{n}(t_{k}))\overset{d}{\rightarrow}(X(t_{1}),\ldots,X(t_{k}),Y(t_{1}),\ldots,Y(t_{k}))
\]
for any $k\in\mathbb{N}$ and $-\infty<t_{1}<\cdots<t_{k}<\infty$,
\end{description}
then $(X_{n},Y_{n})\Rightarrow(X,Y)$.
\end{theorem}

\section{General Assumptions and Results for Covariance Function Estimation}\label{sec:general_assumptions}

In this section, we restate theorems regarding the covariance function estimation under more general assumptions than the finite-horizon model. We first present these assumptions, most of which are related to the smoothness of the output distribution function $Q(t,\underline{\hat{P}})$ and $Q(t,\underline{\hat{P}}_{\theta}^{\ast})$. In fact, all of them can be proved under Assumptions \ref{balanced_data} and \ref{finite_horizon_model} (see Theorem \ref{verification_general_assumptions} below). 

\begin{assumption}[Consistency of the output distribution function]\label{convergence_in_p_to_truth}
For any $t\in\mathbb{R}$,
$Q(t,\underline{\hat{P}})$ is measurable with respect to the $\sigma$-field
generated by the data and $Q(t,\underline{\hat{P}})\overset{p}{\rightarrow}%
{Q}(t,\underline{P})$ as $n\rightarrow\infty$.
\end{assumption}

\begin{assumption}[First order expansion at true input distributions]
\label{1st_expansion_truth}The output distribution function
$Q(t,\underline{\hat{P}})$ satisfies the following first order expansion%
\[
Q(t,\underline{\hat{P}})=Q(t,\underline{P})+\sum_{i=1}^{m}\int IF_{i}%
(t,x;\underline{P})d(\hat{P}_{i}(x)-P_{i}(x))+\varepsilon(t),
\]
where the influence function and the remainder term satisfy%
\[
\int IF_{i}(t,x;\underline{P})dP_{i}(x)=\mathbb{E}_{P_{i}}[IF_{i}%
(t,X_{i};\underline{P})]=0,\forall t\in\mathbb{R},
\]%
\[
\mathbb{E}_{P_{i}}[IF_{i}^{4}(t,X_{i};\underline{P})]<\infty,\forall
t\in\mathbb{R},
\]
and%
\[
\mathbb{E}[\varepsilon(t)^{2}]=o(n^{-1}),\forall t\in\mathbb{R}.
\]
\end{assumption}

\begin{assumption}[First order expansion at empirical input distributions]
\label{1st_expansion_empirical} The output distribution function $Q(t,\underline{\hat{P}}_{\theta}^{\ast})$ satisfies the following first order expansion%
\[
Q(t,\underline{\hat{P}}_{\theta}^{\ast})=Q(t,\underline{\hat{P}})+\sum_{i=1}^{m}\int IF_{i}(t,x;\underline{\hat{P}})d(\hat{P}_{i,s_{i}}^{\ast}-\hat{P}_{i})(x)+\varepsilon^{\ast}(t),
\]
where the influence function and the remainder term satisfy%
\[
\mathbb{E}_{\ast}[IF_{i}(t,X_{i,1}^{\ast};\underline{\hat{P}})]=\int IF_{i}(t,x;\underline{\hat{P}})d\hat{P}_{i}(x)=\sum_{j=1}^{n_{i}}\frac{1}{n_{i}}IF_{i}(t,X_{i,j};\underline{\hat{P}})=0,\forall t\in\mathbb{R},
\]%
\[
\mathbb{E}[(IF_{i}(t,X_{i};\underline{\hat{P}})-IF_{i}(t,X_{i};\underline{P}))^{4}]\rightarrow0,\forall t\in\mathbb{R},
\]%
\[
\mathbb{E}_{\ast}[\varepsilon^{\ast}(t)]=O_{p}((\theta n)^{-1}),\forall t\in\mathbb{R},%
\]%
\[
\mathbb{E}_{\ast}[(\varepsilon^{\ast}(t)-\mathbb{E}_{\ast}[\varepsilon^{\ast}(t)])^{2}]=o_{p}((\theta n)^{-1}),\forall t\in\mathbb{R},
\]%
\[
\mathbb{E}_{\ast}[\varepsilon^{\ast}(t)^{4}]=o_{p}((\theta n)^{-2}),\forall t\in\mathbb{R}.
\]
\end{assumption}

In section \ref{sec:theory_cov}, to study the order of $\theta$ in $\sigma
^{2}(t,t^{\prime})-\mathrm{Cov}(\mathbb{G}(t),\mathbb{G}(t^{\prime}))$, we
need a more precise expansion of the output distribution function. In the
following, we define and present the third order expansion of $Q(t,\cdot)$ with
respect to the input distributions.

\begin{definition}
[Third Order Expansion]Let $\underline{P}^{\prime}=(P_{1}^{\prime}%
,\ldots,P_{m}^{\prime})$ and $\underline{P}^{\prime\prime}=(P_{1}%
^{\prime\prime},\ldots,P_{m}^{\prime\prime})$ be two arbitrary input
distributions. For $v=(v_{1},\ldots,v_{m})\in\lbrack0,1]^{m}$, let
$\underline{P}^{v}=((1-v_{1})P_{1}^{\prime}+v_{1}P_{1}^{\prime\prime}%
,\ldots,(1-v_{m})P_{m}^{\prime}+v_{m}P_{m}^{\prime\prime})$. If there exist
functions $IF_{i_{1},i_{2}}(t,x_{1},x_{2};\underline{P}^{\prime})$ and
$IF_{i_{1},i_{2},i_{3}}(t,x_{1},x_{2},x_{3};\underline{P}^{\prime})$ for any
$1\leq i_{1},i_{2},i_{3}\leq m$ such that:

(i) they are symmetric under
permutations, i.e.,
\[
IF_{i_{1},i_{2}}(t,x_{1},x_{2};\underline{P}^{\prime})=IF_{i_{\pi(1)}%
,i_{\pi(2)}}(t,x_{\pi(1)},x_{\pi(2)};\underline{P}^{\prime})\text{ for any
permutation }\pi\text{ on }\{1,2\},
\]%
\[
IF_{i_{1},i_{2},i_{3}}(t,x_{1},x_{2},x_{3};\underline{P}^{\prime}%
)=IF_{i_{\pi(1)},i_{\pi(2)},i_{\pi(3)}}(t,x_{\pi(1)},x_{\pi(2)},x_{\pi
(3)};\underline{P}^{\prime})\text{ for any permutation }\pi\text{
on }\{1,2,3\},
\]
(ii) they have marginal zero mean under $\underline{P}^{\prime}$, i.e.,%
\[
\mathbb{E}_{P_{i_{2}}^{\prime}}[IF_{i_{1},i_{2}}(t,x_{1},X_{2}%
;\underline{P}^{\prime})]=0,\forall x_{1}\in\mathbb{R},
\]%
\[
\mathbb{E}_{P_{i_{3}}^{\prime}}[IF_{i_{1},i_{2},i_{3}}(t,x_{1},x_{2}%
,X_{3};\underline{P}^{\prime})]=0,\forall x_{1},x_{2}\in\mathbb{R},
\]
(iii) the following expansion holds%
\begin{align*}
&  Q(t,\underline{P}^{v})\\
&  =Q(t,\underline{P}^{\prime})+\sum_{i=1}^{m}v_{i}\int IF_{i}%
(t,x;\underline{P}^{\prime})d(P_{i}^{\prime\prime}-P_{i}^{\prime})(x)\\
&  +\frac{1}{2}\sum_{i_{1},i_{2}=1}^{m}v_{i_{1}}v_{i_{2}}\int IF_{i_{1},i_{2}%
}(t,x_{1},x_{2};\underline{P}^{\prime})d(P_{i_{1}}^{\prime\prime}-P_{i_{1}%
}^{\prime})(x_{1})d(P_{i_{2}}^{\prime\prime}-P_{i_{2}}^{\prime})(x_{2})\\
&  +\frac{1}{6}\sum_{i_{1},i_{2},i_{3}=1}^{m}v_{i_{1}}v_{i_{2}}v_{i_{3}}\int
IF_{i_{1},i_{2},i_{3}}(t,x_{1},x_{2},x_{3};\underline{P}^{\prime})d(P_{i_{1}%
}^{\prime\prime}-P_{i_{1}}^{\prime})(x_{1})d(P_{i_{2}}^{\prime\prime}%
-P_{i_{2}}^{\prime})(x_{2})d(P_{i_{3}}^{\prime\prime}-P_{i_{3}}^{\prime
})(x_{3})\\
&  +o\left(  \left(  \sum_{i=1}^{m}v_{i}^{2}\right)  ^{3/2}\right)  ,
\end{align*}
then $IF_{i_{1},i_{2}}(t,x_{1},x_{2};\underline{P}^{\prime})$ and
$IF_{i_{1},i_{2},i_{3}}(t,x_{1},x_{2},x_{3};\underline{P}^{\prime})$ are
called the second order and third order influence functions of $Q(t,\cdot)$ at
$\underline{P}^{\prime}$, and the expansion is called the third order
expansion of $Q(t,\underline{P}^{v})$ at $\underline{P}^{\prime}$.
\end{definition}

\begin{assumption}[Higher order influence functions at true input distributions]
\label{3rd_IF_at_truth}The second order and third order
influence functions of $Q(t,\cdot)$ at $\underline{P}$ exist. They have finite
moments as follows:
\[
\mathbb{E}[IF_{i_{1}i_{2}}^{4}(t,X_{i_{1},1},X_{i_{2},j_{2}};\underline{P}%
)]<\infty,\quad\mathbb{E}[IF_{i_{1}i_{2}i_{3}}^{2}(t,X_{i_{1},1}%
,X_{i_{2},j_{2}},X_{i_{3},j_{3}};\underline{P})]<\infty,
\]
for any $1\leq i_{1},i_{2},i_{3}\leq m,1\leq j_{2}\leq2,1\leq j_{3}\leq3$ and
$\forall t\in\mathbb{R}$.
\end{assumption}

\begin{assumption}[Third order expansion at empirical input distributions]
\label{3rd_expansion_empirical}$Q(t,\underline{\hat{P}}_{\theta}^{\ast
})$ satisfies the following third order expansion for any $t\in\mathbb{R}$%
\begin{align*}
&  Q(t,\underline{\hat{P}}_{\theta}^{\ast})\\
&  =Q(t,\underline{\hat{P}})+\sum_{i=1}^{m}\int IF_{i}(t,x;\underline{\hat{P}%
})d(\hat{P}_{i,s_{i}}^{\ast}-\hat{P}_{i})(x)\\
&  +\frac{1}{2}\sum_{i_{1},i_{2}=1}^{m}\int IF_{i_{1}i_{2}}(t,x_{1}%
,x_{2};\underline{\hat{P}})d(\hat{P}_{i_{1},s_{i_{1}}}^{\ast}-\hat{P}_{i_{1}%
})(x_{1})d(\hat{P}_{i_{2},s_{i_{2}}}^{\ast}-\hat{P}_{i_{2}})(x_{2})\\
&  +\frac{1}{6}\sum_{i_{1},i_{2},i_{3}=1}^{m}\int IF_{i_{1}i_{2}i_{3}}%
(t,x_{1},x_{2},x_{3};\underline{\hat{P}})d(\hat{P}_{i_{1},s_{i_{1}}}^{\ast
}-\hat{P}_{i_{1}})(x_{1})d(\hat{P}_{i_{2},s_{i_{2}}}^{\ast}-\hat{P}_{i_{2}%
})(x_{2})d(\hat{P}_{i_{3},s_{i_{3}}}^{\ast}-\hat{P}_{i_{3}})(x_{3}%
)+\varepsilon_{3}^{\ast}(t),
\end{align*}
where the influence function and the remainder term satisfy%
\[
\mathbb{E}[(IF_{i_{1}i_{2}}(t,X_{i_{1},1},X_{i_{2},j_{2}};\underline{\hat{P}%
})-IF_{i_{1}i_{2}}(t,X_{i_{1},1},X_{i_{2},j_{2}};\underline{P}))^{2}%
]\rightarrow0,
\]%
\[
\mathbb{E}[(IF_{i_{1}i_{2}i_{3}}(t,X_{i_{1},1},X_{i_{2},j_{2}},X_{i_{3},j_{3}%
};\underline{\hat{P}})-IF_{i_{1}i_{2}i_{3}}(t,X_{i_{1},1},X_{i_{2},j_{2}%
},X_{i_{3},j_{3}};\underline{P}))^{2}]\rightarrow0,
\]%
\[
\mathbb{E}_{\ast}[\varepsilon_{3}^{\ast}(t)^{2}]=o_{p}\left(  s^{-3}\right),
\]
for any $1\leq i_{1},i_{2},i_{3}\leq m,1\leq j_{2}\leq2,1\leq j_{3}\leq3$ and
$\forall t\in\mathbb{R}$. Additionally, the first order influence function
$IF_{i}(t,X_{i,1};\underline{\hat{P}})$ admits the following Taylor expansion%
\begin{align*}
&  IF_{i}(t,X_{i,1};\underline{\hat{P}})\\
&  =IF_{i}(t,X_{i,1};\underline{P})+\sum_{i^{\prime}=1}^{m}\int IF_{ii^{\prime
}}(t,X_{i,1},x;\underline{P})d(\hat{P}_{i^{\prime}}-P_{i^{\prime}})(x)-\int
IF_{i}(t,x;\underline{P})d(\hat{P}_{i}-P_{i})(x)+\varepsilon_{IF_{i}%
}(t,X_{i,1})\\
&  =IF_{i}(t,X_{i,1};\underline{P})+\sum_{i^{\prime}=1}^{m}\int IF_{ii^{\prime
}}(t,X_{i,1},x;\underline{P})d\hat{P}_{i^{\prime}}(x)-\int IF_{i}%
(t,x;\underline{P})d\hat{P}_{i}(x)+\varepsilon_{IF_{i}}(t,X_{i,1})
\end{align*}
where $\varepsilon_{IF_{i}}(t,X_{i,1})$ satisfies $\mathbb{E}[(\varepsilon
_{IF_{i}}(t,X_{i,1}))^{2}]=o(1/n)$ for any $t\in\mathbb{R}$.
\end{assumption}

All assumptions in this section are general versions of the finite-horizon model to establish the results for covariance function estimation. In fact, we can prove all these assumptions under the finite-horizon model.

\begin{theorem}\label{verification_general_assumptions}
Under Assumptions \ref{balanced_data} and \ref{finite_horizon_model}, Assumptions \ref{convergence_in_p_to_truth}-\ref{3rd_expansion_empirical} hold.
\end{theorem}

Under these general assumptions, theorems regarding the covariance function hold as well. Since they could be of independent interest, we restate them in the following.


\begin{theorem}
\label{consistency_sigma2} Suppose Assumptions \ref{balanced_data}, \ref{1st_expansion_truth} and \ref{1st_expansion_empirical} hold. Then the same conclusion in Theorem \ref{consistency_sigma} holds.
\end{theorem}

\begin{theorem}
\label{MSE_sigma_hat2}Suppose Assumptions \ref{balanced_data}, \ref{1st_expansion_truth} and \ref{1st_expansion_empirical} hold and the subsample ratio $\theta$ is chosen such that $\theta n=\omega(1)$. Then the same conclusion in Theorem \ref{MSE_sigma_hat} holds.
\end{theorem}

\begin{theorem}
\label{validity_alg_cov2} Suppose Assumptions \ref{balanced_data}, \ref{1st_expansion_truth} and \ref{1st_expansion_empirical} hold. Moreover, suppose the configuration in Algorithm \ref{alg:subsample} satisfies%
\begin{equation*}
\theta=\omega(1/n),\quad B=\omega(1),\quad R_{s}=\omega(s).
\end{equation*}
Then the same conclusion in Theorem \ref{validity_alg_cov} holds.
\end{theorem}

\begin{corollary}
\label{min_budget2} Suppose Assumptions \ref{balanced_data}, \ref{1st_expansion_truth} and \ref{1st_expansion_empirical} hold. Then the same conclusion in Corollary \ref{min_budget} holds.
\end{corollary}

\begin{theorem}
\label{optimal_B_R2} Suppose Assumptions \ref{balanced_data}, \ref{1st_expansion_truth}, \ref{1st_expansion_empirical} and the non-degeneracy condition (\ref{non-degeneracy}) hold. Then the same conclusion in Theorem \ref{optimal_B_R} holds.
\end{theorem}

\begin{theorem}
\label{error_true_bootstrap2}Suppose Assumptions \ref{balanced_data}, \ref{1st_expansion_truth}, \ref{1st_expansion_empirical}, \ref{3rd_IF_at_truth} and \ref{3rd_expansion_empirical} hold. Then the same conclusion in Theorem \ref{error_true_bootstrap} holds.
\end{theorem}

\begin{theorem}
\label{overall_opt_config2}Suppose Assumptions \ref{balanced_data}, \ref{1st_expansion_truth}, \ref{1st_expansion_empirical}, \ref{3rd_IF_at_truth} and \ref{3rd_expansion_empirical} hold. Then the same conclusion in Theorem \ref{overall_opt_config} holds.
\end{theorem}

\section{Estimating the Covariance of Two Conditional Expectations}\label{sec: var_of_cov}

We consider two random variables $X,Y$ and their conditional expectations
given another random element $Z$. We are interested in the covariance of the
conditional expectations, i.e., $\mathrm{Cov}(\mathbb{E}[X|Z],\mathbb{E}%
[Y|Z])$. The two random variables $X,Y$ are not necessarily independent. This
is a generalization of the variance estimation in \cite{sun2011efficient}. By
modifying the procedure in \cite{sun2011efficient}, we obtain the following
Algorithm \ref{alg:general_covariance_undebias} (assuming $B\geq2,R\geq2$).

\begin{algorithm}[hbt!]
\caption{Covariance Estimation for Two Conditional Expectations}
\label{alg:general_covariance_undebias}
\textbf{Inputs:} number of outer simulation runs $B$, number of inner simulation runs $R$
\begin{algorithmic}
\FOR{$b= 1, \dots, B$ }
\STATE Sample $Z_b$ from $F_Z$.
\FOR{$r= 1, \dots, R$ }
\STATE Sample $(X_{br},Y_{br})$ from $F_{(X,Y)|Z=Z_b}$.
\ENDFOR
\STATE Compute $\bar{X}_{b}=\sum_{r=1}^{R}X_{br}/R$ and $\bar{Y}_{b}=\sum_{r=1}^{R}Y_{br}/R$.
\ENDFOR
\STATE Compute $\bar{\bar{X}}=\sum_{b=1}^{B}\bar{X}_{b}/B$ and $\bar{\bar{Y}}=\sum_{b=1}^{B}\bar{Y}_{b}/B$.
\RETURN
\[
\hat{\sigma}_{Cov}^{2}=\frac{1}{B-1}\sum_{b=1}^{B}(\bar{X}_{b}-\bar{\bar{X}})(\bar{Y}_{b}-\bar{\bar{Y}}).
\]
\end{algorithmic}
\end{algorithm}

In particular, the simulated random elements (or variables) in Algorithm \ref{alg:general_covariance_undebias} satisfy the two-layer
simulation assumptions:
\begin{assumption}\label{assumption_cov}
We have the following:
\begin{description}
\item[(1)] $Z_{1},\ldots, Z_{B}$ are mutually independent.
\item[(2)] The random vectors $(X_{b1},Y_{b1}),\ldots,(X_{bR},Y_{bR})$ are conditionally mutually independent given $Z_{b}$.
\item[(3)] The random vectors $(X_{b1},Y_{b1},\ldots,X_{bR},Y_{bR}),b=1,\ldots,B$ are mutually independent.
\end{description}
\end{assumption}

We have the following lemma for the mean and variance of $\hat{\sigma}%
_{Cov}^{2}$.

\begin{lemma}
\label{general_covariance_undebias}The expectation and variance of
$\hat{\sigma}_{Cov}^{2}$ are given by%
\[
\mathbb{E}[\hat{\sigma}_{Cov}^{2}]=\mathbb{E}[\tau^{X}\tau^{Y}]+\frac{1}%
{R}\mathbb{E}[\varepsilon^{X}\varepsilon^{Y}]=\mathrm{Cov}(\mathbb{E}%
[X|Z],\mathbb{E}[Y|Z])+\frac{1}{R}\mathbb{E}[\mathrm{Cov}(X,Y|Z)],
\]%
\begin{align*}
&  \mathrm{Var}(\hat{\sigma}_{Cov}^{2})\\
&  =\frac{1}{B}\mathbb{E}[(\tau^{X})^{2}(\tau^{Y})^{2}]+\frac{1}{BR}%
\mathbb{E}[(\tau^{X})^{2}(\varepsilon^{Y})^{2}]+\frac{4}{BR}\mathbb{E}%
[\tau^{X}\tau^{Y}\varepsilon^{X}\varepsilon^{Y}]+\frac{2}{BR^{2}}%
\mathbb{E}[\tau^{X}\varepsilon^{X}(\varepsilon^{Y})^{2}]\\
&  +\frac{1}{BR}\mathbb{E}[(\varepsilon^{X})^{2}(\tau^{Y})^{2}]+\frac
{2}{BR^{2}}\mathbb{E}[(\varepsilon^{X})^{2}\tau^{Y}\varepsilon^{Y}]+\frac
{1}{BR^{3}}\mathbb{E}[(\varepsilon^{X})^{2}(\varepsilon^{Y})^{2}]+\frac
{R-1}{BR^{3}}\mathbb{E}[\mathbb{E}[(\varepsilon^{X})^{2}|Z]\mathbb{E}%
[(\varepsilon^{Y})^{2}|Z]]\\
&  +\frac{2(R-1)}{BR^{3}}\mathbb{E}[(\mathbb{E}[\varepsilon^{X}\varepsilon
^{Y}|Z])^{2}]+\frac{1}{B(B-1)}\mathbb{E}[(\tau^{X})^{2}]\mathbb{E}[(\tau
^{Y})^{2}]+\frac{1}{B(B-1)R}\mathbb{E}[(\tau^{X})^{2}]\mathbb{E}%
[(\varepsilon^{Y})^{2}]\\
&  +\frac{1}{B(B-1)R}\mathbb{E}[(\varepsilon^{X})^{2}]\mathbb{E}[(\tau
^{Y})^{2}]+\frac{1}{B(B-1)R^{2}}\mathbb{E}[(\varepsilon^{X})^{2}%
]\mathbb{E}[(\varepsilon^{Y})^{2}]\\
&  -\frac{B-2}{B(B-1)}(\mathbb{E}[\tau^{X}\tau^{Y}])^{2}-\frac{2(B-2)}%
{B(B-1)R}\mathbb{E}[\tau^{X}\tau^{Y}]\mathbb{E}[\varepsilon^{X}\varepsilon
^{Y}]-\frac{B-2}{B(B-1)R^{2}}(\mathbb{E}[\varepsilon^{X}\varepsilon^{Y}%
])^{2}\\
&  =\frac{1}{B}\left(  \mathbb{E}\left[  \left(  \tau^{X}\tau^{Y}+\frac{1}%
{R}\mathbb{E}[\varepsilon^{X}\varepsilon^{Y}|Z]\right)  ^{2}\right]  -\left(
\mathbb{E}\left[  \tau^{X}\tau^{Y}+\frac{1}{R}\varepsilon^{X}\varepsilon
^{Y}\right]  \right)  ^{2}\right)  \\
&  +\frac{1}{BR}\mathbb{E}\left[  \left(  \tau^{X}\varepsilon^{Y}%
+\varepsilon^{X}\tau^{Y}+\frac{1}{R}\varepsilon^{X}\varepsilon^{Y}\right)
^{2}\right]  +\frac{R-1}{BR^{3}}\mathbb{E}[\mathbb{E}[(\varepsilon^{X}%
)^{2}|Z]\mathbb{E}[(\varepsilon^{Y})^{2}|Z]]\\
&  +\frac{R-2}{BR^{3}}\mathbb{E}[(\mathbb{E}[\varepsilon^{X}\varepsilon
^{Y}|Z])^{2}]+\frac{1}{B(B-1)}\mathbb{E}[(\tau^{X})^{2}]\mathbb{E}[(\tau
^{Y})^{2}]+\frac{1}{B(B-1)R}\mathbb{E}[(\tau^{X})^{2}]\mathbb{E}%
[(\varepsilon^{Y})^{2}]\\
&  +\frac{1}{B(B-1)R}\mathbb{E}[(\varepsilon^{X})^{2}]\mathbb{E}[(\tau
^{Y})^{2}]+\frac{1}{B(B-1)R^{2}}\mathbb{E}[(\varepsilon^{X})^{2}%
]\mathbb{E}[(\varepsilon^{Y})^{2}]+\frac{1}{B(B-1)}\left(  \mathbb{E}\left[
\tau^{X}\tau^{Y}+\frac{1}{R}\varepsilon^{X}\varepsilon^{Y}\right]  \right)
^{2}%
\end{align*}
provided these expectations are finite, where $(X,Y,Z)$ is a generic random
vector generated in Algorithm \ref{alg:general_covariance_undebias}, i.e.,
$(X,Y,Z)\overset{d}{=}(X_{11},Y_{11},Z_{1})$ and%
\[
\tau^{X}=\mathbb{E}[X|Z]-\mathbb{E}[X],\quad\tau^{Y}=\mathbb{E}%
[Y|Z]-\mathbb{E}[Y],\quad\varepsilon^{X}=X-\mathbb{E}[X|Z],\quad
\varepsilon^{Y}=Y-\mathbb{E}[Y|Z].
\]
Moreover, each (bracketed) term in the second formula of $\mathrm{Var}%
(\hat{\sigma}_{Cov}^{2})$ is non-negative.
\end{lemma}

\section{Further Details on the Computer Network}\label{sec:details_network}

The network can be represented by an undirected graph in Figure \ref{network}. The true inter-arrival rates are displayed in Table \ref{network_parameter1}.

\begin{figure}[htbp]
\centering
\includegraphics[width=0.45\textwidth]{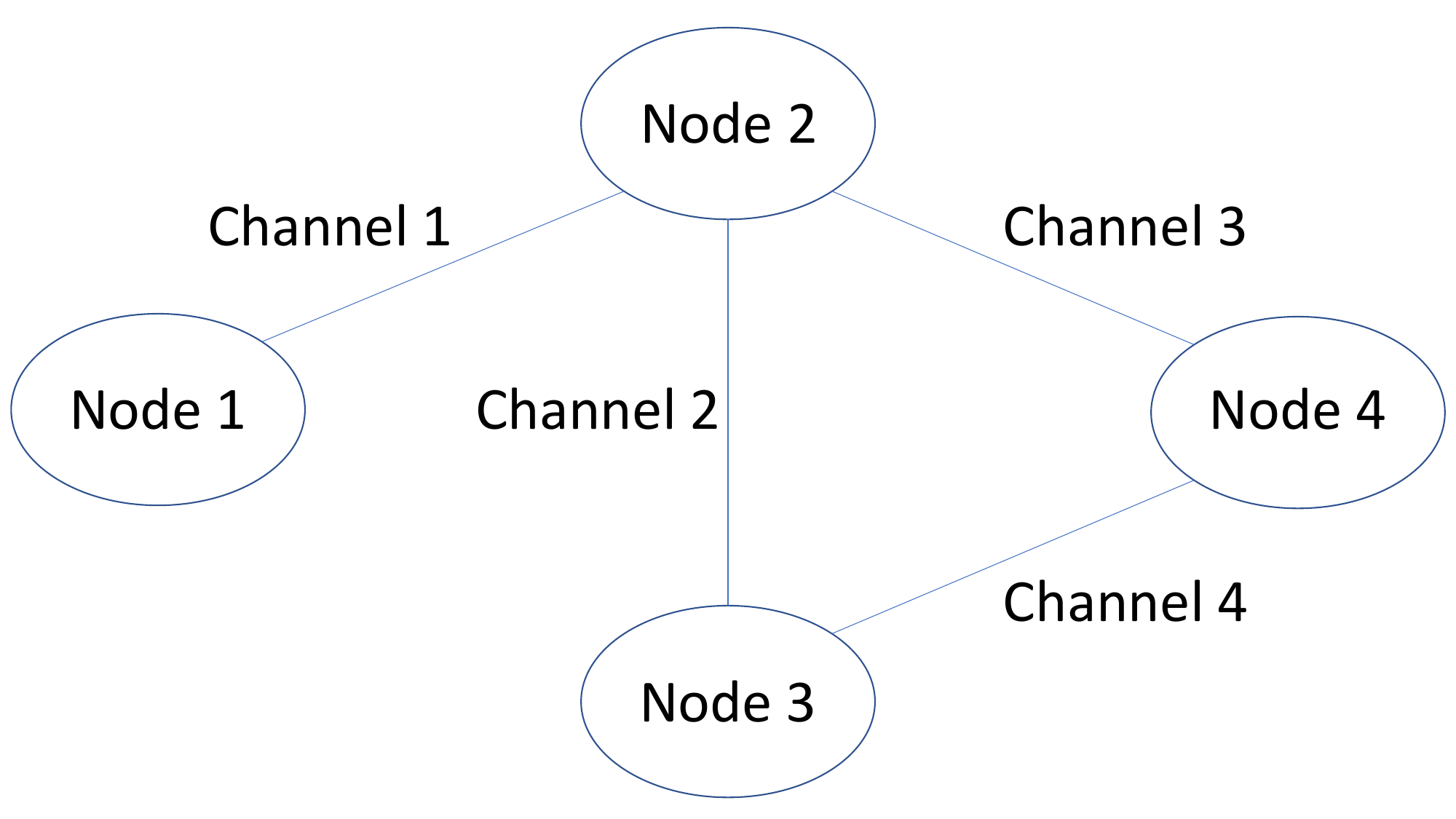}

\caption{A computer network with four nodes and four channels.}
\label{network}
\end{figure}

\begin{table}[ht]
\caption{Arrival rates $\lambda_{i,j}$ of messages to be transmitted from node $i$ to node $j$.}
\medskip
\label{network_parameter1}
\centering
\begin{tabular}{@{}ccccc@{}}
\toprule
         & \multicolumn{4}{c}{Node $j$} \\ \cmidrule(l){2-5}
Node $i$ & 1     & 2     & 3     & 4    \\ 
\midrule
1        & N.A.  & 40    & 30    & 35   \\
2        & 50    & N.A.  & 45    & 15   \\
3        & 60    & 15    & N.A.  & 20   \\
4        & 25    & 30    & 40    & N.A. \\ 
\bottomrule
\end{tabular}
\end{table}

\section{Proofs}\label{sec:proofs}


The proofs are presented in the logical order, which may not be the same as the order the results appear. Section \ref{sec:proofs_convergence} contains proofs of results in Sections \ref{sec:main theory} and \ref{sec:theory_convergence} that are related to weak convergence. Section \ref{sec:proofs_subsampling} contains proofs of results in Sections \ref{sec:cov subsampling} and \ref{sec:theory_cov} that are related to subsampling. Section \ref{sec:proofs_general_assumptions} contains proofs of results in Appendix \ref{sec:general_assumptions}. Section \ref{sec:proofs_var_formula} contains proofs of results in Appendix \ref{sec: var_of_cov}. 

\subsection{Proofs of Results in Sections \ref{sec:main theory} and \ref{sec:theory_convergence}}\label{sec:proofs_convergence}
\begin{proof}{Proof of Lemma \ref{measurability}.}
Since the Borel $\sigma$-field of $D_{\pm\infty}\times D_{\pm\infty}$ is the
product $\sigma$-field $\mathcal{D}_{\pm\infty}\times\mathcal{D}_{\pm\infty}$,
to show $(\sqrt{R}(\hat{Q}(\cdot,\underline{\hat{P}})-Q(\cdot,\underline{\hat
{P}})),\sqrt{n}(Q(\cdot,\underline{\hat{P}})-Q(\cdot,\underline{P})))$ is a
measurable map into $(D_{\pm\infty}\times D_{\pm\infty},\mathcal{D}_{\pm
\infty}\times\mathcal{D}_{\pm\infty})$, it suffices to show the individual
components $\sqrt{R}(\hat{Q}(\cdot,\underline{\hat{P}})-Q(\cdot
,\underline{\hat{P}}))$ and $\sqrt{n}(Q(\cdot,\underline{\hat{P}}%
)-Q(\cdot,\underline{P}))$ are measurable maps into $(D_{\pm\infty
},\mathcal{D}_{\pm\infty})$. Based on the assumption, it suffices to show the
following result: the map $X:(\Omega,\mathcal{F})\mapsto(D_{\pm\infty
},\mathcal{D}_{\pm\infty})$ is measurable if and only if each coordinate
$X(t):(\Omega,\mathcal{F})\mapsto(\mathbb{R},\mathcal{B}(\mathbb{R}))$ is measurable.

For $k\in\mathbb{N}$ and $-\infty<t_{1}<\cdots<t_{k}<\infty$, let $\pi
_{t_{1},\ldots,t_{k}}:(D_{\pm\infty},\mathcal{D}_{\pm\infty})\mapsto
(\mathbb{R}^{k},\mathcal{B}(\mathbb{R}^{k}))$ be the projection at
$\{t_{1},\ldots,t_{k}\}$, i.e., $\pi_{t_{1},\ldots t_{k}}(x)=(x(t_{1}%
),\ldots,x(t_{k}))$ for $x\in D_{\pm\infty}$. It can be shown that
$\mathcal{D}_{\pm\infty}=\sigma(\pi_{t},t\in\mathbb{R})$ (a proof in the space
$D[0,\infty)$ can be found in \cite{billingsley2013convergence} Theorem 16.6
(iii) and a similar argument also holds in $D_{\pm\infty}$). Therefore, the
``only if'' part follows from the
measurability of the composition $\pi_{t}\circ X=X(t)$ if $X$ is measurable.

Now let's show the ``if'' part. Since each
$\pi_{t}$ is measurable and $\mathcal{B}(\mathbb{R}^{k})=\mathcal{B}%
(\mathbb{R})^{k}$ (the product $\sigma$-field), we know that $\pi
_{t_{1},\ldots,t_{k}}$ is also measurable and consequently%
\begin{equation}
\mathcal{D}_{\pm\infty}=\sigma(\pi_{t_{1},\ldots,t_{k}},t_{1}<\cdots
<t_{k},k\in\mathbb{N}). \label{generating_system}%
\end{equation}
Now we define the finite dimensional set%
\[
\mathcal{P}=\{A\in\mathcal{D}_{\pm\infty}:A=\pi_{t_{1},\ldots,t_{k}}%
^{-1}H,t_{1}<\cdots<t_{k},H\in\mathcal{B}(\mathbb{R}^{k}),k\in\mathbb{N}\}
\]
and
\[
\mathcal{L}=\{A\in\mathcal{D}_{\pm\infty}:X^{-1}A\in\mathcal{F}\}.
\]
Then $\mathcal{P}$ is a $\pi$-system and $\mathcal{L}$ is a $\lambda$-system.
Since each coordinate $X(t):(\Omega,\mathcal{F})\mapsto(\mathbb{R}%
,\mathcal{B}(\mathbb{R}))$ is measurable and $\mathcal{B}(\mathbb{R}%
^{k})=\mathcal{B}(\mathbb{R})^{k}$, we know that $(X(t_{1}),\ldots
,X(t_{k})):(\Omega,\mathcal{F})\mapsto(\mathbb{R}^{k},\mathcal{B}%
(\mathbb{R}^{k}))$ are also measurable for any $k\in\mathbb{N}$ and
$-\infty<t_{1}<\cdots<t_{k}<\infty$. Therefore, for each $H\in\mathcal{B}%
(\mathbb{R}^{k})$,%
\[
X^{-1}(\pi_{t_{1},\ldots,t_{k}}^{-1}H)=(\pi_{t_{1},\ldots,t_{k}}\circ
X)^{-1}H=(X(t_{1}),\ldots,X(t_{k}))^{-1}H\in\mathcal{F},
\]
which implies $\mathcal{P}\subset\mathcal{L}$. By the $\pi$-$\lambda$ Theorem
and (\ref{generating_system}), we know that%
\[
\mathcal{D}_{\pm\infty}=\sigma(\pi_{t_{1},\ldots,t_{k}},t_{1}<\cdots
<t_{k},k\in\mathbb{N})\equiv\sigma(\mathcal{P})\subset\mathcal{L\subset
D}_{\pm\infty},
\]
i.e., $\mathcal{L=D}_{\pm\infty}$ and thus $X:(\Omega,\mathcal{F}%
)\mapsto(D_{\pm\infty},\mathcal{D}_{\pm\infty})$ is measurable.

Using the ``if'' part and the assumption
(note that $Q(t,\underline{P})$ is a constant and clearly measurable), we know
that the components $\sqrt{R}(\hat{Q}(\cdot,\underline{\hat{P}})-Q(\cdot
,\underline{\hat{P}}))$ and $\sqrt{n}(Q(\cdot,\underline{\hat{P}}%
)-Q(\cdot,\underline{P}))$ are measurable maps into $(D_{\pm\infty
},\mathcal{D}_{\pm\infty})$ and hence $(\sqrt{R}(\hat{Q}(\cdot,\underline{\hat
{P}})-Q(\cdot,\underline{\hat{P}})),\sqrt{n}(Q(\cdot,\underline{\hat{P}%
})-Q(\cdot,\underline{P})))$ is a measurable map into $(D_{\pm\infty}\times
D_{\pm\infty},\mathcal{D}_{\pm\infty}\times\mathcal{D}_{\pm\infty})$.
\end{proof}

To prove Proposition \ref{fidis_convergence}, we need the following results
from Lemma 3.3.19, Lemma 3.4.3 and Lemma 3.4.4 from
\cite{durrett2019probability}.

\begin{lemma}
\label{bounds_complex_number}We have the following inequalities for the
complex number:%
\begin{equation}
\left\vert e^{ix}-\sum_{m=0}^{k}\frac{(ix)^{m}}{m!}\right\vert \leq
\frac{|x|^{k+1}}{(k+1)!},\forall x\in\mathbb{R},k\in\mathbb{N},
\label{inequality_complex_Taylor}%
\end{equation}%
\begin{equation}
\left\vert z^{k}-w^{k}\right\vert \leq k|z-w|(\max\{|z|,|w|\})^{k-1},\forall
z,w\in\mathbb{C},k\in\mathbb{N}, \label{inequality_complex_power}%
\end{equation}
and%
\begin{equation}
|e^{z}-1-z|\leq|z|^{2},\forall|z|\leq1,z\in\mathbb{C}.
\label{inequality_1st_complex_Taylor}%
\end{equation}

\end{lemma}

By means of Lemma \ref{bounds_complex_number}, we show the following lemma.

\begin{lemma}
\label{fidis_separate}Suppose Assumptions \ref{balanced_data},
\ref{convergence_in_p_to_truth} and \ref{1st_expansion_truth} hold. Then for
any $k\in\mathbb{N}$, any $t_{1}<t_{2}<\cdots<t_{k}$ and any $u=(u_{1}%
,\ldots,u_{k})^{\top}\in\mathbb{R}^{k}$, the (conditional) characteristic
functions satisfy%
\begin{equation}
\mathbb{E}_{\ast}\left[  \exp\left(  i\sum_{j=1}^{k}u_{j}\sqrt{R}(\hat
{Q}(t_{j},\underline{\hat{P}})-Q(t_{j},\underline{\hat{P}}))\right)  \right]
\overset{p}{\rightarrow}\exp\left(  -\frac{1}{2}u^{\top}\Sigma_{BB}u\right)  ,
\label{convergence_con_CF}%
\end{equation}
and%
\begin{equation}
\mathbb{E}\left[  \exp\left(  i\sum_{j=1}^{k}u_{j}\sqrt{n}(Q(t_{j}%
,\underline{\hat{P}})-Q(t_{j},\underline{P}))\right)  \right]  \rightarrow
\exp\left(  -\frac{1}{2}u^{\top}\Sigma_{GP}u\right)  ,
\label{convergence_GP_CF}%
\end{equation}
where $\Sigma_{BB}$ and $\Sigma_{GP}$ are the covariance matrices of
$BB({Q}(\cdot,\underline{P}))$ and $\mathbb{G}(\cdot)$ respectively at
$\{t_{1},\ldots,t_{k}\}$.
\end{lemma}

\begin{proof}{Proof of Lemma \ref{fidis_separate}.}
We first prove (\ref{convergence_con_CF}). According to the covariance
function of the Brownian bridge, the matrix $\Sigma_{BB}$ is given by%
\begin{equation}
\Sigma_{BB,jl}=\min\{{Q}(t_{j},\underline{P}),{Q}(t_{l},\underline{P}%
)\}-{Q}(t_{j},\underline{P}){Q}(t_{l},\underline{P}),\forall1\leq j,l\leq k.
\label{BB_matrix}%
\end{equation}
Recall that
\[
\hat{Q}(t,\underline{\hat{P}})=\frac{1}{R}\sum_{r=1}^{R}I(Y_{r}\leq t),
\]
where $Y_{r},r=1,\ldots,R$ denote i.i.d. Monte Carlo outputs from the
simulation model with input distributions $\underline{\hat{P}}$. The
conditional characteristic function can be simplified as%
\begin{align}
&  \mathbb{E}_{\ast}\left[  \exp\left(  i\sum_{j=1}^{k}u_{j}\sqrt{R}(\hat
{Q}(t_{j},\underline{\hat{P}})-Q(t_{j},\underline{\hat{P}}))\right)  \right]
\nonumber\\
&  =\mathbb{E}_{\ast}\left[  \exp\left(  i\sum_{j=1}^{k}u_{j}\sqrt{R}\hat
{Q}(t_{j},\underline{\hat{P}})\right)  \right]  \exp\left(  -i\sum_{j=1}%
^{k}u_{j}\sqrt{R}Q(t_{j},\underline{\hat{P}})\right) \nonumber\\
&  =\mathbb{E}_{\ast}\left[  \exp\left(  i\sum_{j=1}^{k}\frac{u_{j}}{\sqrt{R}%
}\sum_{r=1}^{R}I(Y_{r}\leq t_{j})\right)  \right]  \exp\left(  -i\sum
_{j=1}^{k}u_{j}\sqrt{R}Q(t_{j},\underline{\hat{P}})\right) \nonumber\\
&  =\mathbb{E}_{\ast}\left[  \prod_{r=1}^{R}\exp\left(  i\sum_{j=1}^{k}%
\frac{u_{j}}{\sqrt{R}}I(Y_{r}\leq t_{j})\right)  \right]  \exp\left(
-i\sum_{j=1}^{k}u_{j}\sqrt{R}Q(t_{j},\underline{\hat{P}})\right) \nonumber\\
&  =\left(  \mathbb{E}_{\ast}\left[  \exp\left(  i\sum_{j=1}^{k}\frac{u_{j}%
}{\sqrt{R}}I(Y_{1}\leq t_{j})\right)  \right]  \right)  ^{R}\exp\left(
-i\sum_{j=1}^{k}u_{j}\sqrt{R}Q(t_{j},\underline{\hat{P}})\right)  .
\label{conputation_con_CF1}%
\end{align}
where the first inequality is due to the measurability of $Q(t_{j}%
,\underline{\hat{P}})$ with respect to the $\sigma$-field generated by the
data and the last inequality is due to the conditional independence of $Y_{r}%
$'s. By direct calculation, we have%
\begin{equation}
\mathbb{E}_{\ast}\left[  \exp\left(  i\sum_{j=1}^{k}\frac{u_{j}}{\sqrt{R}%
}I(Y_{1}\leq t_{j})\right)  \right]  =\sum_{j=1}^{k+1}(Q(t_{j},\underline{\hat
{P}})-Q(t_{j-1},\underline{\hat{P}}))\exp\left(  i\sum_{l=j}^{k}\frac{u_{l}%
}{\sqrt{R}}\right)  , \label{conputation_con_CF2}%
\end{equation}
where $t_{0}:=-\infty$ and $t_{k+1}:=\infty$ with $Q(t_{0},\underline{\hat{P}%
})=0$ and $Q(t_{k+1},\underline{\hat{P}})=1$. Plugging
(\ref{conputation_con_CF2}) into (\ref{conputation_con_CF1}), the conditional
characteristic function can be further simplified as%
\begin{align}
&  \mathbb{E}_{\ast}\left[  \exp\left(  i\sum_{j=1}^{k}u_{j}\sqrt{R}(\hat
{Q}(t_{j},\underline{\hat{P}})-Q(t_{j},\underline{\hat{P}}))\right)  \right]
\nonumber\\
&  =\left(  \sum_{j=1}^{k+1}(Q(t_{j},\underline{\hat{P}})-Q(t_{j-1}%
,\underline{\hat{P}}))\exp\left(  i\sum_{l=j}^{k}\frac{u_{l}}{\sqrt{R}%
}\right)  \right)  ^{R}\left(  \exp\left(  -i\sum_{m=1}^{k}\frac{u_{m}}%
{\sqrt{R}}Q(t_{m},\underline{\hat{P}})\right)  \right)  ^{R}\nonumber\\
&  =\left(  \sum_{j=1}^{k+1}(Q(t_{j},\underline{\hat{P}})-Q(t_{j-1}%
,\underline{\hat{P}}))\exp\left(  i\sum_{l=j}^{k}\frac{u_{l}}{\sqrt{R}}%
-i\sum_{m=1}^{k}\frac{u_{m}}{\sqrt{R}}Q(t_{m},\underline{\hat{P}})\right)
\right)  ^{R}. \label{conputation_con_CF3}%
\end{align}
Next, by (\ref{inequality_complex_Taylor}), we have%
\begin{align*}
&  \exp\left(  i\sum_{l=j}^{k}\frac{u_{l}}{\sqrt{R}}-i\sum_{m=1}^{k}%
\frac{u_{m}}{\sqrt{R}}Q(t_{m},\underline{\hat{P}})\right) \\
&  =1+\frac{i}{\sqrt{R}}\left(  \sum_{l=j}^{k}u_{l}-\sum_{m=1}^{k}u_{m}%
Q(t_{m},\underline{\hat{P}})\right)  -\frac{1}{2R}\left(  \sum_{l=j}^{k}%
u_{l}-\sum_{m=1}^{k}u_{m}Q(t_{m},\underline{\hat{P}})\right)  ^{2}%
+\varepsilon_{j},
\end{align*}
where the error term $\varepsilon_{j}$ satisfies%
\[
|\varepsilon_{j}|\leq\frac{1}{6R^{3/2}}\left\vert \sum_{l=j}^{k}u_{l}%
-\sum_{m=1}^{k}u_{m}Q(t_{m},\underline{\hat{P}})\right\vert ^{3}\leq\frac
{1}{6R^{3/2}}\left(  \sum_{m=1}^{k}|u_{m}|\right)  ^{3}.
\]
By direct calculation, we can get%
\begin{align}
&  \sum_{j=1}^{k+1}(Q(t_{j},\underline{\hat{P}})-Q(t_{j-1},\underline{\hat{P}%
}))\exp\left(  i\sum_{l=j}^{k}\frac{u_{l}}{\sqrt{R}}-i\sum_{m=1}^{k}%
\frac{u_{m}}{\sqrt{R}}Q(t_{m},\underline{\hat{P}})\right) \nonumber\\
&  =\sum_{j=1}^{k+1}(Q(t_{j},\underline{\hat{P}})-Q(t_{j-1},\underline{\hat
{P}}))\left(  1+\frac{i}{\sqrt{R}}\left(  \sum_{l=j}^{k}u_{l}-\sum_{m=1}%
^{k}u_{m}Q(t_{m},\underline{\hat{P}})\right)  \right. \nonumber\\
&  \left.  -\frac{1}{2R}\left(  \sum_{l=j}^{k}u_{l}-\sum_{m=1}^{k}u_{m}%
Q(t_{m},\underline{\hat{P}})\right)  ^{2}+\varepsilon_{j}\right) \nonumber\\
&  =1-\frac{1}{2R}u^{\top}\hat{\Sigma}_{BB}u+\varepsilon,
\label{conputation_con_CF4}%
\end{align}
where the matrix $\hat{\Sigma}_{BB}\in\mathbb{R}^{k\times k}$ is given by
\begin{equation}
\hat{\Sigma}_{BB,jl}=\min\{{Q}(t_{j},\underline{\hat{P}}),{Q}(t_{l}%
,\underline{\hat{P}})\}-{Q}(t_{j},\underline{\hat{P}}){Q}(t_{l}%
,\underline{\hat{P}}), \label{Sigma_hat}%
\end{equation}
and the cumulative error $\varepsilon$ is given by
\[
\varepsilon=\sum_{j=1}^{k+1}(Q(t_{j},\underline{\hat{P}})-Q(t_{j-1}%
,\underline{\hat{P}}))\varepsilon_{j}%
\]
with the bound%
\[
|\varepsilon|\leq\sum_{j=1}^{k+1}(Q(t_{j},\underline{\hat{P}})-Q(t_{j-1}%
,\underline{\hat{P}}))|\varepsilon_{j}|\leq\frac{1}{6R^{3/2}}\left(
\sum_{m=1}^{k}|u_{m}|\right)  ^{3}.
\]
Plugging (\ref{conputation_con_CF4}) into (\ref{conputation_con_CF3}), we have%
\[
\mathbb{E}_{\ast}\left[  \exp\left(  i\sum_{j=1}^{k}u_{j}\sqrt{R}(\hat
{Q}(t_{j},\underline{\hat{P}})-Q(t_{j},\underline{\hat{P}}))\right)  \right]
=\left(  1-\frac{1}{2R}u^{\top}\hat{\Sigma}_{BB}u+\varepsilon\right)  ^{R}%
\]
and thus by the inequality (\ref{inequality_complex_power})%
\begin{align}
&  \left\vert \mathbb{E}_{\ast}\left[  \exp\left(  i\sum_{j=1}^{k}u_{j}%
\sqrt{R}(\hat{Q}(t_{j},\underline{\hat{P}})-Q(t_{j},\underline{\hat{P}%
}))\right)  \right]  -\exp\left(  -\frac{1}{2}u^{\top}\Sigma_{BB}u\right)
\right\vert \nonumber\\
&  =\left\vert \left(  1-\frac{1}{2R}u^{\top}\hat{\Sigma}_{BB}u+\varepsilon
\right)  ^{R}-\left(  \exp\left(  -\frac{1}{2R}u^{\top}\Sigma_{BB}u\right)
\right)  ^{R}\right\vert \nonumber\\
&  \leq R\left(  \max\left\{  \left\vert 1-\frac{1}{2R}u^{\top}\hat{\Sigma
}_{BB}u+\varepsilon\right\vert ,\left\vert \exp\left(  -\frac{1}{2R}u^{\top
}\Sigma_{BB}u\right)  \right\vert \right\}  \right)  ^{R-1}\nonumber\\
&  \times\left\vert 1-\frac{1}{2R}u^{\top}\hat{\Sigma}_{BB}u+\varepsilon
-\exp\left(  -\frac{1}{2R}u^{\top}\Sigma_{BB}u\right)  \right\vert .
\label{conputation_con_CF5}%
\end{align}
Now we analyze the terms in (\ref{conputation_con_CF5}). By the definition in
(\ref{Sigma_hat}), we have $|\hat{\Sigma}_{BB,jl}|\leq1$ and consequently
$|u^{\top}\hat{\Sigma}_{BB}u|\leq(\sum_{j=1}^{k}|u_{j}|)^{2}$. Therefore,
\begin{align}
&  \left(  \max\left\{  \left\vert 1-\frac{1}{2R}u^{\top}\hat{\Sigma}%
_{BB}u+\varepsilon\right\vert ,\left\vert \exp\left(  -\frac{1}{2R}u^{\top
}\Sigma_{BB}u\right)  \right\vert \right\}  \right)  ^{R-1}\nonumber\\
&  \leq\left(  1+\frac{1}{2R}|u^{\top}\hat{\Sigma}_{BB}u|+|\varepsilon
|\right)  ^{R-1}\nonumber\\
&  \leq\left(  1+\frac{1}{2R}\left(  \sum_{j=1}^{k}|u_{j}|\right)  ^{2}%
+\frac{1}{6R^{3/2}}\left(  \sum_{m=1}^{k}|u_{m}|\right)  ^{3}\right)  ^{R-1}.
\label{conputation_con_CF6}%
\end{align}
Besides, when $R\geq u^{\top}\Sigma_{BB}u/2$, by
(\ref{inequality_1st_complex_Taylor}), we have
\begin{align}
&  \left\vert 1-\frac{1}{2R}u^{\top}\hat{\Sigma}_{BB}u+\varepsilon-\exp\left(
-\frac{1}{2R}u^{\top}\Sigma_{BB}u\right)  \right\vert \nonumber\\
&  =\left\vert 1-\frac{1}{2R}u^{\top}\Sigma_{BB}u-\exp\left(  -\frac{1}%
{2R}u^{\top}\Sigma_{BB}u\right)  +\frac{1}{2R}u^{\top}\Sigma_{BB}u-\frac
{1}{2R}u^{\top}\hat{\Sigma}_{BB}u+\varepsilon\right\vert \nonumber\\
&  \leq\left\vert 1-\frac{1}{2R}u^{\top}\Sigma_{BB}u-\exp\left(  -\frac{1}%
{2R}u^{\top}\Sigma_{BB}u\right)  \right\vert +\frac{1}{2R}|u^{\top}%
(\Sigma_{BB}-\hat{\Sigma}_{BB})u|+|\varepsilon|\nonumber\\
&  \leq\frac{1}{4R^{2}}(u^{\top}\Sigma_{BB}u)^{2}+\frac{1}{2R}|u^{\top}%
(\Sigma_{BB}-\hat{\Sigma}_{BB})u|+\frac{1}{6R^{3/2}}\left(  \sum_{m=1}%
^{k}|u_{m}|\right)  ^{3}. \label{conputation_con_CF7}%
\end{align}
Plugging (\ref{conputation_con_CF6}) and (\ref{conputation_con_CF7}) into
(\ref{conputation_con_CF5}), we obtain%
\begin{align}
&  \left\vert \mathbb{E}_{\ast}\left[  \exp\left(  i\sum_{j=1}^{k}u_{j}%
\sqrt{R}(\hat{Q}(t_{j},\underline{\hat{P}})-Q(t_{j},\underline{\hat{P}%
}))\right)  \right]  -\exp\left(  -\frac{1}{2}u^{\top}\Sigma_{BB}u\right)
\right\vert \nonumber\\
&  \leq\left(  1+\frac{1}{2R}\left(  \sum_{j=1}^{k}|u_{j}|\right)  ^{2}%
+\frac{1}{6R^{3/2}}\left(  \sum_{m=1}^{k}|u_{m}|\right)  ^{3}\right)
^{R-1}\nonumber\\
&  \times\left(  \frac{1}{4R^{2}}(u^{\top}\Sigma_{BB}u)^{2}+\frac{1}%
{2R}|u^{\top}(\Sigma_{BB}-\hat{\Sigma}_{BB})u|+\frac{1}{6R^{3/2}}\left(
\sum_{m=1}^{k}|u_{m}|\right)  ^{3}\right)  . \label{conputation_con_CF8}%
\end{align}
By Assumption \ref{convergence_in_p_to_truth} and continuous mapping theorem,
we have $|u^{\top}(\Sigma_{BB}-\hat{\Sigma}_{BB})u|\overset{p}{\rightarrow}0$
as $n\rightarrow\infty$ and thus%
\[
\frac{1}{4R^{2}}(u^{\top}\Sigma_{BB}u)^{2}+\frac{1}{2R}|u^{\top}(\Sigma
_{BB}-\hat{\Sigma}_{BB})u|+\frac{1}{6R^{3/2}}\left(  \sum_{m=1}^{k}%
|u_{m}|\right)  ^{3}\overset{p}{\rightarrow}0
\]
as $n,R\rightarrow\infty$. Moreover, as $R\rightarrow\infty$,%
\[
\left(  1+\frac{1}{2R}\left(  \sum_{j=1}^{k}|u_{j}|\right)  ^{2}+\frac
{1}{6R^{3/2}}\left(  \sum_{m=1}^{k}|u_{m}|\right)  ^{3}\right)  ^{R-1}%
\rightarrow\exp\left(  \frac{1}{2}\left(  \sum_{j=1}^{k}|u_{j}|\right)
^{2}\right)  <\infty,
\]
which, by (\ref{conputation_con_CF8}), finally implies that%
\[
\left\vert \mathbb{E}_{\ast}\left[  \exp\left(  i\sum_{j=1}^{k}u_{j}\sqrt
{R}(\hat{Q}(t_{j},\underline{\hat{P}})-Q(t_{j},\underline{\hat{P}}))\right)
\right]  -\exp\left(  -\frac{1}{2}u^{\top}\Sigma_{BB}u\right)  \right\vert
\overset{p}{\rightarrow}0,
\]
This proves (\ref{convergence_con_CF}).

Next, we prove (\ref{convergence_GP_CF}). It suffices to show%
\[
(\sqrt{n}(Q(t_{1},\underline{\hat{P}})-Q(t_{1},\underline{P})),\ldots,\sqrt
{n}(Q(t_{k},\underline{\hat{P}})-Q(t_{k},\underline{P}%
)))\overset{d}{\rightarrow}N(0,\Sigma_{GP}),
\]
where the matrix $\Sigma_{GP}\in\mathbb{R}^{k\times k}$ is given by%
\begin{equation}
\Sigma_{GP,jl}=\sum_{i=1}^{m}\frac{1}{\beta_{i}}\mathrm{Cov}_{P_{i}}%
(IF_{i}(t_{j},X_{i};\underline{P}),IF_{i}(t_{l},X_{i};\underline{P}%
)),\forall1\leq j,l\leq k. \label{GP_matrix}%
\end{equation}
By Assumption \ref{1st_expansion_truth}, we have%
\begin{align*}
Q(t,\underline{\hat{P}})  &  =Q(t,\underline{P})+\sum_{i=1}^{m}\int
IF_{i}(t,x;\underline{P})d(\hat{P}_{i}(x)-P_{i}(x))+\varepsilon(t)\\
&  =Q(t,\underline{P})+\sum_{i=1}^{m}\frac{1}{n_{i}}\sum_{j=1}^{n_{i}}%
IF_{i}(t,X_{i,j};\underline{P})+\varepsilon(t).
\end{align*}
Therefore, we can get%
\begin{align}
&  (\sqrt{n}(Q(t_{1},\underline{\hat{P}})-Q(t_{1},\underline{P})),\ldots
,\sqrt{n}(Q(t_{k},\underline{\hat{P}})-Q(t_{k},\underline{P})))\nonumber\\
&  =\left(  \sum_{i=1}^{m}\frac{\sqrt{n}}{n_{i}}\sum_{j=1}^{n_{i}}IF_{i}%
(t_{1},X_{i,j};\underline{P}),\ldots,\sum_{i=1}^{m}\frac{\sqrt{n}}{n_{i}}%
\sum_{j=1}^{n_{i}}IF_{i}(t_{k},X_{i,j};\underline{P})\right)  +(\sqrt
{n}\varepsilon(t_{1}),\ldots,\sqrt{n}\varepsilon(t_{k}))\nonumber\\
&  =\sum_{i=1}^{m}\sqrt{\frac{n}{n_{i}}}\left(  \frac{1}{\sqrt{n_{i}}}%
\sum_{j=1}^{n_{i}}IF_{i}(t_{1},X_{i,j};\underline{P}),\ldots,\frac{1}%
{\sqrt{n_{i}}}\sum_{j=1}^{n_{i}}IF_{i}(t_{k},X_{i,j};\underline{P})\right)
+(\sqrt{n}\varepsilon(t_{1}),\ldots,\sqrt{n}\varepsilon(t_{k})).
\label{decomposition_fd_GP}%
\end{align}
For each $i$, by the multivariate central limit theorem, we have%
\begin{equation}
\left(  \frac{1}{\sqrt{n_{i}}}\sum_{j=1}^{n_{i}}IF_{i}(t_{1},X_{i,j}%
;\underline{P}),\ldots,\frac{1}{\sqrt{n_{i}}}\sum_{j=1}^{n_{i}}IF_{i}%
(t_{k},X_{i,j};\underline{P})\right)  \overset{d}{\rightarrow}N(0,\Sigma_{i}),
\label{ith_convergence}%
\end{equation}
where the matrix $\Sigma_{i}\in\mathbb{R}^{k\times k}$ is given by%
\[
\Sigma_{i,jl}=\mathrm{Cov}_{P_{i}}(IF_{i}(t_{j},X_{i};\underline{P}%
),IF_{i}(t_{l},X_{i};\underline{P})),\forall1\leq j,l\leq k.
\]
Since the left hand sides of (\ref{ith_convergence}) are independent for
different $i$ and $n/n_{i}\rightarrow1/\beta_{i}$ by Assumption
\ref{balanced_data}, we have%
\begin{equation}
\sum_{i=1}^{m}\sqrt{\frac{n}{n_{i}}}\left(  \frac{1}{\sqrt{n_{i}}}\sum
_{j=1}^{n_{i}}IF_{i}(t_{1},X_{i,j};\underline{P}),\ldots,\frac{1}{\sqrt{n_{i}%
}}\sum_{j=1}^{n_{i}}IF_{i}(t_{k},X_{i,j};\underline{P})\right)
\overset{d}{\rightarrow}N\left(  0,\sum_{i=1}^{m}\frac{1}{\beta_{i}}\Sigma
_{i}\right)  \overset{d}{=}N(0,\Sigma_{GP}). \label{leading_convergence1}%
\end{equation}
Moreover, by Assumption \ref{1st_expansion_truth}, we know that $\mathbb{E}%
[(\sqrt{n}\varepsilon(t))^{2}]=o(1)$ for each $t\in\mathbb{R}$. Therefore, we
have $\sqrt{n}\varepsilon(t)\overset{p}{\rightarrow}0$ for each $t\in
\mathbb{R}$ and thus%
\begin{equation}
(\sqrt{n}\varepsilon(t_{1}),\ldots,\sqrt{n}\varepsilon(t_{k}%
))\overset{p}{\rightarrow}\mathbf{0}. \label{negligible_term1}%
\end{equation}
By Slutsky's Theorem and plugging (\ref{leading_convergence1}) and
(\ref{negligible_term1}) into (\ref{decomposition_fd_GP}), we obtain%
\[
(\sqrt{n}(Q(t_{1},\underline{\hat{P}})-Q(t_{1},\underline{P})),\ldots,\sqrt
{n}(Q(t_{k},\underline{\hat{P}})-Q(t_{k},\underline{P}%
)))\overset{d}{\rightarrow}N(0,\Sigma_{GP}).
\]
This concludes our proof.
\end{proof}

Now we prove Proposition \ref{fidis_convergence}.

\begin{proof}{Proof of Proposition \ref{fidis_convergence}.}
First, by Theorem \ref{verification_general_assumptions}, Assumptions \ref{convergence_in_p_to_truth}-\ref{3rd_expansion_empirical} hold under Assumptions \ref{balanced_data} and \ref{finite_horizon_model}. Therefore, the conclusions in Lemma \ref{fidis_separate} hold. Fix $k\in\mathbb{N}$ and $t_{1}<t_{2}<\cdots<t_{k}$. We need to show%
\begin{align*}
&  (\sqrt{R}(\hat{Q}(t_{1},\underline{\hat{P}})-Q(t_{1},\underline{\hat{P}%
})),\ldots,\sqrt{R}(\hat{Q}(t_{k},\underline{\hat{P}})-Q(t_{k},\underline{\hat
{P}})),\\
&  \sqrt{n}(Q(t_{1},\underline{\hat{P}})-Q(t_{1},\underline{P})),\ldots
,\sqrt{n}(Q(t_{k},\underline{\hat{P}})-Q(t_{k},\underline{P})))\\
&  \overset{d}{\rightarrow}N\left(
\begin{pmatrix}
0_{k\times1}\\
0_{k\times1}%
\end{pmatrix}
,%
\begin{pmatrix}
\Sigma_{BB} & 0_{k\times k}\\
0_{k\times k} & \Sigma_{GP}%
\end{pmatrix}
\right)  ,
\end{align*}
where $\Sigma_{BB}$ and $\Sigma_{GP}$ are defined in (\ref{BB_matrix}) and
(\ref{GP_matrix}) respectively. For any $u=(u_{1},\ldots,u_{k})^{\top}%
\in\mathbb{R}^{k}$ and $v=(v_{1},\ldots,v_{k})^{\top}\in\mathbb{R}^{k}$, we
have%
\begin{align}
&  \mathbb{E}\left[  \exp\left(  i\sum_{j=1}^{k}u_{j}\sqrt{R}(\hat{Q}%
(t_{j},\underline{\hat{P}})-Q(t_{j},\underline{\hat{P}}))+i\sum_{j=1}^{k}%
v_{j}\sqrt{n}(Q(t_{j},\underline{\hat{P}})-Q(t_{j},\underline{P}))\right)
\right] \nonumber\\
&  =\mathbb{E}\left[  \mathbb{E}_{\ast}\left[  \exp\left(  i\sum_{j=1}%
^{k}u_{j}\sqrt{R}(\hat{Q}(t_{j},\underline{\hat{P}})-Q(t_{j},\underline{\hat
{P}}))+i\sum_{j=1}^{k}v_{j}\sqrt{n}(Q(t_{j},\underline{\hat{P}})-Q(t_{j}%
,\underline{P}))\right)  \right]  \right] \nonumber\\
&  =\mathbb{E}\left[  \exp\left(  i\sum_{j=1}^{k}v_{j}\sqrt{n}(Q(t_{j}%
,\underline{\hat{P}})-Q(t_{j},\underline{P}))\right)  \mathbb{E}_{\ast}\left[
\exp\left(  i\sum_{j=1}^{k}u_{j}\sqrt{R}(\hat{Q}(t_{j},\underline{\hat{P}%
})-Q(t_{j},\underline{\hat{P}}))\right)  \right]  \right] \nonumber\\
&  =\exp\left(  -\frac{1}{2}u^{\top}\Sigma_{BB}u\right)  \mathbb{E}\left[
\exp\left(  i\sum_{j=1}^{k}v_{j}\sqrt{n}(Q(t_{j},\underline{\hat{P}}%
)-Q(t_{j},\underline{P}))\right)  \right] \nonumber\\
&  +\mathbb{E}\left[  \left(  \mathbb{E}_{\ast}\left[  \exp\left(  i\sum
_{j=1}^{k}u_{j}\sqrt{R}(\hat{Q}(t_{j},\underline{\hat{P}})-Q(t_{j}%
,\underline{\hat{P}}))\right)  \right]  -\exp\left(  -\frac{1}{2}u^{\top
}\Sigma_{BB}u\right)  \right)  \right. \nonumber\\
&  \times\left.  \exp\left(  i\sum_{j=1}^{k}v_{j}\sqrt{n}(Q(t_{j}%
,\underline{\hat{P}})-Q(t_{j},\underline{P}))\right)  \right]  .
\label{decomposition_fd_whole}%
\end{align}
By (\ref{convergence_con_CF}) and the dominated convergence theorem (which
also holds if almost sure convergence is replaced by convergence in
probability), we have%
\begin{align}
&  \left\vert \mathbb{E}\left[  \left(  \mathbb{E}_{\ast}\left[  \exp\left(
i\sum_{j=1}^{k}u_{j}\sqrt{R}(\hat{Q}(t_{j},\underline{\hat{P}})-Q(t_{j}%
,\underline{\hat{P}}))\right)  \right]  -\exp\left(  -\frac{1}{2}u^{\top
}\Sigma_{BB}u\right)  \right)  \right.  \right. \nonumber\\
&  \left.  \times\left.  \exp\left(  i\sum_{j=1}^{k}v_{j}\sqrt{n}%
(Q(t_{j},\underline{\hat{P}})-Q(t_{j},\underline{P}))\right)  \right]
\right\vert \nonumber\\
&  \leq\mathbb{E}\left[  \left\vert \mathbb{E}_{\ast}\left[  \exp\left(
i\sum_{j=1}^{k}u_{j}\sqrt{R}(\hat{Q}(t_{j},\underline{\hat{P}})-Q(t_{j}%
,\underline{\hat{P}}))\right)  \right]  -\exp\left(  -\frac{1}{2}u^{\top
}\Sigma_{BB}u\right)  \right\vert \right]  \rightarrow0.
\label{negligible_term2}%
\end{align}
Plugging (\ref{convergence_GP_CF}) and (\ref{negligible_term2}) into
(\ref{decomposition_fd_whole}), we obtain%
\begin{align*}
\mathbb{E}\left[  \exp\left(  i\sum_{j=1}^{k}u_{j}\sqrt{R}(\hat{Q}%
(t_{j},\underline{\hat{P}})-Q(t_{j},\underline{P}))\right)  \right]   &
\rightarrow\exp\left(  -\frac{1}{2}u^{\top}\Sigma_{BB}u-\frac{1}{2}v^{\top
}\Sigma_{GP}v\right) \\
&  =\exp\left(  -\frac{1}{2}(u^{\top},v^{\top})%
\begin{pmatrix}
\Sigma_{BB} & 0_{k\times k}\\
0_{k\times k} & \Sigma_{GP}%
\end{pmatrix}%
\begin{pmatrix}
u\\
v
\end{pmatrix}
\right)
\end{align*}
which implies the desired finite-dimensional convergence. This completes our proof.
\end{proof}

The proof of Theorem \ref{weak_convergence} relies on two tightness results.
Lemma \ref{tightness_BB} establishes the tightness of the Brownian bridge part while Lemma \ref{tightness_GP} establishes the tightness of the Gaussian process part. Recall that $||x||_{\infty}=\sup_{t\in\mathbb{R}}|x(t)|$ which could be $\infty$. If $X$ is a random element in $D_{\pm\infty}$ and $||X||_{\infty}$ is finite almost surely, then $||X||_{\infty}$ is measurable (i.e., a random variable) since%
\[
||X||_{\infty}=\sup_{t\in\mathbb{Q}}||X(t)||
\]
by right-continuity and each $X(t)$ is a random variable from the proof of Lemma \ref{measurability}.

\begin{lemma}
\label{tightness_BB}Suppose for each $t\in\mathbb{R}$, $Q(t,\underline{\hat
{P}})$ is measurable with respect to the $\sigma$-field generated by the data
and only takes finitely many values. Further, assume that%
\[
||Q(\cdot,\underline{\hat{P}})-Q(\cdot,\underline{P})||_{\infty}%
\overset{p}{\rightarrow}0
\]
as $n\rightarrow\infty$. Then we have%
\[
\sqrt{R}(\hat{Q}(\cdot,\underline{\hat{P}})-Q(\cdot,\underline{\hat{P}%
}))\Rightarrow BB({Q}(\cdot,\underline{P}))
\]
in the space $(D_{\pm\infty},\mathcal{D}_{\pm\infty},d_{\pm\infty}^{\circ})$
as $n,R\rightarrow\infty$, where $BB(\cdot)$ is the standard Brownian bridge.
In particular, $\sqrt{R}(\hat{Q}(\cdot,\underline{\hat{P}})-Q(\cdot
,\underline{\hat{P}}))$ is tight.
\end{lemma}

\begin{proof}{Proof of Lemma \ref{tightness_BB}.}
By the Portmanteau Theorem, to show the weak convergence, it suffices to show
\[
\mathbb{E}f(\sqrt{R}(\hat{Q}(\cdot,\underline{\hat{P}})-Q(\cdot
,\underline{\hat{P}})))\rightarrow\mathbb{E}f(BB({Q}(\cdot,\underline{P})))
\]
for each bounded and Lipschitz function $f$ on $(D_{\pm\infty},\mathcal{D}%
_{\pm\infty},d_{\pm\infty}^{\circ})$. By possibly enlarging the underlying
probability space, there exists a standard Brownian bridge $BB(\cdot)$ which
is independent of $Q(\cdot,\underline{\hat{P}})$. Since $Q(t,\underline{\hat
{P}})$ only takes finitely many values, say $q_{1},\ldots q_{l}$, we can see
$BB({Q}(t,\underline{\hat{P}}))$ is a random variable by%
\[
\{BB({Q}(t,\underline{\hat{P}}))\in H\}=\bigcup_{i=1}^{l}\{BB(q_{i})\in
H\}\cap\{{Q}(t,\underline{\hat{P}})=q_{i}\},H\in\mathcal{B}(\mathbb{R}).
\]
Then from the proof of Lemma \ref{measurability}, $BB({Q}(\cdot
,\underline{\hat{P}}))$ is a random element of $D_{\pm\infty}$ (in the
possibly enlarged probability space). Therefore, we can write%
\begin{align}
&  \mathbb{E}f(\sqrt{R}(\hat{Q}(\cdot,\underline{\hat{P}})-Q(\cdot
,\underline{\hat{P}})))-\mathbb{E}f(BB({Q}(\cdot,\underline{P})))\nonumber\\
&  =\mathbb{E[}f(\sqrt{R}(\hat{Q}(\cdot,\underline{\hat{P}})-Q(\cdot
,\underline{\hat{P}})))-f(BB({Q}(\cdot,\underline{\hat{P}})))]+\mathbb{E[}%
f(BB({Q}(\cdot,\underline{\hat{P}})))-f(BB({Q}(\cdot,\underline{P}%
)))]\nonumber\\
&  =\mathbb{E[E}^{\ast}\mathbb{[}f(\sqrt{R}(\hat{Q}(\cdot,\underline{\hat{P}%
})-Q(\cdot,\underline{\hat{P}})))-f(BB({Q}(\cdot,\underline{\hat{P}%
})))]]+\mathbb{E[}f(BB({Q}(\cdot,\underline{\hat{P}})))-f(BB({Q}%
(\cdot,\underline{P})))] \label{decomposition_weak_conv}%
\end{align}
We will show each term in (\ref{decomposition_weak_conv}) converges to $0$.

To handle the first term, recall that Koml\'{o}s-Major-Tusn\'{a}dy
approximation (\cite{komlos1975approximation}) tells us that on a suitable
probability space say $(\Omega_{1},\mathcal{F}_{1},\mathbb{P}_{1})$ (the expectation in this space is written as $\mathbb{E}_1$),
\[
\mathbb{P}_{1}\left(  \sup_{0\leq t\leq1}|\sqrt{R}(\hat{F}_{R}%
(t)-t)-BB(t)|>\frac{C\log R+x}{\sqrt{R}}\right)  <Ke^{-\lambda x}%
\]
for all $x$, where $C$, $K$, $\lambda$ are positive absolute constant and
$\hat{F}_{R}(\cdot)$ is the empirical process of $R$ i.i.d. $U[0,1]$ random
variables. Therefore, for any distribution function $G(t)$ and setting
$x=C\log R$, we obtain%
\[
\mathbb{P}_{1}\left(  ||\sqrt{R}(\hat{F}_{R}(G(t))-G(t))-BB(G(t))||_{\infty
}>\frac{2C\log R}{\sqrt{R}}\right)  <\frac{K}{R^{C\lambda}},
\]
where $\hat{F}_{R}(G(\cdot))$ is the empirical process of $R$ i.i.d. random
variables with distribution $G$. By measurability assumption,
$Q(t,\underline{\hat{P}})$ is a deterministic distribution function
conditional on the data (we take it as the above $G(t)$). Let $L$ and $M$ be the Lipschitz
constant (with respect to $d_{\pm\infty}^{\circ}$) and bound of $f$
respectively. Since $d_{\pm\infty}^{\circ}(x,y)\leq||x-y||_{\infty}$, $f$ is
also $L$-Lipschitz with respect to $||\cdot||_{\infty}$. Therefore,%
\begin{align*}
&  |\mathbb{E}^{\ast}[f(\sqrt{R}(\hat{Q}(\cdot,\underline{\hat{P}}%
)-Q(\cdot,\underline{\hat{P}})))-f(BB({Q}(\cdot,\underline{\hat{P}})))]|\\
&  =|\mathbb{E}_{1}[f(\sqrt{R}(\hat{F}_{R}(G(\cdot))-G(\cdot)))-f(BB(G(\cdot
)))]|\\
&  \leq\mathbb{E}_{1}[|f(\sqrt{R}(\hat{F}_{R}(G(\cdot))-G(\cdot
)))-f(BB(G(\cdot)))|]\\
&  \leq\frac{2CL\log R}{\sqrt{R}}\mathbb{P}_{1}\left(  ||\sqrt{R}(\hat{F}%
_{R}(G(t))-G(t))-BB(G(t))||_{\infty}\leq\frac{2C\log R}{\sqrt{R}}\right) \\
&  +2M\mathbb{P}_{1}\left(  ||\sqrt{R}(\hat{F}_{R}%
(G(t))-G(t))-BB(G(t))||_{\infty}>\frac{2C\log R}{\sqrt{R}}\right) \\
&  \leq\frac{2CL\log R}{\sqrt{R}}+\frac{2MK}{R^{C\lambda}}\rightarrow0
\end{align*}
as $R\rightarrow\infty$. Therefore, by the dominated convergence theorem,
\[
\mathbb{E[E}^{\ast}\mathbb{[}f(\sqrt{R}(\hat{Q}(\cdot,\underline{\hat{P}%
})-Q(\cdot,\underline{\hat{P}})))-f(BB({Q}(\cdot,\underline{\hat{P}%
})))]]\rightarrow0
\]
as $R\rightarrow\infty$. This proves the first term in
(\ref{decomposition_weak_conv}) converges to $0$.

Now we consider the second term in (\ref{decomposition_weak_conv}). We define
the modulus of continuity of an arbitrary function $x(\cdot)$ on $[0,1]$ as%
\begin{equation}
w_{x}(\delta)=\sup_{|s-t|\leq\delta}|x(s)-x(t)|. \label{modulus_of_continuity}%
\end{equation}
If $x$ is continuous on $[0,1]$ (and thus uniformly continuous on $[0,1]$),
then $w_{x}(\delta)\rightarrow0$ as $\delta\rightarrow0$. From the assumption
$||Q(\cdot,\underline{\hat{P}})-Q(\cdot,\underline{P})||_{\infty
}\overset{p}{\rightarrow}0$ and the property of convergence in probability, we
know that each subsequence (we omit the index for simplicity) contains a
further subsequence such that
\[
||Q(\cdot,\underline{\hat{P}})-Q(\cdot,\underline{P})||_{\infty}%
\overset{a.s.}{\rightarrow}0.
\]
Since the sample paths of $BB(\cdot)$ are all continuous, we can deduce that%
\[
||BB(Q(\cdot,\underline{\hat{P}}))-BB(Q(\cdot,\underline{P}))||_{\infty}\leq
w_{BB}(||Q(\cdot,\underline{\hat{P}})-Q(\cdot,\underline{P})||_{\infty
})\overset{a.s.}{\rightarrow}0.
\]
In particular, this implies $BB(Q(\cdot,\underline{\hat{P}}%
))\overset{a.s.}{\rightarrow}BB(Q(\cdot,\underline{P}))$ under the metric
$d_{\pm\infty}^{\circ}$. By the continuity of $f$, we have%
\[
f(BB({Q}(\cdot,\underline{\hat{P}})))\overset{a.s.}{\rightarrow}f(BB({Q}%
(\cdot,\underline{P}))),
\]
which implies
\begin{equation}
\mathbb{E[}f(BB({Q}(\cdot,\underline{\hat{P}})))-f(BB({Q}(\cdot,\underline{P}%
)))]\rightarrow0 \label{subsequence_convergence}%
\end{equation}
by the dominated convergence theorem. The above argument shows that for any
subsequence, there is a further subsequence such that
(\ref{subsequence_convergence}). So (\ref{subsequence_convergence}) must hold
as $n\rightarrow\infty$.

Combining the convergence results of the two terms in
(\ref{decomposition_weak_conv}), we know that%
\[
\mathbb{E}f(\sqrt{R}(\hat{Q}(\cdot,\underline{\hat{P}})-Q(\cdot
,\underline{\hat{P}})))\rightarrow\mathbb{E}f(BB({Q}(\cdot,\underline{P})))
\]
as $n,R\rightarrow\infty$ for any bounded and Lipschitz function $f$.
Therefore, we have%
\[
\sqrt{R}(\hat{Q}(\cdot,\underline{\hat{P}})-Q(\cdot,\underline{\hat{P}%
}))\Rightarrow BB({Q}(\cdot,\underline{P}))
\]
as $n,R\rightarrow\infty$. Since $(D_{\pm\infty},\mathcal{D}_{\pm\infty
},d_{\pm\infty}^{\circ})$ is complete and separable, by the Prokhorov's
theorem (\cite{billingsley2013convergence} Theorem 5.2), $\sqrt{R}(\hat
{Q}(\cdot,\underline{\hat{P}})-Q(\cdot,\underline{\hat{P}}))$ is tight.
\end{proof}

We next establish the tightness of $\sqrt{n}(Q(\cdot,\underline{\hat{P}%
})-Q(\cdot,\underline{P}))$ which highly relies on the structure of the V-statistic.

\begin{lemma}
\label{tightness_GP}Suppose Assumptions \ref{balanced_data} and
\ref{finite_horizon_model} hold. Then
\[
\sqrt{n}(Q(\cdot,\underline{\hat{P}})-Q(\cdot,\underline{P}))\Rightarrow
\mathbb{G}(\cdot),
\]
where $\mathbb{G}(\cdot)$ has the finite-dimensional distribution specified in
Proposition \ref{fidis_convergence} and
\[
\mathbb{P}(\mathbb{G}\in C_{b}(\mathbb{R}))=1.
\]

\end{lemma}

\begin{proof}{Proof of Lemma \ref{tightness_GP}.}
Based on Assumption \ref{finite_horizon_model}, $Q(\cdot,\underline{\hat{P}})$
is a multi-sample V-statistic given by%
\[
Q(t,\underline{\hat{P}})=\frac{1}{\prod_{i=1}^{m}n_{i}^{T_{i}}}\sum_{i=1}%
^{m}\sum_{1\leq j_{i1},\ldots,j_{iT_{i}}\leq n_{i}}I(h(X_{1,j_{11}}%
,\ldots,X_{1,j_{1T_{1}}},\ldots,X_{m,j_{m1}},\ldots,X_{m,j_{mT_{m}}})\leq t).
\]
When $n_{i}\geq T_{i}$, we define the corresponding multi-sample U-statistic
by%
\begin{equation}
Q_{U}(t,\underline{\hat{P}})=\frac{1}{\prod_{i=1}^{m}\prod_{j=0}^{T_{i}%
-1}(n_{i}-j)}\sum_{i=1}^{m}\sum_{1\leq j_{i1}\neq\cdots\neq j_{iT_{i}}\leq
n_{i}}I(h(X_{1,j_{11}},\ldots,X_{1,j_{1T_{1}}},\ldots,X_{m,j_{m1}}%
,\ldots,X_{m,j_{mT_{m}}})\leq t). \label{U_statistic}%
\end{equation}
We first show the pointwise convergence%
\[
\sqrt{n}||Q(\cdot,\underline{\hat{P}})-Q_{U}(\cdot,\underline{\hat{P}%
})||_{\infty}\rightarrow0.
\]
In fact, for any $t\in\mathbb{R}$%
\begin{align*}
&  |Q(t,\underline{\hat{P}})-Q_{U}(t,\underline{\hat{P}})|\\
&  \leq\left\vert \frac{1}{\prod_{i=1}^{m}n_{i}^{T_{i}}}\sum_{i=1}^{m}%
\sum_{1\leq j_{i1},\ldots,j_{iT_{i}}\leq n_{i}}I(h(X_{1,j_{11}},\ldots
,X_{1,j_{1T_{1}}},\ldots,X_{m,j_{m1}},\ldots,X_{m,j_{mT_{m}}})\leq t)\right.
\\
&  \left.  -\frac{1}{\prod_{i=1}^{m}n_{i}^{T_{i}}}\sum_{i=1}^{m}\sum_{1\leq
j_{i1}\neq\cdots\neq j_{iT_{i}}\leq n_{i}}I(h(X_{1,j_{11}},\ldots
,X_{1,j_{1T_{1}}},\ldots,X_{m,j_{m1}},\ldots,X_{m,j_{mT_{m}}})\leq
t)\right\vert \\
&  +\left\vert \frac{1}{\prod_{i=1}^{m}n_{i}^{T_{i}}}\sum_{i=1}^{m}\sum_{1\leq
j_{i1}\neq\cdots\neq j_{iT_{i}}\leq n_{i}}I(h(X_{1,j_{11}},\ldots
,X_{1,j_{1T_{1}}},\ldots,X_{m,j_{m1}},\ldots,X_{m,j_{mT_{m}}})\leq t)\right.
\\
&  \left.  -\frac{1}{\prod_{i=1}^{m}\prod_{j=0}^{T_{i}-1}(n_{i}-j)}\sum
_{i=1}^{m}\sum_{1\leq j_{i1}\neq\cdots\neq j_{iT_{i}}\leq n_{i}}%
I(h(X_{1,j_{11}},\ldots,X_{1,j_{1T_{1}}},\ldots,X_{m,j_{m1}},\ldots
,X_{m,j_{mT_{m}}})\leq t)\right\vert \\
&  \leq\frac{\prod_{i=1}^{m}n_{i}^{T_{i}}-\prod_{i=1}^{m}\prod_{j=0}^{T_{i}%
-1}(n_{i}-j)}{\prod_{i=1}^{m}n_{i}^{T_{i}}}+\prod_{i=1}^{m}\prod_{j=0}%
^{T_{i}-1}(n_{i}-j)\left(  \frac{1}{\prod_{i=1}^{m}\prod_{j=0}^{T_{i}-1}%
(n_{i}-j)}-\frac{1}{\prod_{i=1}^{m}n_{i}^{T_{i}}}\right)  ,
\end{align*}
which implies%
\begin{align}
\sqrt{n}||Q(\cdot,\underline{\hat{P}})-Q_{U}(\cdot,\underline{\hat{P}%
})||_{\infty}  &  \leq\sqrt{n}\left(  2-2\frac{\prod_{i=1}^{m}\prod
_{j=0}^{T_{i}-1}(n_{i}-j)}{\prod_{i=1}^{m}n_{i}^{T_{i}}}\right) \nonumber\\
&  =2\sqrt{n}\frac{\prod_{i=1}^{m}n_{i}^{T_{i}}-\prod_{i=1}^{m}\prod
_{j=0}^{T_{i}-1}(n_{i}-j)}{\prod_{i=1}^{m}n_{i}^{T_{i}}}.
\label{UV_difference}%
\end{align}
Note that for each $i=1,\ldots,m$,%
\[
\prod_{j=0}^{T_{i}-1}(n_{i}-j)=n_{i}^{T_{i}}-\frac{(T_{i}-1)T_{i}}{2}%
n_{i}^{T_{i}-1}+o(n_{i}^{T_{i}-1})
\]
and thus%
\begin{align}
\prod_{i=1}^{m}\prod_{j=0}^{T_{i}-1}(n_{i}-j)  &  =\prod_{i=1}^{m}\left(
n_{i}^{T_{i}}-\frac{(T_{i}-1)T_{i}}{2}n_{i}^{T_{i}-1}+o(n_{i}^{T_{i}%
-1})\right) \nonumber\\
&  =\prod_{i=1}^{m}n_{i}^{T_{i}}-\left(  \sum_{i=1}^{m}\frac{(T_{i}-1)T_{i}%
}{2n_{i}}\right)  \prod_{i=1}^{m}n_{i}^{T_{i}}+o\left(  n^{\left(  \sum
_{i=1}^{m}T_{i}\right)  -1}\right)  , \label{product_order1}%
\end{align}
where the $o(\cdot)$ in the second inequality is due to $n_{i}=\Theta(n)$ by
Assumption \ref{balanced_data}. Plugging (\ref{product_order1}) into
(\ref{UV_difference}), we have%
\begin{equation}
\sqrt{n}||Q(\cdot,\underline{\hat{P}})-Q_{U}(\cdot,\underline{\hat{P}%
})||_{\infty}\leq2\sqrt{n}\frac{\left(  \sum_{i=1}^{m}\frac{(T_{i}-1)T_{i}%
}{2n_{i}}\right)  \prod_{i=1}^{m}n_{i}^{T_{i}}-o\left(  n^{\left(  \sum
_{i=1}^{m}T_{i}\right)  -1}\right)  }{\prod_{i=1}^{m}n_{i}^{T_{i}}}=O\left(
\frac{1}{\sqrt{n}}\right)  \rightarrow0 \label{UV_negligible_diff}%
\end{equation}
as $n\rightarrow\infty$. Therefore, to prove the tightness of $\sqrt
{n}(Q(\cdot,\underline{\hat{P}})-Q(\cdot,\underline{P}))$, it suffices to
study the weak convergence of $\sqrt{n}(Q_{U}(\cdot,\underline{\hat{P}%
})-Q(\cdot,\underline{P}))$ since by Theorem 3.1 in
\cite{billingsley2013convergence}, the weak convergence of $\sqrt{n}%
(Q_{U}(\cdot,\underline{\hat{P}})-Q(\cdot,\underline{P}))$ and%
\begin{align*}
&  d_{\pm\infty}^{\circ}(\sqrt{n}(Q(\cdot,\underline{\hat{P}})-Q(\cdot
,\underline{P})),\sqrt{n}(Q_{U}(\cdot,\underline{\hat{P}})-Q(\cdot
,\underline{P})))\\
&  \leq||\sqrt{n}(Q(\cdot,\underline{\hat{P}})-Q(\cdot,\underline{P}%
))-\sqrt{n}(Q_{U}(\cdot,\underline{\hat{P}})-Q(\cdot,\underline{P}%
))||_{\infty}\\
&  =\sqrt{n}||Q(\cdot,\underline{\hat{P}})-Q_{U}(\cdot,\underline{\hat{P}%
})||_{\infty}\rightarrow0
\end{align*}
imply the weak convergence (and thus tightness) of $\sqrt{n}(Q(\cdot
,\underline{\hat{P}})-Q(\cdot,\underline{P}))$.

Now we study the weak convergence of $\sqrt{n}(Q_{U}(\cdot,\underline{\hat{P}%
})-Q(\cdot,\underline{P}))$. We start by showing the weak convergence of
\[
\sqrt{n}(Q_{U}(Q^{-1}(s,\underline{P}),\underline{\hat{P}})-Q(Q^{-1}%
(s,\underline{P}),\underline{P}))=\sqrt{n}(Q_{U}(Q^{-1}(s,\underline{P}%
),\underline{\hat{P}})-s),s\in\lbrack0,1]
\]
on the space $D[0,1]$ where $Q^{-1}(s,\underline{P})$ is the generalized
inverse defined by%
\[
Q^{-1}(s,\underline{P})=\inf\{t:Q(t,\underline{P})\geq s\}
\]
and the equality $Q(Q^{-1}(s,\underline{P}),\underline{P})=s$ holds because
$Q(t,\underline{P})$ is a continuous function. First, we need to verify that
$Q_{U}(Q^{-1}(\cdot,\underline{P}),\underline{\hat{P}})$ is indeed a random
element in $D[0,1]$, i.e., $Q_{U}(Q^{-1}(\cdot,\underline{P}),\underline{\hat
{P}})$ has paths in $D[0,1]$ almost surely. We note that $Q^{-1}%
(\cdot,\underline{P})$ is left-continuous and has right limits, and
$Q_{U}(\cdot,\underline{\hat{P}})$ is right-continuous and has left limits.
Moreover, both $Q^{-1}(\cdot,\underline{P})$ and $Q_{U}(\cdot,\underline{\hat
{P}})$ are non-decreasing. Therefore, $Q_{U}(Q^{-1}(\cdot,\underline{P}%
),\underline{\hat{P}})$ has left limits. In order to show almost surely all
paths of $Q_{U}(Q^{-1}(\cdot,\underline{P}),\underline{\hat{P}})$ are
right-continuous, we consider two cases based on whether $Q^{-1}%
(\cdot,\underline{P})$ is right-continuous at $s\in\lbrack0,1)$. If
$Q^{-1}(\cdot,\underline{P})$ is right-continuous at $s$, then $Q_{U}%
(Q^{-1}(\cdot,\underline{P}),\underline{\hat{P}})$ is clearly right-continuous
at $s$ as well. If $Q^{-1}(\cdot,\underline{P})$ is not right-continuous at
$s$, i.e., $\lim_{t\downarrow s}Q^{-1}(t,\underline{P}):=x_{0}>Q^{-1}%
(s,\underline{P})$, by the continuity of $Q(\cdot,\underline{P})$, we can see
that $Q(\cdot,\underline{P})$ takes constant value $s$ in the interval
$[Q^{-1}(s,\underline{P}),x]$, i.e., the probability distribution of
$h(\mathbf{X}_{1},\ldots,\mathbf{X}_{m})$ puts zero mass in $(Q^{-1}%
(s,\underline{P}),x]$. Therefore,%
\[
\lim_{t\downarrow s}Q_{U}(Q^{-1}(t,\underline{P}),\underline{\hat{P}}%
)=Q_{U}\left(  \lim_{t\downarrow s}Q^{-1}(t,\underline{P}),\underline{\hat{P}%
}\right)  =Q_{U}(x_{0},\underline{\hat{P}})=Q_{U}(Q^{-1}(s,\underline{P}%
),\underline{\hat{P}}),\text{a.s.,}%
\]
where the first equality uses the right-continuity of $Q_{U}$ and the last
continuity holds almost surely since the empirical distribution of the
U-statistics also puts zero mass in $(Q^{-1}(s,\underline{P}),x]$ almost
surely. This proves that, almost surely all paths of $Q_{U}(Q^{-1}%
(\cdot,\underline{P}),\underline{\hat{P}})$ are right-continuous at a jump
point $s$ of $Q^{-1}(\cdot,\underline{P})$. The fact that all jump points of
$Q^{-1}(\cdot,\underline{P})$ are countable finally implies that almost surely
all paths of $Q_{U}(Q^{-1}(\cdot,\underline{P}),\underline{\hat{P}})$ are
right-continuous in $[0,1]$, which concludes that $Q_{U}(Q^{-1}(\cdot
,\underline{P}),\underline{\hat{P}})$ is a random element in $D[0,1]$. Now we
start to show the weak convergence. Finite-dimensional convergence follows
from Proposition \ref{fidis_convergence}, (\ref{UV_negligible_diff}) and
Slutsky's Theorem. So it suffices to show the tightness of $Y_{n}(s):=\sqrt
{n}(Q_{U}(Q^{-1}(s,\underline{P}),\underline{\hat{P}})-s)$. The proof is
generalized from \cite{silverman1976limit} which proves the tightness of the
single-sample U empirical process. The main idea is to partition the summation
in $Q_{U}$ into a few parts such that the summation is conducted for i.i.d.
random variables within each part. Recall that%
\[
Q_{U}(\cdot,\underline{\hat{P}})=\frac{1}{\prod_{i=1}^{m}\prod_{j=0}^{T_{i}%
-1}(n_{i}-j)}\sum_{i=1}^{m}\sum_{1\leq j_{i1}\neq\cdots\neq j_{iT_{i}}\leq
n_{i}}I(h(X_{1,j_{11}},\ldots,X_{1,j_{1T_{1}}},\ldots,X_{m,j_{m1}}%
,\ldots,X_{m,j_{mT_{m}}})\leq\cdot).
\]
In the following, we assume $n_{i}\geq T_{i}$. We collect the indices used in
the summation of $Q_{U}$:%
\[
S(\mathbf{n},\mathbf{T}):=\{(j_{11},\ldots,j_{1T_{1}},\ldots,j_{m1}%
,\ldots,j_{mT_{m}}):1\leq j_{i1}\neq\cdots\neq j_{iT_{i}}\leq n_{i},1\leq
i\leq m\}.
\]
We construct a graph $\mathcal{G}$ with vertices labelled by the elements in
$S(\mathbf{n},\mathbf{T})$ and with an edge between%
\[
J=(j_{11},\ldots,j_{1T_{1}},\ldots,j_{m1},\ldots,j_{mT_{m}}),\quad
K=(k_{11},\ldots,k_{1T_{1}},\ldots,k_{m1},\ldots,k_{mT_{m}})
\]
if and only if%
\[
\bigcup_{i=1}^{m}(\{j_{i1},\ldots,j_{iT_{i}}\}\cap\{k_{i1},\ldots,k_{iT_{i}%
}\})\neq\emptyset.
\]
In the graph theory, the chromatic number $\chi(\mathcal{G})$ is the smallest
number of colors needed to color the vertices of $\mathcal{G}$ such that no
adjacent vertices have the same color. Let $\Delta(\mathcal{G})$ be the
maximum degree of vertices in $\mathcal{G}$. Then the greedy coloring
algorithm (\cite{west2001introduction} Proposition 5.1.13) shows that%
\[
\chi(\mathcal{G})\leq\Delta(\mathcal{G})+1,
\]
which implies that we can use $\Delta(\mathcal{G})+1$ colors to color the
vertices of $\mathcal{G}$ so that no adjacent vertices have the same color. So
we can partition the vertices
\[
S(\mathbf{n},\mathbf{T})=\bigcup_{i=1}^{\Delta(\mathcal{G})+1}S_{i},
\]
where $S_{i}$ is composed of vertices with color $i$. By the property of the
coloring, we know that the random variables
\[
h(X_{1,j_{11}},\ldots,X_{1,j_{1T_{1}}},\ldots,X_{m,j_{m1}},\ldots
,X_{m,j_{mT_{m}}}),\quad(j_{11},\ldots,j_{1T_{1}},\ldots,j_{m1},\ldots
,j_{mT_{m}})\in S_{i}%
\]
are mutually independent. We write $|S_{i}|$ as the cardinality of $S_{i}$ and
define%
\begin{align*}
&  H_{i}(s)\\
&  =\sqrt{|S_{i}|}\left(  \frac{1}{|S_{i}|}\sum_{(j_{11},\ldots,j_{mT_{m}})\in
S_{i}}I(h(X_{1,j_{11}},\ldots,X_{1,j_{1T_{1}}},\ldots,X_{m,j_{m1}}%
,\ldots,X_{m,j_{mT_{m}}})\leq Q^{-1}(s,\underline{P}))-s\right)  .
\end{align*}
Note that $H_{i}(s)$ has the same finite-dimensional distributions as the
empirical process of $|S_{i}|$ i.i.d. $U[0,1]$ random variables due to
$Q(Q^{-1}(s,\underline{P}),\underline{P})=s$ from the continuity of
$Q(t,\underline{P})$. Besides, by (\ref{generating_system}), the class of
finite-dimensional sets generates $\mathcal{D}_{\pm\infty}$ and is a $\pi
$-system. Therefore, by Theorem 3.3 in \cite{billingsley2008probability}, as a random element on $D_{\pm\infty}$, $H_{i}(s)$ has the same distribution as the empirical process of $|S_{i}|$
i.i.d. $U[0,1]$ random variables. Recall the
modulus of continuity defined in (\ref{modulus_of_continuity}):%
\[
w_{x}(\delta)=\sup_{|s-t|\leq\delta}|x(s)-x(t)|,\quad x:[0,1]\mapsto
\mathbb{R}.
\]
Then by Lemma 3 in \cite{silverman1976limit}, we have%
\begin{equation}
\mathbb{E}[w_{H_{i}}(\delta)]\leq K(\alpha(|S_{i}|)+\beta(\delta)),
\label{modulus_bound_Hi}%
\end{equation}
where $K$ is a constant,%
\[
\alpha(x)=\left\{
\begin{array}
[c]{l}%
x^{-1/2}\log x\\
x^{-1/2}(\log2-(2-x)/2)
\end{array}
\left.
\begin{array}
[c]{l}%
\text{if }x\geq2\\
\text{if }1\leq x\leq2
\end{array}
\right.  \right.  ,
\]
and%
\[
\beta(x)=\left\{
\begin{array}
[c]{l}%
(-x\log x)^{1/2}\\
e^{-1/2}%
\end{array}
\left.
\begin{array}
[c]{l}%
\text{if }0<x\leq e^{-1}\\
\text{if }x>e^{-1}%
\end{array}
\right.  \right.  .
\]
Besides, we can represent $Y_{n}$ as%
\[
Y_{n}(s)=\sqrt{n}(Q_{U}(Q^{-1}(s,\underline{P}),\underline{\hat{P}}%
)-s)=\frac{\sqrt{n}}{|S(\mathbf{n},\mathbf{T})|}\sum_{i=1}^{\Delta
(\mathcal{G})+1}\sqrt{|S_{i}|}H_{i}(s),
\]
which implies%
\[
w_{Y_{n}}(\delta)\leq\frac{\sqrt{n}}{|S(\mathbf{n},\mathbf{T})|}\sum
_{i=1}^{\Delta(\mathcal{G})+1}\sqrt{|S_{i}|}w_{H_{i}}(\delta).
\]
Therefore, by (\ref{modulus_bound_Hi}) we have%
\begin{equation}
\mathbb{E}[w_{Y_{n}}(\delta)]\leq\frac{K\sqrt{n}}{|S(\mathbf{n},\mathbf{T}%
)|}\sum_{i=1}^{\Delta(\mathcal{G})+1}\sqrt{|S_{i}|}(\alpha(|S_{i}%
|)+\beta(\delta)). \label{modulus_bound_Y1}%
\end{equation}
Since $x^{1/2}$ and $x^{1/2}\alpha(x)$ are concave functions in $[1,\infty)$,
by Jensen's inequality, we have%
\begin{equation}
\frac{1}{\Delta(\mathcal{G})+1}\sum_{i=1}^{\Delta(\mathcal{G})+1}\sqrt
{|S_{i}|}\leq\sqrt{\frac{1}{\Delta(\mathcal{G})+1}\sum_{i=1}^{\Delta
(\mathcal{G})+1}|S_{i}|}=\sqrt{\frac{|S(\mathbf{n},\mathbf{T})|}%
{\Delta(\mathcal{G})+1}}, \label{Jensen1}%
\end{equation}
and%
\begin{equation}
\frac{1}{\Delta(\mathcal{G})+1}\sum_{i=1}^{\Delta(\mathcal{G})+1}\sqrt
{|S_{i}|}\alpha(|S_{i}|)\leq\sqrt{\frac{|S(\mathbf{n},\mathbf{T})|}%
{\Delta(\mathcal{G})+1}}\alpha\left(  \frac{|S(\mathbf{n},\mathbf{T})|}%
{\Delta(\mathcal{G})+1}\right)  . \label{Jensen2}%
\end{equation}
Plugging (\ref{Jensen1}) and (\ref{Jensen2}) into (\ref{modulus_bound_Y1}), we
obtain%
\begin{align}
\mathbb{E}[w_{Y_{n}}(\delta)]  &  \leq\frac{K\sqrt{n}(\Delta(\mathcal{G}%
)+1)}{|S(\mathbf{n},\mathbf{T})|}\sqrt{\frac{|S(\mathbf{n},\mathbf{T}%
)|}{\Delta(\mathcal{G})+1}}\left(  \alpha\left(  \frac{|S(\mathbf{n}%
,\mathbf{T})|}{\Delta(\mathcal{G})+1}\right)  +\beta(\delta)\right)
\nonumber\\
&  =K\sqrt{n}\sqrt{\frac{\Delta(\mathcal{G})+1}{|S(\mathbf{n},\mathbf{T})|}%
}\left(  \alpha\left(  \frac{|S(\mathbf{n},\mathbf{T})|}{\Delta(\mathcal{G}%
)+1}\right)  +\beta(\delta)\right)  . \label{modulus_bound_Y2}%
\end{align}
From the construction of $\mathcal{G}$, we know that%
\[
\Delta(\mathcal{G})+1=\prod_{i=1}^{m}\prod_{j=0}^{T_{i}-1}(n_{i}%
-j)-\prod_{i=1}^{m}\prod_{j=0}^{T_{i}-1}(n_{i}-T_{i}-j),
\]
and
\[
|S(\mathbf{n},\mathbf{T})|=\prod_{i=1}^{m}\prod_{j=0}^{T_{i}-1}(n_{i}-j).
\]
A similar argument like (\ref{product_order1}) shows that%
\begin{align}
\prod_{i=1}^{m}\prod_{j=0}^{T_{i}-1}(n_{i}-T_{i}-j)  &  =\prod_{i=1}%
^{m}\left(  n_{i}^{T_{i}}-\frac{(3T_{i}-1)T_{i}}{2}n_{i}^{T_{i}-1}%
+o(n_{i}^{T_{i}-1})\right) \nonumber\\
&  =\prod_{i=1}^{m}n_{i}^{T_{i}}-\left(  \sum_{i=1}^{m}\frac{(3T_{i}-1)T_{i}%
}{2n_{i}}\right)  \prod_{i=1}^{m}n_{i}^{T_{i}}+o\left(  n^{\left(  \sum
_{i=1}^{m}T_{i}\right)  -1}\right)  . \label{product_order2}%
\end{align}
Therefore, by (\ref{product_order1}) and (\ref{product_order2}),
\begin{align*}
n\frac{\Delta(\mathcal{G})+1}{|S(\mathbf{n},\mathbf{T})|}  &  =n\frac
{\prod_{i=1}^{m}\prod_{j=0}^{T_{i}-1}(n_{i}-j)-\prod_{i=1}^{m}\prod
_{j=0}^{T_{i}-1}(n_{i}-T_{i}-j)}{\prod_{i=1}^{m}\prod_{j=0}^{T_{i}-1}%
(n_{i}-j)}\\
&  =n\frac{\left(  \sum_{i=1}^{m}\frac{T_{i}^{2}}{n_{i}}\right)  \prod
_{i=1}^{m}n_{i}^{T_{i}}+o\left(  n^{\left(  \sum_{i=1}^{m}T_{i}\right)
-1}\right)  }{\prod_{i=1}^{m}n_{i}^{T_{i}}-\left(  \sum_{i=1}^{m}\frac
{(T_{i}-1)T_{i}}{2n_{i}}\right)  \prod_{i=1}^{m}n_{i}^{T_{i}}+o\left(
n^{\left(  \sum_{i=1}^{m}T_{i}\right)  -1}\right)  }\\
&  =\frac{\sum_{i=1}^{m}T_{i}^{2}\frac{n}{n_{i}}+o(1)}{1-O\left(  \frac{1}%
{n}\right)  }\rightarrow\sum_{i=1}^{m}\frac{T_{i}^{2}}{\beta_{i}},
\end{align*}
where we use $\lim_{n\rightarrow\infty}n_{i}/n=\beta_{i}$ from Assumption
\ref{balanced_data} when analyzing the orders and obtaining the final limit.
With this limit and $\lim_{x\rightarrow\infty}\alpha(x)=\lim_{x\rightarrow
0}\beta(x)=0$, it follows from (\ref{modulus_bound_Y2}) that%
\[
\lim_{\delta\rightarrow0,n\rightarrow\infty}\mathbb{E}[w_{Y_{n}}(\delta)]=0.
\]
So we can see (7.7) in \cite{billingsley2013convergence} holds. (7.6) in
\cite{billingsley2013convergence} is trivial because $Y_{n}(0)=0$ a.s.. Thus,
by Corollary below Theorem 13.4 in \cite{billingsley2013convergence}, we have%
\[
Y_{n}(s)=\sqrt{n}(Q_{U}(Q^{-1}(s,\underline{P}),\underline{\hat{P}%
})-s)\Rightarrow\mathbb{G}^{\prime}(s),
\]
where the finite-dimensional distributions of $\mathbb{G}^{\prime}$ agrees
with the limiting finite-dimensional distributions of $Y_{n}$ and
$\mathbb{P}(\mathbb{G}^{\prime}\in C[0,1])=1$.

Next, we deduce the weak convergence of $\sqrt{n}(Q_{U}(t,\underline{\hat{P}%
})-Q(t,\underline{P}))$ from $Y_{n}$. We define a function $\psi:D[0,1]\mapsto
D_{\pm\infty}$ by $(\psi x)(\cdot)=x(Q(\cdot,\underline{P}))$ for $x\in
D[0,1]$. $\psi$ is a measurable function since the finite-dimensional sets are
measurable:%
\[
\{x:((\psi x)(t_{1}),\ldots,(\psi x)(t_{k}))\in H\}=\{x:(x(Q(t_{1}%
,\underline{P})),\ldots,x(Q(t_{1},\underline{P})))\in H\}\in\mathcal{D}[0,1]
\]
for $H\in\mathcal{B}(\mathbb{R}^{k})$. We show $\psi$ is continuous on
$C[0,1]$. Suppose $d^{\circ}(x_{n},x)\rightarrow0$ for $x_{n}\in D[0,1]$ and
$x\in C[0,1]$, where $d^{\circ}$ is the metric making $D[0,1]$ complete (see
\cite{billingsley2013convergence} Section 12). When $x\in$ $C[0,1]$, Skorohod
convergence implies uniform convergence, we have
\[
\sup_{t\in\lbrack0,1]}|x_{n}(t)-x(t)|\rightarrow0
\]
and thus%
\[
d_{\pm\infty}^{\circ}(\psi x_{n},\psi x)\leq||\psi x_{n}-\psi x||_{\infty}%
\leq\sup_{t\in\lbrack0,1]}|x_{n}(t)-x(t)|\rightarrow0,
\]
which means $\psi$ is continuous on $C[0,1]$. Since $\mathbb{P}(\mathbb{G}%
^{\prime}\in C[0,1])=1$, by continuous mapping theorem, we know that%
\[
\psi(Y_{n})\Rightarrow\psi(\mathbb{G}^{\prime}),
\]
i.e.,%
\begin{equation}
\sqrt{n}(Q_{U}(Q^{-1}(Q(\cdot,\underline{P}),\underline{P}),\underline{\hat
{P}})-Q(\cdot,\underline{P}))\Rightarrow\mathbb{G}^{\prime}(Q(\cdot
,\underline{P})):=\mathbb{G}(\cdot) \label{weak_convergence_U1}%
\end{equation}
on $D_{\pm\infty}$. Then we need to show $Q_{U}(Q^{-1}(Q(\cdot,\underline{P}%
),\underline{P}),\underline{\hat{P}})=Q_{U}(\cdot,\underline{\hat{P}})$ a.s..
Since $Q^{-1}(Q(t,\underline{P}),\underline{P})\leq t$, we can see%
\begin{align*}
&  ||Q_{U}(Q^{-1}(Q(\cdot,\underline{P}),\underline{P}),\underline{\hat{P}%
})-Q_{U}(\cdot,\underline{\hat{P}})||_{\infty}\\
&  =\sup_{t\in\mathbb{Q}}|Q_{U}(Q^{-1}(Q(t,\underline{P}),\underline{P}%
),\underline{\hat{P}})-Q_{U}(t,\underline{\hat{P}})|\\
&  =\sup_{t\in\mathbb{Q}}\frac{1}{|S(\mathbf{n},\mathbf{T})|}\sum
_{(j_{11},\ldots,j_{mT_{m}})\in S(\mathbf{n},\mathbf{T})}I(Q^{-1}%
(Q(t,\underline{P}),\underline{P})<h(X_{1,j_{11}},\ldots,X_{m,j_{mT_{m}}})\leq
t).
\end{align*}
However, by the continuity of $Q(t,\underline{P})$, we have
\begin{align*}
&  P(Q^{-1}(Q(t,\underline{P}),\underline{P})<h(X_{1,j_{11}},\ldots
,X_{m,j_{mT_{m}}})\leq t)\\
&  =Q(t,\underline{P})-Q(Q^{-1}(Q(t,\underline{P}),\underline{P}%
),\underline{P})\\
&  =0,
\end{align*}
which implies%
\[
\frac{1}{|S(\mathbf{n},\mathbf{T})|}\sum_{(j_{11},\ldots,j_{mT_{m}})\in
S(\mathbf{n},\mathbf{T})}I(Q^{-1}(Q(t,\underline{P}),\underline{P}%
)<h(X_{1,j_{11}},\ldots,X_{m,j_{mT_{m}}})\leq t)=0\text{ a.s.,}%
\]
and%
\begin{align*}
&  ||Q_{U}(Q^{-1}(Q(\cdot,\underline{P}),\underline{P}),\underline{\hat{P}%
})-Q_{U}(\cdot,\underline{\hat{P}})||_{\infty}\\
&  =\sup_{t\in\mathbb{Q}}\frac{1}{|S(\mathbf{n},\mathbf{T})|}\sum
_{(j_{11},\ldots,j_{mT_{m}})\in S(\mathbf{n},\mathbf{T})}I(Q^{-1}%
(Q(t,\underline{P}),\underline{P})<h(X_{1,j_{11}},\ldots,X_{m,j_{mT_{m}}})\leq
t)\\
&  =0\text{ a.s..}%
\end{align*}
In other words, $Q_{U}(Q^{-1}(Q(\cdot,\underline{P}),\underline{P}%
),\underline{\hat{P}})$ and $Q_{U}(\cdot,\underline{\hat{P}})$ are
indistinguishable as stochastic processes. Then from
(\ref{weak_convergence_U1}), we obtain%
\[
\sqrt{n}(Q_{U}(\cdot,\underline{\hat{P}})-Q(\cdot,\underline{P}))\Rightarrow
\mathbb{G}(\cdot).
\]
Further, by (\ref{UV_negligible_diff}) and the discussion below it, we obtain%
\[
\sqrt{n}(Q(\cdot,\underline{\hat{P}})-Q(\cdot,\underline{P}))\Rightarrow
\mathbb{G}(\cdot),
\]
and $\mathbb{G}(\cdot)$ must have the finite-dimensional distributions
specified in Proposition \ref{fidis_convergence}.

Finally, we show that $\mathbb{P}(\mathbb{G}\in C_{b}(\mathbb{R}))=1$. Since
$\mathbb{P}(\mathbb{G}^{\prime}\in C[0,1])=1$ and $\mathbb{G}(\cdot
)=\mathbb{G}^{\prime}(Q(\cdot,\underline{P}))$, we know that almost all the
paths of $\mathbb{G}$ are in $C_{b}(\mathbb{R})$. Therefore, it suffices to
show $C_{b}(\mathbb{R})$ is a measurable set in $(D_{\pm\infty},\mathcal{D}%
_{\pm\infty})$. Note that
\begin{equation}
C_{b}(\mathbb{R})=\left(  \bigcup_{k=1}^{\infty}\{x:||x||_{\infty}\leq
k\}\right)  \bigcap\left(  \bigcap_{k=1}^{\infty}\left\{  x:\sup
_{t\in\mathbb{R}}|x(t)-x(t-)|\leq\frac{1}{k}\right\}  \right)  . \label{C_b}%
\end{equation}
So it suffices to show each set is in $\mathcal{D}_{\pm\infty}$. For
$\{x:||x||_{\infty}\leq k\}$, we have%
\begin{equation}
\{x:||x||_{\infty}\leq k\}=\bigcap_{t\in\mathbb{Q}}\{x:|x(t)|\leq
k\}\in\mathcal{D}_{\pm\infty} \label{measurability_supremum}%
\end{equation}
since each finite-dimensional set is measurable. Next, we consider
$\{x:\sup_{t\in\mathbb{R}}|x(t)-x(t-)|\leq1/k\}$. We will prove it is a closed
set. Suppose that $d_{\pm\infty}^{\circ}(x_{n},x)\rightarrow0$ and $\sup
_{t\in\mathbb{R}}|x_{n}(t)-x_{n}(t-)|\leq1/k$. We need to show $\sup
_{t\in\mathbb{R}}|x(t)-x(t-)|\leq1/k$. First, by (5) in \cite{ferger2010weak},
$d_{\pm\infty}^{\circ}(x_{n},x)\rightarrow0$ implies that $\exists$ strictly
increasing, continuous, surjective mappings $\lambda_{n}$ from $\mathbb{R}$
onto itself such that%
\begin{equation}
\sup_{|t|\leq m}|x(t)-x_{n}\circ\lambda_{n}(t)|\rightarrow0
\label{deformation_uni_conv}%
\end{equation}
for any $m\in\mathbb{N}$. For $|t|\leq m$, we know that%
\begin{align*}
&  |x(t)-x(t-)-(x_{n}\circ\lambda_{n}(t)-x_{n}\circ\lambda_{n}(t-))|\\
&  \leq|x(t)-x_{n}\circ\lambda_{n}(t)|+|x(t-)-x_{n}\circ\lambda_{n}(t-)|\\
&  =|x(t)-x_{n}\circ\lambda_{n}(t)|+\lim_{s\uparrow t}|x(s)-x_{n}\circ
\lambda_{n}(s)|\\
&  \leq2\sup_{|t|\leq m+1}|x(t)-x_{n}\circ\lambda_{n}(t)|,
\end{align*}
which further implies that%
\begin{align*}
&  \left\vert \sup_{|t|\leq m}|x(t)-x(t-)|-\sup_{|t|\leq m}|x_{n}\circ
\lambda_{n}(t)-x_{n}\circ\lambda_{n}(t-)|\right\vert \\
&  \leq\sup_{|t|\leq m}|x(t)-x(t-)-(x_{n}\circ\lambda_{n}(t)-x_{n}\circ
\lambda_{n}(t-))|\\
&  \leq2\sup_{|t|\leq m+1}|x(t)-x_{n}\circ\lambda_{n}(t)|.
\end{align*}
Therefore, we have%
\begin{align*}
\sup_{|t|\leq m}|x(t)-x(t-)|  &  \leq\sup_{|t|\leq m}|x_{n}\circ\lambda
_{n}(t)-x_{n}\circ\lambda_{n}(t-)|+2\sup_{|t|\leq m+1}|x(t)-x_{n}\circ
\lambda_{n}(t)|\\
&  \leq\sup_{t\in\mathbb{R}}|x_{n}\circ\lambda_{n}(t)-x_{n}\circ\lambda
_{n}(t-)|+2\sup_{|t|\leq m+1}|x(t)-x_{n}\circ\lambda_{n}(t)|\\
&  =\sup_{t\in\mathbb{R}}|x_{n}(t)-x_{n}(t-)|+2\sup_{|t|\leq m+1}%
|x(t)-x_{n}\circ\lambda_{n}(t)|\\
&  \leq\frac{1}{k}+2\sup_{|t|\leq m+1}|x(t)-x_{n}\circ\lambda_{n}(t)|.
\end{align*}
Letting $n\rightarrow\infty$ and by means of (\ref{deformation_uni_conv}), we
have%
\[
\sup_{|t|\leq m}|x(t)-x(t-)|\leq\frac{1}{k}.
\]
Letting $k\rightarrow\infty$, we have%
\[
\sup_{t\in\mathbb{R}}|x(t)-x(t-)|=\lim_{m\rightarrow\infty}\sup_{|t|\leq
m}|x(t)-x(t-)|\leq\frac{1}{k},
\]
which proves $\{x:\sup_{t\in\mathbb{R}}|x(t)-x(t-)|\leq1/k\}$ is a closed set.
So by the representation in (\ref{C_b}), $C_{b}(\mathbb{R})\in\mathcal{D}%
_{\pm\infty}$. This completes our proof.
\end{proof}

Now we are ready to prove Theorem \ref{weak_convergence}.

\begin{proof}{Proof of Theorem \ref{weak_convergence}.}
According to Theorem 5.1 in \cite{ferger2010weak}, it suffices to show the tightness of the processes in each component and also show the convergence of finite-dimensional distribution. So it suffices to verify the conditions in Proposition \ref{fidis_convergence}, Lemma \ref{tightness_BB} and Lemma \ref{tightness_GP}. Since the conditions of Proposition \ref{fidis_convergence} and Lemma \ref{tightness_GP} are the same conditions assumed in this theorem, it suffices to verify the condition $||Q(\cdot,\underline{\hat{P}})-Q(\cdot,\underline{P})||_{\infty}\overset{p}{\rightarrow}0$ in Lemma \ref{tightness_BB}.

We will show a stronger result $||Q(\cdot,\underline{\hat{P}})-Q(\cdot,\underline{P})||_{\infty}\overset{a.s.}{\rightarrow}0$. In fact, by the small discrepancy of the U-statistic and V-statistic in (\ref{UV_negligible_diff}), we obtain $Q(t,\underline{\hat{P}})\overset{a.s.}{\rightarrow}Q(t,\underline{P})$ for any $t\in\mathbb{R}$. Similar, we can also get $Q(t-,\underline{\hat{P}})\overset{a.s.}{\rightarrow}Q(t-,\underline{P})$ for any $t\in\mathbb{R}$. By the similar proof of classic Glivenko-Cantelli Theorem (\cite{van2000asymptotic} Theorem 19.1), we obtain $||Q(\cdot,\underline{\hat{P}})-Q(\cdot,\underline{P})||_{\infty}\overset{a.s.}{\rightarrow}0$. This concludes our proof.
\end{proof}

\begin{proof}{Proof of Corollary \ref{weak_convergence_sum}.}
We first unify the normalizing constant by proving%
\begin{equation}
(\sqrt{R}(\hat{Q}(\cdot,\underline{\hat{P}})-Q(\cdot,\underline{\hat{P}%
})),\sqrt{R}(Q(\cdot,\underline{\hat{P}})-Q(\cdot,\underline{P})))\Rightarrow
(BB({Q}(\cdot,\underline{P})),\gamma\mathbb{G}(\cdot)).
\label{normalized_weak_conv}%
\end{equation}
The finite-dimensional convergence follows from Proposition
\ref{fidis_convergence} and Slutsky's Theorem. Tightness of $\sqrt{R}(\hat
{Q}(\cdot,\underline{\hat{P}})-Q(\cdot,\underline{\hat{P}}))$ is automatic by
Theorem \ref{weak_convergence}. Then by Theorem 5.1 in \cite{ferger2010weak},
it suffices to show the tightness of $\sqrt{R}(Q(\cdot,\underline{\hat{P}%
})-Q(\cdot,\underline{P}))$. By Theorem \ref{weak_convergence}, $\sqrt
{n}(Q(\cdot,\underline{\hat{P}})-Q(\cdot,\underline{P}))\Rightarrow
\mathbb{G}(\cdot)$. Besides, $\sqrt{R/n}\Rightarrow\gamma$ when viewed as the
degenerate random variables and $\sqrt{R/n}$ is clearly independent of
$\sqrt{n}(Q(\cdot,\underline{\hat{P}})-Q(\cdot,\underline{P}))$. Then by
Theorem 2.8 in \cite{billingsley2013convergence} (the product space is
separable because $D_{\pm\infty}$ and $\mathbb{R}$ are separable),%
\[
(\sqrt{n}(Q(\cdot,\underline{\hat{P}})-Q(\cdot,\underline{P})),\sqrt
{R/n})\Rightarrow(\mathbb{G}(\cdot),\gamma).
\]
It's easy to see the scalar multiplication $f(x,a)=ax$ for $x\in D_{\pm\infty
}$ and $a\in\mathbb{R}$ is a continuous function from $D_{\pm\infty}%
\times\mathbb{R}$ to $D_{\pm\infty}$. So by continuous mapping theorem, we
have%
\[
\sqrt{R}(Q(\cdot,\underline{\hat{P}})-Q(\cdot,\underline{P}))=\sqrt{\frac
{R}{n}}\sqrt{n}(Q(\cdot,\underline{\hat{P}})-Q(\cdot,\underline{P}%
))\Rightarrow\gamma\mathbb{G}(\cdot),
\]
which implies the tightness of $\sqrt{R}(Q(\cdot,\underline{\hat{P}}%
)-Q(\cdot,\underline{P}))$ and thus proves (\ref{normalized_weak_conv}).

Next, we show the weak convergence of the sum of two processes in
(\ref{normalized_weak_conv}). We define $f:D_{\pm\infty}\times D_{\pm\infty
}\mapsto D_{\pm\infty}$ as $f(x,y)=x+y$. We first show its measurability. By
(\ref{generating_system}), it suffices to show the finite-dimensional
projections are measurable, i.e., for any $t_{1}<\cdots<t_{k}$ and
$H\in\mathcal{B}(\mathbb{R}^{k})$, we have%
\begin{equation}
\{(x,y):(x(t_{1})+y(t_{1}),\ldots,x(t_{k})+y(t_{k}))\in H\}\in\mathcal{D}%
_{\pm\infty}\times\mathcal{D}_{\pm\infty}. \label{measurability_proj_add}%
\end{equation}
When $H$ is the one-sided open rectangle%
\begin{equation}
H=(-\infty,a_{1})\times\cdots\times(-\infty,a_{k}), \label{rectangle}%
\end{equation}
we obtain%
\begin{align*}
&  \{(x,y):(x(t_{1})+y(t_{1}),\ldots,x(t_{k})+y(t_{k}))\in H\}\\
=  &  \{(x,y):x(t_{1})+y(t_{1})<a_{1},\ldots,x(t_{k})+y(t_{k})<a_{k}\}\\
=  &  \bigcup_{(r_{1},\ldots,r_{k}),(s_{1},\ldots,s_{k})\in\mathbb{Q}%
^{k},r_{i}+s_{i}<a_{i}}\{x:x(t_{i})<r_{i}\}\times\{y:y(t_{i})<s_{i}%
\}\in\mathcal{D}_{\pm\infty}\times\mathcal{D}_{\pm\infty},
\end{align*}
i.e., (\ref{measurability_proj_add}) holds if $H$ is in the form of
(\ref{rectangle}). Since all one-sided open rectangles is a $\pi$-system
generating $\mathcal{B}(\mathbb{R}^{k})$, it follows that
(\ref{measurability_proj_add}) holds for each $H\in\mathcal{B}(\mathbb{R}%
^{k})$. This proves the measurability of $f$. Then we show that $f$ is
continuous on $C(\mathbb{R})\times C(\mathbb{R})$. Recall that according to
(5) in \cite{ferger2010weak}, $d_{\pm\infty}^{\circ}(x_{n},x)\rightarrow0$ if
and only if there is a sequence of $\lambda_{n}(\cdot)$ that are strictly
increasing, continuous and surjective from $\mathbb{R}$ onto $\mathbb{R}$ such
that%
\begin{equation}
\left\{
\begin{array}
[c]{l}%
\sup_{t\in\mathbb{R}}|\lambda_{n}(t)-t|\rightarrow0\\
\sup_{|t|\leq m}|x(t)-x_{n}(\lambda_{n}(t))|\rightarrow0\text{ for any }%
m\in\mathbb{N}%
\end{array}
\right.  . \label{convergence_in_D1}%
\end{equation}
We first show that (\ref{convergence_in_D1}) is equivalent to: there is a
sequence of $\mu_{n}(\cdot)$ that are strictly increasing, continuous and
surjective from $\mathbb{R}$ onto $\mathbb{R}$ such that%
\begin{equation}
\left\{
\begin{array}
[c]{l}%
\sup_{t\in\mathbb{R}}|\mu_{n}(t)-t|\rightarrow0\\
\sup_{|t|\leq m}|x(\mu_{n}(t))-x_{n}(t)|\rightarrow0\text{ for any }%
m\in\mathbb{N}%
\end{array}
\right.  . \label{convergence_in_D2}%
\end{equation}
Assume (\ref{convergence_in_D1}) holds. We show $\mu_{n}=\lambda_{n}^{-1}$
satisfies (\ref{convergence_in_D2}). The first result in
(\ref{convergence_in_D2}) follows from%
\[
\sup_{t\in\mathbb{R}}|\mu_{n}(t)-t|=\sup_{t\in\mathbb{R}}|\lambda_{n}%
^{-1}(t)-t|=\sup_{t\in\mathbb{R}}|\lambda_{n}^{-1}(\lambda_{n}(t))-\lambda
_{n}(t)|=\sup_{t\in\mathbb{R}}|\lambda_{n}(t)-t|\rightarrow0.
\]
Moreover, we know that when $n$ is large, $|\mu_{n}(t)-t|\leq1/2$ for any
$t\in\mathbb{R}$. Therefore, when $n$ is large, $|t|\leq m$ implies $|\mu
_{n}(t)|\leq m+1$. We fix any $t\in\lbrack-m,m]$ and obtain%
\[
|x(\mu_{n}(t))-x_{n}(t)|=|x(\mu_{n}(t))-x_{n}(\lambda_{n}(\mu_{n}%
(t)))|\leq\sup_{|u|\leq m+1}|x(u)-x_{n}(\lambda_{n}(u))|.
\]
So when $n$ is large, we have%
\[
\sup_{|t|\leq m}|x(\mu_{n}(t))-x_{n}(t)|\leq\sup_{|u|\leq m+1}|x(u)-x_{n}%
(\lambda_{n}(u))|\rightarrow0,
\]
which proves (\ref{convergence_in_D2}). A similar argument tells that
(\ref{convergence_in_D2}) implies (\ref{convergence_in_D1}). Thus,
(\ref{convergence_in_D2}) and (\ref{convergence_in_D1}) are equivalent and
both are equivalent to $d_{\pm\infty}^{\circ}(x_{n},x)\rightarrow0$. Then we
show $d_{\pm\infty}^{\circ}(x_{n},x)\rightarrow0$ and $x\in C(\mathbb{R})$
implies locally uniform convergence, i.e., for any $m\in\mathbb{N}$%
\[
\sup_{|t|\leq m}|x(t)-x_{n}(t)|\rightarrow0.
\]
We define the local modulus of continuity as
\[
w_{x,m}(\delta)=\sup_{|s-t|\leq\delta,|t|\leq m,|s|\leq m}|x(t)-x(s)|.
\]
If $x\in C(\mathbb{R})$, then $x$ is uniformly continuous on $[-m,m]$ and thus
$w_{x,m}(\delta)\rightarrow0$ as $\delta\rightarrow0$ for any $m\in\mathbb{N}%
$. As proved above, $d_{\pm\infty}^{\circ}(x_{n},x)\rightarrow0$ implies
(\ref{convergence_in_D2}). So when $n$ is large, $|\mu_{n}(t)-t|\leq1/2$ for
any $t\in\mathbb{R}$ and thus $|t|\leq m$ implies $|\mu_{n}(t)|\leq m+1$. It
follows that%
\begin{align*}
&  \sup_{|t|\leq m}|x(t)-x_{n}(t)|\\
&  \leq\sup_{|t|\leq m}|x(t)-x(\mu_{n}(t))|+\sup_{|t|\leq m}|x(\mu
_{n}(t))-x_{n}(t)|\\
&  \leq w_{x,m+1}\left(  \sup_{t\in\mathbb{R}}|\mu_{n}(t)-t|\right)
+\sup_{|t|\leq m}|x(\mu_{n}(t))-x_{n}(t)|\rightarrow0.
\end{align*}
So the locally uniform convergence holds. Now let us prove $f(x,y)=x+y$ is
continuous on $C(\mathbb{R})\times C(\mathbb{R})$. Assume $d_{\pm\infty
}^{\circ}(x_{n},x)\rightarrow0$ and $d_{\pm\infty}^{\circ}(y_{n}%
,y)\rightarrow0$ with $x_{n},y_{n}\in D_{\pm\infty}$ and $x,y\in
C(\mathbb{R})$. From the locally uniform convergence, we have%
\[
\sup_{|t|\leq m}|x(t)+y(t)-x_{n}(t)-y_{n}(t)|\leq\sup_{|t|\leq m}%
|x(t)-x_{n}(t)|+\sup_{|t|\leq m}|y(t)-y_{n}(t)|\rightarrow0,
\]
which by the criterion (\ref{convergence_in_D1}) with $\lambda_{n}(t)\equiv t$
implies that $d_{\pm\infty}^{\circ}(x_{n}+y_{n},x+y)\rightarrow0$, i.e.,
$f(x,y)=x+y$ is continuous on $C(\mathbb{R})\times C(\mathbb{R})$. Now from
(\ref{normalized_weak_conv}), $\mathbb{P}((BB({Q}(\cdot,\underline{P}%
)),\mathbb{G}(\cdot))\in C_{b}(\mathbb{R})\times C_{b}(\mathbb{R}))=1$ and
continuous mapping theorem, we obtain%
\[
\sqrt{R}(\hat{Q}(\cdot,\underline{\hat{P}})-Q(\cdot,\underline{P}%
))\Rightarrow\gamma\mathbb{G}(\cdot)+BB({Q}(\cdot,\underline{P})).
\]

\end{proof}

\begin{proof}{Proof of Theorem \ref{weak_convergence_sup}.}
In the same simulation model specified in Assumption
\ref{finite_horizon_model}, we consider another output $Q(h(\mathbf{X}%
_{1},\ldots,\mathbf{X}_{m}),\underline{P})$ whose distribution is $U[0,1]$ due
to the continuity of $Q(t,\underline{P})$. Here we assume the two kinds of
output $h(\mathbf{X}_{1},\ldots,\mathbf{X}_{m})$ and $Q(h(\mathbf{X}%
_{1},\ldots,\mathbf{X}_{m}),\underline{P})$ are returned by the same
simulation model so that they are in the same probability space and their
behaviors are comparable. The new output is supported on $[0,1]$ and Corollary
\ref{weak_convergence_sum} for the new output should read as follows:%
\begin{equation}
\sqrt{R}(\hat{\tilde{Q}}(t,\underline{\hat{P}})-t)\Rightarrow\gamma
\mathbb{\tilde{G}}(t)+BB({t}) \label{new_weak_conv}%
\end{equation}
on $D[0,1]$ (we use $\tilde{\cdot}$ to denote the quantities for the new
output). A small difference between Corollary \ref{weak_convergence_sum} and
(\ref{new_weak_conv}) is that the result in Corollary
\ref{weak_convergence_sum} is actually proved in $D_{\pm\infty}$ not $D[0,1]$.
So to make the argument rigorous, we consider a function $f:D_{\pm\infty
}\mapsto D[0,1]$ defined as $f(x)=x|_{D[0,1]}$ where $x|_{D[0,1]}$ is the
restriction of $x$ on $D[0,1]$. $f$ is clearly measurable (by considering the
finite-dimensional sets) and also continuous on $C(\mathbb{R})$ since
$d_{\pm\infty}^{\circ}(x_{n},x)\rightarrow0$ and $x\in C(\mathbb{R})$ implies
locally uniform convergence (proof of Corollary \ref{weak_convergence_sum})
and thus the Skorohod convergence on $D[0,1]$. Therefore, by the continuous
mapping theorem, (\ref{new_weak_conv}) holds on $D[0,1]$.

Now we define the supremum operator $f(x)=\sup_{t\in\lbrack0,1]}|x(t)|$ on
$D[0,1]$. $f(x)$ is measurable because $f(x)=\sup_{t\in\mathbb{Q}\cap
\lbrack0,1]}|x(t)|$. Besides, it is continuous on $C[0,1]$ because Skorohod
convergence to $x\in C[0,1]$ implies uniform convergence to $x$ (proof of
Corollary \ref{weak_convergence_sum} or \cite{billingsley2013convergence}
p.124). Since $\mathbb{P}(\gamma\mathbb{\tilde{G}}(t)+BB({t})\in C[0,1])=1$,
by the continuous mapping theorem, we have%
\begin{equation}
\sup_{t\in\lbrack0,1]}|\sqrt{R}(\hat{\tilde{Q}}(t,\underline{\hat{P}%
})-t)|\overset{d}{\rightarrow}\sup_{t\in\lbrack0,1]}|\gamma\mathbb{\tilde{G}%
}(t)+BB({t})|. \label{new_weak_conv_sup}%
\end{equation}

Next, we show%
\begin{equation}
\sup_{t\in\mathbb{R}}|\sqrt{R}(\hat{Q}(t,\underline{\hat{P}}%
)-Q(t,\underline{P}))|-\sup_{t\in\lbrack0,1]}|\sqrt{R}(\hat{\tilde{Q}%
}(t,\underline{\hat{P}})-t)|\overset{L_{1}}{\rightarrow}0.
\label{L1_convergence}%
\end{equation}
First, $\sup_{t\in\mathbb{R}}|\sqrt{R}(\hat{Q}(t,\underline{\hat{P}%
})-Q(t,\underline{P}))|$ is a random variable because
\[
\sup_{t\in\mathbb{R}}|\sqrt{R}(\hat{Q}(t,\underline{\hat{P}}%
)-Q(t,\underline{P}))|=\sup_{t\in\mathbb{Q}}|\sqrt{R}(\hat{Q}%
(t,\underline{\hat{P}})-Q(t,\underline{P}))|
\]
and it is bounded by $\sqrt{R}$. Since the range of $Q(t,\underline{P})$
contains $(0,1)$ due to its continuity, we can see%
\[
\sup_{t\in\lbrack0,1]}|\sqrt{R}(\hat{\tilde{Q}}(t,\underline{\hat{P}%
})-t)|=\sup_{t\in\mathbb{R}}|\sqrt{R}(\hat{\tilde{Q}}(Q(t,\underline{P}%
),\underline{\hat{P}})-Q(t,\underline{P}))|,
\]
and%
\begin{equation}
\sup_{t\in\lbrack0,1]}|\gamma\mathbb{\tilde{G}}(t)+BB({t})|=\sup
_{t\in\mathbb{R}}|\gamma\mathbb{\tilde{G}}(Q(t,\underline{P}%
))+BB(Q(t,\underline{P}))|. \label{change_of_variable}%
\end{equation}
To show (\ref{L1_convergence}), we will bound the following expectation%
\begin{align}
&  \mathbb{E}\left[  \left\vert \sup_{t\in\mathbb{R}}|\sqrt{R}(\hat
{Q}(t,\underline{\hat{P}})-Q(t,\underline{P}))|-\sup_{t\in\lbrack0,1]}%
|\sqrt{R}(\hat{\tilde{Q}}(t,\underline{\hat{P}})-t)|\right\vert \right]
\nonumber\\
&  =\mathbb{E}\left[  \left\vert \sup_{t\in\mathbb{R}}|\sqrt{R}(\hat
{Q}(t,\underline{\hat{P}})-Q(t,\underline{P}))|-\sup_{t\in\mathbb{R}}|\sqrt
{R}(\hat{\tilde{Q}}(Q(t,\underline{P}),\underline{\hat{P}})-Q(t,\underline{P}%
))|\right\vert \right] \nonumber\\
&  \leq\mathbb{E}\left[  \sup_{t\in\mathbb{R}}\sqrt{R}|\hat{Q}%
(t,\underline{\hat{P}})-\hat{\tilde{Q}}(Q(t,\underline{P}),\underline{\hat{P}%
})|\right]  . \label{difference_expectation1}%
\end{align}
By the definition of $\hat{Q}$ and the construction of the new output, we have%
\begin{equation}
\hat{Q}(t,\underline{\hat{P}})=\frac{1}{R}\sum_{r=1}^{R}I(Y_{r}\leq
t),\quad\hat{\tilde{Q}}(Q(t,\underline{P}),\underline{\hat{P}})=\frac{1}%
{R}\sum_{r=1}^{R}I(Q(Y_{r},\underline{P})\leq Q(t,\underline{P})),
\label{definition_two_Qs}%
\end{equation}
where $Y_{r},r=1,\ldots,R$ are the i.i.d. original outputs of the simulation
model with the empirical input distributions $\underline{\hat{P}}$. Besides,
for any random variable $X$ and any $t\in\mathbb{R}$, we have the simple fact%
\[
I(X\leq t)\leq I(Q(X,\underline{P})\leq Q(t,\underline{P}))
\]
since $Q(\cdot,\underline{P})$ is nondecreasing. We will intensively use this
fact to remove the absolute values during the following computation. In the
following, we assume $n$ is large such that $n_{i}\geq T_{i}$. Plugging
(\ref{definition_two_Qs}) into (\ref{difference_expectation1}), we obtain%
\begin{align}
&  \mathbb{E}\left[  \left\vert \sup_{t\in\mathbb{R}}|\sqrt{R}(\hat
{Q}(t,\underline{\hat{P}})-Q(t,\underline{P}))|-\sup_{t\in\lbrack0,1]}%
|\sqrt{R}(\hat{\tilde{Q}}(t,\underline{\hat{P}})-t)|\right\vert \right]
\nonumber\\
&  \leq\mathbb{E}\left[  \sup_{t\in\mathbb{R}}\sqrt{R}\frac{1}{R}\sum
_{r=1}^{R}[I(Q(Y_{r},\underline{P})\leq Q(t,\underline{P}))-I(Y_{r}\leq
t)]\right] \nonumber\\
&  \leq\mathbb{E}\left[  \frac{1}{\sqrt{R}}\sum_{r=1}^{R}\sup_{t\in\mathbb{R}%
}[I(Q(Y_{r},\underline{P})\leq Q(t,\underline{P}))-I(Y_{r}\leq t)]\right]
\nonumber\\
&  =\mathbb{E}\left[  \mathbb{E}_{\ast}\left[  \frac{1}{\sqrt{R}}\sum
_{r=1}^{R}\sup_{t\in\mathbb{R}}[I(Q(Y_{r},\underline{P})\leq Q(t,\underline{P}%
))-I(Y_{r}\leq t)]\right]  \right] \nonumber\\
&  =\mathbb{E}\left[  \sqrt{R}\mathbb{E}_{\ast}\left[  \sup_{t\in\mathbb{R}%
}[I(Q(Y_{1},\underline{P})\leq Q(t,\underline{P}))-I(Y_{1}\leq t)]\right]
\right] \nonumber\\
&  =\mathbb{E}\left[  \sqrt{R}\frac{1}{\prod_{i=1}^{m}n_{i}^{T_{i}}}\sum
_{i=1}^{m}\sum_{1\leq j_{i1},\ldots,j_{iT_{i}}\leq n_{i}}\sup_{t\in\mathbb{R}%
}[I(Q(h(X_{1,j_{11}},\ldots,X_{m,j_{mT_{m}}}),\underline{P})\leq
Q(t,\underline{P}))\right. \nonumber\\
&  -I(h(X_{1,j_{11}},\ldots,X_{m,j_{mT_{m}}})\leq t)\Bigg]\nonumber\\
&  =\mathbb{E}\left[  \sqrt{R}\frac{1}{\prod_{i=1}^{m}n_{i}^{T_{i}}}\sum
_{i=1}^{m}\sum_{1\leq j_{i1}\neq\cdots\neq j_{iT_{i}}\leq n_{i}}\sup
_{t\in\mathbb{R}}[I(Q(h(X_{1,j_{11}},\ldots,X_{m,j_{mT_{m}}}),\underline{P}%
)\leq Q(t,\underline{P}))\right. \nonumber\\
&  -I(h(X_{1,j_{11}},\ldots,X_{m,j_{mT_{m}}})\leq t)]\Bigg]+\varepsilon,
\label{difference_expectation2}%
\end{align}
where the remainder term $\varepsilon$ is the difference between the ``V-type'' summation and ``U-type'' summation and it satisfies%
\begin{align}
|\varepsilon|  &  \leq\sqrt{R}\frac{1}{\prod_{i=1}^{m}n_{i}^{T_{i}}}\left(
\sum_{i=1}^{m}\sum_{1\leq j_{i1},\ldots,j_{iT_{i}}\leq n_{i}}1-\sum_{i=1}%
^{m}\sum_{1\leq j_{i1}\neq\cdots\neq j_{iT_{i}}\leq n_{i}}1\right) \nonumber\\
&  =\sqrt{R}\frac{\prod_{i=1}^{m}n_{i}^{T_{i}}-\prod_{i=1}^{m}\prod
_{j=0}^{T_{i}-1}(n_{i}-j)}{\prod_{i=1}^{m}n_{i}^{T_{i}}}\nonumber\\
&  =\sqrt{R}\frac{\left(  \sum_{i=1}^{m}\frac{(T_{i}-1)T_{i}}{2n_{i}}\right)
\prod_{i=1}^{m}n_{i}^{T_{i}}-o\left(  n^{\left(  \sum_{i=1}^{m}T_{i}\right)
-1}\right)  }{\prod_{i=1}^{m}n_{i}^{T_{i}}}\nonumber\\
&  =O\left(  \frac{\sqrt{R}}{n}\right)  \rightarrow0
\label{difference_expectation3}%
\end{align}
where the second equality follows from (\ref{product_order1}) and the order
analysis uses Assumptions \ref{balanced_data} and \ref{balanced_randomness}.
As for the leading term in (\ref{difference_expectation2}), we have%
\begin{align*}
&  \mathbb{E}\left[  \sqrt{R}\frac{1}{\prod_{i=1}^{m}n_{i}^{T_{i}}}\sum
_{i=1}^{m}\sum_{1\leq j_{i1}\neq\cdots\neq j_{iT_{i}}\leq n_{i}}\sup
_{t\in\mathbb{R}}[I(Q(h(X_{1,j_{11}},\ldots,X_{m,j_{mT_{m}}}),\underline{P}%
)\leq Q(t,\underline{P}))\right. \\
&  -I(h(X_{1,j_{11}},\ldots,X_{m,j_{mT_{m}}})\leq t)]\Bigg]\\
&  =\sqrt{R}\frac{1}{\prod_{i=1}^{m}n_{i}^{T_{i}}}\sum_{i=1}^{m}\sum_{1\leq
j_{i1}\neq\cdots\neq j_{iT_{i}}\leq n_{i}}\mathbb{E}\left[  \sup
_{t\in\mathbb{R}}[I(Q(h(X_{1,j_{11}},\ldots,X_{m,j_{mT_{m}}}),\underline{P}%
)\leq Q(t,\underline{P}))\right. \\
&  -I(h(X_{1,j_{11}},\ldots,X_{m,j_{mT_{m}}})\leq t)\bigg]\\
&  =\sqrt{R}\frac{\prod_{i=1}^{m}\prod_{j=0}^{T_{i}-1}(n_{i}-j)}{\prod
_{i=1}^{m}n_{i}^{T_{i}}}\mathbb{E}\left[  \sup_{t\in\mathbb{R}}[I(Q(h(X_{1,1}%
,\ldots,X_{1,T_{1}},\ldots,X_{m,1},\ldots,X_{m,T_{m}}),\underline{P})\leq
Q(t,\underline{P}))\right. \\
&  -I(h(X_{1,1},\ldots,X_{1,T_{1}},\ldots,X_{m,1},\ldots,X_{m,T_{m}})\leq
t)\bigg]\\
&  =\sqrt{R}\frac{\prod_{i=1}^{m}\prod_{j=0}^{T_{i}-1}(n_{i}-j)}{\prod
_{i=1}^{m}n_{i}^{T_{i}}}\mathbb{E}\left[  \sup_{t\in\mathbb{Q}}[I(Q(h(X_{1,1}%
,\ldots,X_{1,T_{1}},\ldots,X_{m,1},\ldots,X_{m,T_{m}}),\underline{P})\leq
Q(t,\underline{P}))\right. \\
&  -I(h(X_{1,1},\ldots,X_{1,T_{1}},\ldots,X_{m,1},\ldots,X_{m,T_{m}})\leq
t)\bigg].
\end{align*}
However, for any $t\in\mathbb{R}$,%
\begin{align*}
&  \mathbb{P}(Q(h(X_{1,1},\ldots,X_{1,T_{1}},\ldots,X_{m,1},\ldots,X_{m,T_{m}%
}),\underline{P})\leq Q(t,\underline{P}))\\
&  -\mathbb{P}(h(X_{1,1},\ldots,X_{1,T_{1}},\ldots,X_{m,1},\ldots,X_{m,T_{m}%
})\leq t)\\
&  =Q(t,\underline{P})-Q(t,\underline{P})=0,
\end{align*}
which implies%
\[
I(Q(h(X_{1,1},\ldots,X_{m,T_{m}}),\underline{P})\leq Q(t,\underline{P}%
))-I(h(X_{1,1},\ldots,X_{m,T_{m}})\leq t)\overset{a.s.}{=}0
\]
for any $t\in\mathbb{R}$. For the supremum over a countable set $\mathbb{Q}$,
we still have%
\[
\sup_{t\in\mathbb{Q}}[I(Q(h(X_{1,1},\ldots,X_{m,T_{m}}),\underline{P})\leq
Q(t,\underline{P}))-I(h(X_{1,1},\ldots,X_{m,T_{m}})\leq t)]\overset{a.s.}{=}%
0.
\]
Therefore, we obtain%
\begin{align}
&  \mathbb{E}\left[  \sqrt{R}\frac{1}{\prod_{i=1}^{m}n_{i}^{T_{i}}}\sum
_{i=1}^{m}\sum_{1\leq j_{i1}\neq\cdots\neq j_{iT_{i}}\leq n_{i}}\sup
_{t\in\mathbb{R}}[I(Q(h(X_{1,j_{11}},\ldots,X_{m,j_{mT_{m}}}),\underline{P}%
)\leq Q(t,\underline{P}))\right. \nonumber\\
&  -I(h(X_{1,j_{11}},\ldots,X_{m,j_{mT_{m}}})\leq t)]\Bigg]=0.
\label{difference_expectation4}%
\end{align}
Plugging (\ref{difference_expectation3}) and (\ref{difference_expectation4})
into (\ref{difference_expectation2}), we obtain the $L_{1}$ convergence in
(\ref{L1_convergence}). Since $L_{1}$ convergence implies convergence in
probability, by Slutsky's theorem and (\ref{new_weak_conv_sup}), we have%
\begin{equation}
\sup_{t\in\mathbb{R}}|\sqrt{R}(\hat{Q}(t,\underline{\hat{P}}%
)-Q(t,\underline{P}))|\overset{d}{\rightarrow}\sup_{t\in\lbrack0,1]}%
|\gamma\mathbb{\tilde{G}}(t)+BB({t})|=\sup_{t\in\mathbb{R}}|\gamma
\mathbb{\tilde{G}}(Q(t,\underline{P}))+BB(Q(t,\underline{P}))|
\label{weak_conv_tilde}%
\end{equation}
as $n,R\rightarrow\infty$, where the equality follows from
(\ref{change_of_variable}).

To show the desired weak convergence, it remains to prove%
\[
\sup_{t\in\mathbb{R}}|\gamma\mathbb{\tilde{G}}(Q(t,\underline{P}%
))+BB(Q(t,\underline{P}))|\overset{d}{=}\sup_{t\in\mathbb{R}}|\gamma
\mathbb{G}(t)+BB(Q(t,\underline{P}))|.
\]
It suffices to show%
\[
\gamma\mathbb{\tilde{G}}(Q(\cdot,\underline{P}))+BB(Q(\cdot,\underline{P}%
))\overset{d}{=}\gamma\mathbb{G}(\cdot)+BB(Q(\cdot,\underline{P}))
\]
on $D_{\pm\infty}$. Since finite-dimensional sets generate $\mathcal{D}%
_{\pm\infty}$, it suffices to show the two processes have the same
finite-dimensional distribution. Since both $\mathbb{\tilde{G}}$ and
$\mathbb{G}$ are independent of $BB$, it suffices to show $\mathbb{\tilde{G}%
}(Q(\cdot,\underline{P}))$ and $\mathbb{G}(\cdot)$ have the same
finite-dimensional distributions. By Proposition \ref{fidis_convergence}, the
finite-dimensional distributions of $\mathbb{G}$ and $\mathbb{\tilde{G}}$ are
determined by
\[
\sqrt{n}(Q(\cdot,\underline{\hat{P}})-Q(\cdot,\underline{P}%
))\overset{f.d.}{\rightarrow}\mathbb{G}(\cdot),\quad\sqrt{n}(\tilde{Q}%
(\cdot,\underline{\hat{P}})-\cdot)\overset{f.d.}{\rightarrow}\mathbb{\tilde
{G}}(\cdot),
\]
and thus%
\[
\sqrt{n}(Q(\cdot,\underline{\hat{P}})-Q(\cdot,\underline{P}%
))\overset{f.d.}{\rightarrow}\mathbb{G}(\cdot),\quad\sqrt{n}(\tilde{Q}%
(Q(\cdot,\underline{P}),\underline{\hat{P}})-Q(\cdot,\underline{P}%
))\overset{f.d.}{\rightarrow}\mathbb{\tilde{G}}(Q(\cdot,\underline{P})).
\]
Therefore, by Slutsky's theorem, it suffices to show%
\[
\sqrt{n}(Q(t,\underline{\hat{P}})-\tilde{Q}(Q(t,\underline{P}),\underline{\hat
{P}}))\overset{p}{\rightarrow}0
\]
for any $t\in\mathbb{R}$. This is actually automatic from the results we
proved just now. Notice that the previous proof shows that expectation in
(\ref{difference_expectation1}) converges to $0$, i.e.,%
\[
\mathbb{E}\left[  \sup_{t\in\mathbb{R}}\sqrt{R}|\hat{Q}(t,\underline{\hat{P}%
})-\hat{\tilde{Q}}(Q(t,\underline{P}),\underline{\hat{P}})|\right]
\rightarrow0.
\]
This implies%
\[
\sqrt{R}|\hat{Q}(t,\underline{\hat{P}})-\hat{\tilde{Q}}(Q(t,\underline{P}%
),\underline{\hat{P}})|\leq\sup_{t\in\mathbb{R}}\sqrt{R}|\hat{Q}%
(t,\underline{\hat{P}})-\hat{\tilde{Q}}(Q(t,\underline{P}),\underline{\hat{P}%
})|\overset{p}{\rightarrow}0.
\]
By Assumption \ref{balanced_randomness}, $R=\Theta(n)$ and thus we know
\[
\sqrt{n}(Q(t,\underline{\hat{P}})-\tilde{Q}(Q(t,\underline{P}),\underline{\hat
{P}}))\overset{p}{\rightarrow}0.
\]
This proves the desired weak convergence result.

Finally, we show the absolute continuity of the limiting distribution. We use
the form of limiting distribution in (\ref{weak_conv_tilde}), i.e., we need to
prove the absolute continuity of the following random variable%
\begin{equation}
\sup_{t\in\mathbb{R}}|\gamma\mathbb{\tilde{G}}(Q(t,\underline{P}%
))+BB(Q(t,\underline{P}))|=\sup_{t\in\lbrack0,1]}|\gamma\mathbb{\tilde{G}%
}(t)+BB(t)|, \label{form_limiting_distribution}%
\end{equation}
where the equality follows from the continuity of $Q(t,\underline{P})$. By
Theorem \ref{weak_convergence} we know that
\[
\mathbb{P}(\mathbb{\tilde{G}}(t)\in C[0,1])=1.
\]
According to \cite{billingsley2013convergence} p.124, the Skorohod topology on
$D[0,1]$ relativized to $C[0,1]$ coincides with the uniform topology there. So
we can view $\gamma\mathbb{\tilde{G}}(t)+BB(t)$ as random elements in $C[0,1]$
with the sup norm. Now according to Corollary in \cite{lifshits1983absolute},
(\ref{form_limiting_distribution}) can have an atom only at the point%
\[
a_{M}=\sup_{t:Var(\gamma\mathbb{\tilde{G}}(t)+BB(t))=0}|\mathbb{E}%
[\gamma\mathbb{\tilde{G}}(t)+BB(t)]|=\sup_{t=0,1}|\mathbb{E}[\gamma
\mathbb{\tilde{G}}(t)+BB(t)]|=0,
\]
it vanishes on the ray $(-\infty,0)$ and is absolutely continuous with respect
to Lebesgue measure on $(0,\infty)$. However, it is easy to see
(\ref{form_limiting_distribution}) cannot have positive probability at $0$
because%
\begin{align*}
&  \mathbb{P}\left(  \sup_{t\in\lbrack0,1]}|\gamma\mathbb{\tilde{G}%
}(t)+BB(t)|=0\right) \\
&  =\mathbb{P(}|\gamma\mathbb{\tilde{G}}(t)+BB(t)|=0,\forall t\in
\lbrack0,1])\\
&  \leq\mathbb{P}\left(  \left\vert \gamma\mathbb{\tilde{G}}\left(  \frac
{1}{2}\right)  +BB\left(  \frac{1}{2}\right)  \right\vert =0\right) \\
&  =\mathbb{P}\left(  \left\vert N\left(  0,\gamma^{2}Var\left(
\mathbb{\tilde{G}}\left(  \frac{1}{2}\right)  \right)  +\frac{1}{4}\right)
\right\vert =0\right) \\
&  =0.
\end{align*}
Therefore, the limiting distribution is absolutely continuous with respect to
Lebesgue measure on $[0,\infty)$.
\end{proof}

\begin{proof}{Proof of Corollary \ref{CB_direct}.}
This follows directly from Theorem \ref{weak_convergence_sup}.
\end{proof}

\subsection{Proofs of Results in Sections \ref{sec:cov subsampling} and \ref{sec:theory_cov}}\label{sec:proofs_subsampling}
To show Theorem \ref{consistency_sigma}. we first prove the following lemma.

\begin{lemma}
\label{influence_product}Suppose Assumptions \ref{balanced_data},
\ref{1st_expansion_truth} and \ref{1st_expansion_empirical} hold. We have%
\[
\mathbb{E}_{\ast}[IF_{i}(t,X_{i,1}^{\ast};\underline{\hat{P}})IF_{i}%
(t^{\prime},X_{i,1}^{\ast};\underline{\hat{P}})]\overset{p}{\rightarrow
}\mathrm{Cov}_{P_{i}}(IF_{i}(t,X_{i};\underline{P}),IF_{i}(t^{\prime}%
,X_{i};\underline{P})),\forall t,t^{\prime}\in\mathbb{R},
\]%
\[
\mathbb{E}_{\ast}[IF_{i}^{2}(t,X_{i,1}^{\ast};\underline{\hat{P}})IF_{i}%
^{2}(t^{\prime},X_{i,1}^{\ast};\underline{\hat{P}})]\overset{p}{\rightarrow
}\mathbb{E}_{P_{i}}[IF_{i}^{2}(t,X_{i};\underline{P})IF_{i}^{2}(t^{\prime
},X_{i};\underline{P})],\forall t,t^{\prime}\in\mathbb{R},
\]
as $n\rightarrow\infty$. Consequently,
\[
\mathbb{E}_{\ast}[IF_{i}(t,X_{i,1}^{\ast};\underline{\hat{P}})IF_{i}%
(t^{\prime},X_{i,1}^{\ast};\underline{\hat{P}})]=O_{p}(1),\forall t,t^{\prime
}\in\mathbb{R},
\]%
\[
\mathbb{E}_{\ast}[IF_{i}^{2}(t,X_{i,1}^{\ast};\underline{\hat{P}})IF_{i}%
^{2}(t^{\prime},X_{i,1}^{\ast};\underline{\hat{P}})]=O_{p}(1),\forall
t,t^{\prime}\in\mathbb{R},
\]
as $n\rightarrow\infty$.
\end{lemma}

\begin{proof}{Proof of Lemma \ref{influence_product}.}
By definition, we have%
\[
\mathbb{E}_{\ast}[IF_{i}(t,X_{i,1}^{\ast};\underline{\hat{P}})IF_{i}%
(t^{\prime},X_{i,1}^{\ast};\underline{\hat{P}})]=\frac{1}{n_{i}}\sum
_{j=1}^{n_{i}}IF_{i}(t,X_{i,j};\underline{\hat{P}})IF_{i}(t^{\prime}%
,X_{i,j};\underline{\hat{P}}).
\]
Notice that%
\begin{align*}
&  \mathbb{E}\left[  \left\vert \frac{1}{n_{i}}\sum_{j=1}^{n_{i}}%
IF_{i}(t,X_{i,j};\underline{\hat{P}})IF_{i}(t^{\prime},X_{i,j};\underline{\hat
{P}})-\frac{1}{n_{i}}\sum_{j=1}^{n_{i}}IF_{i}(t,X_{i,j};\underline{P}%
)IF_{i}(t^{\prime},X_{i,j};\underline{P})\right\vert \right] \\
&  \leq\frac{1}{n_{i}}\sum_{j=1}^{n_{i}}\mathbb{E}[|IF_{i}(t,X_{i,j}%
;\underline{\hat{P}})IF_{i}(t^{\prime},X_{i,j};\underline{\hat{P}}%
)-IF_{i}(t,X_{i,j};\underline{P})IF_{i}(t^{\prime},X_{i,j};\underline{P})|]\\
&  =\mathbb{E}[|IF_{i}(t,X_{i,1};\underline{\hat{P}})IF_{i}(t^{\prime}%
,X_{i,1};\underline{\hat{P}})-IF_{i}(t,X_{i,1};\underline{P})IF_{i}(t^{\prime
},X_{i,1};\underline{P})|]\\
&  \leq\mathbb{E}[|IF_{i}(t,X_{i,1};\underline{\hat{P}})IF_{i}(t^{\prime
},X_{i,1};\underline{\hat{P}})-IF_{i}(t,X_{i,1};\underline{P})IF_{i}%
(t^{\prime},X_{i,1};\underline{\hat{P}})|]\\
&  +\mathbb{E}[|IF_{i}(t,X_{i,1};\underline{P})IF_{i}(t^{\prime}%
,X_{i,1};\underline{\hat{P}})-IF_{i}(t,X_{i,1};\underline{P})IF_{i}(t^{\prime
},X_{i,1};\underline{P})|]\\
&  \leq\sqrt{\mathbb{E}[(IF_{i}(t,X_{i,1};\underline{\hat{P}})-IF_{i}%
(t,X_{i,1};\underline{P}))^{2}]\mathbb{E}[IF_{i}^{2}(t^{\prime},X_{i,1}%
;\underline{\hat{P}})]}\\
&  +\sqrt{\mathbb{E}[IF_{i}^{2}(t,X_{i,1};\underline{P})]\mathbb{E}%
[(IF_{i}(t^{\prime},X_{i,1};\underline{\hat{P}})-IF_{i}(t^{\prime}%
,X_{i,1};\underline{P}))^{2}]},
\end{align*}
where the equality follows from $IF_{i}(t,X_{i,j};\underline{\hat{P}%
})\overset{d}{=}IF_{i}(t,X_{i,1};\underline{\hat{P}})$. By Assumptions
\ref{1st_expansion_truth} and \ref{1st_expansion_empirical}, we know that
\[
\mathbb{E}[(IF_{i}(t,X_{i};\underline{\hat{P}})-IF_{i}(t,X_{i};\underline{P}%
))^{2}]\rightarrow0,\mathbb{E}[IF_{i}^{2}(t,X_{i};\underline{P})]<\infty
,\forall t\in\mathbb{R}.
\]
Therefore,%
\[
\mathbb{E}\left[  \left\vert \frac{1}{n_{i}}\sum_{j=1}^{n_{i}}IF_{i}%
(t,X_{i,j};\underline{\hat{P}})IF_{i}(t^{\prime},X_{i,j};\underline{\hat{P}%
})-\frac{1}{n_{i}}\sum_{j=1}^{n_{i}}IF_{i}(t,X_{i,j};\underline{P}%
)IF_{i}(t^{\prime},X_{i,j};\underline{P})\right\vert \right]  \rightarrow0.
\]
Combining it with the weak law of large numbers%
\[
\frac{1}{n_{i}}\sum_{j=1}^{n_{i}}IF_{i}(t,X_{i,j};\underline{P})IF_{i}%
(t^{\prime},X_{i,j};\underline{P})\overset{p}{\rightarrow}\mathrm{Cov}_{P_{i}%
}(IF_{i}(t,X_{i};\underline{P}),IF_{i}(t^{\prime},X_{i};\underline{P})),
\]
we obtain%
\[
\mathbb{E}_{\ast}[IF_{i}(t,X_{i,1}^{\ast};\underline{\hat{P}})IF_{i}%
(t^{\prime},X_{i,1}^{\ast};\underline{\hat{P}})]\overset{p}{\rightarrow
}\mathrm{Cov}_{P_{i}}(IF_{i}(t,X_{i};\underline{P}),IF_{i}(t^{\prime}%
,X_{i};\underline{P})).
\]

Result for the fourth moment can be proved by a similar argument. This
concludes our proof.
\end{proof}

Now let us prove Theorem \ref{consistency_sigma}.

\begin{proof}{Proof of Theorem \ref{consistency_sigma}.}
First, by Theorem \ref{verification_general_assumptions}, Assumptions \ref{convergence_in_p_to_truth}-\ref{3rd_expansion_empirical} hold under Assumptions \ref{balanced_data} and \ref{finite_horizon_model}. Moreover, the conclusions in Lemma \ref{influence_product} hold. By definition, we have%
\[
\sigma^{2}(t,t^{\prime})=\theta n\mathrm{Cov}_{\ast}(Q(t,\underline{\hat{P}%
}_{\theta}^{\ast}),Q(t^{\prime},\underline{\hat{P}}_{\theta}^{\ast}))=\theta
n\mathbb{E}_{\ast}[(Q(t,\underline{\hat{P}}_{\theta}^{\ast})-\mathbb{E}_{\ast
}[Q(t,\underline{\hat{P}}_{\theta}^{\ast})])(Q(t^{\prime},\underline{\hat{P}%
}_{\theta}^{\ast})-\mathbb{E}_{\ast}[Q(t^{\prime},\underline{\hat{P}}_{\theta
}^{\ast})])].
\]
By Assumption \ref{1st_expansion_empirical}, we know that%
\begin{align*}
Q(t,\underline{\hat{P}}_{\theta}^{\ast})  &  =Q(t,\underline{\hat{P}}%
)+\sum_{i=1}^{m}\int IF_{i}(t,x;\underline{\hat{P}})d(\hat{P}_{i,s_{i}}^{\ast
}(x)-\hat{P}_{i}(x))+\varepsilon^{\ast}(t)\\
&  =Q(t,\underline{\hat{P}})+\sum_{i=1}^{m}\frac{1}{s_{i}}\sum_{k=1}^{s_{i}%
}IF_{i}(t,X_{i,k}^{\ast};\underline{\hat{P}})+\varepsilon^{\ast}(t).
\end{align*}
Besides, from Assumption \ref{1st_expansion_empirical}, we also know%
\[
\mathbb{E}_{\ast}[IF_{i}(t,X_{i,k}^{\ast};\underline{\hat{P}})]=0,\forall
i=1,\ldots,m,\forall k=1,\ldots,s_{i}.
\]
So we obtain%
\[
Q(t,\underline{\hat{P}}_{\theta}^{\ast})-\mathbb{E}_{\ast}[Q(t,\underline{\hat
{P}}_{\theta}^{\ast})]=\sum_{i=1}^{m}\frac{1}{s_{i}}\sum_{k=1}^{s_{i}}%
IF_{i}(t,X_{i,k}^{\ast};\underline{\hat{P}})+\varepsilon^{\ast}(t)-\mathbb{E}%
_{\ast}[\varepsilon^{\ast}(t)].
\]
Therefore, $\sigma^{2}(t,t^{\prime})$ can be written as
\begin{align}
&  \sigma^{2}(t,t^{\prime})\nonumber\\
&  =\theta n\mathbb{E}_{\ast}\left[  \left(  \sum_{i=1}^{m}\frac{1}{s_{i}}%
\sum_{k=1}^{s_{i}}IF_{i}(t,X_{i,k}^{\ast};\underline{\hat{P}})+\varepsilon
^{\ast}(t)-\mathbb{E}_{\ast}[\varepsilon^{\ast}(t)]\right)  \right.
\nonumber\\
&  \left.  \left.  \left(  \sum_{i=1}^{m}\frac{1}{s_{i}}\sum_{k=1}^{s_{i}%
}IF_{i}(t^{\prime},X_{i,k}^{\ast};\underline{\hat{P}})+\varepsilon^{\ast
}(t^{\prime})-\mathbb{E}_{\ast}[\varepsilon^{\ast}(t^{\prime})]\right)
\right)  \right]  . \label{computation_sigma}%
\end{align}
The leading term is given by%
\begin{align*}
&  \theta n\mathbb{E}_{\ast}\left[  \left(  \sum_{i=1}^{m}\frac{1}{s_{i}}%
\sum_{k=1}^{s_{i}}IF_{i}(t,X_{i,k}^{\ast};\underline{\hat{P}})\right)  \left(
\sum_{i=1}^{m}\frac{1}{s_{i}}\sum_{k=1}^{s_{i}}IF_{i}(t^{\prime},X_{i,k}%
^{\ast};\underline{\hat{P}})\right)  \right] \\
&  =\theta n\sum_{i=1}^{m}\sum_{k=1}^{s_{i}}\frac{1}{s_{i}^{2}}\mathbb{E}%
_{\ast}[IF_{i}(t,X_{i,k}^{\ast};\underline{\hat{P}})IF_{i}(t^{\prime}%
,X_{i,k}^{\ast};\underline{\hat{P}})]\\
&  =\theta n\sum_{i=1}^{m}\frac{1}{s_{i}}\mathbb{E}_{\ast}[IF_{i}%
(t,X_{i,1}^{\ast};\underline{\hat{P}})IF_{i}(t^{\prime},X_{i,1}^{\ast
};\underline{\hat{P}})].
\end{align*}
where the first inequality uses the mutual conditional independence of
$X_{i,k}^{\ast}$ and the mean-zero property of $IF_{i}(\cdot,X_{i,k}^{\ast
};\underline{\hat{P}})$. By Lemma \ref{influence_product} and noticing that
$\theta n/s_{i}=\theta n/\lfloor\theta n_{i}\rfloor\rightarrow1/\beta_{i}$, we
have%
\begin{align}
&  \theta n\mathbb{E}_{\ast}\left[  \left(  \sum_{i=1}^{m}\frac{1}{s_{i}}%
\sum_{k=1}^{s_{i}}IF_{i}(t,X_{i,k}^{\ast};\underline{\hat{P}})\right)  \left(
\sum_{i=1}^{m}\frac{1}{s_{i}}\sum_{k=1}^{s_{i}}IF_{i}(t^{\prime},X_{i,k}%
^{\ast};\underline{\hat{P}})\right)  \right] \nonumber\\
&  \overset{p}{\rightarrow}\sum_{i=1}^{m}\frac{1}{\beta_{i}}\mathrm{Cov}%
_{P_{i}}(IF_{i}(t,X_{i};\underline{P}),IF_{i}(t^{\prime},X_{i};\underline{P}%
))=\mathrm{Cov}(\mathbb{G}(t),\mathbb{G}(t^{\prime})). \label{leading_sigma}%
\end{align}
As for the remainder term in (\ref{computation_sigma}), by Assumption
\ref{1st_expansion_empirical} and the order in (\ref{leading_sigma}), we have
\begin{align*}
&  \theta n\mathbb{E}_{\ast}\left[  \left(  \sum_{i=1}^{m}\frac{1}{s_{i}}%
\sum_{k=1}^{s_{i}}IF_{i}(t,X_{i,k}^{\ast};\underline{\hat{P}})\right)
(\varepsilon^{\ast}(t^{\prime})-\mathbb{E}_{\ast}[\varepsilon^{\ast}%
(t^{\prime})])\right] \\
&  \leq\theta n\sqrt{\mathbb{E}_{\ast}\left[  \left(  \sum_{i=1}^{m}\frac
{1}{s_{i}}\sum_{k=1}^{s_{i}}IF_{i}(t,X_{i,k}^{\ast};\underline{\hat{P}%
})\right)  ^{2}\right]  }\sqrt{\mathbb{E}_{\ast}[(\varepsilon^{\ast}%
(t^{\prime})-\mathbb{E}_{\ast}[\varepsilon^{\ast}(t^{\prime})])^{2}]}\\
&  =\theta n\times O_{p}\left(  \frac{1}{\sqrt{\theta n}}\right)  \times
o_{p}\left(  \frac{1}{\sqrt{\theta n}}\right)  =o_{p}(1),
\end{align*}
and similarly%
\[
\theta n\mathbb{E}_{\ast}\left[  \left(  \sum_{i=1}^{m}\frac{1}{s_{i}}%
\sum_{k=1}^{s_{i}}IF_{i}(t^{\prime},X_{i,k}^{\ast};\underline{\hat{P}%
})\right)  (\varepsilon^{\ast}(t)-\mathbb{E}_{\ast}[\varepsilon^{\ast
}(t)])\right]  =o_{p}(1),
\]%
\[
\theta n\mathbb{E}_{\ast}[(\varepsilon^{\ast}(t)-\mathbb{E}_{\ast}%
[\varepsilon^{\ast}(t)])(\varepsilon^{\ast}(t^{\prime})-\mathbb{E}_{\ast
}[\varepsilon^{\ast}(t^{\prime})])]=o_{p}(1).
\]
Therefore, we conclude that%
\[
\sigma^{2}(t,t^{\prime})\overset{p}{\rightarrow}\mathrm{Cov}(\mathbb{G}%
(t),\mathbb{G}(t^{\prime})).
\]
\end{proof}

To prove Theorem \ref{MSE_sigma_hat}, we need the variance formula presented in Lemma \ref{general_covariance_undebias} and the following lemma that establishes the order of some conditional expectations.

\begin{lemma}
\label{computation_moments}Suppose Assumptions \ref{balanced_data},
\ref{1st_expansion_truth} and \ref{1st_expansion_empirical} hold. We define%
\[
\tau^{X}=Q(t,\underline{\hat{P}}_{\theta}^{\ast})-\mathbb{E}_{\ast
}[Q(t,\underline{\hat{P}}_{\theta}^{\ast})],\quad\tau^{Y}=Q(t^{\prime
},\underline{\hat{P}}_{\theta}^{\ast})-\mathbb{E}_{\ast}[Q(t^{\prime
},\underline{\hat{P}}_{\theta}^{\ast})],
\]
and%
\[
\varepsilon^{X}=I(Y^{\ast}\leq t)-Q(t,\underline{\hat{P}}_{\theta}^{\ast
}),\quad\varepsilon^{Y}=I(Y^{\ast}\leq t^{\prime})-Q(t^{\prime}%
,\underline{\hat{P}}_{\theta}^{\ast}),
\]
where $Y^{\ast}$ is the output of the simulation model when the input
distribution is $\underline{\hat{P}}_{\theta}^{\ast}$. Assume $t\leq
t^{\prime}$. Then we have%
\begin{align*}
\mathbb{E}_{\ast}[(\tau^{X})^{2}(\tau^{Y})^{2}] &  =\left(  \sum_{i=1}%
^{m}\frac{1}{s_{i}}\mathbb{E}_{P_{i}}[IF_{i}^{2}(t,X_{i};\underline{P}%
)]\right)  \left(  \sum_{i=1}^{m}\frac{1}{s_{i}}\mathbb{E}_{P_{i}}[IF_{i}%
^{2}(t^{\prime},X_{i};\underline{P})]\right)  \\
&  +2\left(  \sum_{i=1}^{m}\frac{1}{s_{i}}\mathrm{Cov}_{P_{i}}(IF_{i}%
(t,X_{i};\underline{P}),IF_{i}(t^{\prime},X_{i};\underline{P}))\right)
^{2}+o_{p}\left(  \frac{1}{s^{2}}\right)  =O_{p}\left(  \frac{1}{s^{2}%
}\right)  ,
\end{align*}%
\[
\mathbb{E}_{\ast}[(\tau^{X})^{2}(\varepsilon^{Y})^{2}]=\sum_{i=1}^{m}\frac
{1}{s_{i}}\mathbb{E}_{P_{i}}[IF_{i}^{2}(t,X_{i};\underline{P})]Q(t^{\prime
},\underline{P})(1-Q(t^{\prime},\underline{P}))+o_{p}\left(  \frac{1}%
{s}\right)  =O_{p}\left(  \frac{1}{s}\right)  ,
\]%
\[
\mathbb{E}_{\ast}[\tau^{X}\tau^{Y}\varepsilon^{X}\varepsilon^{Y}]=\sum
_{i=1}^{m}\frac{1}{s_{i}}\mathrm{Cov}_{P_{i}}(IF_{i}(t,X_{i};\underline{P}%
),IF_{i}(t^{\prime},X_{i};\underline{P}))Q(t,\underline{P})(1-Q(t^{\prime
},\underline{P}))+o_{p}\left(  \frac{1}{s}\right)  =O_{p}\left(  \frac{1}%
{s}\right)  ,
\]%
\[
\mathbb{E}_{\ast}[(\varepsilon^{X})^{2}(\tau^{Y})^{2}]=\sum_{i=1}^{m}\frac
{1}{s_{i}}\mathbb{E}_{P_{i}}[IF_{i}^{2}(t^{\prime},X_{i};\underline{P}%
)]Q(t,\underline{P})(1-Q(t,\underline{P}))+o_{p}\left(  \frac{1}{s}\right)
=O_{p}\left(  \frac{1}{s}\right)  ,
\]%
\begin{align*}
& \mathbb{E}_{\ast}[\tau^{X}\varepsilon^{X}(\varepsilon^{Y})^{2}]\\
& =(1-2Q(t^{\prime},\underline{P}))(1-Q(t^{\prime},\underline{P}))\sum
_{i=1}^{m}\frac{1}{s_{i}}\mathbb{E}_{P_{i}}(IF_{i}^{2}(t,X_{i};\underline{P}%
))\\
& +(4Q(t,\underline{P})Q(t^{\prime},\underline{P})-3Q(t,\underline{P}%
))\sum_{i=1}^{m}\frac{1}{s_{i}}\mathrm{Cov}_{P_{i}}(IF_{i}(t,X_{i}%
;\underline{P}),IF_{i}(t^{\prime},X_{i};\underline{P}))+o_{p}\left(  \frac
{1}{s}\right)  \\
& =O_{p}\left(  \frac{1}{s}\right)  ,
\end{align*}%
\begin{align*}
& \mathbb{E}_{\ast}[(\varepsilon^{X})^{2}\tau^{Y}\varepsilon^{Y}]\\
& =(1-4Q(t,\underline{P}))(1-Q(t^{\prime},\underline{P}))\sum_{i=1}^{m}%
\frac{1}{s_{i}}\mathrm{Cov}_{P_{i}}(IF_{i}(t,X_{i};\underline{P}%
),IF_{i}(t^{\prime},X_{i};\underline{P}))\\
& -Q(t,\underline{P})(1-2Q(t,\underline{P}))\sum_{i=1}^{m}\frac{1}{s_{i}%
}\mathbb{E}_{P_{i}}(IF_{i}^{2}(t^{\prime},X_{i};\underline{P}))+o_{p}\left(
\frac{1}{s}\right)  \\
& =O_{p}\left(  \frac{1}{s}\right)  ,
\end{align*}%
\begin{align*}
& \mathbb{E}_{\ast}[(\varepsilon^{X})^{2}(\varepsilon^{Y})^{2}]\\
& =Q(t,\underline{P})(1-Q(t^{\prime},\underline{P}))(1-Q(t^{\prime
},\underline{P})-2Q(t,\underline{P})+3Q(t,\underline{P})Q(t^{\prime
},\underline{P}))+o_{p}(1)\\
& =O_{p}(1),
\end{align*}%
\[
\mathbb{E}_{\ast}[\mathbb{E}_{\ast}[(\varepsilon^{X})^{2}|\underline{\hat{P}%
}_{\theta}^{\ast}]\mathbb{E}_{\ast}[(\varepsilon^{Y})^{2}|\underline{\hat{P}%
}_{\theta}^{\ast}]]=Q(t,\underline{P})(1-Q(t,\underline{P}))Q(t^{\prime
},\underline{P})(1-Q(t^{\prime},\underline{P}))+o_{p}(1)=O_{p}(1),
\]%
\[
\mathbb{E}_{\ast}[(\mathbb{E}_{\ast}[\varepsilon^{X}\varepsilon^{Y}%
|\underline{\hat{P}}_{\theta}^{\ast}])^{2}]=Q(t,\underline{P})^{2}%
(1-Q(t^{\prime},\underline{P}))^{2}+o_{p}(1)=O_{p}(1),
\]%
\[
\mathbb{E}_{\ast}[(\tau^{X})^{2}]\mathbb{E}_{\ast}[(\tau^{Y})^{2}]=\left(
\sum_{i=1}^{m}\frac{1}{s_{i}}\mathbb{E}_{P_{i}}[IF_{i}^{2}(t,X_{i}%
;\underline{P})]\right)  \left(  \sum_{i=1}^{m}\frac{1}{s_{i}}\mathbb{E}%
_{P_{i}}[IF_{i}^{2}(t^{\prime},X_{i};\underline{P})]\right)  +o_{p}\left(
\frac{1}{s^{2}}\right)  =O_{p}\left(  \frac{1}{s^{2}}\right)  ,
\]%
\[
\mathbb{E}_{\ast}[(\tau^{X})^{2}]\mathbb{E}_{\ast}[(\varepsilon^{Y})^{2}%
]=\sum_{i=1}^{m}\frac{1}{s_{i}}\mathbb{E}_{P_{i}}[IF_{i}^{2}(t,X_{i}%
;\underline{P})]Q(t^{\prime},\underline{P})(1-Q(t^{\prime},\underline{P}%
))+o_{p}\left(  \frac{1}{s}\right)  =O_{p}\left(  \frac{1}{s}\right)  ,
\]%
\[
\mathbb{E}_{\ast}[(\varepsilon^{X})^{2}]\mathbb{E}_{\ast}[(\tau^{Y})^{2}%
]=\sum_{i=1}^{m}\frac{1}{s_{i}}\mathbb{E}_{P_{i}}[IF_{i}^{2}(t^{\prime}%
,X_{i};\underline{P})]Q(t,\underline{P})(1-Q(t,\underline{P}))+o_{p}\left(
\frac{1}{s}\right)  =O_{p}\left(  \frac{1}{s}\right)  ,
\]%
\[
\mathbb{E}_{\ast}[(\varepsilon^{X})^{2}]\mathbb{E}_{\ast}[(\varepsilon
^{Y})^{2}]=Q(t,\underline{P})(1-Q(t,\underline{P}))Q(t^{\prime},\underline{P}%
)(1-Q(t^{\prime},\underline{P}))+o_{p}(1)=O_{p}(1),
\]%
\[
(\mathbb{E}_{\ast}[\tau^{X}\tau^{Y}])^{2}=\left(  \sum_{i=1}^{m}\frac{1}%
{s_{i}}\mathrm{Cov}_{P_{i}}(IF_{i}(t,X_{i};\underline{P}),IF_{i}(t^{\prime
},X_{i};\underline{P}))\right)  ^{2}+o_{p}\left(  \frac{1}{s^{2}}\right)
=O_{p}\left(  \frac{1}{s^{2}}\right)  ,
\]%
\[
\mathbb{E}_{\ast}[\tau^{X}\tau^{Y}]\mathbb{E}_{\ast}[\varepsilon
^{X}\varepsilon^{Y}]=\sum_{i=1}^{m}\frac{1}{s_{i}}\mathrm{Cov}_{P_{i}}%
(IF_{i}(t,X_{i};\underline{P}),IF_{i}(t^{\prime},X_{i};\underline{P}%
))Q(t,\underline{P})(1-Q(t^{\prime},\underline{P}))+o_{p}\left(  \frac{1}%
{s}\right)  =O_{p}\left(  \frac{1}{s}\right)  ,
\]%
\[
(\mathbb{E}_{\ast}[\varepsilon^{X}\varepsilon^{Y}])^{2}=Q(t,\underline{P}%
)^{2}(1-Q(t^{\prime},\underline{P}))^{2}+o_{p}(1)=O_{p}(1).
\]

\end{lemma}

\begin{proof}{Proof of Lemma \ref{computation_moments}.}
By Assumption \ref{1st_expansion_empirical}, we know that%
\begin{equation}
\tau^{X}=\sum_{i=1}^{m}\frac{1}{s_{i}}\sum_{k=1}^{s_{i}}IF_{i}(t,X_{i,k}%
^{\ast};\underline{\hat{P}})+\varepsilon^{\ast}(t)-\mathbb{E}_{\ast
}[\varepsilon^{\ast}(t)],\label{tau_X}%
\end{equation}%
\begin{equation}
\tau^{Y}=\sum_{i=1}^{m}\frac{1}{s_{i}}\sum_{k=1}^{s_{i}}IF_{i}(t^{\prime
},X_{i,k}^{\ast};\underline{\hat{P}})+\varepsilon^{\ast}(t^{\prime
})-\mathbb{E}_{\ast}[\varepsilon^{\ast}(t^{\prime})].\label{tau_Y}%
\end{equation}
We first compute a few conditional second moments. The first one is
$\mathbb{E}_{\ast}[\tau^{X}\tau^{Y}]$. The leading term of $\mathbb{E}_{\ast
}[\tau^{X}\tau^{Y}]$ is
\begin{align*}
&  \mathbb{E}_{\ast}\left[  \left(  \sum_{i=1}^{m}\frac{1}{s_{i}}\sum
_{k=1}^{s_{i}}IF_{i}(t,X_{i,k}^{\ast};\underline{\hat{P}})\right)  \left(
\sum_{i=1}^{m}\frac{1}{s_{i}}\sum_{k=1}^{s_{i}}IF_{i}(t^{\prime},X_{i,k}%
^{\ast};\underline{\hat{P}})\right)  \right]  \\
&  =\sum_{i=1}^{m}\sum_{k=1}^{s_{i}}\frac{1}{s_{i}^{2}}\mathbb{E}_{\ast
}[IF_{i}(t,X_{i,k}^{\ast};\underline{\hat{P}})IF_{i}(t^{\prime},X_{i,k}^{\ast
};\underline{\hat{P}})]\\
&  =\sum_{i=1}^{m}\frac{1}{s_{i}}\mathbb{E}_{\ast}[IF_{i}(t,X_{i,1}^{\ast
};\underline{\hat{P}})IF_{i}(t^{\prime},X_{i,1}^{\ast};\underline{\hat{P}})].
\end{align*}
where the first equality uses the mutual conditional independence of
$X_{i,k}^{\ast}$ and the mean-zero property of $IF_{i}(\cdot,X_{i,k}^{\ast
};\underline{\hat{P}})$. By Lemma \ref{influence_product}, this leading term
can be further written as%
\begin{align}
&  \mathbb{E}_{\ast}\left[  \left(  \sum_{i=1}^{m}\frac{1}{s_{i}}\sum
_{k=1}^{s_{i}}IF_{i}(t,X_{i,k}^{\ast};\underline{\hat{P}})\right)  \left(
\sum_{i=1}^{m}\frac{1}{s_{i}}\sum_{k=1}^{s_{i}}IF_{i}(t^{\prime},X_{i,k}%
^{\ast};\underline{\hat{P}})\right)  \right]  \nonumber\\
&  =\sum_{i=1}^{m}\frac{1}{s_{i}}\mathrm{Cov}_{P_{i}}(IF_{i}(t,X_{i}%
;\underline{P}),IF_{i}(t^{\prime},X_{i};\underline{P}))+o_{p}\left(  \frac
{1}{s}\right)  =O_{p}\left(  \frac{1}{s}\right)  ,\forall t,t^{\prime}%
\in\mathbb{R}.\label{order_E[IF^2]}%
\end{align}
Since $\mathbb{E}_{\ast}[(\varepsilon^{\ast}(t)-\mathbb{E}_{\ast}%
[\varepsilon^{\ast}(t)])^{2}]=o_{p}(1/s)$ according to Assumption
\ref{1st_expansion_empirical}, by H\"{o}lder's inequality, $\mathbb{E}_{\ast
}[\tau^{X}\tau^{Y}]$ can be expressed as
\[
\mathbb{E}_{\ast}[\tau^{X}\tau^{Y}]=\sum_{i=1}^{m}\frac{1}{s_{i}}%
\mathrm{Cov}_{P_{i}}(IF_{i}(t,X_{i};\underline{P}),IF_{i}(t^{\prime}%
,X_{i};\underline{P}))+o_{p}\left(  \frac{1}{s}\right)  =O_{p}\left(  \frac
{1}{s}\right)  .
\]
Similarly, we have%
\[
\mathbb{E}_{\ast}[(\tau^{X})^{2}]=\sum_{i=1}^{m}\frac{1}{s_{i}}\mathbb{E}%
_{P_{i}}[IF_{i}^{2}(t,X_{i};\underline{P})]+o_{p}\left(  \frac{1}{s}\right)
=O_{p}\left(  \frac{1}{s}\right)  ,
\]
and%
\[
\mathbb{E}_{\ast}[(\tau^{Y})^{2}]=\sum_{i=1}^{m}\frac{1}{s_{i}}\mathbb{E}%
_{P_{i}}[IF_{i}^{2}(t^{\prime},X_{i};\underline{P})]+o_{p}\left(  \frac{1}%
{s}\right)  =O_{p}\left(  \frac{1}{s}\right)  .
\]
Next, by direct computation, we have%
\[
\mathbb{E}_{\ast}[\varepsilon^{X}\varepsilon^{Y}|\underline{\hat{P}}_{\theta
}^{\ast}]=Q(t,\underline{\hat{P}}_{\theta}^{\ast})-Q(t,\underline{\hat{P}%
}_{\theta}^{\ast})Q(t^{\prime},\underline{\hat{P}}_{\theta}^{\ast
})=Q(t,\underline{\hat{P}}_{\theta}^{\ast})(1-Q(t^{\prime},\underline{\hat{P}%
}_{\theta}^{\ast})),
\]%
\[
\mathbb{E}_{\ast}[(\varepsilon^{X})^{2}|\underline{\hat{P}}_{\theta}^{\ast
}]=Q(t,\underline{\hat{P}}_{\theta}^{\ast})-Q^{2}(t,\underline{\hat{P}%
}_{\theta}^{\ast})=Q(t,\underline{\hat{P}}_{\theta}^{\ast}%
)(1-Q(t,\underline{\hat{P}}_{\theta}^{\ast})),
\]%
\begin{equation}
\mathbb{E}_{\ast}[(\varepsilon^{Y})^{2}|\underline{\hat{P}}_{\theta}^{\ast
}]=Q(t^{\prime},\underline{\hat{P}}_{\theta}^{\ast})-Q^{2}(t^{\prime
},\underline{\hat{P}}_{\theta}^{\ast})=Q(t^{\prime},\underline{\hat{P}%
}_{\theta}^{\ast})(1-Q(t^{\prime},\underline{\hat{P}}_{\theta}^{\ast
})),\label{conditional_E[epsilon_Y^2]}%
\end{equation}%
\begin{align}
\mathbb{E}_{\ast}[\varepsilon^{X}(\varepsilon^{Y})^{2}|\underline{\hat{P}%
}_{\theta}^{\ast}]  & =Q(t,\underline{\hat{P}}_{\theta}^{\ast}%
)-3Q(t,\underline{\hat{P}}_{\theta}^{\ast})Q(t^{\prime},\underline{\hat{P}%
}_{\theta}^{\ast})+2Q(t,\underline{\hat{P}}_{\theta}^{\ast})Q^{2}(t^{\prime
},\underline{\hat{P}}_{\theta}^{\ast})\nonumber\\
& =Q(t,\underline{\hat{P}}_{\theta}^{\ast})(1-2Q(t^{\prime},\underline{\hat
{P}}_{\theta}^{\ast}))(1-Q(t^{\prime},\underline{\hat{P}}_{\theta}^{\ast
})),\label{conditional_E[epsilon_Xepsilon_Y^2]}%
\end{align}
\begin{align*}
\mathbb{E}_{\ast}[(\varepsilon^{X})^{2}\varepsilon^{Y}|\underline{\hat{P}%
}_{\theta}^{\ast}]  & =Q(t,\underline{\hat{P}}_{\theta}^{\ast}%
)-Q(t,\underline{\hat{P}}_{\theta}^{\ast})Q(t^{\prime},\underline{\hat{P}%
}_{\theta}^{\ast})-2Q^{2}(t,\underline{\hat{P}}_{\theta}^{\ast})+2Q^{2}%
(t,\underline{\hat{P}}_{\theta}^{\ast})Q(t^{\prime},\underline{\hat{P}%
}_{\theta}^{\ast})\\
& =Q(t,\underline{\hat{P}}_{\theta}^{\ast})(1-2Q(t,\underline{\hat{P}}%
_{\theta}^{\ast}))(1-Q(t^{\prime},\underline{\hat{P}}_{\theta}^{\ast})),
\end{align*}
\begin{align*}
\mathbb{E}_{\ast}[(\varepsilon^{X})^{2}(\varepsilon^{Y})^{2}|\underline{\hat
{P}}_{\theta}^{\ast}]  & =Q(t,\underline{\hat{P}}_{\theta}^{\ast
})-2Q(t,\underline{\hat{P}}_{\theta}^{\ast})Q(t^{\prime},\underline{\hat{P}%
}_{\theta}^{\ast})-2Q^{2}(t,\underline{\hat{P}}_{\theta}^{\ast})\\
& +Q(t,\underline{\hat{P}}_{\theta}^{\ast})Q^{2}(t^{\prime},\underline{\hat
{P}}_{\theta}^{\ast})+5Q^{2}(t,\underline{\hat{P}}_{\theta}^{\ast}%
)Q(t^{\prime},\underline{\hat{P}}_{\theta}^{\ast})-3Q^{2}(t,\underline{\hat
{P}}_{\theta}^{\ast})Q^{2}(t^{\prime},\underline{\hat{P}}_{\theta}^{\ast})\\
& =Q(t,\underline{\hat{P}}_{\theta}^{\ast})(1-Q(t^{\prime},\underline{\hat{P}%
}_{\theta}^{\ast}))(1-Q(t^{\prime},\underline{\hat{P}}_{\theta}^{\ast
})-2Q(t,\underline{\hat{P}}_{\theta}^{\ast})+3Q(t,\underline{\hat{P}}_{\theta
}^{\ast})Q(t^{\prime},\underline{\hat{P}}_{\theta}^{\ast}))
\end{align*}
Based on $\mathbb{E}_{\ast}[\varepsilon^{X}\varepsilon^{Y}|\underline{\hat{P}%
}_{\theta}^{\ast}]$, we can compute $\mathbb{E}_{\ast}[\varepsilon
^{X}\varepsilon^{Y}]$. By Assumption \ref{1st_expansion_empirical}, we have%
\begin{align*}
&  \mathbb{E}_{\ast}[Q(t,\underline{\hat{P}}_{\theta}^{\ast})(1-Q(t^{\prime
},\underline{\hat{P}}_{\theta}^{\ast}))]\\
&  =\mathbb{E}_{\ast}\left[  \left(  Q(t,\underline{\hat{P}})+\sum_{i=1}%
^{m}\frac{1}{s_{i}}\sum_{k=1}^{s_{i}}IF_{i}(t,X_{i,k}^{\ast};\underline{\hat
{P}})+\varepsilon^{\ast}(t)\right)  \left(  1-Q(t^{\prime},\underline{\hat{P}%
})-\sum_{i=1}^{m}\frac{1}{s_{i}}\sum_{k=1}^{s_{i}}IF_{i}(t^{\prime}%
,X_{i,k}^{\ast};\underline{\hat{P}})-\varepsilon^{\ast}(t^{\prime})\right)
\right]  .
\end{align*}
The leading term is clearly $Q(t,\underline{\hat{P}})(1-Q(t^{\prime
},\underline{\hat{P}}))=Q(t,\underline{P})(1-Q(t^{\prime},\underline{P}%
))+o_{p}(1)$. Moreover, the mean-zero property of $IF_{i}$ leads to
\[
(1-Q(t^{\prime},\underline{\hat{P}}))\mathbb{E}_{\ast}\left[  \sum_{i=1}%
^{m}\frac{1}{s_{i}}\sum_{k=1}^{s_{i}}IF_{i}(t,X_{i,k}^{\ast};\underline{\hat
{P}})\right]  =Q(t,\underline{\hat{P}})\mathbb{E}_{\ast}\left[  \sum_{i=1}%
^{m}\frac{1}{s_{i}}\sum_{k=1}^{s_{i}}IF_{i}(t^{\prime},X_{i,k}^{\ast
};\underline{\hat{P}})\right]  =0.
\]
So the remainder term is of order $O_{p}(1/s)$ by Assumption
\ref{1st_expansion_empirical}, (\ref{order_E[IF^2]}) and H\"{o}lder's
inequality. Therefore, we have%
\begin{equation}
\mathbb{E}_{\ast}[\varepsilon^{X}\varepsilon^{Y}]=\mathbb{E}_{\ast
}[Q(t,\underline{\hat{P}}_{\theta}^{\ast})(1-Q(t^{\prime},\underline{\hat{P}%
}_{\theta}^{\ast}))]=Q(t,\underline{P})(1-Q(t^{\prime},\underline{P}%
))+o_{p}(1)=O_{p}(1).\label{E[epsilon_X*epsilon_Y]}%
\end{equation}
Similarly, we have%
\[
\mathbb{E}_{\ast}[(\varepsilon^{X})^{2}]=Q(t,\underline{P}%
)(1-Q(t,\underline{P}))+o_{p}(1)=O_{p}(1),
\]
and%
\[
\mathbb{E}_{\ast}[(\varepsilon^{Y})^{2}]=Q(t^{\prime},\underline{P}%
)(1-Q(t^{\prime},\underline{P}))+o_{p}(1)=O_{p}(1).
\]

Now let us compute the desired quantities. First, we compute $\mathbb{E}%
_{\ast}[(\tau^{X})^{2}(\tau^{Y})^{2}]$. By (\ref{tau_X}) and (\ref{tau_Y}), we
have
\begin{align*}
&  \mathbb{E}_{\ast}[(\tau^{X})^{2}(\tau^{Y})^{2}]\\
&  =\mathbb{E}_{\ast}\left[  \left(  \sum_{i=1}^{m}\frac{1}{s_{i}}\sum
_{k=1}^{s_{i}}IF_{i}(t,X_{i,k}^{\ast};\underline{\hat{P}})+\varepsilon^{\ast
}(t)-\mathbb{E}_{\ast}[\varepsilon^{\ast}(t)]\right)  ^{2}\left(  \sum
_{i=1}^{m}\frac{1}{s_{i}}\sum_{k=1}^{s_{i}}IF_{i}(t^{\prime},X_{i,k}^{\ast
};\underline{\hat{P}})+\varepsilon^{\ast}(t^{\prime})-\mathbb{E}_{\ast
}[\varepsilon^{\ast}(t^{\prime})]\right)  ^{2}\right]  .
\end{align*}
The leading term is given by%
\begin{align*}
&  \mathbb{E}_{\ast}\left[  \left(  \sum_{i=1}^{m}\frac{1}{s_{i}}\sum
_{k=1}^{s_{i}}IF_{i}(t,X_{i,k}^{\ast};\underline{\hat{P}})\right)  ^{2}\left(
\sum_{i=1}^{m}\frac{1}{s_{i}}\sum_{k=1}^{s_{i}}IF_{i}(t^{\prime},X_{i,k}%
^{\ast};\underline{\hat{P}})\right)  ^{2}\right]  \\
&  =\sum_{i=1}^{m}\sum_{k=1}^{s_{i}}\frac{1}{s_{i}^{4}}\mathbb{E}_{\ast
}[IF_{i}^{2}(t,X_{i,k}^{\ast};\underline{\hat{P}})IF_{i}^{2}(t^{\prime
},X_{i,k}^{\ast};\underline{\hat{P}})]+\sum_{(i_{1},k_{1})\neq(i_{2},k_{2}%
)}\frac{1}{s_{i_{1}}^{2}s_{i_{2}}^{2}}\mathbb{E}_{\ast}[IF_{i_{1}}%
^{2}(t,X_{i_{1},k_{1}}^{\ast};\underline{\hat{P}})]\mathbb{E}_{\ast}%
[IF_{i_{2}}^{2}(t^{\prime},X_{i_{2},k_{2}}^{\ast};\underline{\hat{P}})]\\
&  +\sum_{(i_{1},k_{1})\neq(i_{2},k_{2})}\frac{2}{s_{i_{1}}^{2}s_{i_{2}}^{2}%
}\mathbb{E}_{\ast}[IF_{i_{1}}(t,X_{i_{1},k_{1}}^{\ast};\underline{\hat{P}%
})IF_{i_{1}}(t^{\prime},X_{i_{1},k_{1}}^{\ast};\underline{\hat{P}}%
)]\mathbb{E}_{\ast}[IF_{i_{2}}(t,X_{i_{2},k_{2}}^{\ast};\underline{\hat{P}%
})IF_{i_{2}}(t^{\prime},X_{i_{2},k_{2}}^{\ast};\underline{\hat{P}})]\\
&  =\sum_{i_{1}=1}^{m}\sum_{k_{1}=1}^{s_{i_{1}}}\sum_{i_{2}=1}^{m}\sum
_{k_{2}=1}^{s_{i_{2}}}\frac{1}{s_{i_{1}}^{2}s_{i_{2}}^{2}}\mathbb{E}_{\ast
}[IF_{i_{1}}^{2}(t,X_{i_{1},k_{1}}^{\ast};\underline{\hat{P}})]\mathbb{E}%
_{\ast}[IF_{i_{2}}^{2}(t^{\prime},X_{i_{2},k_{2}}^{\ast};\underline{\hat{P}%
})]\\
&  +2\sum_{i_{1}=1}^{m}\sum_{k_{1}=1}^{s_{i_{1}}}\sum_{i_{2}=1}^{m}\sum
_{k_{2}=1}^{s_{i_{2}}}\frac{1}{s_{i_{1}}^{2}s_{i_{2}}^{2}}\mathbb{E}_{\ast
}[IF_{i_{1}}(t,X_{i_{1},k_{1}}^{\ast};\underline{\hat{P}})IF_{i_{1}}%
(t^{\prime},X_{i_{1},k_{1}}^{\ast};\underline{\hat{P}})]\mathbb{E}_{\ast
}[IF_{i_{2}}(t,X_{i_{2},k_{2}}^{\ast};\underline{\hat{P}})IF_{i_{2}}%
(t^{\prime},X_{i_{2},k_{2}}^{\ast};\underline{\hat{P}})]\\
&  +\sum_{i=1}^{m}\sum_{k=1}^{s_{i}}\frac{1}{s_{i}^{4}}(\mathbb{E}_{\ast
}[IF_{i}^{2}(t,X_{i,k}^{\ast};\underline{\hat{P}})IF_{i}^{2}(t^{\prime
},X_{i,k}^{\ast};\underline{\hat{P}})]-\mathbb{E}_{\ast}[IF_{i}^{2}%
(t,X_{i,k}^{\ast};\underline{\hat{P}})]\mathbb{E}_{\ast}[IF_{i}^{2}(t^{\prime
},X_{i,k}^{\ast};\underline{\hat{P}})]\\
&  -(\mathbb{E}_{\ast}[IF_{i}^{2}(t,X_{i,k}^{\ast};\underline{\hat{P}}%
)IF_{i}^{2}(t^{\prime},X_{i,k}^{\ast};\underline{\hat{P}})])^{2})\\
&  =\left(  \sum_{i=1}^{m}\frac{1}{s_{i}}\mathbb{E}_{\ast}[IF_{i}%
^{2}(t,X_{i,1}^{\ast};\underline{\hat{P}})]\right)  \left(  \sum_{i=1}%
^{m}\frac{1}{s_{i}}\mathbb{E}_{\ast}[IF_{i}^{2}(t^{\prime},X_{i,1}^{\ast
};\underline{\hat{P}})]\right)  \\
&  +2\left(  \sum_{i=1}^{m}\frac{1}{s_{i}}\mathbb{E}_{\ast}[IF_{i}%
(t,X_{i,1}^{\ast};\underline{\hat{P}})IF_{i}(t^{\prime},X_{i,1}^{\ast
};\underline{\hat{P}})]\right)  ^{2}+O_{p}\left(  \frac{1}{s^{3}}\right)  \\
&  =\left(  \sum_{i=1}^{m}\frac{1}{s_{i}}\mathbb{E}_{P_{i}}[IF_{i}^{2}%
(t,X_{i};\underline{P})]\right)  \left(  \sum_{i=1}^{m}\frac{1}{s_{i}%
}\mathbb{E}_{P_{i}}[IF_{i}^{2}(t^{\prime},X_{i};\underline{P})]\right)  \\
&  +2\left(  \sum_{i=1}^{m}\frac{1}{s_{i}}\mathrm{Cov}_{P_{i}}(IF_{i}%
(t,X_{i};\underline{P}),IF_{i}(t^{\prime},X_{i};\underline{P}))\right)
^{2}+o_{p}\left(  \frac{1}{s^{2}}\right)  +O_{p}\left(  \frac{1}{s^{3}%
}\right)  ,
\end{align*}
where the first equality follows from the conditional zero mean and
conditional mutual independence of $X_{i,k}^{\ast},k=1,\ldots,s_{i}%
,i=1,\ldots,m$ and the order of the remainder term in the last two equalities
follows from Lemma \ref{influence_product}. Besides, the above computation
reveals that the leading term itself is of order $O_{p}(1/s^{2})$, i.e.,%
\begin{equation}
\mathbb{E}_{\ast}\left[  \left(  \sum_{i=1}^{m}\frac{1}{s_{i}}\sum
_{k=1}^{s_{i}}IF_{i}(t,X_{i,k}^{\ast};\underline{\hat{P}})\right)  ^{2}\left(
\sum_{i=1}^{m}\frac{1}{s_{i}}\sum_{k=1}^{s_{i}}IF_{i}(t^{\prime},X_{i,k}%
^{\ast};\underline{\hat{P}})\right)  ^{2}\right]  =O_{p}\left(  \frac{1}%
{s^{2}}\right)  ,\forall t,t^{\prime}\in\mathbb{R}.\label{order_E[IF^4]}%
\end{equation}
By Assumption \ref{1st_expansion_empirical}, (\ref{order_E[IF^4]}) and
H\"{o}lder's inequality, we can see the other terms in $\mathbb{E}_{\ast}%
[(\tau^{X})^{2}(\tau^{Y})^{2}]$ are all of order $o_{p}(1/s^{2})$. Therefore,
we obtain%
\begin{align*}
\mathbb{E}_{\ast}[(\tau^{X})^{2}(\tau^{Y})^{2}] &  =\left(  \sum_{i=1}%
^{m}\frac{1}{s_{i}}\mathbb{E}_{P_{i}}[IF_{i}^{2}(t,X_{i};\underline{P}%
)]\right)  \left(  \sum_{i=1}^{m}\frac{1}{s_{i}}\mathbb{E}_{P_{i}}[IF_{i}%
^{2}(t^{\prime},X_{i};\underline{P})]\right)  \\
&  +2\left(  \sum_{i=1}^{m}\frac{1}{s_{i}}\mathrm{Cov}_{P_{i}}(IF_{i}%
(t,X_{i};\underline{P}),IF_{i}(t^{\prime},X_{i};\underline{P}))\right)
^{2}+o_{p}\left(  \frac{1}{s^{2}}\right)  =O_{p}\left(  \frac{1}{s^{2}%
}\right)  .
\end{align*}
Next, we compute $\mathbb{E}_{\ast}[(\tau^{X})^{2}(\varepsilon^{Y})^{2}]$. By
Assumption \ref{1st_expansion_empirical}, (\ref{tau_X}) and
(\ref{conditional_E[epsilon_Y^2]}), we have%
\begin{align*}
&  \mathbb{E}_{\ast}[(\tau^{X})^{2}(\varepsilon^{Y})^{2}]\\
&  =\mathbb{E}_{\ast}[(\tau^{X})^{2}\mathbb{E}_{\ast}[(\varepsilon^{Y}%
)^{2}|\underline{\hat{P}}_{\theta}^{\ast}]]\\
&  =\mathbb{E}_{\ast}\left[  \left(  \sum_{i=1}^{m}\frac{1}{s_{i}}\sum
_{k=1}^{s_{i}}IF_{i}(t,X_{i,k}^{\ast};\underline{\hat{P}})+\varepsilon^{\ast
}(t)-\mathbb{E}_{\ast}[\varepsilon^{\ast}(t)]\right)  ^{2}Q(t^{\prime
},\underline{\hat{P}}_{\theta}^{\ast})(1-Q(t^{\prime},\underline{\hat{P}%
}_{\theta}^{\ast}))\right]  \\
&  =\mathbb{E}_{\ast}\left[  \left(  \sum_{i=1}^{m}\frac{1}{s_{i}}\sum
_{k=1}^{s_{i}}IF_{i}(t,X_{i,k}^{\ast};\underline{\hat{P}})+\varepsilon^{\ast
}(t)-\mathbb{E}_{\ast}[\varepsilon^{\ast}(t)]\right)  ^{2}\left(  Q(t^{\prime
},\underline{\hat{P}})+\sum_{i=1}^{m}\frac{1}{s_{i}}\sum_{k=1}^{s_{i}}%
IF_{i}(t^{\prime},X_{i,k}^{\ast};\underline{\hat{P}})+\varepsilon^{\ast
}(t^{\prime})\right)  \right.  \\
&  \left.  \times\left(  1-Q(t^{\prime},\underline{\hat{P}})-\sum_{i=1}%
^{m}\frac{1}{s_{i}}\sum_{k=1}^{s_{i}}IF_{i}(t^{\prime},X_{i,k}^{\ast
};\underline{\hat{P}})-\varepsilon^{\ast}(t^{\prime})\right)  \right]  .
\end{align*}
The leading term is clearly%
\begin{align*}
&  \mathbb{E}_{\ast}\left[  \left(  \sum_{i=1}^{m}\frac{1}{s_{i}}\sum
_{k=1}^{s_{i}}IF_{i}(t,X_{i,k}^{\ast};\underline{\hat{P}})\right)
^{2}\right]  Q(t^{\prime},\underline{\hat{P}})(1-Q(t^{\prime},\underline{\hat
{P}}))\\
&  =\sum_{i=1}^{m}\frac{1}{s_{i}}\mathbb{E}_{\ast}[IF_{i}^{2}(t,X_{i,1}^{\ast
};\underline{\hat{P}})]Q(t^{\prime},\underline{\hat{P}})(1-Q(t^{\prime
},\underline{\hat{P}}))\\
&  =\sum_{i=1}^{m}\frac{1}{s_{i}}\mathbb{E}_{P_{i}}[IF_{i}^{2}(t,X_{i}%
;\underline{P})]Q(t^{\prime},\underline{P})(1-Q(t^{\prime},\underline{P}%
))+o_{p}\left(  \frac{1}{s}\right)  =O_{p}\left(  \frac{1}{s}\right)  .
\end{align*}
By Assumption \ref{1st_expansion_empirical}, (\ref{order_E[IF^4]}) and
H\"{o}lder's inequality, we can see the remaining terms in $\mathbb{E}_{\ast
}[(\tau^{X})^{2}(\varepsilon^{Y})^{2}]$ are all of order $o_{p}(1/s)$. Hence,
\[
\mathbb{E}_{\ast}[(\tau^{X})^{2}(\varepsilon^{Y})^{2}]=\sum_{i=1}^{m}\frac
{1}{s_{i}}\mathbb{E}_{P_{i}}[IF_{i}^{2}(t,X_{i};\underline{P})]Q(t^{\prime
},\underline{P})(1-Q(t^{\prime},\underline{P}))+o_{p}\left(  \frac{1}%
{s}\right)  =O_{p}\left(  \frac{1}{s}\right)  .
\]
Similarly, we obtain
\[
\mathbb{E}_{\ast}[\tau^{X}\tau^{Y}\varepsilon^{X}\varepsilon^{Y}]=\sum
_{i=1}^{m}\frac{1}{s_{i}}\mathrm{Cov}_{P_{i}}(IF_{i}(t,X_{i};\underline{P}%
),IF_{i}(t^{\prime},X_{i};\underline{P}))Q(t,\underline{P})(1-Q(t^{\prime
},\underline{P}))+o_{p}\left(  \frac{1}{s}\right)  =O_{p}\left(  \frac{1}%
{s}\right)  ,
\]
and
\[
\mathbb{E}_{\ast}[(\varepsilon^{X})^{2}(\tau^{Y})^{2}]=\sum_{i=1}^{m}\frac
{1}{s_{i}}\mathbb{E}_{P_{i}}[IF_{i}^{2}(t^{\prime},X_{i};\underline{P}%
)]Q(t,\underline{P})(1-Q(t,\underline{P}))+o_{p}\left(  \frac{1}{s}\right)
=O_{p}\left(  \frac{1}{s}\right)  .
\]
Then we compute $\mathbb{E}_{\ast}[\tau^{X}\varepsilon^{X}(\varepsilon
^{Y})^{2}]$. By (\ref{tau_X}), (\ref{conditional_E[epsilon_Xepsilon_Y^2]}) and
Assumption \ref{1st_expansion_empirical}, we have%
\begin{align*}
& \mathbb{E}_{\ast}[\tau^{X}\varepsilon^{X}(\varepsilon^{Y})^{2}]\\
& =\mathbb{E}_{\ast}[\tau^{X}\mathbb{E}_{\ast}[\varepsilon^{X}(\varepsilon
^{Y})^{2}|\underline{\hat{P}}_{\theta}^{\ast}]]\\
& =\mathbb{E}_{\ast}[\tau^{X}Q(t,\underline{\hat{P}}_{\theta}^{\ast
})(1-2Q(t^{\prime},\underline{\hat{P}}_{\theta}^{\ast}))(1-Q(t^{\prime
},\underline{\hat{P}}_{\theta}^{\ast}))]\\
& =\mathbb{E}_{\ast}\left[  \left(  \sum_{i=1}^{m}\frac{1}{s_{i}}\sum
_{k=1}^{s_{i}}IF_{i}(t,X_{i,k}^{\ast};\underline{\hat{P}})+\varepsilon^{\ast
}(t)-\mathbb{E}_{\ast}[\varepsilon^{\ast}(t)]\right)  \left(
Q(t,\underline{\hat{P}})+\sum_{i=1}^{m}\frac{1}{s_{i}}\sum_{k=1}^{s_{i}}%
IF_{i}(t,X_{i,k}^{\ast};\underline{\hat{P}})+\varepsilon^{\ast}(t)\right)
\right.  \\
& \times\left(  1-2Q(t^{\prime},\underline{\hat{P}})-2\sum_{i=1}^{m}\frac
{1}{s_{i}}\sum_{k=1}^{s_{i}}IF_{i}(t^{\prime},X_{i,k}^{\ast};\underline{\hat
{P}})-2\varepsilon^{\ast}(t^{\prime})\right)  \\
& \left.  \times\left(  1-Q(t^{\prime},\underline{\hat{P}})-\sum_{i=1}%
^{m}\frac{1}{s_{i}}\sum_{k=1}^{s_{i}}IF_{i}(t^{\prime},X_{i,k}^{\ast
};\underline{\hat{P}})-\varepsilon^{\ast}(t^{\prime})\right)  \right]  .
\end{align*}
By the mean-zero property of $IF_{i}$, the leading term in the above expansion
is 0:%
\[
Q(t,\underline{\hat{P}})(1-2Q(t^{\prime},\underline{\hat{P}}))(1-Q(t^{\prime
},\underline{\hat{P}}))\mathbb{E}_{\ast}\left[  \sum_{i=1}^{m}\frac{1}{s_{i}%
}\sum_{k=1}^{s_{i}}IF_{i}(t,X_{i,k}^{\ast};\underline{\hat{P}})\right]  =0.
\]
Among the remaining terms, the leading order term is%
\begin{align*}
& (1-2Q(t^{\prime},\underline{\hat{P}}))(1-Q(t^{\prime},\underline{\hat{P}%
}))\mathbb{E}_{\ast}\left[  \left(  \sum_{i=1}^{m}\frac{1}{s_{i}}\sum
_{k=1}^{s_{i}}IF_{i}(t,X_{i,k}^{\ast};\underline{\hat{P}})\right)
^{2}\right]  \\
& -Q(t,\underline{\hat{P}})(1-2Q(t^{\prime},\underline{\hat{P}}))\mathbb{E}%
_{\ast}\left[  \left(  \sum_{i=1}^{m}\frac{1}{s_{i}}\sum_{k=1}^{s_{i}}%
IF_{i}(t,X_{i,k}^{\ast};\underline{\hat{P}})\right)  \left(  \sum_{i=1}%
^{m}\frac{1}{s_{i}}\sum_{k=1}^{s_{i}}IF_{i}(t^{\prime},X_{i,k}^{\ast
};\underline{\hat{P}})\right)  \right]  \\
& -2Q(t,\underline{\hat{P}})(1-Q(t^{\prime},\underline{\hat{P}}))\mathbb{E}%
_{\ast}\left[  \left(  \sum_{i=1}^{m}\frac{1}{s_{i}}\sum_{k=1}^{s_{i}}%
IF_{i}(t,X_{i,k}^{\ast};\underline{\hat{P}})\right)  \left(  \sum_{i=1}%
^{m}\frac{1}{s_{i}}\sum_{k=1}^{s_{i}}IF_{i}(t^{\prime},X_{i,k}^{\ast
};\underline{\hat{P}})\right)  \right]  \\
& =(1-2Q(t^{\prime},\underline{P}))(1-Q(t^{\prime},\underline{P}))\sum
_{i=1}^{m}\frac{1}{s_{i}}\mathbb{E}_{P_{i}}(IF_{i}^{2}(t,X_{i};\underline{P}%
))\\
& +(4Q(t,\underline{P})Q(t^{\prime},\underline{P})-3Q(t,\underline{P}%
))\sum_{i=1}^{m}\frac{1}{s_{i}}\mathrm{Cov}_{P_{i}}(IF_{i}(t,X_{i}%
;\underline{P}),IF_{i}(t^{\prime},X_{i};\underline{P}))+o_{p}\left(  \frac
{1}{s}\right)  \\
& =O_{p}\left(  \frac{1}{s}\right)  ,
\end{align*}
where the equalities follow from (\ref{order_E[IF^2]}) and consistency of
$Q(\cdot,\underline{\hat{P}})$. Based on the orders in Assumption
\ref{1st_expansion_empirical}, (\ref{order_E[IF^2]}) and (\ref{order_E[IF^4]})
and by means of H\"{o}lder's inequality, we can see all other remaining terms are
of order $o_{p}\left(  1/s\right)  $. Therefore, we obtain
\begin{align*}
& \mathbb{E}_{\ast}[\tau^{X}\varepsilon^{X}(\varepsilon^{Y})^{2}]\\
& =(1-2Q(t^{\prime},\underline{P}))(1-Q(t^{\prime},\underline{P}))\sum
_{i=1}^{m}\frac{1}{s_{i}}\mathbb{E}_{P_{i}}(IF_{i}^{2}(t,X_{i};\underline{P}%
))\\
& +(4Q(t,\underline{P})Q(t^{\prime},\underline{P})-3Q(t,\underline{P}%
))\sum_{i=1}^{m}\frac{1}{s_{i}}\mathrm{Cov}_{P_{i}}(IF_{i}(t,X_{i}%
;\underline{P}),IF_{i}(t^{\prime},X_{i};\underline{P}))+o_{p}\left(  \frac
{1}{s}\right)  \\
& =O_{p}\left(  \frac{1}{s}\right)  .
\end{align*}
Similarly, $\mathbb{E}_{\ast}[(\varepsilon^{X})^{2}\tau^{Y}\varepsilon^{Y}]$
can be written as%
\begin{align*}
& \mathbb{E}_{\ast}[(\varepsilon^{X})^{2}\tau^{Y}\varepsilon^{Y}]\\
& =(1-4Q(t,\underline{P}))(1-Q(t^{\prime},\underline{P}))\sum_{i=1}^{m}%
\frac{1}{s_{i}}\mathrm{Cov}_{P_{i}}(IF_{i}(t,X_{i};\underline{P}%
),IF_{i}(t^{\prime},X_{i};\underline{P}))\\
& -Q(t,\underline{P})(1-2Q(t,\underline{P}))\sum_{i=1}^{m}\frac{1}{s_{i}%
}\mathbb{E}_{P_{i}}(IF_{i}^{2}(t^{\prime},X_{i};\underline{P}))+o_{p}\left(
\frac{1}{s}\right)  \\
& =O_{p}\left(  \frac{1}{s}\right)  .
\end{align*}
Next, the computation of $\mathbb{E}_{\ast}[(\varepsilon^{X})^{2}%
(\varepsilon^{Y})^{2}]$, $\mathbb{E}_{\ast}[\mathbb{E}_{\ast}[(\varepsilon
^{X})^{2}|\underline{\hat{P}}_{\theta}^{\ast}]\mathbb{E}_{\ast}[(\varepsilon
^{Y})^{2}|\underline{\hat{P}}_{\theta}^{\ast}]]$ and $\mathbb{E}_{\ast
}[(\mathbb{E}_{\ast}[\varepsilon^{X}\varepsilon^{Y}|\underline{\hat{P}%
}_{\theta}^{\ast}])^{2}]$ is similar to the computation of $\mathbb{E}_{\ast
}[\varepsilon^{X}\varepsilon^{Y}]$ in (\ref{E[epsilon_X*epsilon_Y]}).
Therefore, we can get%
\begin{align*}
& \mathbb{E}_{\ast}[(\varepsilon^{X})^{2}(\varepsilon^{Y})^{2}]\\
& =Q(t,\underline{P})(1-Q(t^{\prime},\underline{P}))(1-Q(t^{\prime
},\underline{P})-2Q(t,\underline{P})+3Q(t,\underline{P})Q(t^{\prime
},\underline{P}))+o_{p}(1)\\
& =O_{p}(1),
\end{align*}
\[
\mathbb{E}_{\ast}[\mathbb{E}_{\ast}[(\varepsilon^{X})^{2}|\underline{\hat{P}%
}_{\theta}^{\ast}]\mathbb{E}_{\ast}[(\varepsilon^{Y})^{2}|\underline{\hat{P}%
}_{\theta}^{\ast}]]=Q(t,\underline{P})(1-Q(t,\underline{P}))Q(t^{\prime
},\underline{P})(1-Q(t^{\prime},\underline{P}))+o_{p}(1)=O_{p}(1),
\]
and%
\[
\mathbb{E}_{\ast}[(\mathbb{E}_{\ast}[\varepsilon^{X}\varepsilon^{Y}%
|\underline{\hat{P}}_{\theta}^{\ast}])^{2}]=Q(t,\underline{P})^{2}%
(1-Q(t^{\prime},\underline{P}))^{2}+o_{p}(1)=O_{p}(1).
\]
Finally we compute the product of second moments. By the formulas of the
second moments at the beginning of this proof, we can get%
\[
\mathbb{E}_{\ast}[(\tau^{X})^{2}]\mathbb{E}_{\ast}[(\tau^{Y})^{2}]=\left(
\sum_{i=1}^{m}\frac{1}{s_{i}}\mathbb{E}_{P_{i}}[IF_{i}^{2}(t,X_{i}%
;\underline{P})]\right)  \left(  \sum_{i=1}^{m}\frac{1}{s_{i}}\mathbb{E}%
_{P_{i}}[IF_{i}^{2}(t^{\prime},X_{i};\underline{P})]\right)  +o_{p}\left(
\frac{1}{s^{2}}\right)  =O_{p}\left(  \frac{1}{s^{2}}\right)  ,
\]%
\[
\mathbb{E}_{\ast}[(\tau^{X})^{2}]\mathbb{E}_{\ast}[(\varepsilon^{Y})^{2}%
]=\sum_{i=1}^{m}\frac{1}{s_{i}}\mathbb{E}_{P_{i}}[IF_{i}^{2}(t,X_{i}%
;\underline{P})]Q(t^{\prime},\underline{P})(1-Q(t^{\prime},\underline{P}%
))+o_{p}\left(  \frac{1}{s}\right)  =O_{p}\left(  \frac{1}{s}\right)  ,
\]%
\[
\mathbb{E}_{\ast}[(\varepsilon^{X})^{2}]\mathbb{E}_{\ast}[(\tau^{Y})^{2}%
]=\sum_{i=1}^{m}\frac{1}{s_{i}}\mathbb{E}_{P_{i}}[IF_{i}^{2}(t^{\prime}%
,X_{i};\underline{P})]Q(t,\underline{P})(1-Q(t,\underline{P}))+o_{p}\left(
\frac{1}{s}\right)  =O_{p}\left(  \frac{1}{s}\right)  ,
\]%
\[
\mathbb{E}_{\ast}[(\varepsilon^{X})^{2}]\mathbb{E}_{\ast}[(\varepsilon
^{Y})^{2}]=Q(t,\underline{P})(1-Q(t,\underline{P}))Q(t^{\prime},\underline{P}%
)(1-Q(t^{\prime},\underline{P}))+o_{p}(1)=O_{p}(1),
\]%
\[
(\mathbb{E}_{\ast}[\tau^{X}\tau^{Y}])^{2}=\left(  \sum_{i=1}^{m}\frac{1}%
{s_{i}}\mathrm{Cov}_{P_{i}}(IF_{i}(t,X_{i};\underline{P}),IF_{i}(t^{\prime
},X_{i};\underline{P}))\right)  ^{2}+o_{p}\left(  \frac{1}{s^{2}}\right)
=O_{p}\left(  \frac{1}{s^{2}}\right)  ,
\]%
\[
\mathbb{E}_{\ast}[\tau^{X}\tau^{Y}]\mathbb{E}_{\ast}[\varepsilon
^{X}\varepsilon^{Y}]=\sum_{i=1}^{m}\frac{1}{s_{i}}\mathrm{Cov}_{P_{i}}%
(IF_{i}(t,X_{i};\underline{P}),IF_{i}(t^{\prime},X_{i};\underline{P}%
))Q(t,\underline{P})(1-Q(t^{\prime},\underline{P}))+o_{p}\left(  \frac{1}%
{s}\right)  =O_{p}\left(  \frac{1}{s}\right)  ,
\]%
\[
(\mathbb{E}_{\ast}[\varepsilon^{X}\varepsilon^{Y}])^{2}=Q(t,\underline{P}%
)^{2}(1-Q(t^{\prime},\underline{P}))^{2}+o_{p}(1)=O_{p}(1).
\]
\end{proof}

Now let us prove Theorem \ref{MSE_sigma_hat}.

\begin{proof}{Proof of Theorem \ref{MSE_sigma_hat}.}
First, by Theorem \ref{verification_general_assumptions}, Assumptions \ref{convergence_in_p_to_truth}-\ref{3rd_expansion_empirical} hold under Assumptions \ref{balanced_data} and \ref{finite_horizon_model}. Therefore, the conclusions in Lemma \ref{computation_moments} hold. Recall that $s=(1/m)\sum_{i=1}^{m}s_{i}=(1/m)\sum_{i=1}^{m}\lfloor\theta
n_{i}\rfloor$ is the average subsample size. We have $s/\theta n\rightarrow1$
under Assumption \ref{balanced_data}.

Notice that given the input data, the procedure that leads to $\hat{\sigma
}^{2}(t,t^{\prime})$ is the same as Algorithm
\ref{alg:general_covariance_undebias} where $Z_{b}=\underline{\hat{P}}%
_{\theta}^{\ast b}$, $X_{br}=I(Y_{r}^{\ast b}\leq t)$ and $Y_{br}%
=I(Y_{r}^{\ast b}\leq t^{\prime})$. Without loss of generality, assume $t\leq
t^{\prime}$. By Lemma \ref{general_covariance_undebias}, the conditional
expectation of $\hat{\sigma}^{2}(t,t^{\prime})$ is given by%
\begin{align*}
\mathbb{E}_{\ast}[\hat{\sigma}^{2}(t,t^{\prime})]  & =\theta n\mathrm{Cov}%
_{\ast}(Q(t,\underline{\hat{P}}_{\theta}^{\ast}),Q(t^{\prime},\underline{\hat
{P}}_{\theta}^{\ast}))+\frac{\theta n}{R_s}\mathbb{E}_{\ast}[Q(t,\underline{\hat
{P}}_{\theta}^{\ast})-Q(t,\underline{\hat{P}}_{\theta}^{\ast})Q(t^{\prime
},\underline{\hat{P}}_{\theta}^{\ast})]\\
& =\sigma^{2}(t,t^{\prime})+\frac{\theta n}{R_s}\mathbb{E}_{\ast}[\varepsilon
^{X}\varepsilon^{Y}],
\end{align*}
which implies, by Lemma \ref{computation_moments}, that%
\[
(\mathbb{E}_{\ast}[\hat{\sigma}^{2}(t,t^{\prime})]-\sigma^{2}(t,t^{\prime
}))^{2}=O_{p}\left(  \frac{(\theta n)^{2}}{R_s^{2}}\right)  .
\]
Moreover, plugging the results of Lemma \ref{computation_moments} into the
formula of the variance in Lemma \ref{general_covariance_undebias}, we obtain%
\begin{align}
\mathrm{Var}_{\ast}(\hat{\sigma}^{2}(t,t^{\prime})) &  =(\theta n)^{2}O\left(
\frac{1}{Bs^{2}}+\frac{1}{BR_{s}s}+\frac{1}{BR_{s}^{2}s}\right.  \nonumber\\
&  \left.  +\frac{1}{BR_{s}^{3}}+\frac{1}{BR_{s}^{2}}+\frac{1}{B^{2}s^{2}}+\frac
{1}{B^{2}R_{s}s}+\frac{1}{B^{2}R_{s}^{2}}\right)  (1+o_{p}%
(1)).\label{order_conditional_variance}%
\end{align}
Recall that $s/\theta n\rightarrow1$. So under (\ref{configuration_B_R}), it
is easy to see
\[
\mathbb{E}_{\ast}[(\hat{\sigma}^{2}(t,t^{\prime})-\sigma^{2}(t,t^{\prime
}))^{2}]=o_{p}(1).
\]
Now suppose%
\[
0<Q(t,\underline{P}),Q(t^{\prime},\underline{P})<1,\quad\sum_{i=1}%
^{m}\mathbb{E}_{P_{i}}[IF_{i}^{2}(t,X_{i};\underline{P})]>0,\quad\sum
_{i=1}^{m}\mathbb{E}_{P_{i}}[IF_{i}^{2}(t^{\prime},X_{i};\underline{P})]>0.
\]
Lemma \ref{computation_moments} tells us that%
\[
\mathbb{E}_{\ast}[\mathbb{E}_{\ast}[(\varepsilon^{X})^{2}|\underline{\hat{P}%
}_{\theta}^{\ast}]\mathbb{E}_{\ast}[(\varepsilon^{Y})^{2}|\underline{\hat{P}%
}_{\theta}^{\ast}]]=Q(t,\underline{P})(1-Q(t,\underline{P}))Q(t^{\prime
},\underline{P})(1-Q(t^{\prime},\underline{P}))+o_{p}(1)=\Theta_{p}(1),
\]%
\[
(\mathbb{E}_{\ast}[\varepsilon^{X}\varepsilon^{Y}])^{2}=Q(t,\underline{P}%
)^{2}(1-Q(t^{\prime},\underline{P}))^{2}+o_{p}(1)=\Theta_{p}(1),
\]
and%
\[
\mathbb{E}_{\ast}[(\tau^{X})^{2}]\mathbb{E}_{\ast}[(\tau^{Y})^{2}]=\left(
\sum_{i=1}^{m}\frac{1}{s_{i}}\mathbb{E}_{P_{i}}[IF_{i}^{2}(t,X_{i}%
;\underline{P})]\right)  \left(  \sum_{i=1}^{m}\frac{1}{s_{i}}\mathbb{E}%
_{P_{i}}[IF_{i}^{2}(t^{\prime},X_{i};\underline{P})]\right)  +o_{p}\left(
\frac{1}{s^{2}}\right)  =\Theta_{p}\left(  \frac{1}{s^{2}}\right)  .
\]
Therefore, by Lemma \ref{general_covariance_undebias},%
\begin{align*}
\mathrm{Var}_{\ast}(\hat{\sigma}^{2}(t,t^{\prime})) &  \geq(\theta
n)^{2}\left(  \frac{(R_{s}-1)}{BR_{s}^{3}}\mathbb{E}_{\ast}[\mathbb{E}_{\ast
}[(\varepsilon^{X})^{2}|\underline{\hat{P}}_{\theta}^{\ast}]\mathbb{E}_{\ast
}[(\varepsilon^{Y})^{2}|\underline{\hat{P}}_{\theta}^{\ast}]]+\frac{1}%
{B(B-1)}\mathbb{E}_{\ast}[(\tau^{X})^{2}]\mathbb{E}_{\ast}[(\tau^{Y}%
)^{2}]\right)  \\
&  =(\theta n)^{2}\left(  \Theta_{p}\left(  \frac{1}{BR_{s}^{2}}\right)
+\Theta_{p}\left(  \frac{1}{B^{2}s^{2}}\right)  \right)  ,
\end{align*}
and%
\[
(\mathbb{E}_{\ast}[\hat{\sigma}^{2}(t,t^{\prime})]-\sigma^{2}(t,t^{\prime
}))^{2}=\Theta_{p}\left(  \frac{(\theta n)^{2}}{R_{s}^{2}}\right)
\]
In order to make sure $\mathbb{E}_{\ast}[(\hat{\sigma}^{2}(t,t^{\prime
})-\sigma^{2}(t,t^{\prime}))^{2}]=o_{p}(1)$, we must have $B=\omega(1)$ and
$R_{s}=\omega(s)$.
\end{proof}

\begin{proof}{Proof of Theorem \ref{validity_alg_cov}.}
By Theorem \ref{consistency_sigma}, if $\theta n=\omega(1)$, then $\sigma
^{2}(t,t^{\prime})-\mathrm{Cov}(\mathbb{G}(t),\mathbb{G}(t^{\prime}%
))=o_{p}(1)$. By Theorem \ref{MSE_sigma_hat}, if $\theta n=\omega(1)$,
$B=\omega(1)$ and $R_{s}=\omega(s)$, then we have $\mathbb{E}_{\ast}%
[(\hat{\sigma}^{2}(t,t^{\prime})-\sigma^{2}(t,t^{\prime}))^{2}]=o_{p}(1)$,
which implies $\hat{\sigma}^{2}(t,t^{\prime})-\sigma^{2}(t,t^{\prime}%
)=o_{p}(1)$. Therefore, if $\theta n=\omega(1)$, $B=\omega(1)$ and
$R_{s}=\omega(s)$, we have $\hat{\sigma}^{2}(t,t^{\prime})-\mathrm{Cov}%
(\mathbb{G}(t),\mathbb{G}(t^{\prime}))=o_{p}(1)$.
\end{proof}

\begin{proof}{Proof of Corollary \ref{min_budget}.}
It is easy to see the necessary condition on $N$ s.t.
(\ref{configuration_B_R_theta}) holds is $N=\omega(1)$. Next, we show
$N=\omega(1)$ is also sufficient for (\ref{configuration_B_R_theta}). In fact,
by choosing $\theta=\sqrt{N}/n$, $B=N^{1/3}$ and $R_{s}=N^{2/3}$,
(\ref{configuration_B_R_theta}) holds. Therefore, $N=\omega(1)$ is the minimum budget s.t. (\ref{configuration_B_R_theta}) holds.
\end{proof}

\begin{proof}{Proof of Theorem \ref{optimal_B_R}.}
Without loss of generality, assume $t\leq t^{\prime}$. First, by Theorem \ref{verification_general_assumptions}, Assumptions \ref{convergence_in_p_to_truth}-\ref{3rd_expansion_empirical} hold under Assumptions \ref{balanced_data} and \ref{finite_horizon_model}. Therefore, the conclusions in Lemma \ref{computation_moments} hold. According to Theorem \ref{MSE_sigma_hat}, to minimize the order of $\mathbb{E}_{\ast}[(\hat{\sigma
}^{2}(t,t^{\prime})-\sigma^{2}(t,t^{\prime}))^{2}]$, $B$ and $R_{s}$ must
satisfy $B=\omega(1)$ and $R_{s}=\omega(s)$ otherwise $\mathbb{E}_{\ast}%
[(\hat{\sigma}^{2}(t,t^{\prime})-\sigma^{2}(t,t^{\prime}))^{2}]$ is even not
$o_{p}(1)$. By plugging the results of Lemma \ref{computation_moments} into
the formula in Lemma \ref{general_covariance_undebias}, and using the
constraints $B=\omega(1),R_{s}=\omega(s)$, we can see the leading term of
$\mathbb{E}_{\ast}[(\hat{\sigma}^{2}(t,t^{\prime})-\sigma^{2}(t,t^{\prime
}))^{2}]$ is given by%
\begin{align*}
& \mathbb{E}_{\ast}[(\hat{\sigma}^{2}(t,t^{\prime})-\sigma^{2}(t,t^{\prime
}))^{2}]\\
& =(\theta n)^{2}\left[  \frac{1}{R_{s}^{2}}Q(t,\underline{P})^{2}%
(1-Q(t^{\prime},\underline{P}))^{2}\right.  \\
& \left.  +\frac{1}{B}\left(  \sum_{i=1}^{m}\frac{1}{s_{i}}\mathbb{E}_{P_{i}%
}[IF_{i}^{2}(t,X_{i};\underline{P})]\right)  \left(  \sum_{i=1}^{m}\frac
{1}{s_{i}}\mathbb{E}_{P_{i}}[IF_{i}^{2}(t^{\prime},X_{i};\underline{P}%
)]\right)  \right]  (1+o_{p}(1))\\
& =\Theta\left(  \frac{s^{2}}{R_{s}^{2}}+\frac{1}{B}\right)  .
\end{align*}
Therefore, the minimizer of the order of $\mathbb{E}_{\ast}[(\hat{\sigma}%
^{2}(t,t^{\prime})-\sigma^{2}(t,t^{\prime}))^{2}]$ subject to $BR_{s}=N$ is
$R_{s}^{\ast}=\Theta(N^{1/3}s^{2/3})$ and consequently $B^{\ast}=N/R_{s}%
^{\ast}$. Since $N=\omega(\theta n)=\omega(s)$, the optimal $R_{s}^{\ast}$ and
$B^{\ast}$ also satisfies the conditions $B=\omega(1)$ and $R_{s}=\omega(s)$
which we assume at the beginning of the proof. With the optimal $R_{s}^{\ast}$ and
$B^{\ast}$, we can see $\mathbb{E}_{\ast}[(\hat{\sigma}^{2}(t,t^{\prime
})-\sigma^{2}(t,t^{\prime}))^{2}]=\Theta((\theta n)^{2/3}/N^{2/3}%
)(1+o_{p}(1))$.

If the non-degeneracy condition (\ref{non-degeneracy}) does not hold, by
plugging $B^{\ast}$ and $R_{s}^{\ast}$ into (\ref{order_conditional_variance}%
), we can get%
\[
\mathbb{E}_{\ast}[(\hat{\sigma}^{2}(t,t^{\prime})-\sigma^{2}(t,t^{\prime
}))^{2}]=O((\theta n)^{2/3}/N^{2/3})(1+o_{p}(1)).
\]

\end{proof}

\begin{proof}{Proof of Theorem \ref{error_true_bootstrap}.}
First, by Theorem \ref{verification_general_assumptions}, Assumptions \ref{convergence_in_p_to_truth}-\ref{3rd_expansion_empirical} hold under Assumptions \ref{balanced_data} and \ref{finite_horizon_model}. According to the computation in the proof of Theorem \ref{consistency_sigma},
we have%
\begin{align*}
\sigma^{2}(t,t^{\prime})  &  =\theta n\mathbb{E}_{\ast}\left[  \left(
\sum_{i=1}^{m}\frac{1}{s_{i}}\sum_{k=1}^{s_{i}}IF_{i}(t,X_{i,k}^{\ast
};\underline{\hat{P}})\right)  \left(  \sum_{i=1}^{m}\frac{1}{s_{i}}\sum
_{k=1}^{s_{i}}IF_{i}(t^{\prime},X_{i,k}^{\ast};\underline{\hat{P}})\right)
\right] \\
&  +\theta n\mathbb{E}_{\ast}\left[  \left(  \sum_{i=1}^{m}\frac{1}{s_{i}}%
\sum_{k=1}^{s_{i}}IF_{i}(t,X_{i,k}^{\ast};\underline{\hat{P}})\right)
(\varepsilon^{\ast}(t^{\prime})-\mathbb{E}_{\ast}[\varepsilon^{\ast}%
(t^{\prime})])\right] \\
&  +\theta n\mathbb{E}_{\ast}\left[  \left(  \sum_{i=1}^{m}\frac{1}{s_{i}}%
\sum_{k=1}^{s_{i}}IF_{i}(t^{\prime},X_{i,k}^{\ast};\underline{\hat{P}%
})\right)  (\varepsilon^{\ast}(t)-\mathbb{E}_{\ast}[\varepsilon^{\ast
}(t)])\right] \\
&  +\theta n\mathbb{E}_{\ast}[(\varepsilon^{\ast}(t)-\mathbb{E}_{\ast
}[\varepsilon^{\ast}(t)])(\varepsilon^{\ast}(t^{\prime})-\mathbb{E}_{\ast
}[\varepsilon^{\ast}(t^{\prime})])].
\end{align*}
By the mean-zero property of the influence function and conditional
independence among $X_{i,k}^{\ast}$, the first term can be further simplified
as%
\begin{align*}
&  \theta n\mathbb{E}_{\ast}\left[  \left(  \sum_{i=1}^{m}\frac{1}{s_{i}}%
\sum_{k=1}^{s_{i}}IF_{i}(t,X_{i,k}^{\ast};\underline{\hat{P}})\right)  \left(
\sum_{i=1}^{m}\frac{1}{s_{i}}\sum_{k=1}^{s_{i}}IF_{i}(t^{\prime},X_{i,k}%
^{\ast};\underline{\hat{P}})\right)  \right] \\
&  =\theta n\sum_{i=1}^{m}\sum_{k=1}^{s_{i}}\frac{1}{s_{i}^{2}}\mathbb{E}%
_{\ast}[IF_{i}(t,X_{i,k}^{\ast};\underline{\hat{P}})IF_{i}(t^{\prime}%
,X_{i,k}^{\ast};\underline{\hat{P}})]\\
&  =\theta n\sum_{i=1}^{m}\frac{1}{s_{i}}\mathbb{E}_{\ast}[IF_{i}%
(t,X_{i,1}^{\ast};\underline{\hat{P}})IF_{i}(t^{\prime},X_{i,1}^{\ast
};\underline{\hat{P}})]\\
&  =\theta n\sum_{i=1}^{m}\frac{1}{s_{i}}\sum_{j=1}^{n_{i}}\frac{1}{n_{i}%
}IF_{i}(t,X_{i,j};\underline{\hat{P}})IF_{i}(t^{\prime},X_{i,j}%
;\underline{\hat{P}}).
\end{align*}
The second term can be simplified as%
\begin{align*}
&  \theta n\mathbb{E}_{\ast}\left[  \left(  \sum_{i=1}^{m}\frac{1}{s_{i}}%
\sum_{k=1}^{s_{i}}IF_{i}(t,X_{i,k}^{\ast};\underline{\hat{P}})\right)
(\varepsilon^{\ast}(t^{\prime})-\mathbb{E}_{\ast}[\varepsilon^{\ast}%
(t^{\prime})])\right] \\
&  =\theta n\sum_{i=1}^{m}\frac{1}{s_{i}}\sum_{k=1}^{s_{i}}\mathbb{E}_{\ast
}[IF_{i}(t,X_{i,k}^{\ast};\underline{\hat{P}})(\varepsilon^{\ast}(t^{\prime
})-\mathbb{E}_{\ast}[\varepsilon^{\ast}(t^{\prime})])]\\
&  =\theta n\sum_{i=1}^{m}\frac{1}{s_{i}}\sum_{k=1}^{s_{i}}\mathbb{E}_{\ast
}[IF_{i}(t,X_{i,k}^{\ast};\underline{\hat{P}})(\mathbb{E}_{\ast}%
[\varepsilon^{\ast}(t^{\prime})|X_{i,k}^{\ast}]-\mathbb{E}_{\ast}%
[\varepsilon^{\ast}(t^{\prime})])]\\
&  =\theta n\sum_{i=1}^{m}\mathbb{E}_{\ast}[IF_{i}(t,X_{i,1}^{\ast
};\underline{\hat{P}})(\mathbb{E}_{\ast}[\varepsilon^{\ast}(t^{\prime
})|X_{i,1}^{\ast}]-\mathbb{E}_{\ast}[\varepsilon^{\ast}(t^{\prime})])],
\end{align*}
where the last equality follows from the symmetry of $\varepsilon^{\ast
}(t^{\prime})$ about $X_{i,1}^{\ast},\ldots,X_{i,s_{i}}$ (see its definition
in Assumption \ref{1st_expansion_empirical}). Similarly,
\begin{align*}
&  \theta n\mathbb{E}_{\ast}\left[  \left(  \sum_{i=1}^{m}\frac{1}{s_{i}}%
\sum_{k=1}^{s_{i}}IF_{i}(t^{\prime},X_{i,k}^{\ast};\underline{\hat{P}%
})\right)  (\varepsilon^{\ast}(t)-\mathbb{E}_{\ast}[\varepsilon^{\ast
}(t)])\right] \\
&  =\theta n\sum_{i=1}^{m}\mathbb{E}_{\ast}[IF_{i}(t^{\prime},X_{i,1}^{\ast
};\underline{\hat{P}})(\mathbb{E}_{\ast}[\varepsilon^{\ast}(t)|X_{i,1}^{\ast
}]-\mathbb{E}_{\ast}[\varepsilon^{\ast}(t)])].
\end{align*}
Therefore, $\sigma^{2}(t,t^{\prime})$ is rewritten as%
\begin{align}
\sigma^{2}(t,t^{\prime})  &  =\theta n\sum_{i=1}^{m}\frac{1}{s_{i}}\sum
_{j=1}^{n_{i}}\frac{1}{n_{i}}IF_{i}(t,X_{i,j};\underline{\hat{P}}%
)IF_{i}(t^{\prime},X_{i,j};\underline{\hat{P}})\nonumber\\
&  +\theta n\sum_{i=1}^{m}\mathbb{E}_{\ast}[IF_{i}(t,X_{i,1}^{\ast
};\underline{\hat{P}})(\mathbb{E}_{\ast}[\varepsilon^{\ast}(t^{\prime
})|X_{i,1}^{\ast}]-\mathbb{E}_{\ast}[\varepsilon^{\ast}(t^{\prime
})])]\nonumber\\
&  +\theta n\sum_{i=1}^{m}\mathbb{E}_{\ast}[IF_{i}(t^{\prime},X_{i,1}^{\ast
};\underline{\hat{P}})(\mathbb{E}_{\ast}[\varepsilon^{\ast}(t)|X_{i,1}^{\ast
}]-\mathbb{E}_{\ast}[\varepsilon^{\ast}(t)])]\nonumber\\
&  +\theta n\mathbb{E}_{\ast}[(\varepsilon^{\ast}(t)-\mathbb{E}_{\ast
}[\varepsilon^{\ast}(t)])(\varepsilon^{\ast}(t^{\prime})-\mathbb{E}_{\ast
}[\varepsilon^{\ast}(t^{\prime})])]. \label{computation_sigma2}%
\end{align}
Next, we study the order of each term in (\ref{computation_sigma2}). We will
show that%
\begin{align}
&  \theta n\sum_{i=1}^{m}\frac{1}{s_{i}}\sum_{j=1}^{n_{i}}\frac{1}{n_{i}%
}IF_{i}(t,X_{i,j};\underline{\hat{P}})IF_{i}(t^{\prime},X_{i,j}%
;\underline{\hat{P}})\nonumber\\
&  =\mathrm{Cov}(\mathbb{G}(t),\mathbb{G}(t^{\prime}))+\sum_{i=1}^{m}\left(
\frac{n\mathrm{frac}(\theta n_{i})}{s_{i}n_{i}}+\left(  \frac{n}{n_{i}}%
-\frac{1}{\beta_{i}}\right)  \right)  \mathrm{Cov}(IF_{i}(t,X_{i}%
;\underline{P}),IF_{i}(t^{\prime},X_{i};\underline{P}))\nonumber\\
&  +\sum_{i=1}^{m}\frac{1}{\beta_{i}}\left(  \frac{1}{n_{i}}\sum_{j=1}^{n_{i}%
}IF_{i}(t,X_{i,j};\underline{P})IF_{i}(t^{\prime},X_{i,j};\underline{P}%
)-\mathrm{Cov}(IF_{i}(t,X_{i};\underline{P}),IF_{i}(t^{\prime},X_{i}%
;\underline{P}))\right) \nonumber\\
&  +\sum_{i=1}^{m}\frac{1}{n_{i}}\sum_{j=1}^{n_{i}}\sum_{i^{\prime}=1}%
^{m}\frac{1}{\beta_{i^{\prime}}}\mathbb{E[}IF_{i^{\prime}}(t,X_{i^{\prime}%
};\underline{P})IF_{i^{\prime}i}(t^{\prime},X_{i^{\prime}},X_{i,j}%
;\underline{P})|X_{i,j}]\nonumber\\
&  +\sum_{i=1}^{m}\frac{1}{n_{i}}\sum_{j=1}^{n_{i}}\sum_{i^{\prime}=1}%
^{m}\frac{1}{\beta_{i^{\prime}}}\mathbb{E[}IF_{i^{\prime}}(t^{\prime
},X_{i^{\prime}};\underline{P})IF_{i^{\prime}i}(t,X_{i^{\prime}}%
,X_{i,j};\underline{P})|X_{i,j}]+o_{p}\left(  \frac{1}{\sqrt{n}}\right)  ,
\label{computation_sigma_term10}%
\end{align}%
\begin{align}
&  \theta n\sum_{i=1}^{m}\mathbb{E}_{\ast}[IF_{i}(t,X_{i,1}^{\ast
};\underline{\hat{P}})(\mathbb{E}_{\ast}[\varepsilon^{\ast}(t^{\prime
})|X_{i,1}^{\ast}]-\mathbb{E}_{\ast}[\varepsilon^{\ast}(t^{\prime
})])]\nonumber\\
&  =\theta n\sum_{i=1}^{m}\frac{1}{2s_{i}^{2}}\mathrm{Cov}(IF_{i}%
(t,X_{i};\underline{P}),IF_{ii}(t^{\prime},X_{i},X_{i};\underline{P}%
))\nonumber\\
&  +\theta n\sum_{i=1}^{m}\sum_{i^{\prime}=1}^{m}\frac{1}{2s_{i}s_{i^{\prime}%
}}\mathrm{Cov}(IF_{i}(t,X_{i,1};\underline{P}),IF_{ii^{\prime}i^{\prime}%
}(t^{\prime},X_{i,1},X_{i^{\prime},2},X_{i^{\prime},2};\underline{P}%
))+o_{p}\left(  \frac{1}{s}\right)  , \label{computation_sigma_term20}%
\end{align}%
\begin{align}
&  \theta n\sum_{i=1}^{m}\mathbb{E}_{\ast}[IF_{i}(t^{\prime},X_{i,1}^{\ast
};\underline{\hat{P}})(\mathbb{E}_{\ast}[\varepsilon^{\ast}(t)|X_{i,1}^{\ast
}]-\mathbb{E}_{\ast}[\varepsilon^{\ast}(t)])]\nonumber\\
&  =\theta n\sum_{i=1}^{m}\frac{1}{2s_{i}^{2}}\mathrm{Cov}(IF_{i}(t^{\prime
},X_{i};\underline{P}),IF_{ii}(t,X_{i},X_{i};\underline{P}))\nonumber\\
&  +\theta n\sum_{i=1}^{m}\sum_{i^{\prime}=1}^{m}\frac{1}{2s_{i}s_{i^{\prime}%
}}\mathrm{Cov}(IF_{i}(t^{\prime},X_{i,1};\underline{P}),IF_{ii^{\prime
}i^{\prime}}(t,X_{i,1},X_{i^{\prime},2},X_{i^{\prime},2};\underline{P}%
))+o_{p}\left(  \frac{1}{s}\right)  , \label{computation_sigma_term30}%
\end{align}%
\begin{align}
&  \theta n\mathbb{E}_{\ast}[(\varepsilon^{\ast}(t)-\mathbb{E}_{\ast
}[\varepsilon^{\ast}(t)])(\varepsilon^{\ast}(t^{\prime})-\mathbb{E}_{\ast
}[\varepsilon^{\ast}(t^{\prime})])]\nonumber\\
&  =\theta n\sum_{i_{1},i_{2}=1}^{m}\frac{1}{2s_{i_{1}}s_{i_{2}}}%
\mathrm{Cov}(IF_{i_{1}i_{2}}(t,X_{i_{1},1},X_{i_{2},2};\underline{P}%
),IF_{i_{1}i_{2}}(t^{\prime},X_{i_{1},1},X_{i_{2},2};\underline{P}%
))+o_{p}\left(  \frac{1}{s}\right)  . \label{computation_sigma_term40}%
\end{align}

We first show (\ref{computation_sigma_term10}). According to Assumption
\ref{3rd_expansion_empirical}, we have%
\begin{align}
&  \sum_{j=1}^{n_{i}}\frac{1}{n_{i}}IF_{i}(t,X_{i,j};\underline{\hat{P}%
})IF_{i}(t^{\prime},X_{i,j};\underline{\hat{P}})\nonumber\\
&  =\frac{1}{n_{i}}\sum_{j=1}^{n_{i}}IF_{i}(t,X_{i,j};\underline{P}%
)IF_{i}(t^{\prime},X_{i,j};\underline{P})\nonumber\\
&  +\frac{1}{n_{i}}\sum_{j=1}^{n_{i}}IF_{i}(t,X_{i,j};\underline{P})\left(
\sum_{i^{\prime}=1}^{m}\frac{1}{n_{i^{\prime}}}\sum_{j^{\prime}=1}%
^{n_{i^{\prime}}}IF_{ii^{\prime}}(t^{\prime},X_{i,j},X_{i^{\prime},j^{\prime}%
};\underline{P})-\frac{1}{n_{i}}\sum_{j^{\prime}=1}^{n_{i}}IF_{i}(t^{\prime
},X_{i,j^{\prime}};\underline{P})\right) \nonumber\\
&  +\frac{1}{n_{i}}\sum_{j=1}^{n_{i}}IF_{i}(t^{\prime},X_{i,j};\underline{P}%
)\left(  \sum_{i^{\prime}=1}^{m}\frac{1}{n_{i^{\prime}}}\sum_{j^{\prime}%
=1}^{n_{i^{\prime}}}IF_{ii^{\prime}}(t,X_{i,j},X_{i^{\prime},j^{\prime}%
};\underline{P})-\frac{1}{n_{i}}\sum_{j^{\prime}=1}^{n_{i}}IF_{i}%
(t,X_{i,j^{\prime}};\underline{P})\right)  +o_{p}\left(  \frac{1}{\sqrt{n}%
}\right)  . \label{computation_sigma_term11}%
\end{align}
To further simplify the second term, we notice that
\[
\frac{1}{n_{i}n_{i^{\prime}}}\sum_{j=1}^{n_{i}}\sum_{j^{\prime}=1}%
^{n_{i^{\prime}}}IF_{i}(t,X_{i,j};\underline{P})IF_{ii^{\prime}}(t^{\prime
},X_{i,j},X_{i^{\prime},j^{\prime}};\underline{P})
\]
is a (one/two-sample depending on whether $i=i^{\prime}$) V-statistics with
$\mathbb{E[}IF_{i}(t,x;\underline{P})IF_{ii^{\prime}}(t^{\prime}%
,x,X_{i^{\prime},j^{\prime}};\underline{P})]=0$. Using H\'{a}jek projection on
$X_{i^{\prime},j^{\prime}}$, we can obtain that%
\[
\frac{1}{n_{i}n_{i^{\prime}}}\sum_{j=1}^{n_{i}}\sum_{j^{\prime}=1}%
^{n_{i^{\prime}}}IF_{i}(t,X_{i,j};\underline{P})IF_{ii^{\prime}}(t^{\prime
},X_{i,j},X_{i^{\prime},j^{\prime}};\underline{P})=\frac{1}{n_{i^{\prime}}%
}\sum_{j^{\prime}=1}^{n_{i^{\prime}}}\mathbb{E[}IF_{i}(t,X_{i};\underline{P}%
)IF_{ii^{\prime}}(t^{\prime},X_{i},X_{i^{\prime},j^{\prime}};\underline{P}%
)|X_{i^{\prime},j^{\prime}}]+O_{p}\left(  \frac{1}{n}\right)  ,
\]
where $X_{i}\sim P_{i}$ denotes a generic random variable independent of the
data $X_{i,j}$. Besides,%
\[
\frac{1}{n_{i}^{2}}\sum_{j=1}^{n_{i}}\sum_{j^{\prime}=1}^{n_{i}}%
IF_{i}(t,X_{i,j};\underline{P})IF_{i}(t^{\prime},X_{i,j^{\prime}%
};\underline{P})
\]
is a degenerate V-statistics, which implies that it is of order $O_{p}(1/n)$.
Therefore, the second term in (\ref{computation_sigma_term11}) can be
simplified as%
\begin{align*}
&  \frac{1}{n_{i}}\sum_{j=1}^{n_{i}}IF_{i}(t,X_{i,j};\underline{P})\left(
\sum_{i^{\prime}=1}^{m}\frac{1}{n_{i^{\prime}}}\sum_{j^{\prime}=1}%
^{n_{i^{\prime}}}IF_{ii^{\prime}}(t^{\prime},X_{i,j},X_{i^{\prime},j^{\prime}%
};\underline{P})-\frac{1}{n_{i}}\sum_{j^{\prime}=1}^{n_{i}}IF_{i}(t^{\prime
},X_{i,j^{\prime}};\underline{P})\right) \\
&  =\sum_{i^{\prime}=1}^{m}\frac{1}{n_{i^{\prime}}}\sum_{j^{\prime}%
=1}^{n_{i^{\prime}}}\mathbb{E[}IF_{i}(t,X_{i};\underline{P})IF_{ii^{\prime}%
}(t^{\prime},X_{i},X_{i^{\prime},j^{\prime}};\underline{P})|X_{i^{\prime
},j^{\prime}}]+O_{p}\left(  \frac{1}{n}\right)  .
\end{align*}
Similarly, the third term in (\ref{computation_sigma_term11}) can be
simplified as%
\begin{align*}
&  \frac{1}{n_{i}}\sum_{j=1}^{n_{i}}IF_{i}(t^{\prime},X_{i,j};\underline{P}%
)\left(  \sum_{i^{\prime}=1}^{m}\frac{1}{n_{i^{\prime}}}\sum_{j^{\prime}%
=1}^{n_{i^{\prime}}}IF_{ii^{\prime}}(t,X_{i,j},X_{i^{\prime},j^{\prime}%
};\underline{P})-\frac{1}{n_{i}}\sum_{j^{\prime}=1}^{n_{i}}IF_{i}%
(t,X_{i,j^{\prime}};\underline{P})\right) \\
&  =\sum_{i^{\prime}=1}^{m}\frac{1}{n_{i^{\prime}}}\sum_{j^{\prime}%
=1}^{n_{i^{\prime}}}\mathbb{E[}IF_{i}(t^{\prime},X_{i};\underline{P}%
)IF_{ii^{\prime}}(t,X_{i},X_{i^{\prime},j^{\prime}};\underline{P}%
)|X_{i^{\prime},j^{\prime}}]+O_{p}\left(  \frac{1}{n}\right)  .
\end{align*}
Hence, (\ref{computation_sigma_term11}) can be rewritten as%
\begin{align*}
&  \sum_{j=1}^{n_{i}}\frac{1}{n_{i}}IF_{i}(t,X_{i,j};\underline{\hat{P}%
})IF_{i}(t^{\prime},X_{i,j};\underline{\hat{P}})\\
&  =\frac{1}{n_{i}}\sum_{j=1}^{n_{i}}IF_{i}(t,X_{i,j};\underline{P}%
)IF_{i}(t^{\prime},X_{i,j};\underline{P})+\sum_{i^{\prime}=1}^{m}\frac
{1}{n_{i^{\prime}}}\sum_{j^{\prime}=1}^{n_{i^{\prime}}}\mathbb{E[}%
IF_{i}(t,X_{i};\underline{P})IF_{ii^{\prime}}(t^{\prime},X_{i},X_{i^{\prime
},j^{\prime}};\underline{P})|X_{i^{\prime},j^{\prime}}]\\
&  +\sum_{i^{\prime}=1}^{m}\frac{1}{n_{i^{\prime}}}\sum_{j^{\prime}%
=1}^{n_{i^{\prime}}}\mathbb{E[}IF_{i}(t^{\prime},X_{i};\underline{P}%
)IF_{ii^{\prime}}(t,X_{i},X_{i^{\prime},j^{\prime}};\underline{P}%
)|X_{i^{\prime},j^{\prime}}]+o_{p}\left(  \frac{1}{\sqrt{n}}\right)  .
\end{align*}
Plugging it back into the LHS of (\ref{computation_sigma_term10}), we have%
\begin{align*}
&  \theta n\sum_{i=1}^{m}\frac{1}{s_{i}}\sum_{j=1}^{n_{i}}\frac{1}{n_{i}%
}IF_{i}(t,X_{i,j};\underline{\hat{P}})IF_{i}(t^{\prime},X_{i,j}%
;\underline{\hat{P}})\\
&  =\sum_{i=1}^{m}\frac{\theta n}{s_{i}}\frac{1}{n_{i}}\sum_{j=1}^{n_{i}%
}IF_{i}(t,X_{i,j};\underline{P})IF_{i}(t^{\prime},X_{i,j};\underline{P}%
)+\sum_{i=1}^{m}\frac{\theta n}{s_{i}}\sum_{i^{\prime}=1}^{m}\frac
{1}{n_{i^{\prime}}}\sum_{j^{\prime}=1}^{n_{i^{\prime}}}\mathbb{E}_{X_{i}\sim
P_{i}}\mathbb{[}IF_{i}(t,X_{i};\underline{P})IF_{ii^{\prime}}(t^{\prime}%
,X_{i},X_{i^{\prime},j^{\prime}};\underline{P})|X_{i^{\prime},j^{\prime}}]\\
&  +\sum_{i=1}^{m}\frac{\theta n}{s_{i}}\sum_{i^{\prime}=1}^{m}\frac
{1}{n_{i^{\prime}}}\sum_{j^{\prime}=1}^{n_{i^{\prime}}}\mathbb{E[}%
IF_{i}(t^{\prime},X_{i};\underline{P})IF_{ii^{\prime}}(t,X_{i},X_{i^{\prime
},j^{\prime}};\underline{P})|X_{i^{\prime},j^{\prime}}]+o_{p}\left(  \frac
{1}{\sqrt{n}}\right) \\
&  =\sum_{i=1}^{m}\frac{\theta n}{s_{i}}\mathrm{Cov}(IF_{i}(t,X_{i}%
;\underline{P}),IF_{i}(t^{\prime},X_{i};\underline{P}))\\
&  +\sum_{i=1}^{m}\frac{\theta n}{s_{i}}\left(  \frac{1}{n_{i}}\sum
_{j=1}^{n_{i}}IF_{i}(t,X_{i,j};\underline{P})IF_{i}(t^{\prime},X_{i,j}%
;\underline{P})-\mathrm{Cov}(IF_{i}(t,X_{i};\underline{P}),IF_{i}(t^{\prime
},X_{i};\underline{P}))\right) \\
&  +\sum_{i=1}^{m}\frac{\theta n}{s_{i}}\sum_{i^{\prime}=1}^{m}\frac
{1}{n_{i^{\prime}}}\sum_{j^{\prime}=1}^{n_{i^{\prime}}}\mathbb{E}_{X_{i}\sim
P_{i}}\mathbb{[}IF_{i}(t,X_{i};\underline{P})IF_{ii^{\prime}}(t^{\prime}%
,X_{i},X_{i^{\prime},j^{\prime}};\underline{P})|X_{i^{\prime},j^{\prime}}]\\
&  +\sum_{i=1}^{m}\frac{\theta n}{s_{i}}\sum_{i^{\prime}=1}^{m}\frac
{1}{n_{i^{\prime}}}\sum_{j^{\prime}=1}^{n_{i^{\prime}}}\mathbb{E[}%
IF_{i}(t^{\prime},X_{i};\underline{P})IF_{ii^{\prime}}(t,X_{i},X_{i^{\prime
},j^{\prime}};\underline{P})|X_{i^{\prime},j^{\prime}}]+o_{p}\left(  \frac
{1}{\sqrt{n}}\right) \\
&  =\mathrm{Cov}(\mathbb{G}(t),\mathbb{G}(t^{\prime}))+\sum_{i=1}^{m}\left(
\frac{n\mathrm{frac}(\theta n_{i})}{s_{i}n_{i}}+\left(  \frac{n}{n_{i}}%
-\frac{1}{\beta_{i}}\right)  \right)  \mathrm{Cov}(IF_{i}(t,X_{i}%
;\underline{P}),IF_{i}(t^{\prime},X_{i};\underline{P}))\\
&  +\sum_{i=1}^{m}\frac{1}{\beta_{i}}\left(  \frac{1}{n_{i}}\sum_{j=1}^{n_{i}%
}IF_{i}(t,X_{i,j};\underline{P})IF_{i}(t^{\prime},X_{i,j};\underline{P}%
)-\mathrm{Cov}(IF_{i}(t,X_{i};\underline{P}),IF_{i}(t^{\prime},X_{i}%
;\underline{P}))\right) \\
&  +\sum_{i=1}^{m}\frac{1}{\beta_{i}}\sum_{i^{\prime}=1}^{m}\frac
{1}{n_{i^{\prime}}}\sum_{j^{\prime}=1}^{n_{i^{\prime}}}\mathbb{E[}%
IF_{i}(t,X_{i};\underline{P})IF_{ii^{\prime}}(t^{\prime},X_{i},X_{i^{\prime
},j^{\prime}};\underline{P})|X_{i^{\prime},j^{\prime}}]\\
&  +\sum_{i=1}^{m}\frac{1}{\beta_{i}}\sum_{i^{\prime}=1}^{m}\frac
{1}{n_{i^{\prime}}}\sum_{j^{\prime}=1}^{n_{i^{\prime}}}\mathbb{E[}%
IF_{i}(t^{\prime},X_{i};\underline{P})IF_{ii^{\prime}}(t,X_{i},X_{i^{\prime
},j^{\prime}};\underline{P})|X_{i^{\prime},j^{\prime}}]+o_{p}\left(  \frac
{1}{\sqrt{n}}\right)  ,
\end{align*}
where the last equality follows from the following facts%
\[
\frac{\theta n}{s_{i}}=\frac{\theta n}{[\theta n_{i}]}=\frac{1}{\beta_{i}%
}+\frac{n\mathrm{frac}(\theta n_{i})}{s_{i}n_{i}}+\left(  \frac{n}{n_{i}%
}-\frac{1}{\beta_{i}}\right)  ,
\]%
\[
\mathrm{Cov}(\mathbb{G}(t),\mathbb{G}(t^{\prime}))=\sum_{i=1}^{m}\frac
{1}{\beta_{i}}\mathrm{Cov}(IF_{i}(t,X_{i};\underline{P}),IF_{i}(t^{\prime
},X_{i};\underline{P}))
\]%
\[
\frac{\theta n}{s_{i}}\rightarrow\frac{1}{\beta_{i}},
\]
\[
\frac{1}{n_{i}}\sum_{j=1}^{n_{i}}IF_{i}(t,X_{i,j};\underline{P})IF_{i}%
(t^{\prime},X_{i,j};\underline{P})-\mathrm{Cov}(IF_{i}(t,X_{i};\underline{P}%
),IF_{i}(t^{\prime},X_{i};\underline{P}))=O_{p}\left(  \frac{1}{\sqrt{n}%
}\right)  ,
\]%
\[
\sum_{i^{\prime}=1}^{m}\frac{1}{n_{i^{\prime}}}\sum_{j^{\prime}=1}%
^{n_{i^{\prime}}}\mathbb{E[}IF_{i}(t,X_{i};\underline{P})IF_{ii^{\prime}%
}(t^{\prime},X_{i},X_{i^{\prime},j^{\prime}};\underline{P})|X_{i^{\prime
},j^{\prime}}]=O_{p}\left(  \frac{1}{\sqrt{n}}\right)  ,
\]%
\[
\sum_{i^{\prime}=1}^{m}\frac{1}{n_{i^{\prime}}}\sum_{j^{\prime}=1}%
^{n_{i^{\prime}}}\mathbb{E[}IF_{i}(t^{\prime},X_{i};\underline{P}%
)IF_{ii^{\prime}}(t,X_{i},X_{i^{\prime},j^{\prime}};\underline{P}%
)|X_{i^{\prime},j^{\prime}}]=O_{p}\left(  \frac{1}{\sqrt{n}}\right)  .
\]
Then in the last two terms, we replace the index name $i,i^{\prime},j^{\prime
}$ by $i^{\prime},i,j$ and rearrange the summation, which leads to
(\ref{computation_sigma_term10}).

Now we prove (\ref{computation_sigma_term20}). We first simplify the term
$\mathbb{E}_{\ast}[\varepsilon^{\ast}(t^{\prime})|X_{i,1}^{\ast}%
]-\mathbb{E}_{\ast}[\varepsilon^{\ast}(t^{\prime})]$. By comparing the
expansions in Assumptions \ref{1st_expansion_empirical} and
\ref{3rd_expansion_empirical}, we can see%
\begin{align*}
&  \varepsilon^{\ast}(t^{\prime})\\
&  =\frac{1}{2}\sum_{i_{1},i_{2}=1}^{m}\int IF_{i_{1}i_{2}}(t^{\prime}%
,x_{1},x_{2};\underline{\hat{P}})d(\hat{P}_{i_{1},s_{i_{1}}}^{\ast}-\hat
{P}_{i_{1}})(x_{1})d(\hat{P}_{i_{2},s_{i_{2}}}^{\ast}-\hat{P}_{i_{2}}%
)(x_{2})\\
&  +\frac{1}{6}\sum_{i_{1},i_{2},i_{3}=1}^{m}\int IF_{i_{1}i_{2}i_{3}%
}(t^{\prime},x_{1},x_{2},x_{3};\underline{\hat{P}})d(\hat{P}_{i_{1},s_{i_{1}}%
}^{\ast}-\hat{P}_{i_{1}})(x_{1})d(\hat{P}_{i_{2},s_{i_{2}}}^{\ast}-\hat
{P}_{i_{2}})(x_{2})d(\hat{P}_{i_{3},s_{i_{3}}}^{\ast}-\hat{P}_{i_{3}}%
)(x_{3})+\varepsilon_{3}^{\ast}(t^{\prime})\\
&  =\frac{1}{2}\sum_{i_{1},i_{2}=1}^{m}\frac{1}{s_{i_{1}}s_{i_{2}}}\sum
_{j_{1}=1}^{s_{i_{1}}}\sum_{j_{2}=1}^{s_{i_{2}}}IF_{i_{1}i_{2}}(t^{\prime
},X_{i_{1},j_{1}}^{\ast},X_{i_{2},j_{2}}^{\ast};\underline{\hat{P}})\\
&  +\frac{1}{6}\sum_{i_{1},i_{2},i_{3}=1}^{m}\frac{1}{s_{i_{1}}s_{i_{2}%
}s_{i_{3}}}\sum_{j_{1}=1}^{s_{i_{1}}}\sum_{j_{2}=1}^{s_{i_{2}}}\sum_{j_{3}%
=1}^{s_{i_{3}}}IF_{i_{1}i_{2}i_{3}}(t^{\prime},X_{i_{1},j_{1}}^{\ast}%
,X_{i_{2},j_{2}}^{\ast},X_{i_{3},j_{3}}^{\ast};\underline{\hat{P}%
})+\varepsilon_{3}^{\ast}(t^{\prime}).
\end{align*}
Therefore, we can see%
\begin{align*}
&  \mathbb{E}_{\ast}[\varepsilon^{\ast}(t^{\prime})|X_{i,1}^{\ast}%
]-\mathbb{E}_{\ast}[\varepsilon^{\ast}(t^{\prime})]\\
&  =\frac{1}{2s_{i}^{2}}(IF_{ii}(t^{\prime},X_{i,1}^{\ast},X_{i,1}^{\ast
};\underline{\hat{P}})-\mathbb{E}_{\ast}[IF_{ii}(t^{\prime},X_{i,1}^{\ast
},X_{i,1}^{\ast};\underline{\hat{P}})])\\
&  +\frac{1}{6s_{i}^{3}}(IF_{iii}(t^{\prime},X_{i,1}^{\ast},X_{i,1}^{\ast
},X_{i,1}^{\ast};\underline{\hat{P}})-\mathbb{E}_{\ast}[IF_{iii}(t^{\prime
},X_{i,1}^{\ast},X_{i,1}^{\ast},X_{i,1}^{\ast};\underline{\hat{P}})])\\
&  +\sum_{i^{\prime}\neq i}\frac{1}{2s_{i}s_{i^{\prime}}}\mathbb{E}_{\ast
}[IF_{ii^{\prime}i^{\prime}}(t^{\prime},X_{i,1}^{\ast},X_{i^{\prime},2}^{\ast
},X_{i^{\prime},2}^{\ast};\underline{\hat{P}})|X_{i,1}^{\ast}]+\frac{s_{i}%
-1}{2s_{i}^{3}}\mathbb{E}_{\ast}[IF_{iii}(t^{\prime},X_{i,1}^{\ast}%
,X_{i,2}^{\ast},X_{i,2}^{\ast};\underline{\hat{P}})|X_{i,1}^{\ast}]\\
&  +\mathbb{E}_{\ast}[\varepsilon_{3}^{\ast}(t^{\prime})|X_{i,1}^{\ast
}]-\mathbb{E}_{\ast}[\varepsilon_{3}^{\ast}(t^{\prime})],\\
&  =\frac{1}{2s_{i}^{2}}(IF_{ii}(t^{\prime},X_{i,1}^{\ast},X_{i,1}^{\ast
};\underline{\hat{P}})-\mathbb{E}_{\ast}[IF_{ii}(t^{\prime},X_{i,1}^{\ast
},X_{i,1}^{\ast};\underline{\hat{P}})])+\sum_{i^{\prime}=1}^{m}\frac{1}%
{2s_{i}s_{i^{\prime}}}\mathbb{E}_{\ast}[IF_{ii^{\prime}i^{\prime}}(t^{\prime
},X_{i,1}^{\ast},X_{i^{\prime},2}^{\ast},X_{i^{\prime},2}^{\ast}%
;\underline{\hat{P}})|X_{i,1}^{\ast}]\\
&  +\frac{1}{6s_{i}^{3}}(IF_{iii}(t^{\prime},X_{i,1}^{\ast},X_{i,1}^{\ast
},X_{i,1}^{\ast};\underline{\hat{P}})-\mathbb{E}_{\ast}[IF_{iii}(t^{\prime
},X_{i,1}^{\ast},X_{i,1}^{\ast},X_{i,1}^{\ast};\underline{\hat{P}})])-\frac
{1}{2s_{i}^{3}}\mathbb{E}_{\ast}[IF_{iii}(t^{\prime},X_{i,1}^{\ast}%
,X_{i,2}^{\ast},X_{i,2}^{\ast};\underline{\hat{P}})|X_{i,1}^{\ast}]\\
&  +\mathbb{E}_{\ast}[\varepsilon_{3}^{\ast}(t^{\prime})|X_{i,1}^{\ast
}]-\mathbb{E}_{\ast}[\varepsilon_{3}^{\ast}(t^{\prime})].
\end{align*}
Plugging it into the LHS of (\ref{computation_sigma_term20}), we obtain%
\begin{align}
&  \theta n\sum_{i=1}^{m}\mathbb{E}_{\ast}[IF_{i}(t,X_{i,1}^{\ast
};\underline{\hat{P}})(\mathbb{E}_{\ast}[\varepsilon^{\ast}(t^{\prime
})|X_{i,1}^{\ast}]-\mathbb{E}_{\ast}[\varepsilon^{\ast}(t^{\prime
})])]\nonumber\\
&  =\theta n\sum_{i=1}^{m}\frac{1}{2s_{i}^{2}}\mathbb{E}_{\ast}[IF_{i}%
(t,X_{i,1}^{\ast};\underline{\hat{P}})(IF_{ii}(t^{\prime},X_{i,1}^{\ast
},X_{i,1}^{\ast};\underline{\hat{P}})-\mathbb{E}_{\ast}[IF_{ii}(t^{\prime
},X_{i,1}^{\ast},X_{i,1}^{\ast};\underline{\hat{P}})])]\nonumber\\
&  +\theta n\sum_{i=1}^{m}\sum_{i^{\prime}=1}^{m}\frac{1}{2s_{i}s_{i^{\prime}%
}}\mathbb{E}_{\ast}[IF_{i}(t,X_{i,1}^{\ast};\underline{\hat{P}})\mathbb{E}%
_{\ast}[IF_{ii^{\prime}i^{\prime}}(t^{\prime},X_{i,1}^{\ast},X_{i^{\prime}%
,2}^{\ast},X_{i^{\prime},2}^{\ast};\underline{\hat{P}})|X_{i,1}^{\ast
}]]\nonumber\\
&  +\theta n\sum_{i=1}^{m}\frac{1}{6s_{i}^{3}}\mathbb{E}_{\ast}[IF_{i}%
(t,X_{i,1}^{\ast};\underline{\hat{P}})(IF_{iii}(t^{\prime},X_{i,1}^{\ast
},X_{i,1}^{\ast},X_{i,1}^{\ast};\underline{\hat{P}})-\mathbb{E}_{\ast
}[IF_{iii}(t^{\prime},X_{i,1}^{\ast},X_{i,1}^{\ast},X_{i,1}^{\ast
};\underline{\hat{P}})])]\nonumber\\
&  -\theta n\sum_{i=1}^{m}\frac{1}{2s_{i}^{3}}\mathbb{E}_{\ast}[IF_{i}%
(t,X_{i,1}^{\ast};\underline{\hat{P}})\mathbb{E}_{\ast}[IF_{iii}(t^{\prime
},X_{i,1}^{\ast},X_{i,2}^{\ast},X_{i,2}^{\ast};\underline{\hat{P}}%
)|X_{i,1}^{\ast}]]\nonumber\\
&  +\theta n\sum_{i=1}^{m}\mathbb{E}_{\ast}[IF_{i}(t,X_{i,1}^{\ast
};\underline{\hat{P}})(\mathbb{E}_{\ast}[\varepsilon_{3}^{\ast}(t^{\prime
})|X_{i,1}^{\ast}]-\mathbb{E}_{\ast}[\varepsilon_{3}^{\ast}(t^{\prime
})])]\nonumber\\
&  =\theta n\sum_{i=1}^{m}\frac{1}{2s_{i}^{2}}\mathrm{Cov}_{\ast}%
(IF_{i}(t,X_{i,1}^{\ast};\underline{\hat{P}}),IF_{ii}(t^{\prime},X_{i,1}%
^{\ast},X_{i,1}^{\ast};\underline{\hat{P}}))\nonumber\\
&  +\theta n\sum_{i=1}^{m}\sum_{i^{\prime}=1}^{m}\frac{1}{2s_{i}s_{i^{\prime}%
}}\mathrm{Cov}_{\ast}(IF_{i}(t,X_{i,1}^{\ast};\underline{\hat{P}%
}),IF_{ii^{\prime}i^{\prime}}(t^{\prime},X_{i,1}^{\ast},X_{i^{\prime},2}%
^{\ast},X_{i^{\prime},2}^{\ast};\underline{\hat{P}}))\nonumber\\
&  +\theta n\sum_{i=1}^{m}\frac{1}{6s_{i}^{3}}\mathrm{Cov}_{\ast}%
(IF_{i}(t,X_{i,1}^{\ast};\underline{\hat{P}}),IF_{iii}(t^{\prime}%
,X_{i,1}^{\ast},X_{i,1}^{\ast},X_{i,1}^{\ast};\underline{\hat{P}}))\nonumber\\
&  -\theta n\sum_{i=1}^{m}\frac{1}{2s_{i}^{3}}\mathrm{Cov}_{\ast}%
(IF_{i}(t,X_{i,1}^{\ast};\underline{\hat{P}}),IF_{iii}(t^{\prime}%
,X_{i,1}^{\ast},X_{i,2}^{\ast},X_{i,2}^{\ast};\underline{\hat{P}}))\nonumber\\
&  +\theta n\sum_{i=1}^{m}\mathrm{Cov}_{\ast}(IF_{i}(t,X_{i,1}^{\ast
};\underline{\hat{P}}),\varepsilon_{3}^{\ast}(t^{\prime})).
\label{computation_sigma_term21}%
\end{align}
By a similar proof for Lemma \ref{influence_product} as well as the law of
large numbers for U/V-statistics (see \cite{korolyuk2013theory} Section 3.1),
we can get%
\[
\mathrm{Cov}_{\ast}(IF_{i}(t,X_{i,1}^{\ast};\underline{\hat{P}}),IF_{ii}%
(t^{\prime},X_{i,1}^{\ast},X_{i,1}^{\ast};\underline{\hat{P}}%
))\overset{p}{\rightarrow}\mathrm{Cov}(IF_{i}(t,X_{i};\underline{P}%
),IF_{ii}(t^{\prime},X_{i},X_{i};\underline{P})),
\]%
\[
\mathrm{Cov}_{\ast}(IF_{i}(t,X_{i,1}^{\ast};\underline{\hat{P}}%
),IF_{ii^{\prime}i^{\prime}}(t^{\prime},X_{i,1}^{\ast},X_{i^{\prime},2}^{\ast
},X_{i^{\prime},2}^{\ast};\underline{\hat{P}}))\overset{p}{\rightarrow
}\mathrm{Cov}(IF_{i}(t,X_{i,1};\underline{P}),IF_{ii^{\prime}i^{\prime}%
}(t^{\prime},X_{i,1},X_{i^{\prime},2},X_{i^{\prime},2};\underline{P})),
\]%
\[
\mathrm{Cov}_{\ast}(IF_{i}(t,X_{i,1}^{\ast};\underline{\hat{P}}),IF_{iii}%
(t^{\prime},X_{i,1}^{\ast},X_{i,1}^{\ast},X_{i,1}^{\ast};\underline{\hat{P}%
}))\overset{p}{\rightarrow}\mathrm{Cov}(IF_{i}(t,X_{i};\underline{P}%
),IF_{iii}(t^{\prime},X_{i},X_{i},X_{i};\underline{P})),
\]%
\[
\mathrm{Cov}_{\ast}(IF_{i}(t,X_{i,1}^{\ast};\underline{\hat{P}}),IF_{iii}%
(t^{\prime},X_{i,1}^{\ast},X_{i,2}^{\ast},X_{i,2}^{\ast};\underline{\hat{P}%
}))\overset{p}{\rightarrow}\mathrm{Cov}(IF_{i}(t,X_{i,1};\underline{P}%
),IF_{iii}(t^{\prime},X_{i,1},X_{i,2},X_{i,2};\underline{P})).
\]
Next, to bound $\mathrm{Cov}_{\ast}(IF_{i}(t,X_{i,1}^{\ast};\underline{\hat
{P}}),\varepsilon_{3}^{\ast}(t^{\prime}))$, we notice that $\varepsilon
_{3}^{\ast}(t^{\prime})$ is symmetric about $X_{i,1}^{\ast},\ldots,X_{i,s_{i}%
}^{\ast}$. Thus,
\begin{align*}
&  |\mathrm{Cov}_{\ast}(IF_{i}(t,X_{i,1}^{\ast};\underline{\hat{P}%
}),\varepsilon_{3}^{\ast}(t^{\prime}))|=\left\vert \frac{1}{s_{i}}%
\mathrm{Cov}_{\ast}\left(  \sum_{j=1}^{s_{i}}IF_{i}(t,X_{i,j}^{\ast
};\underline{\hat{P}}),\varepsilon_{3}^{\ast}(t^{\prime})\right)  \right\vert
\\
&  \leq\frac{1}{s_{i}}\sqrt{\mathrm{Var}_{\ast}\left(  \sum_{j=1}^{s_{i}%
}IF_{i}(t,X_{i,j}^{\ast};\underline{\hat{P}})\right)  \mathrm{Var}_{\ast
}(\varepsilon_{3}^{\ast}(t^{\prime}))}=\frac{1}{s_{i}}\sqrt{s_{i}%
\mathrm{Var}_{\ast}(IF_{i}(t,X_{i,1}^{\ast};\underline{\hat{P}}))\mathrm{Var}%
_{\ast}(\varepsilon_{3}^{\ast}(t^{\prime}))}=o_{p}\left(  \frac{1}{s^{2}%
}\right)  ,
\end{align*}
where the last equality uses $\mathbb{E}_{\ast}[\varepsilon_{3}^{\ast
}(t^{\prime})^{2}]=o_{p}\left(  s^{-3}\right)  $. Plugging this bound and the
above convergence results into (\ref{computation_sigma_term21}), we can get%
\begin{align*}
&  \theta n\sum_{i=1}^{m}\mathbb{E}_{\ast}[IF_{i}(t,X_{i,1}^{\ast
};\underline{\hat{P}})(\mathbb{E}_{\ast}[\varepsilon^{\ast}(t^{\prime
})|X_{i,1}^{\ast}]-\mathbb{E}_{\ast}[\varepsilon^{\ast}(t^{\prime})])]\\
&  =\theta n\sum_{i=1}^{m}\frac{1}{2s_{i}^{2}}\mathrm{Cov}(IF_{i}%
(t,X_{i};\underline{P}),IF_{ii}(t^{\prime},X_{i},X_{i};\underline{P}))\\
&  +\theta n\sum_{i=1}^{m}\sum_{i^{\prime}=1}^{m}\frac{1}{2s_{i}s_{i^{\prime}%
}}\mathrm{Cov}(IF_{i}(t,X_{i,1};\underline{P}),IF_{ii^{\prime}i^{\prime}%
}(t^{\prime},X_{i,1},X_{i^{\prime},2},X_{i^{\prime},2};\underline{P}%
))+o_{p}\left(  \frac{1}{s}\right)  ,
\end{align*}
which proves (\ref{computation_sigma_term20}).

(\ref{computation_sigma_term30}) can be proved by the same argument for
(\ref{computation_sigma_term20}).

Finally, we prove (\ref{computation_sigma_term40}). Recall that $\varepsilon
^{\ast}(t)$ can be expressed as
\begin{align*}
&  \varepsilon^{\ast}(t^{\prime})=\frac{1}{2}\sum_{i_{1},i_{2}=1}^{m}\frac
{1}{s_{i_{1}}s_{i_{2}}}\sum_{j_{1}=1}^{s_{i_{1}}}\sum_{j_{2}=1}^{s_{i_{2}}%
}IF_{i_{1}i_{2}}(t^{\prime},X_{i_{1},j_{1}}^{\ast},X_{i_{2},j_{2}}^{\ast
};\underline{\hat{P}})\\
&  +\frac{1}{6}\sum_{i_{1},i_{2},i_{3}=1}^{m}\frac{1}{s_{i_{1}}s_{i_{2}%
}s_{i_{3}}}\sum_{j_{1}=1}^{s_{i_{1}}}\sum_{j_{2}=1}^{s_{i_{2}}}\sum_{j_{3}%
=1}^{s_{i_{3}}}IF_{i_{1}i_{2}i_{3}}(t^{\prime},X_{i_{1},j_{1}}^{\ast}%
,X_{i_{2},j_{2}}^{\ast},X_{i_{3},j_{3}}^{\ast};\underline{\hat{P}%
})+\varepsilon_{3}^{\ast}(t^{\prime}).
\end{align*}
By the variance formula for the V-statistic, we can get the following order%
\[
\mathrm{Var}_{\ast}\left(  \frac{1}{s_{i_{1}}s_{i_{2}}}\sum_{j_{1}%
=1}^{s_{i_{1}}}\sum_{j_{2}=1}^{s_{i_{2}}}IF_{i_{1}i_{2}}(t^{\prime}%
,X_{i_{1},j_{1}}^{\ast},X_{i_{2},j_{2}}^{\ast};\underline{\hat{P}})\right)
=O_{p}\left(  \frac{1}{s^{2}}\right)  ,
\]%
\[
\mathrm{Var}_{\ast}\left(  \frac{1}{s_{i_{1}}s_{i_{2}}s_{i_{3}}}\sum_{j_{1}%
=1}^{s_{i_{1}}}\sum_{j_{2}=1}^{s_{i_{2}}}\sum_{j_{3}=1}^{s_{i_{3}}}%
IF_{i_{1}i_{2}i_{3}}(t^{\prime},X_{i_{1},j_{1}}^{\ast},X_{i_{2},j_{2}}^{\ast
},X_{i_{3},j_{3}}^{\ast};\underline{\hat{P}})\right)  =O_{p}\left(  \frac
{1}{s^{3}}\right)  .
\]
Assumption \ref{3rd_expansion_empirical} ensures that $\varepsilon_{3}^{\ast
}(t^{\prime})=o_{p}\left(  1/s^{3/2}\right)  $. Therefore, we obtain%
\begin{align}
&  \theta n\mathbb{E}_{\ast}[(\varepsilon^{\ast}(t)-\mathbb{E}_{\ast
}[\varepsilon^{\ast}(t)])(\varepsilon^{\ast}(t^{\prime})-\mathbb{E}_{\ast
}[\varepsilon^{\ast}(t^{\prime})])]\nonumber\\
&  =\frac{\theta n}{4}\mathrm{Cov}_{\ast}\left(  \sum_{i_{1},i_{2}=1}^{m}%
\frac{1}{s_{i_{1}}s_{i_{2}}}\sum_{j_{1}=1}^{s_{i_{1}}}\sum_{j_{2}=1}%
^{s_{i_{2}}}IF_{i_{1}i_{2}}(t,X_{i_{1},j_{1}}^{\ast},X_{i_{2},j_{2}}^{\ast
};\underline{\hat{P}}),\sum_{i_{1},i_{2}=1}^{m}\frac{1}{s_{i_{1}}s_{i_{2}}%
}\sum_{j_{1}=1}^{s_{i_{1}}}\sum_{j_{2}=1}^{s_{i_{2}}}IF_{i_{1}i_{2}}%
(t^{\prime},X_{i_{1},j_{1}}^{\ast},X_{i_{2},j_{2}}^{\ast};\underline{\hat{P}%
})\right) \nonumber\\
&  +o_{p}\left(  \frac{1}{s}\right) \nonumber\\
&  =\frac{\theta n}{4}\sum_{i_{1},i_{2},i_{3},i_{4}=1}^{m}\frac{1}{s_{i_{1}%
}s_{i_{2}}s_{i_{3}}s_{i_{4}}}\mathrm{Cov}_{\ast}\left(  \sum_{j_{1}%
=1}^{s_{i_{1}}}\sum_{j_{2}=1}^{s_{i_{2}}}IF_{i_{1}i_{2}}(t,X_{i_{1},j_{1}%
}^{\ast},X_{i_{2},j_{2}}^{\ast};\underline{\hat{P}}),\sum_{j_{3}=1}^{s_{i_{3}%
}}\sum_{j_{4}=1}^{s_{i_{4}}}IF_{i_{3}i_{4}}(t^{\prime},X_{i_{3},j_{3}}^{\ast
},X_{i_{4},j_{4}}^{\ast};\underline{\hat{P}})\right) \nonumber\\
&  +o_{p}\left(  \frac{1}{s}\right) \nonumber\\
&  =\frac{\theta n}{4}\sum_{1\leq i_{1}\neq i_{2}\leq m}\frac{2}{s_{i_{1}}%
^{2}s_{i_{2}}^{2}}\mathrm{Cov}_{\ast}\left(  \sum_{j_{1}=1}^{s_{i_{1}}}%
\sum_{j_{2}=1}^{s_{i_{2}}}IF_{i_{1}i_{2}}(t,X_{i_{1},j_{1}}^{\ast}%
,X_{i_{2},j_{2}}^{\ast};\underline{\hat{P}}),\sum_{j_{3}=1}^{s_{i_{1}}}%
\sum_{j_{4}=1}^{s_{i_{2}}}IF_{i_{1}i_{2}}(t^{\prime},X_{i_{1},j_{3}}^{\ast
},X_{i_{2},j_{4}}^{\ast};\underline{\hat{P}})\right) \nonumber\\
&  +\frac{\theta n}{4}\sum_{i=1}^{m}\frac{1}{s_{i}^{4}}\mathrm{Cov}_{\ast
}\left(  \sum_{j_{1}=1}^{s_{i}}\sum_{j_{2}=1}^{s_{i}}IF_{ii}(t,X_{i,j_{1}%
}^{\ast},X_{i,j_{2}}^{\ast};\underline{\hat{P}}),\sum_{j_{3}=1}^{s_{i}}%
\sum_{j_{4}=1}^{s_{i}}IF_{ii}(t^{\prime},X_{i,j_{3}}^{\ast},X_{i,j_{4}}^{\ast
};\underline{\hat{P}})\right)  +o_{p}\left(  \frac{1}{s}\right) \nonumber\\
&  =\frac{\theta n}{4}\sum_{1\leq i_{1}\neq i_{2}\leq m}\frac{2}{s_{i_{1}}%
^{2}s_{i_{2}}^{2}}\sum_{j_{1}=1}^{s_{i_{1}}}\sum_{j_{2}=1}^{s_{i_{2}}%
}\mathrm{Cov}_{\ast}(IF_{i_{1}i_{2}}(t,X_{i_{1},j_{1}}^{\ast},X_{i_{2},j_{2}%
}^{\ast};\underline{\hat{P}}),IF_{i_{1}i_{2}}(t^{\prime},X_{i_{1},j_{1}}%
^{\ast},X_{i_{2},j_{2}}^{\ast};\underline{\hat{P}}))\nonumber\\
&  +\frac{\theta n}{4}\sum_{i=1}^{m}\frac{2}{s_{i}^{4}}\sum_{1\leq j_{1}\neq
j_{2}\leq s_{i}}\mathrm{Cov}_{\ast}(IF_{ii}(t,X_{i,j_{1}}^{\ast},X_{i,j_{2}%
}^{\ast};\underline{\hat{P}}),IF_{ii}(t^{\prime},X_{i,j_{1}}^{\ast}%
,X_{i,j_{2}}^{\ast};\underline{\hat{P}}))\nonumber\\
&  +\frac{\theta n}{4}\sum_{i=1}^{m}\frac{1}{s_{i}^{4}}\sum_{j=1}^{s_{i}%
}\mathrm{Cov}_{\ast}(IF_{ii}(t,X_{i,j}^{\ast},X_{i,j}^{\ast};\underline{\hat
{P}}),IF_{ii}(t^{\prime},X_{i,j}^{\ast},X_{i,j}^{\ast};\underline{\hat{P}%
}))+o_{p}\left(  \frac{1}{s}\right) \nonumber\\
&  =\frac{\theta n}{4}\sum_{1\leq i_{1}\neq i_{2}\leq m}\frac{2}{s_{i_{1}%
}s_{i_{2}}}\mathrm{Cov}_{\ast}(IF_{i_{1}i_{2}}(t,X_{i_{1},1}^{\ast}%
,X_{i_{2},1}^{\ast};\underline{\hat{P}}),IF_{i_{1}i_{2}}(t^{\prime}%
,X_{i_{1},1}^{\ast},X_{i_{2},1}^{\ast};\underline{\hat{P}}))\nonumber\\
&  +\frac{\theta n}{4}\sum_{i=1}^{m}\frac{2(s_{i}-1)}{s_{i}^{3}}%
\mathrm{Cov}_{\ast}(IF_{ii}(t,X_{i,1}^{\ast},X_{i,2}^{\ast};\underline{\hat
{P}}),IF_{ii}(t^{\prime},X_{i,1}^{\ast},X_{i,2}^{\ast};\underline{\hat{P}%
}))\nonumber\\
&  +\frac{\theta n}{4}\sum_{i=1}^{m}\frac{1}{s_{i}^{3}}\mathrm{Cov}_{\ast
}(IF_{ii}(t,X_{i,1}^{\ast},X_{i,1}^{\ast};\underline{\hat{P}}),IF_{ii}%
(t^{\prime},X_{i,1}^{\ast},X_{i,1}^{\ast};\underline{\hat{P}}))+o_{p}\left(
\frac{1}{s}\right)  . \label{computation_sigma_term41}%
\end{align}
By a similar proof for Lemma \ref{influence_product} as well as the law of
large numbers for U/V-statistics (see \cite{korolyuk2013theory} Section 3.1),
we can get%
\[
\mathrm{Cov}_{\ast}(IF_{i_{1}i_{2}}(t,X_{i_{1},1}^{\ast},X_{i_{2},1}^{\ast
};\underline{\hat{P}}),IF_{i_{1}i_{2}}(t^{\prime},X_{i_{1},1}^{\ast}%
,X_{i_{2},1}^{\ast};\underline{\hat{P}}))\overset{p}{\rightarrow}%
\mathrm{Cov}(IF_{i_{1}i_{2}}(t,X_{i_{1},1},X_{i_{2},1};\underline{P}%
),IF_{i_{1}i_{2}}(t^{\prime},X_{i_{1},1},X_{i_{2},1};\underline{P})),
\]%
\[
\mathrm{Cov}_{\ast}(IF_{ii}(t,X_{i,1}^{\ast},X_{i,2}^{\ast};\underline{\hat
{P}}),IF_{ii}(t^{\prime},X_{i,1}^{\ast},X_{i,2}^{\ast};\underline{\hat{P}%
}))\overset{p}{\rightarrow}\mathrm{Cov}(IF_{ii}(t,X_{i,1},X_{i,2}%
;\underline{P}),IF_{ii}(t^{\prime},X_{i,1},X_{i,2};\underline{P})),
\]%
\[
\mathrm{Cov}_{\ast}(IF_{ii}(t,X_{i,1}^{\ast},X_{i,1}^{\ast};\underline{\hat
{P}}),IF_{ii}(t^{\prime},X_{i,1}^{\ast},X_{i,1}^{\ast};\underline{\hat{P}%
}))\overset{p}{\rightarrow}\mathrm{Cov}(IF_{ii}(t,X_{i,1},X_{i,1}%
;\underline{P}),IF_{ii}(t^{\prime},X_{i,1},X_{i,1};\underline{P})).
\]
Therefore, (\ref{computation_sigma_term41}) can be simplified as%
\begin{align*}
&  \theta n\mathbb{E}_{\ast}[(\varepsilon^{\ast}(t)-\mathbb{E}_{\ast
}[\varepsilon^{\ast}(t)])(\varepsilon^{\ast}(t^{\prime})-\mathbb{E}_{\ast
}[\varepsilon^{\ast}(t^{\prime})])]\\
&  =\frac{\theta n}{4}\sum_{1\leq i_{1}\neq i_{2}\leq m}\frac{2}{s_{i_{1}%
}s_{i_{2}}}\mathrm{Cov}(IF_{i_{1}i_{2}}(t,X_{i_{1},1},X_{i_{2},1}%
;\underline{P}),IF_{i_{1}i_{2}}(t^{\prime},X_{i_{1},1},X_{i_{2},1}%
;\underline{P}))\\
&  +\frac{\theta n}{4}\sum_{i=1}^{m}\frac{2}{s_{i}^{2}}\mathrm{Cov}%
(IF_{ii}(t,X_{i,1},X_{i,2};\underline{P}),IF_{ii}(t^{\prime},X_{i,1}%
,X_{i,2};\underline{P}))+o_{p}\left(  \frac{1}{s}\right) \\
&  =\theta n\sum_{i_{1},i_{2}=1}^{m}\frac{1}{2s_{i_{1}}s_{i_{2}}}%
\mathrm{Cov}(IF_{i_{1}i_{2}}(t,X_{i_{1},1},X_{i_{2},2};\underline{P}%
),IF_{i_{1}i_{2}}(t^{\prime},X_{i_{1},1},X_{i_{2},2};\underline{P}%
))+o_{p}\left(  \frac{1}{s}\right)  ,
\end{align*}
which proves (\ref{computation_sigma_term40}).

Now we have proved (\ref{computation_sigma_term10}%
)-(\ref{computation_sigma_term40}). Plugging them into
(\ref{computation_sigma2}), we can see%
\begin{align*}
&  \sigma^{2}(t,t^{\prime})-\mathrm{Cov}(\mathbb{G}(t),\mathbb{G}(t^{\prime
}))\\
&  =\sum_{i=1}^{m}\left(  \frac{n\mathrm{frac}(\theta n_{i})}{s_{i}n_{i}%
}+\left(  \frac{n}{n_{i}}-\frac{1}{\beta_{i}}\right)  \right)  \mathrm{Cov}%
(IF_{i}(t,X_{i};\underline{P}),IF_{i}(t^{\prime},X_{i};\underline{P}))\\
&  +\sum_{i=1}^{m}\frac{1}{\beta_{i}}\left(  \frac{1}{n_{i}}\sum_{j=1}^{n_{i}%
}IF_{i}(t,X_{i,j};\underline{P})IF_{i}(t^{\prime},X_{i,j};\underline{P}%
)-\mathrm{Cov}(IF_{i}(t,X_{i};\underline{P}),IF_{i}(t^{\prime},X_{i}%
;\underline{P}))\right) \\
&  +\sum_{i=1}^{m}\frac{1}{n_{i}}\sum_{j=1}^{n_{i}}\sum_{i^{\prime}=1}%
^{m}\frac{1}{\beta_{i^{\prime}}}\mathbb{E[}IF_{i^{\prime}}(t,X_{i^{\prime}%
};\underline{P})IF_{i^{\prime}i}(t^{\prime},X_{i^{\prime}},X_{i,j}%
;\underline{P})|X_{i,j}]\\
&  +\sum_{i=1}^{m}\frac{1}{n_{i}}\sum_{j=1}^{n_{i}}\sum_{i^{\prime}=1}%
^{m}\frac{1}{\beta_{i^{\prime}}}\mathbb{E[}IF_{i^{\prime}}(t^{\prime
},X_{i^{\prime}};\underline{P})IF_{i^{\prime}i}(t,X_{i^{\prime}}%
,X_{i,j};\underline{P})|X_{i,j}]\\
&  +\theta n\sum_{i=1}^{m}\frac{1}{2s_{i}^{2}}\mathrm{Cov}(IF_{i}%
(t,X_{i};\underline{P}),IF_{ii}(t^{\prime},X_{i},X_{i};\underline{P}))\\
&  +\theta n\sum_{i=1}^{m}\sum_{i^{\prime}=1}^{m}\frac{1}{2s_{i}s_{i^{\prime}%
}}\mathrm{Cov}(IF_{i}(t,X_{i,1};\underline{P}),IF_{ii^{\prime}i^{\prime}%
}(t^{\prime},X_{i,1},X_{i^{\prime},2},X_{i^{\prime},2};\underline{P}))\\
&  +\theta n\sum_{i=1}^{m}\frac{1}{2s_{i}^{2}}\mathrm{Cov}(IF_{i}(t^{\prime
},X_{i};\underline{P}),IF_{ii}(t,X_{i},X_{i};\underline{P}))\\
&  +\theta n\sum_{i=1}^{m}\sum_{i^{\prime}=1}^{m}\frac{1}{2s_{i}s_{i^{\prime}%
}}\mathrm{Cov}(IF_{i}(t^{\prime},X_{i,1};\underline{P}),IF_{ii^{\prime
}i^{\prime}}(t,X_{i,1},X_{i^{\prime},2},X_{i^{\prime},2};\underline{P}))\\
&  +\theta n\sum_{i_{1},i_{2}=1}^{m}\frac{1}{2s_{i_{1}}s_{i_{2}}}%
\mathrm{Cov}(IF_{i_{1}i_{2}}(t,X_{i_{1},1},X_{i_{2},2};\underline{P}%
),IF_{i_{1}i_{2}}(t^{\prime},X_{i_{1},1},X_{i_{2},2};\underline{P}))\\
&  +o_{p}\left(  \frac{1}{s}+\frac{1}{\sqrt{n}}\right)  .
\end{align*}
Let the leading order terms in $\sigma^{2}(t,t^{\prime})-\mathrm{Cov}%
(\mathbb{G}(t),\mathbb{G}(t^{\prime}))$ be $\mathcal{E}_{1}$. It's easy to
check that it has desired expectation and variance.
\end{proof}

\begin{proof}{Proof of Theorem \ref{overall_opt_config}.}
Based on the assumption, it suffices to minimize the order of the conditional MSE $\mathbb{E}_*[(\hat{\sigma}^{2}(t,t^{\prime})-\sigma^{2}(t,t^{\prime}))^2]$ and the order of $\sigma^{2}(t,t^{\prime})-\mathrm{Cov}(\mathbb{G}(t),\mathbb{G}(t^{\prime}))$. Since $R_s$ and $B$ only affect the order of the conditional MSE, by Theorem \ref{optimal_B_R}, the optimal $R_s$ and $B$ are given by $R_{s}^{\ast}=\Theta(N^{1/3}s^{2/3}),B^{\ast}=N/R_{s}^{\ast}$, leading to the following order of the conditional MSE
\[
\mathbb{E}_{\ast}[(\hat{\sigma}^{2}(t,t^{\prime})-\sigma^{2}(t,t^{\prime
}))^{2}]=\Theta_p((\theta n)^{2/3}/N^{2/3}).
\]
Moreover, by Theorem \ref{error_true_bootstrap}, when the mean and variance of $\mathcal{E}_{1}$ satisfy
\[
\mathbb{E}[\mathcal{E}_{1}]=\Theta\left(  \frac{1}{s}+\sum_{i=1}^{m}\left\vert\frac{n}{n_{i}}-\frac{1}{\beta_{i}}\right\vert \right)  ,\quad\mathrm{Var}(\mathcal{E}_{1})=\Theta\left(  \frac{1}{n}\right)  ,
\]
the order of $\sigma^{2}(t,t^{\prime})-\mathrm{Cov}(\mathbb{G}(t),\mathbb{G}(t^{\prime}))$ is
\[
\sigma^{2}(t,t^{\prime})-\mathrm{Cov}(\mathbb{G}(t),\mathbb{G}(t^{\prime}))=\Theta_{p}\left(  \frac{1}{s}+\sum_{i=1}^{m}\left\vert \frac{n}{n_{i}}-\frac{1}{\beta_{i}}\right\vert +\frac{1}{\sqrt{n}}\right).
\]
Therefore, the total estimation error $\hat{\sigma}^{2}(t,t^{\prime})-\mathrm{Cov}(\mathbb{G}(t),\mathbb{G}(t^{\prime}))$ has the following order
\[
\hat{\sigma}^{2}(t,t^{\prime})-\mathrm{Cov}(\mathbb{G}(t),\mathbb{G}%
(t^{\prime}))=\Theta_{p}\left(  \frac{1}{s}+\sum_{i=1}^{m}\left\vert \frac
{n}{n_{i}}-\frac{1}{\beta_{i}}\right\vert +\frac{1}{\sqrt{n}}+\frac{s^{1/3}%
}{N^{1/3}}\right)  .
\]
From the above order, we can see the optimal $s$ that minimizes the order of $\hat
{\sigma}^{2}(t,t^{\prime})-\mathrm{Cov}(\mathbb{G}(t),\mathbb{G}(t^{\prime}))$
is $s^{\ast}=\Theta(N^{1/4})$. The corresponding optimal subsample ratio is
$\theta^{\ast}=\Theta(s^{\ast}/n)=\Theta(N^{1/4}/n)$ and the minimal order of
$\hat{\sigma}^{2}(t,t^{\prime})-\mathrm{Cov}(\mathbb{G}(t),\mathbb{G}%
(t^{\prime}))$ is%
\[
\hat{\sigma}^{2}(t,t^{\prime})-\mathrm{Cov}(\mathbb{G}(t),\mathbb{G}%
(t^{\prime}))=\Theta_{p}\left(  \frac{1}{N^{1/4}}+\sum_{i=1}^{m}\left\vert
\frac{n}{n_{i}}-\frac{1}{\beta_{i}}\right\vert +\frac{1}{\sqrt{n}}\right).
\]

\end{proof}

\subsection{Proofs of Results in Appendix \ref{sec:general_assumptions}}\label{sec:proofs_general_assumptions}

\begin{proof}{Proof of Theorem \ref{verification_general_assumptions}.}

Proof of Assumption \ref{convergence_in_p_to_truth}. Under the finite-horizon model, the output distribution function $Q(t,\underline{\hat{P}})$ is just a multi-sample V-statistic (see Section 1.3 in \cite{korolyuk2013theory}), which is measurable. The consistency follows from the Law of Large Numbers for V-statistic (e.g., Theorem 3.3.1 in \cite{korolyuk2013theory} for single-sample V-statistic; similar arguments hold for multi-sample V-statistic).

Proof of Assumption \ref{1st_expansion_truth}. This assumption follows from Theorem 6 in \cite{lam2022subsampling} where their function $h$ corresponds to the indicator function $I(h(\mathbf{X}_{1},\ldots,\mathbf{X}_{m})\leq t)$ in our case. Moreover, their condition Assumption 8 is not needed here since we do not require $\mathrm{Var}_{P_{i}}(IF_{i}(t,X_{i};\underline{P}))>0$ as in their Assumption 3.

Proof of Assumption \ref{1st_expansion_empirical}. All results except 
\[
\mathbb{E}_{\ast}[\varepsilon^{\ast}(t)]=O_{p}((\theta n)^{-1}),\forall \in\mathbb{R},%
\]
and
\[
\mathbb{E}_{\ast}[(\varepsilon^{\ast}(t)-\mathbb{E}_{\ast}[\varepsilon^{\ast}(t)])^{2}]=o_{p}((\theta n)^{-1}),\forall t\in\mathbb{R},
\]%
follow from Theorem 6 in \cite{lam2022subsampling}. To prove the first remaining result, we use the third order expansion in Assumption \ref{3rd_expansion_empirical} (proved later) and get the following representation of $\varepsilon^{\ast}(t)$:
\begin{align*}
&\varepsilon^{\ast}(t)\\
& =\frac{1}{2}\sum_{i_{1},i_{2}=1}^{m}\int IF_{i_{1}i_{2}}(t,x_{1}%
,x_{2};\underline{\hat{P}})d(\hat{P}_{i_{1},s_{i_{1}}}^{\ast}-\hat{P}_{i_{1}%
})(x_{1})d(\hat{P}_{i_{2},s_{i_{2}}}^{\ast}-\hat{P}_{i_{2}})(x_{2})\\
&  +\frac{1}{6}\sum_{i_{1},i_{2},i_{3}=1}^{m}\int IF_{i_{1}i_{2}i_{3}}%
(t,x_{1},x_{2},x_{3};\underline{\hat{P}})d(\hat{P}_{i_{1},s_{i_{1}}}^{\ast
}-\hat{P}_{i_{1}})(x_{1})d(\hat{P}_{i_{2},s_{i_{2}}}^{\ast}-\hat{P}_{i_{2}%
})(x_{2})d(\hat{P}_{i_{3},s_{i_{3}}}^{\ast}-\hat{P}_{i_{3}})(x_{3}%
)\\
&+\varepsilon_{3}^{\ast}(t).
\end{align*}
Since all higher order influence functions have marginal zero mean, we have
\[
\mathbb{E}_*[\varepsilon^{\ast}(t)] = \mathbb{E}_*[\varepsilon_3^{\ast}(t)],
\]
which implies $\mathbb{E}_{\ast}[\varepsilon^{\ast}(t)]=O_{p}((\theta n)^{-1})$ (actually we can get a more precise order) by $\mathbb{E}_{\ast}[\varepsilon_{3}^{\ast}(t)^{2}]=o_{p}\left(  (\theta n)^{-3}\right)$ and H\"{o}lder's inequality.

The second remaining result can be proved by
\[
\mathbb{E}_{\ast}[(\varepsilon^{\ast}(t)-\mathbb{E}_{\ast}[\varepsilon^{\ast}(t)])^{2}]\leq \mathbb{E}_{\ast}[(\varepsilon^{\ast}(t))^{2}]\leq \sqrt{\mathbb{E}_{\ast}[\varepsilon^{\ast}(t)^{4}]}=o_{p}((\theta n)^{-1}).
\]

Proof of Assumption \ref{3rd_IF_at_truth}. This assumption follows from Theorem 9 in \cite{lam2022subsampling}.

Proof of Assumption \ref{3rd_expansion_empirical}. This assumption follows from Theorem 9 in \cite{lam2022subsampling}.
\end{proof}

\begin{proof}{Proofs of Theorems \ref{consistency_sigma2}-\ref{overall_opt_config2}.}
By tracking the proofs of the original results, we can find that the assumptions presented in Theorems \ref{consistency_sigma2}-\ref{overall_opt_config2} are the model properties that we actually need. This concludes our proof.
\end{proof}

\subsection{Proofs of Results in Appendix \ref{sec: var_of_cov}}\label{sec:proofs_var_formula}
\begin{proof}{Proof of Lemma \ref{general_covariance_undebias}.}
The key building block to derive the mean and variance of $\hat{\sigma}_{Cov}^{2}$ is the following decomposition of $X_{br}$ and $Y_{br}$ generalized from \cite{sun2011efficient}. We write
\[
X_{br}=\mathbb{E}[X]+\tau_{b}^{X}+\varepsilon_{br}^{X},\quad Y_{br}=\mathbb{E}[Y]+\tau_{b}^{Y}+\varepsilon_{br}^{Y},
\]
where the ``group effects'' are defined as
\[
\tau_{b}^{X}\coloneqq \mathbb{E}[X|Z_{b}]-\mathbb{E}[X],\quad\tau_{b}^{Y}\coloneqq\mathbb{E}[Y|Z_{b}]-\mathbb{E}[Y],
\]
and the ``residual errors'' are defined as
\[
\varepsilon_{br}^{X}\coloneqq X_{br}-\mathbb{E}[X|Z_{b}],\quad\varepsilon_{br}^{Y} \coloneqq Y_{br}-\mathbb{E}[Y|Z_{b}].
\]
Based on Assumption \ref{assumption_cov}, we make four key observations: (1) $(\tau_{b}^{X},\tau_{b}^{Y}),b=1,\ldots,B$ are mutually independent and have zero mean; (2) $(\varepsilon_{br}^{X},\varepsilon_{br}^{Y}),r=1,\ldots,R$ are conditionally mutually independent and have conditional zero mean given $Z_{b}$ (and thus they have zero mean); (3) The random vectors $(\tau_{b}^{X},\tau_{b}^{Y},\varepsilon_{b1}^{X},\varepsilon_{b1}^{Y},\ldots,\varepsilon_{bR}^{X},\varepsilon_{bR}^{Y}),b=1,\ldots,B$ are mutually independent; (4) Within the same $b$, the residual error and group effect are uncorrelated since%
\[
\mathbb{E}[\tau_{b}^{X}\varepsilon_{br}^{X}]=\mathbb{E}[\mathbb{E}[\tau_{b}^{X}\varepsilon_{br}^{X}|Z_{b}]]=\mathbb{E}[\tau_{b}^{X}\mathbb{E}[\varepsilon_{br}^{X}|Z_{b}]]=0=\mathbb{E}[\tau_{b}^{X}]\mathbb{E}[\varepsilon_{br}^{X}],
\]
and similarly%
\[
\mathbb{E}[\tau_{b}^{X}\varepsilon_{br}^{Y}]=\mathbb{E}[\tau_{b}^{Y}\varepsilon_{br}^{X}]=\mathbb{E}[\tau_{b}^{Y}\varepsilon_{br}^{Y}]=0.
\]
These four properties are essential for the following computations.

We define the average error within each $b$:%
\[
\bar{\varepsilon}_{b}^{X}=\frac{1}{R}\sum_{r=1}^{R}\varepsilon_{br}^{X}%
,\quad\bar{\varepsilon}_{b}^{Y}=\frac{1}{R}\sum_{r=1}^{R}\varepsilon_{br}%
^{Y},
\]
and the overall average effect and error%
\[
\bar{\tau}^{X}=\frac{1}{B}\sum_{b=1}^{B}\tau_{b}^{X},\quad\bar{\tau}^{Y}%
=\frac{1}{B}\sum_{b=1}^{B}\tau_{b}^{Y},\quad\bar{\bar{\varepsilon}}^{X}%
=\frac{1}{B}\sum_{b=1}^{B}\bar{\varepsilon}_{b}^{X},\quad\bar{\bar
{\varepsilon}}^{Y}=\frac{1}{B}\sum_{b=1}^{B}\bar{\varepsilon}_{b}^{Y}.
\]
Thus, $\bar{X}_{b},\bar{Y}_{b},\bar{\bar{X}},\bar{\bar{Y}}$ can be written as%
\begin{equation}
\bar{X}_{b}=\mathbb{E}[X]+\tau_{b}^{X}+\bar{\varepsilon}_{b}^{X},\quad\bar
{Y}_{b}=\mathbb{E}[Y]+\tau_{b}^{Y}+\bar{\varepsilon}_{b}^{Y}, \label{bar_b}%
\end{equation}%
\begin{equation}
\bar{\bar{X}}=\mathbb{E}[X]+\frac{1}{B}\sum_{b=1}^{B}\tau_{b}^{X}+\frac{1}%
{B}\sum_{b=1}^{B}\bar{\varepsilon}_{b}^{X}=\mathbb{E}[X]+\bar{\tau}^{X}%
+\bar{\bar{\varepsilon}}^{X}, \label{barbar_X}%
\end{equation}
and%
\begin{equation}
\bar{\bar{Y}}=\mathbb{E}[Y]+\frac{1}{B}\sum_{b=1}^{B}\tau_{b}^{Y}+\frac{1}%
{B}\sum_{b=1}^{B}\bar{\varepsilon}_{b}^{Y}=\mathbb{E}[Y]+\bar{\tau}^{Y}%
+\bar{\bar{\varepsilon}}^{Y}. \label{barbar_Y}%
\end{equation}

We first compute the mean of $\hat{\sigma}_{Cov}^{2}$. By (\ref{bar_b}),
(\ref{barbar_X}), (\ref{barbar_Y}) and the observations above, we have%
\begin{align*}
&  \mathbb{E}[\hat{\sigma}_{Cov}^{2}]\\
&  =\frac{1}{B-1}\sum_{b=1}^{B}\mathbb{E}\left[  \left(  \tau_{b}^{X}%
+\bar{\varepsilon}_{b}^{X}-\frac{1}{B}\sum_{b=1}^{B}\tau_{b}^{X}-\frac{1}%
{B}\sum_{b=1}^{B}\bar{\varepsilon}_{b}^{X}\right)  \left(  \tau_{b}^{Y}%
+\bar{\varepsilon}_{b}^{Y}-\frac{1}{B}\sum_{b=1}^{B}\tau_{b}^{Y}-\frac{1}%
{B}\sum_{b=1}^{B}\bar{\varepsilon}_{b}^{Y}\right)  \right]  \\
&  =\frac{1}{B-1}\sum_{b=1}^{B}\mathbb{E}\left[  \left(  \left(  1-\frac{1}%
{B}\right)  \tau_{b}^{X}-\frac{1}{B}\sum_{b^{\prime}\neq b}\tau_{b^{\prime}%
}^{X}+\left(  1-\frac{1}{B}\right)  \bar{\varepsilon}_{b}^{X}-\frac{1}{B}%
\sum_{b^{\prime}\neq b}\bar{\varepsilon}_{b^{\prime}}^{X}\right)  \right.  \\
&  \left.  \times\left(  \left(  1-\frac{1}{B}\right)  \tau_{b}^{Y}-\frac
{1}{B}\sum_{b^{\prime}\neq b}\tau_{b^{\prime}}^{Y}+\left(  1-\frac{1}%
{B}\right)  \bar{\varepsilon}_{b}^{Y}-\frac{1}{B}\sum_{b^{\prime}\neq b}%
\bar{\varepsilon}_{b^{\prime}}^{Y}\right)  \right]  \\
&  =\frac{1}{B-1}\sum_{b=1}^{B}\left(  \left(  1-\frac{1}{B}\right)
^{2}+\frac{B-1}{B^{2}}\right)  (\mathbb{E}[\tau_{1}^{X}\tau_{1}^{Y}%
]+\mathbb{E}[\bar{\varepsilon}_{1}^{X}\bar{\varepsilon}_{1}^{Y}])\\
&  =\frac{1}{B}\sum_{b=1}^{B}\left(  \mathbb{E}[\tau^{X}\tau^{Y}]+\frac{1}%
{R}\mathbb{E}[\varepsilon^{X}\varepsilon^{Y}]\right)  \\
&  =\mathbb{E}[\tau^{X}\tau^{Y}]+\frac{1}{R}\mathbb{E}[\varepsilon
^{X}\varepsilon^{Y}].
\end{align*}

Now we compute the variance of $\hat{\sigma}_{Cov}^{2}$. We first simplify the
expression of $\mathrm{Var}(\hat{\sigma}_{Cov}^{2})$
\begin{align}
&  \mathrm{Var}(\hat{\sigma}_{Cov}^{2})\nonumber\\
&  =\mathrm{Var}\left(  \frac{1}{B-1}\sum_{b=1}^{B}(\bar{X}_{b}-\bar{\bar{X}%
})(\bar{Y}_{b}-\bar{\bar{Y}})\right)  \nonumber\\
&  =\frac{1}{(B-1)^{2}}\sum_{b=1}^{B}\mathrm{Var}((\bar{X}_{b}-\bar{\bar{X}%
})(\bar{Y}_{b}-\bar{\bar{Y}}))\nonumber\\
&  +\frac{1}{(B-1)^{2}}\sum_{1\leq b_{1}\neq b_{2}\leq B}\mathrm{Cov}((\bar
{X}_{b_{1}}-\bar{\bar{X}})(\bar{Y}_{b_{1}}-\bar{\bar{Y}}),(\bar{X}_{b_{2}%
}-\bar{\bar{X}})(\bar{Y}_{b_{2}}-\bar{\bar{Y}}))\nonumber\\
&  =\frac{B}{(B-1)^{2}}\mathrm{Var}((\bar{X}_{1}-\bar{\bar{X}})(\bar{Y}%
_{1}-\bar{\bar{Y}}))+\frac{B}{B-1}\mathrm{Cov}((\bar{X}_{1}-\bar{\bar{X}%
})(\bar{Y}_{1}-\bar{\bar{Y}}),(\bar{X}_{2}-\bar{\bar{X}})(\bar{Y}_{2}%
-\bar{\bar{Y}}))\nonumber\\
&  =\frac{B}{(B-1)^{2}}\mathbb{E}[(\bar{X}_{1}-\bar{\bar{X}})^{2}(\bar{Y}%
_{1}-\bar{\bar{Y}})^{2}]+\frac{B}{B-1}\mathbb{E}[(\bar{X}_{1}-\bar{\bar{X}%
})(\bar{Y}_{1}-\bar{\bar{Y}})(\bar{X}_{2}-\bar{\bar{X}})(\bar{Y}_{2}-\bar
{\bar{Y}})]\nonumber\\
&  -\frac{B^{2}}{(B-1)^{2}}(\mathbb{E}[(\bar{X}_{1}-\bar{\bar{X}})(\bar{Y}%
_{1}-\bar{\bar{Y}})])^{2}.\label{computation1_var}%
\end{align}
The first expectation $\mathbb{E}[(\bar{X}_{1}-\bar{\bar{X}})^{2}(\bar{Y}%
_{1}-\bar{\bar{Y}})^{2}]$ can be simplified as
\begin{align*}
&  \mathbb{E}[(\bar{X}_{1}-\bar{\bar{X}})^{2}(\bar{Y}_{1}-\bar{\bar{Y}}%
)^{2}]\\
&  =\mathbb{E}[(\tau_{1}^{X}-\bar{\tau}^{X}+\bar{\varepsilon}_{1}^{X}%
-\bar{\bar{\varepsilon}}^{X})^{2}(\tau_{1}^{Y}-\bar{\tau}^{Y}+\bar
{\varepsilon}_{1}^{Y}-\bar{\bar{\varepsilon}}^{Y})^{2}]\\
&  =\mathbb{E}\left[  \left(  \frac{B-1}{B}(\tau_{1}^{X}+\bar{\varepsilon}%
_{1}^{X})-\frac{1}{B}\sum_{b=2}^{B}(\tau_{b}^{X}+\bar{\varepsilon}_{b}%
^{X})\right)  ^{2}\left(  \frac{B-1}{B}(\tau_{1}^{Y}+\bar{\varepsilon}_{1}%
^{Y})-\frac{1}{B}\sum_{b=2}^{B}(\tau_{b}^{Y}+\bar{\varepsilon}_{b}%
^{Y})\right)  ^{2}\right]  \\
&  =\mathbb{E}\left[  \left(  \frac{(B-1)^{2}}{B^{2}}(\tau_{1}^{X}%
+\bar{\varepsilon}_{1}^{X})^{2}-\frac{2(B-1)}{B^{2}}(\tau_{1}^{X}%
+\bar{\varepsilon}_{1}^{X})\sum_{b=2}^{B}(\tau_{b}^{X}+\bar{\varepsilon}%
_{b}^{X})+\frac{1}{B^{2}}\left(  \sum_{b=2}^{B}(\tau_{b}^{X}+\bar{\varepsilon
}_{b}^{X})\right)  ^{2}\right)  \right.  \\
&  \left.  \times\left(  \frac{(B-1)^{2}}{B^{2}}(\tau_{1}^{Y}+\bar
{\varepsilon}_{1}^{Y})^{2}-\frac{2(B-1)}{B^{2}}(\tau_{1}^{Y}+\bar{\varepsilon
}_{1}^{Y})\sum_{b=2}^{B}(\tau_{b}^{Y}+\bar{\varepsilon}_{b}^{Y})+\frac
{1}{B^{2}}\left(  \sum_{b=2}^{B}(\tau_{b}^{Y}+\bar{\varepsilon}_{b}%
^{Y})\right)  ^{2}\right)  \right]  \\
&  =\frac{(B-1)^{4}}{B^{4}}\mathbb{E}[(\tau_{1}^{X}+\bar{\varepsilon}_{1}%
^{X})^{2}(\tau_{1}^{Y}+\bar{\varepsilon}_{1}^{Y})^{2}]+\frac{(B-1)^{2}}{B^{2}%
}\mathbb{E}[(\tau_{1}^{X}+\bar{\varepsilon}_{1}^{X})^{2}]\frac{1}{B^{2}}%
\sum_{b=2}^{B}\mathbb{E}[(\tau_{b}^{Y}+\bar{\varepsilon}_{b}^{Y})^{2}]\\
&  +\frac{4(B-1)^{2}}{B^{4}}\mathbb{E}[(\tau_{1}^{X}+\bar{\varepsilon}_{1}%
^{X})(\tau_{1}^{Y}+\bar{\varepsilon}_{1}^{Y})]\sum_{b=2}^{B}\mathbb{E}%
[(\tau_{b}^{X}+\bar{\varepsilon}_{b}^{X})(\tau_{b}^{Y}+\bar{\varepsilon}%
_{b}^{Y})]+\frac{(B-1)^{2}}{B^{2}}\mathbb{E}[(\tau_{1}^{Y}+\bar{\varepsilon
}_{1}^{Y})^{2}]\frac{1}{B^{2}}\sum_{b=2}^{B}\mathbb{E}[(\tau_{b}^{X}%
+\bar{\varepsilon}_{b}^{X})^{2}]\\
&  +\frac{1}{B^{4}}\left(  \sum_{b=2}^{B}\mathbb{E}[(\tau_{b}^{X}%
+\bar{\varepsilon}_{b}^{X})^{2}(\tau_{b}^{Y}+\bar{\varepsilon}_{b}^{Y}%
)^{2}]+\sum_{2\leq b_{1}\neq b_{2}\leq B}\mathbb{E}[(\tau_{b_{1}}^{X}%
+\bar{\varepsilon}_{b_{1}}^{X})^{2}]\mathbb{E}[(\tau_{b_{2}}^{Y}%
+\bar{\varepsilon}_{b_{2}}^{Y})^{2}]\right.  \\
&  \left.  +2\sum_{2\leq b_{1}\neq b_{2}\leq B}\mathbb{E}[(\tau_{b_{1}}%
^{X}+\bar{\varepsilon}_{b_{1}}^{X})(\tau_{b_{1}}^{Y}+\bar{\varepsilon}_{b_{1}%
}^{Y})]\mathbb{E}[(\tau_{b_{2}}^{X}+\bar{\varepsilon}_{b_{2}}^{X})(\tau
_{b_{2}}^{Y}+\bar{\varepsilon}_{b_{2}}^{Y})]\right)  \\
&  =\frac{(B-1)^{4}+(B-1)}{B^{4}}\mathbb{E}[(\tau_{1}^{X}+\bar{\varepsilon
}_{1}^{X})^{2}(\tau_{1}^{Y}+\bar{\varepsilon}_{1}^{Y})^{2}]+\frac
{(B-1)(2B-3)}{B^{3}}\mathbb{E}[(\tau_{1}^{X}+\bar{\varepsilon}_{1}^{X}%
)^{2}]\mathbb{E}[(\tau_{1}^{Y}+\bar{\varepsilon}_{1}^{Y})^{2}]\\
&  +\frac{2(B-1)(2B-3)}{B^{3}}(\mathbb{E}[(\tau_{1}^{X}+\bar{\varepsilon}%
_{1}^{X})(\tau_{1}^{Y}+\bar{\varepsilon}_{1}^{Y})]\mathbb{)}^{2}.
\end{align*}
The second expectation $\mathbb{E}[(\bar{X}_{1}-\bar{\bar{X}})(\bar{Y}%
_{1}-\bar{\bar{Y}})(\bar{X}_{2}-\bar{\bar{X}})(\bar{Y}_{2}-\bar{\bar{Y}})]$
can be simplified as%
\begin{align*}
&  \mathbb{E}[(\bar{X}_{1}-\bar{\bar{X}})(\bar{Y}_{1}-\bar{\bar{Y}})(\bar
{X}_{2}-\bar{\bar{X}})(\bar{Y}_{2}-\bar{\bar{Y}})]\\
&  =\mathbb{E}[(\tau_{1}^{X}-\bar{\tau}^{X}+\bar{\varepsilon}_{1}^{X}%
-\bar{\bar{\varepsilon}}^{X})(\tau_{1}^{Y}-\bar{\tau}^{Y}+\bar{\varepsilon
}_{1}^{Y}-\bar{\bar{\varepsilon}}^{Y})(\tau_{2}^{X}-\bar{\tau}^{X}%
+\bar{\varepsilon}_{2}^{X}-\bar{\bar{\varepsilon}}^{X})(\tau_{2}^{Y}-\bar
{\tau}^{Y}+\bar{\varepsilon}_{2}^{Y}-\bar{\bar{\varepsilon}}^{Y})]\\
&  =\mathbb{E}\left[  \left(  \frac{B-1}{B}(\tau_{1}^{X}+\bar{\varepsilon}%
_{1}^{X})-\frac{1}{B}\sum_{b=2}^{B}(\tau_{b}^{X}+\bar{\varepsilon}_{b}%
^{X})\right)  \left(  \frac{B-1}{B}(\tau_{1}^{Y}+\bar{\varepsilon}_{1}%
^{Y})-\frac{1}{B}\sum_{b=2}^{B}(\tau_{b}^{Y}+\bar{\varepsilon}_{b}%
^{Y})\right)  \right.  \\
&  \left.  \times\left(  \frac{B-1}{B}(\tau_{2}^{X}+\bar{\varepsilon}_{2}%
^{X})-\frac{1}{B}\sum_{b\neq2}(\tau_{b}^{X}+\bar{\varepsilon}_{b}^{X})\right)
\left(  \frac{B-1}{B}(\tau_{2}^{Y}+\bar{\varepsilon}_{2}^{Y})-\frac{1}{B}%
\sum_{b\neq2}(\tau_{b}^{Y}+\bar{\varepsilon}_{b}^{Y})\right)  \right]  \\
&  =\frac{2B-3}{B^{3}}\mathbb{E}[(\tau_{1}^{X}+\bar{\varepsilon}_{1}^{X}%
)^{2}(\tau_{1}^{Y}+\bar{\varepsilon}_{1}^{Y})^{2}]+\frac{B^{3}-2B^{2}%
-B+3}{B^{3}}(\mathbb{E}[(\tau_{1}^{X}+\bar{\varepsilon}_{1}^{X})(\tau_{1}%
^{Y}+\bar{\varepsilon}_{1}^{Y})])^{2}\\
&  -\frac{B-3}{B^{3}}\mathbb{E}[(\tau_{1}^{X}+\bar{\varepsilon}_{1}^{X}%
)^{2}]\mathbb{E}[(\tau_{1}^{Y}+\bar{\varepsilon}_{1}^{Y})^{2}]-\frac
{B-3}{B^{3}}(\mathbb{E}[(\tau_{1}^{X}+\bar{\varepsilon}_{1}^{X})(\tau_{1}%
^{Y}+\bar{\varepsilon}_{1}^{Y})])^{2}\\
&  =\frac{2B-3}{B^{3}}\mathbb{E}[(\tau_{1}^{X}+\bar{\varepsilon}_{1}^{X}%
)^{2}(\tau_{1}^{Y}+\bar{\varepsilon}_{1}^{Y})^{2}]+\frac{B^{3}-2B^{2}%
-2B+6}{B^{3}}(\mathbb{E}[(\tau_{1}^{X}+\bar{\varepsilon}_{1}^{X})(\tau_{1}%
^{Y}+\bar{\varepsilon}_{1}^{Y})])^{2}\\
&  -\frac{B-3}{B^{3}}\mathbb{E}[(\tau_{1}^{X}+\bar{\varepsilon}_{1}^{X}%
)^{2}]\mathbb{E}[(\tau_{1}^{Y}+\bar{\varepsilon}_{1}^{Y})^{2}].
\end{align*}
The third expectation $\mathbb{E}[(\bar{X}_{1}-\bar{\bar{X}})(\bar{Y}_{1}%
-\bar{\bar{Y}})]$ can be simplified as%
\begin{align*}
&  \mathbb{E}[(\bar{X}_{1}-\bar{\bar{X}})(\bar{Y}_{1}-\bar{\bar{Y}})]\\
&  =\mathbb{E}[(\tau_{1}^{X}-\bar{\tau}^{X}+\bar{\varepsilon}_{1}^{X}%
-\bar{\bar{\varepsilon}}^{X})(\tau_{1}^{Y}-\bar{\tau}^{Y}+\bar{\varepsilon
}_{1}^{Y}-\bar{\bar{\varepsilon}}^{Y})]\\
&  =\mathbb{E}\left[  \left(  \frac{B-1}{B}(\tau_{1}^{X}+\bar{\varepsilon}%
_{1}^{X})-\frac{1}{B}\sum_{b=2}^{B}(\tau_{b}^{X}+\bar{\varepsilon}_{b}%
^{X})\right)  \left(  \frac{B-1}{B}(\tau_{1}^{Y}+\bar{\varepsilon}_{1}%
^{Y})-\frac{1}{B}\sum_{b=2}^{B}(\tau_{b}^{Y}+\bar{\varepsilon}_{b}%
^{Y})\right)  \right]  \\
&  =\frac{(B-1)^{2}}{B^{2}}\mathbb{E}[(\tau_{1}^{X}+\bar{\varepsilon}_{1}%
^{X})(\tau_{1}^{Y}+\bar{\varepsilon}_{1}^{Y})]+\frac{B-1}{B^{2}}%
\mathbb{E}[(\tau_{1}^{X}+\bar{\varepsilon}_{1}^{X})(\tau_{1}^{Y}%
+\bar{\varepsilon}_{1}^{Y})]\\
&  =\frac{B-1}{B}\mathbb{E}[(\tau_{1}^{X}+\bar{\varepsilon}_{1}^{X})(\tau
_{1}^{Y}+\bar{\varepsilon}_{1}^{Y})].
\end{align*}
Plugging the three formulas back into (\ref{computation1_var}), we obtain%
\begin{align}
&  \mathrm{Var}(\hat{\sigma}_{Cov}^{2})\nonumber\\
&  =\frac{B}{(B-1)^{2}}\mathbb{E}[(\bar{X}_{1}-\bar{\bar{X}})^{2}(\bar{Y}%
_{1}-\bar{\bar{Y}})^{2}]+\frac{B}{B-1}\mathbb{E}[(\bar{X}_{1}-\bar{\bar{X}%
})(\bar{Y}_{1}-\bar{\bar{Y}})(\bar{X}_{2}-\bar{\bar{X}})(\bar{Y}_{2}-\bar
{\bar{Y}})]\nonumber\\
&  -\frac{B^{2}}{(B-1)^{2}}(\mathbb{E}[(\bar{X}_{1}-\bar{\bar{X}})(\bar{Y}%
_{1}-\bar{\bar{Y}})])^{2}\nonumber\\
&  =\frac{1}{B}\mathbb{E}[(\tau_{1}^{X}+\bar{\varepsilon}_{1}^{X})^{2}%
(\tau_{1}^{Y}+\bar{\varepsilon}_{1}^{Y})^{2}]+\frac{1}{B(B-1)}\mathbb{E}%
[(\tau_{1}^{X}+\bar{\varepsilon}_{1}^{X})^{2}]\mathbb{E}[(\tau_{1}^{Y}%
+\bar{\varepsilon}_{1}^{Y})^{2}]\nonumber\\
&  -\frac{B-2}{B(B-1)}(\mathbb{E}[(\tau_{1}^{X}+\bar{\varepsilon}_{1}%
^{X})(\tau_{1}^{Y}+\bar{\varepsilon}_{1}^{Y})])^{2}.\label{computation2_var}%
\end{align}
Now we compute the expectations in (\ref{computation2_var}). First,
$\mathbb{E}[(\tau_{1}^{X}+\bar{\varepsilon}_{1}^{X})^{2}(\tau_{1}^{Y}%
+\bar{\varepsilon}_{1}^{Y})^{2}]$ is given by
\begin{align*}
&  \mathbb{E}[(\tau_{1}^{X}+\bar{\varepsilon}_{1}^{X})^{2}(\tau_{1}^{Y}%
+\bar{\varepsilon}_{1}^{Y})^{2}]\\
&  =\mathbb{E}[((\tau_{1}^{X})^{2}+2\tau_{1}^{X}\bar{\varepsilon}_{1}%
^{X}+(\bar{\varepsilon}_{1}^{X})^{2})((\tau_{1}^{Y})^{2}+2\tau_{1}^{Y}%
\bar{\varepsilon}_{1}^{Y}+(\bar{\varepsilon}_{1}^{Y})^{2})]\\
&  =\mathbb{E}[(\tau_{1}^{X})^{2}(\tau_{1}^{Y})^{2}]+2\mathbb{E}[(\tau_{1}%
^{X})^{2}\tau_{1}^{Y}\bar{\varepsilon}_{1}^{Y}]+\mathbb{E}[(\tau_{1}^{X}%
)^{2}(\bar{\varepsilon}_{1}^{Y})^{2}]+2\mathbb{E}[\tau_{1}^{X}\bar
{\varepsilon}_{1}^{X}(\tau_{1}^{Y})^{2}]\\
&  +4\mathbb{E}[\tau_{1}^{X}\bar{\varepsilon}_{1}^{X}\tau_{1}^{Y}%
\bar{\varepsilon}_{1}^{Y}]+2\mathbb{E}[\tau_{1}^{X}\bar{\varepsilon}_{1}%
^{X}(\bar{\varepsilon}_{1}^{Y})^{2}]+\mathbb{E}[(\bar{\varepsilon}_{1}%
^{X})^{2}(\tau_{1}^{Y})^{2}]+2\mathbb{E}[(\bar{\varepsilon}_{1}^{X})^{2}%
\tau_{1}^{Y}\bar{\varepsilon}_{1}^{Y}]+\mathbb{E}[(\bar{\varepsilon}_{1}%
^{X})^{2}(\bar{\varepsilon}_{1}^{Y})^{2}]\\
&  =\mathbb{E}[(\tau^{X})^{2}(\tau^{Y})^{2}]+0+\frac{1}{R}\mathbb{E}[(\tau
^{X})^{2}(\varepsilon^{Y})^{2}]+0+\frac{4}{R}\mathbb{E}[\tau^{X}\tau
^{Y}\varepsilon^{X}\varepsilon^{Y}]+\frac{2}{R^{2}}\mathbb{E}[\tau
^{X}\varepsilon^{X}(\varepsilon^{Y})^{2}]\\
&  +\frac{1}{R}\mathbb{E}[(\varepsilon^{X})^{2}(\tau^{Y})^{2}]+\frac{2}{R^{2}%
}\mathbb{E}[(\varepsilon^{X})^{2}\tau^{Y}\varepsilon^{Y}]\\
&  +\frac{1}{R^{4}}(R\mathbb{E}[(\varepsilon^{X})^{2}(\varepsilon^{Y}%
)^{2}]+R(R-1)\mathbb{E}[\mathbb{E}[(\varepsilon^{X})^{2}|Z]\mathbb{E}%
[(\varepsilon^{Y})^{2}|Z]]+2R(R-1)\mathbb{E}[(\mathbb{E}[\varepsilon
^{X}\varepsilon^{Y}|Z])^{2}])\\
&  =\mathbb{E}[(\tau^{X})^{2}(\tau^{Y})^{2}]+\frac{1}{R}\mathbb{E}[(\tau
^{X})^{2}(\varepsilon^{Y})^{2}]+\frac{4}{R}\mathbb{E}[\tau^{X}\tau
^{Y}\varepsilon^{X}\varepsilon^{Y}]+\frac{2}{R^{2}}\mathbb{E}[\tau
^{X}\varepsilon^{X}(\varepsilon^{Y})^{2}]\\
&  +\frac{1}{R}\mathbb{E}[(\varepsilon^{X})^{2}(\tau^{Y})^{2}]+\frac{2}{R^{2}%
}\mathbb{E}[(\varepsilon^{X})^{2}\tau^{Y}\varepsilon^{Y}]+\frac{1}{R^{3}%
}\mathbb{E}[(\varepsilon^{X})^{2}(\varepsilon^{Y})^{2}]\\
&  +\frac{R-1}{R^{3}}\mathbb{E}[\mathbb{E}[(\varepsilon^{X})^{2}%
|Z]\mathbb{E}[(\varepsilon^{Y})^{2}|Z]]+\frac{2(R-1)}{R^{3}}\mathbb{E}%
[(\mathbb{E}[\varepsilon^{X}\varepsilon^{Y}|Z])^{2}].
\end{align*}
Next,
\[
\mathbb{E}[(\tau_{1}^{X}+\bar{\varepsilon}_{1}^{X})(\tau_{1}^{Y}%
+\bar{\varepsilon}_{1}^{Y})]=\mathbb{E}[\tau_{1}^{X}\tau_{1}^{Y}%
]+\mathbb{E}[\tau_{1}^{X}\bar{\varepsilon}_{1}^{Y}]+\mathbb{E}[\bar
{\varepsilon}_{1}^{X}\tau_{1}^{Y}]+\mathbb{E}[\bar{\varepsilon}_{1}^{X}%
\bar{\varepsilon}_{1}^{Y}]=\mathbb{E}[\tau^{X}\tau^{Y}]+\frac{1}{R}%
\mathbb{E}[\varepsilon^{X}\varepsilon^{Y}],
\]
and similarly%
\[
\mathbb{E}[(\tau_{1}^{X}+\bar{\varepsilon}_{1}^{X})^{2}]=\mathbb{E}[(\tau
^{X})^{2}]+\frac{1}{R}\mathbb{E}[(\varepsilon^{X})^{2}],
\]%
\[
\mathbb{E}[(\tau_{1}^{Y}+\bar{\varepsilon}_{1}^{Y})^{2}]=\mathbb{E}[(\tau
^{Y})^{2}]+\frac{1}{R}\mathbb{E}[(\varepsilon^{Y})^{2}].
\]
Plugging these expectations into (\ref{computation2_var}), we obtain%
\begin{align*}
&  \mathrm{Var}(\hat{\sigma}_{Cov}^{2})\\
&  =\frac{1}{B}\mathbb{E}[(\tau_{1}^{X}+\bar{\varepsilon}_{1}^{X})^{2}%
(\tau_{1}^{Y}+\bar{\varepsilon}_{1}^{Y})^{2}]+\frac{1}{B(B-1)}\mathbb{E}%
[(\tau_{1}^{X}+\bar{\varepsilon}_{1}^{X})^{2}]\mathbb{E}[(\tau_{1}^{Y}%
+\bar{\varepsilon}_{1}^{Y})^{2}]\\
&  -\frac{B-2}{B(B-1)}(\mathbb{E}[(\tau_{1}^{X}+\bar{\varepsilon}_{1}%
^{X})(\tau_{1}^{Y}+\bar{\varepsilon}_{1}^{Y})])^{2}\\
&  =\frac{1}{B}\left(  \mathbb{E}[(\tau^{X})^{2}(\tau^{Y})^{2}]+\frac{1}%
{R}\mathbb{E}[(\tau^{X})^{2}(\varepsilon^{Y})^{2}]+\frac{4}{R}\mathbb{E}%
[\tau^{X}\tau^{Y}\varepsilon^{X}\varepsilon^{Y}]+\frac{2}{R^{2}}%
\mathbb{E}[\tau^{X}\varepsilon^{X}(\varepsilon^{Y})^{2}]+\frac{1}{R}%
\mathbb{E}[(\varepsilon^{X})^{2}(\tau^{Y})^{2}]\right.  \\
&  \left.  +\frac{2}{R^{2}}\mathbb{E}[(\varepsilon^{X})^{2}\tau^{Y}%
\varepsilon^{Y}]+\frac{1}{R^{3}}\mathbb{E}[(\varepsilon^{X})^{2}%
(\varepsilon^{Y})^{2}]+\frac{R-1}{R^{3}}\mathbb{E}[\mathbb{E}[(\varepsilon
^{X})^{2}|Z]\mathbb{E}[(\varepsilon^{Y})^{2}|Z]]+\frac{2(R-1)}{R^{3}%
}\mathbb{E}[(\mathbb{E}[\varepsilon^{X}\varepsilon^{Y}|Z])^{2}]\right)  \\
&  +\frac{1}{B(B-1)}\left(  \mathbb{E}[(\tau^{X})^{2}]+\frac{1}{R}%
\mathbb{E}[(\varepsilon^{X})^{2}]\right)  \left(  \mathbb{E}[(\tau^{Y}%
)^{2}]+\frac{1}{R}\mathbb{E}[(\varepsilon^{Y})^{2}]\right)  \\
&  -\frac{B-2}{B(B-1)}\left(  \mathbb{E}[\tau^{X}\tau^{Y}]+\frac{1}%
{R}\mathbb{E}[\varepsilon^{X}\varepsilon^{Y}]\right)  ^{2}\\
= &  \frac{1}{B}\mathbb{E}[(\tau^{X})^{2}(\tau^{Y})^{2}]+\frac{1}%
{BR}\mathbb{E}[(\tau^{X})^{2}(\varepsilon^{Y})^{2}]+\frac{4}{BR}%
\mathbb{E}[\tau^{X}\tau^{Y}\varepsilon^{X}\varepsilon^{Y}]+\frac{2}{BR^{2}%
}\mathbb{E}[\tau^{X}\varepsilon^{X}(\varepsilon^{Y})^{2}]\\
&  +\frac{1}{BR}\mathbb{E}[(\varepsilon^{X})^{2}(\tau^{Y})^{2}]+\frac
{2}{BR^{2}}\mathbb{E}[(\varepsilon^{X})^{2}\tau^{Y}\varepsilon^{Y}]+\frac
{1}{BR^{3}}\mathbb{E}[(\varepsilon^{X})^{2}(\varepsilon^{Y})^{2}]+\frac
{R-1}{BR^{3}}\mathbb{E}[\mathbb{E}[(\varepsilon^{X})^{2}|Z]\mathbb{E}%
[(\varepsilon^{Y})^{2}|Z]]\\
&  +\frac{2(R-1)}{BR^{3}}\mathbb{E}[(\mathbb{E}[\varepsilon^{X}\varepsilon
^{Y}|Z])^{2}]+\frac{1}{B(B-1)}\mathbb{E}[(\tau^{X})^{2}]\mathbb{E}[(\tau
^{Y})^{2}]+\frac{1}{B(B-1)R}\mathbb{E}[(\tau^{X})^{2}]\mathbb{E}%
[(\varepsilon^{Y})^{2}]\\
&  +\frac{1}{B(B-1)R}\mathbb{E}[(\varepsilon^{X})^{2}]\mathbb{E}[(\tau
^{Y})^{2}]+\frac{1}{B(B-1)R^{2}}\mathbb{E}[(\varepsilon^{X})^{2}%
]\mathbb{E}[(\varepsilon^{Y})^{2}]\\
&  -\frac{B-2}{B(B-1)}(\mathbb{E}[\tau^{X}\tau^{Y}])^{2}-\frac{2(B-2)}%
{B(B-1)R}\mathbb{E}[\tau^{X}\tau^{Y}]\mathbb{E}[\varepsilon^{X}\varepsilon
^{Y}]-\frac{B-2}{B(B-1)R^{2}}(\mathbb{E}[\varepsilon^{X}\varepsilon^{Y}])^{2}%
\end{align*}
Notice that%
\begin{align*}
&  -\frac{B-2}{B(B-1)}(\mathbb{E}[\tau^{X}\tau^{Y}])^{2}-\frac{2(B-2)}%
{B(B-1)R}\mathbb{E}[\tau^{X}\tau^{Y}]\mathbb{E}[\varepsilon^{X}\varepsilon
^{Y}]-\frac{B-2}{B(B-1)R^{2}}(\mathbb{E}[\varepsilon^{X}\varepsilon^{Y}%
])^{2}\\
&  =-\frac{(B-2)}{B(B-1)}\left(  \mathbb{E}\left[  \tau^{X}\tau^{Y}+\frac
{1}{R}\varepsilon^{X}\varepsilon^{Y}\right]  \right)  ^{2}\\
&  =-\frac{1}{B}\left(  \mathbb{E}\left[  \tau^{X}\tau^{Y}+\frac{1}%
{R}\varepsilon^{X}\varepsilon^{Y}\right]  \right)  ^{2}+\frac{1}%
{B(B-1)}\left(  \mathbb{E}\left[  \tau^{X}\tau^{Y}+\frac{1}{R}\varepsilon
^{X}\varepsilon^{Y}\right]  \right)  ^{2},
\end{align*}%
\begin{align*}
&  \frac{1}{B}\mathbb{E}[(\tau^{X})^{2}(\tau^{Y})^{2}]+\frac{2}{BR}%
\mathbb{E}[\tau^{X}\tau^{Y}\varepsilon^{X}\varepsilon^{Y}]+\frac{1}{BR^{2}%
}\mathbb{E}[(\mathbb{E}[\varepsilon^{X}\varepsilon^{Y}|Z])^{2}]\\
&  =\frac{1}{B}\mathbb{E}\left[  \left(  \tau^{X}\tau^{Y}+\frac{1}%
{R}\mathbb{E}[\varepsilon^{X}\varepsilon^{Y}|Z]\right)  ^{2}\right]  ,
\end{align*}%
\begin{align*}
&  \frac{1}{BR}\mathbb{E}[(\tau^{X})^{2}(\varepsilon^{Y})^{2}]+\frac{2}%
{BR}\mathbb{E}[\tau^{X}\tau^{Y}\varepsilon^{X}\varepsilon^{Y}]+\frac{2}%
{BR^{2}}\mathbb{E}[\tau^{X}\varepsilon^{X}(\varepsilon^{Y})^{2}]\\
&  +\frac{1}{BR}\mathbb{E}[(\varepsilon^{X})^{2}(\tau^{Y})^{2}]+\frac
{2}{BR^{2}}\mathbb{E}[(\varepsilon^{X})^{2}\tau^{Y}\varepsilon^{Y}]+\frac
{1}{BR^{3}}\mathbb{E}[(\varepsilon^{X})^{2}(\varepsilon^{Y})^{2}]\\
&  =\frac{1}{BR}\mathbb{E}\left[  \left(  \tau^{X}\varepsilon^{Y}%
+\varepsilon^{X}\tau^{Y}+\frac{1}{R}\varepsilon^{X}\varepsilon^{Y}\right)
^{2}\right]  .
\end{align*}
Therefore, $\mathrm{Var}(\hat{\sigma}_{Cov}^{2})$ can be written as%
\begin{align*}
&  \mathrm{Var}(\hat{\sigma}_{Cov}^{2})\\
&  =\frac{1}{B}\left(  \mathbb{E}\left[  \left(  \tau^{X}\tau^{Y}+\frac{1}%
{R}\mathbb{E}[\varepsilon^{X}\varepsilon^{Y}|Z]\right)  ^{2}\right]  -\left(
\mathbb{E}\left[  \tau^{X}\tau^{Y}+\frac{1}{R}\varepsilon^{X}\varepsilon
^{Y}\right]  \right)  ^{2}\right)  \\
&  +\frac{1}{BR}\mathbb{E}\left[  \left(  \tau^{X}\varepsilon^{Y}%
+\varepsilon^{X}\tau^{Y}+\frac{1}{R}\varepsilon^{X}\varepsilon^{Y}\right)
^{2}\right]  +\frac{R-1}{BR^{3}}\mathbb{E}[\mathbb{E}[(\varepsilon^{X}%
)^{2}|Z]\mathbb{E}[(\varepsilon^{Y})^{2}|Z]]\\
&  +\frac{R-2}{BR^{3}}\mathbb{E}[(\mathbb{E}[\varepsilon^{X}\varepsilon
^{Y}|Z])^{2}]+\frac{1}{B(B-1)}\mathbb{E}[(\tau^{X})^{2}]\mathbb{E}[(\tau
^{Y})^{2}]+\frac{1}{B(B-1)R}\mathbb{E}[(\tau^{X})^{2}]\mathbb{E}%
[(\varepsilon^{Y})^{2}]\\
&  +\frac{1}{B(B-1)R}\mathbb{E}[(\varepsilon^{X})^{2}]\mathbb{E}[(\tau
^{Y})^{2}]+\frac{1}{B(B-1)R^{2}}\mathbb{E}[(\varepsilon^{X})^{2}%
]\mathbb{E}[(\varepsilon^{Y})^{2}]+\frac{1}{B(B-1)}\left(  \mathbb{E}\left[
\tau^{X}\tau^{Y}+\frac{1}{R}\varepsilon^{X}\varepsilon^{Y}\right]  \right)
^{2}.
\end{align*}

\end{proof}

\end{appendix}

\end{document}